\documentclass[11pt]{article}
\usepackage[a4paper,margin=3cm]{geometry}
\usepackage{amsmath, amsthm, amssymb}
\usepackage{graphicx}
\usepackage[font = small]{caption}
\usepackage{subcaption}
\usepackage{pgf,tikz}
\usepackage{txfonts}
\usetikzlibrary{arrows}
\usepackage{hyperref}
\usepackage{epstopdf}
\usepackage{enumitem}

\newcommand{\R}{\mathbb{R}}\newcommand{\Z}{\mathbb{Z}}\newcommand{\F}{\mathcal{F}}\newcommand{\C}{\mathbb{C}}\newcommand{\N}{\mathbb{N}}
\newcommand{\tf}{\mathrm{tf}}
\newcommand{\lin}{\mathrm{lin}}
\newcommand{\abs}{\mathrm{abs}}
\newcommand{\ri}{\mathrm{i}}
\newcommand{\El}{\mathcal{L}}

\newcommand{\ord}{\mathcal{O}}

\newcommand{\A}{\mathcal{A}}
\newcommand{\re}{\mathrm{e}}
\newcommand{\de}{\mathrm{d}}
\newcommand{\Es}{\mathcal{S}}
\newcommand{\E}{\mathcal{E}}
\newcommand{\D}{\mathcal{D}}
\newcommand{\dell}{\epsilon_2}

\newcommand{\Y}{\mathcal{Y}}
\newcommand{\W}{\mathcal{W}}
\newcommand{\G}{\mathcal{G}}
\newcommand{\Zet}{\mathcal{Z}}
\newcommand{\U}{\mathcal{U}}
\newcommand{\B}{\mathcal{B}}
\newcommand{\ad}{\mathrm{ad}}
\numberwithin{equation}{section}
\def\epsilon{\varepsilon}

\def\Re{\mathrm{Re}}
\def\Im{\mathrm{Im}}

\newtheorem{theorem}{Theorem}[section]\newtheorem{lemma}[theorem]{Lemma}
\newtheorem{definition}[theorem]{Definition}
\newtheorem{proposition}[theorem]{Proposition}
\newtheorem{corollary}[theorem]{Corollary}

\newtheorem{remark}[theorem]{Remark}
\newenvironment{Acknowledgment}%
 {\begin{trivlist}\item[]\textbf{Acknowledgments.}}{\end{trivlist}}

\begin{document}

\title{Spectral stability of pattern-forming fronts in the complex Ginzburg-Landau equation with a quenching mechanism}
\author{Ryan Goh and Bj\"{o}rn de Rijk}

\newcommand{\Addresses}{{
  \bigskip
  \footnotesize

  R.~Goh, \textsc{Department of Mathematics and Statistics, Boston University, 111 Cummington Mall, Boston,  MA 02215, USA. ~}\nopagebreak
  \textit{E-mail address:} \texttt{rgoh@bu.edu}

  \medskip

  B.~de Rijk, \textsc{Zentrum Mathematik, Technische Universit\"at M\"unchen, Boltzmannstr.~3, 85748 Garching bei M\"unchen, Germany. ~}\nopagebreak
  \textit{E-mail address:} \texttt{bjoern.de-rijk@tum.de}
}}

\maketitle
\begin{abstract}
We consider pattern-forming fronts in the complex Ginzburg-Landau equation with a traveling spatial heterogeneity which destabilizes, or quenches, the trivial ground state while progressing through the domain. We consider the regime where the heterogeneity propagates with speed $c$ just below the linear invasion speed of the pattern-forming front in the associated homogeneous system. In this situation, the front locks to the interface of the heterogeneity leaving a long intermediate state lying near the unstable ground state, possibly allowing for growth of perturbations. This manifests itself in the spectrum of the linearization about the front through the accumulation of eigenvalues onto the absolute spectrum associated with the unstable ground state. As the quench speed $c$ increases towards the linear invasion speed, the absolute spectrum stabilizes with the same rate at which eigenvalues accumulate onto it allowing us to rigorously establish spectrally stability of the front in $L^2(\R)$.

The presence of unstable absolute spectrum poses a technical challenge as spatial eigenvalues along the intermediate state no longer admit a hyperbolic splitting and standard tools such as exponential dichotomies are unavailable. Instead, we projectivize the linear flow, and use Riemann surface unfolding in combination with a superposition principle to study the evolution of subspaces as solutions to the associated matrix Riccati differential equation on the Grassmannian manifold.  Eigenvalues can then be identified as the roots of the meromorphic Riccati-Evans function, and can be located using winding number and parity arguments.


\end{abstract}

\Addresses

\section{Introduction}
In many physical settings spatial heterogeneities and growth processes have been shown to effectively mediate and control the formation of regular spatial patterns. Here, instead of building a periodic structure from small random fluctuations or a localized perturbation of a homogeneous unstable state, a heterogeneity moves through a spatial domain, progressively exciting, ``triggering," or ``quenching" a system into an unstable state, from which patterns can nucleate. By controlling this excitation, patterns can be precisely selected, while suppressing the formation of defects. Examples of such phenomena occur in various natural and experimental settings such as light-sensing reaction-diffusion systems~\cite{konow2019turing,cima,somberg2019growth}, directional solidification of crystals~\cite{zigzageutectic}, ion-beam milling~\cite{bradley}, and phase separative systems~\cite{foard, krekhov,ly2020two}.

Most of the recent mathematical work has focused on bifurcation and existence of pattern-forming fronts in the presence of such a heterogeneity, connecting a stable pattern to the stable ground state ahead of the heterogeneity. In a variety of prototypical partial differential equation models with such spatial heterogeneities~\cite{avery2019growing,beekie,GS14,gs2,gs3,monteiro2017phase,monteiro2020swift}, tools from dynamical systems theory and functional analysis have been employed to establish fronts for various nonlinearities and quenching speeds in one or two spatial dimensions. Broadly speaking, these works focus on how characteristics of the pattern (i.e.~wavenumber, amplitude, and orientation) depend on the speed and structure of the heterogeneity. In comparison, relatively little work has been done to characterize temporal stability of such pattern-forming fronts. Here one wishes to understand how localized or bounded perturbations of the front grow or decay in a suitable norm.

In this work, we consider stability of pattern-forming fronts (see Figure~\ref{fig:prof}) in a prototypical model for pattern-formation: the  complex Ginzburg-Landau (cGL) equation with \emph{supercritical} cubic nonlinearity,
\begin{equation}
A_t = (1+\ri \alpha) A_{xx} + \chi(x - ct) A - (1+\ri \gamma) A|A|^2, \quad \chi(\xi) = -\mathrm{sgn}(\xi),\quad A\in \C, \,\, x,t\in \R,\label{e:cgl0}
\end{equation}
posed in one space dimension, with dispersion parameters $\alpha,\gamma \in \R$, and where $\chi \colon \R \to \R$ is a heterogeneity, traveling with fixed speed $c \in \R$. The spatially homogeneous cGL equation with $\chi \equiv 1$ can be derived and justified as a universal amplitude equation near the onset of an oscillatory instability of a trivial ground state in dissipative systems, such as a reaction-diffusion system or the Couette-Taylor problem, see~\cite{AK02,MIE02,SCH94b,SCHN98}. Pattern formation in such systems is initiated by a localized perturbation of the destabilized ground state. Under such idealized conditions, the localized perturbation leads to a spatial invasion process which leaves a periodic pattern in its wake, whose wavenumber does not depend on the perturbation but only on the system parameters;  see for instance~\cite{DOELEX} and references therein. Thus, through the Ginzburg-Landau approximation, this spatial invasion process corresponds to the existence of an invading front solution to the cGL equation connecting a periodic pattern to the unstable ground state $A = 0$.

In spatially homogeneous models, it cannot be expected that such pattern-forming fronts are stable against bounded or $L^2$-localized perturbations (or in fact against perturbations that are small in any translational invariant norm) as any perturbation ahead of the front will grow exponentially in time due to the instability of the ground state. In contrast, for heterogeneous models which progressively excite a system into an unstable state, the unstable state is only established in the wake of the heterogeneity after which patterns start to nucleate. Consequently, perturbations cannot grow far ahead of the  interface of the pattern-forming front. This begs the question of whether stability of quenched fronts can be rigorously established against perturbations which are small in a translational invariant norm.

In this paper, we make the first step towards answering this question in the affirmative. We prove spectral stability in $L^2(\R)$ of pattern-forming fronts in the cGL equation~\eqref{e:cgl0} with a step-function heterogeneity $\chi(\xi) = -\text{sgn}(\xi)$. That is, we establish that the spectrum of the linearization of~\eqref{e:cgl0} about the front is confined to the open left-half plane, except for a simple (embedded) eigenvalue at zero (due to gauge invariance of the cGL equation) and a parabolic touching of continuous spectrum at the origin (due to the diffusive stability of the periodic pattern).

Under similar spectral conditions, nonlinear stability of source defects in the cGL equation has been obtained in~\cite{BEC}. Thus, we strongly expect that, using a similar approach as in~\cite{BEC}, our spectral results can be employed to prove nonlinear stability of the pattern-forming front as a solution to~\eqref{e:cgl0} against $L^2$-localized perturbations.

\begin{figure}[h!]
\centering
\vspace{-0.0in}
  \hspace{-0.5in}\includegraphics[trim = 0.5in 0.05in 0.5in 0.2in, clip,width = 0.99\textwidth]{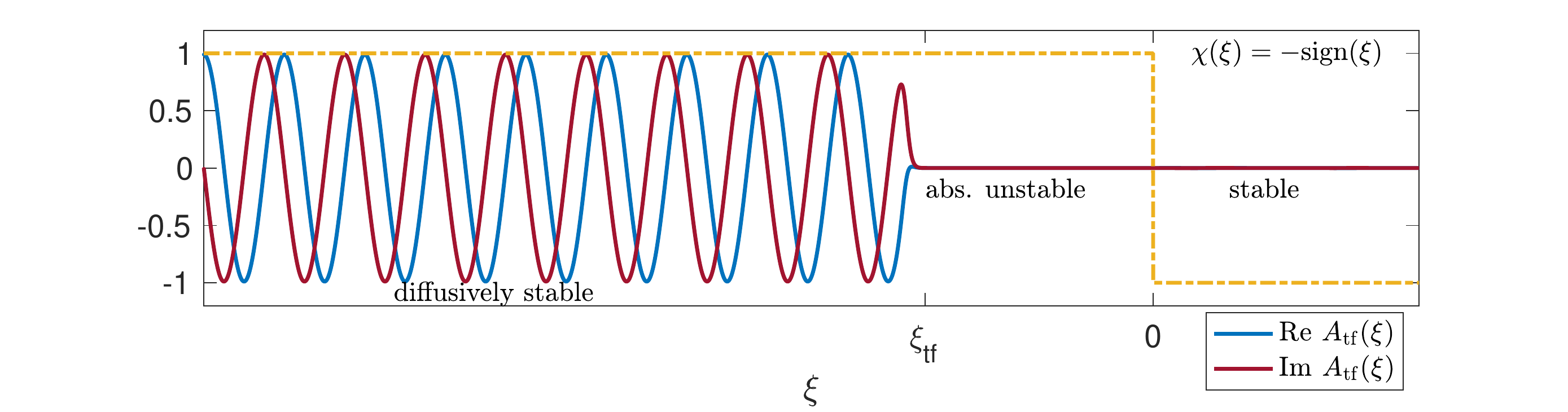}
   \hspace{-0.5in}
   \caption{Plot of the real and imaginary parts of the triggered/quenched front solution $A_\mathrm{tf}(\xi)$ to~\eqref{e:cgl0} in the co-moving frame with speed $c$. Included is a plot of the step-function heterogeneity and the stability of the pure associated homogeneous solutions, which are the diffusively stable periodic state $A_p$ behind the front interface for $\xi\leq \xi_\tf$, the absolutely unstable trivial state $A = 0$ for $\xi\in(\xi_\mathrm{tf},0)$, and the stable homogeneous state $A = 0$ ahead of the quenching interface for $\xi\geq0$, which are ``glued" together to form the front solution.}
   \label{fig:prof}
\end{figure}

We again emphasize that, in contrast to the spatially homogeneous situation, we are able to obtain spectral stability of the invading front in the space $L^2(\R)$, which has a translational invariant norm. In order to properly position our result in a broader perspective, it is useful to explore this dichotomy in more depth and first discuss spectral stability and instability of the base state and of pattern-forming fronts in the \emph{spatially homogeneous} system in~\S\ref{ss:inv} before stating our main result in~\S\ref{ss:main}.

\subsection{Invasion into the unstable trivial state: convective and absolute instability}\label{ss:inv} 

In the spatially homogeneous cGL equation~\eqref{e:cgl0} with $\chi \equiv 1$, a pattern-forming front connects the unstable equilibrium $A = 0$ at $x = +\infty$ to a periodic pattern at $x =-\infty$. The speed at which this front travels through the domain is commonly referred to as the \emph{spreading speed}. It is often the case that the linear information about the rest state can be used to predict properties, including the spreading speed and the spatio-temporal oscillation frequency of the pattern formed in the wake, of such an invasion front resulting from a compact perturbation of the unstable ground state. In this case, since the linear information of the state ahead of the front dictates its invasion properties, such fronts are referred to as \emph{pulled} fronts. 

The linearly selected speed can be heuristically thought of as the co-moving frame speed at which the base state transitions from \emph{convective} to \emph{absolute} instability. In the case of the former, perturbations of the unstable state grow but are convected into the far-field, while in the latter  perturbations grow both in $L^2$-norm and also pointwise. In other words, one can think of this speed as the minimal co-moving frame speed for which perturbations of the constant state do not grow pointwise; see ~\cite{holz,SAN} and~\cite[\S 2]{vS}.

The transition between convective and absolute instability can also be understood in terms of essential and absolute spectrum of the associated linearization $\smash{\mathcal{L}_0:= (1+\ri\alpha)\partial_{\xi}^2 + 1 + c\partial_{\xi}}$, written in a co-moving frame $\xi = x - ct$ and posed on $L^2(\R,\C)$. The essential spectrum $\sigma_\mathrm{ess}(\El_0)$ of this operator is defined as the set of $\lambda\in\C$ for which $\mathcal{L}_0 - \lambda$ is not Fredholm with index 0; see Appendix~\ref{appess1}. To characterize this set,  one inserts $A = \re^{\lambda t + \nu \xi}$ into the linearized equation $A_t = \mathcal{L}_0 A$, obtaining a \emph{linear dispersion relation}
\begin{equation}
0=d(\lambda,\nu;c) = (1 + \ri\alpha)\nu^2 + 1 - \lambda, \qquad\nu,\lambda\in \C. \label{e:ldisp}
\end{equation}
Discontinuous changes in the Fredholm index of $\mathcal{L}_0 - \lambda$ are then found when~\eqref{e:ldisp} has a root $\nu(\lambda)\in \ri\R$, and thus the essential spectrum is given by
$$
 \sigma_{\mathrm{ess}}(\El_0) := \left\{-(1+\ri \alpha) \ell^2 + 1 + c\ri \ell  : k \in \R\right\}.
$$
To understand the behavior of localized perturbations of the base state in the co-moving frame, one can pose $\mathcal{L}_0$ on a weighted $L^2$-space with norm $\smash{\| f\|_{L^2_\beta(\R,\C)}^2:= } \smash{\int |\re^{\beta \xi}f(\xi)|^2 d\xi}$, which penalizes or allows asymptotic growth depending on the weight $\beta\in\R$. In this space, for $\beta>0$, the Fredholm boundaries are shifted so that the essential spectrum is given by
$$
\sigma_{\mathrm{ess},\beta}(\El_0) := \left\{(1+\ri \alpha)(\beta + \ri k)^2 + 1 + c (\beta + \ri k) \colon k \in \R\right\},
$$
penalizing perturbations at $\xi=+\infty$, and allowing growth at $\xi = -\infty$. If there exists a $\beta$ such that $\sigma_{\mathrm{ess},\beta}(\El_0)$ is contained in the open left-half plane, the base state is only convectively unstable.  If there does not exist such a $\beta$, perturbations grow pointwise, and thus the base state is absolutely unstable.

 \paragraph{Absolute spectrum} The transition between both types of instabilities, as $c$ is varied, is mediated by the location of a set in the complex plane known as the \emph{absolute spectrum}, see Appendix~\ref{ss:app_absp}. In our case, the absolute spectrum, $\Sigma_{*,\mathrm{abs}}$, consists of $\lambda$-values for which the linear dispersion relation~\eqref{e:ldisp} has a pair of roots with the same real part. One readily observes that there are always points in the set $\sigma_{\mathrm{ess},\beta}(\El_0)$, which consists of $\lambda$-values for which~\eqref{e:ldisp} possesses a root $\nu(\lambda)$ with $\Re \, \nu(\lambda) = \beta$, lying to the right of $\Sigma_{*,\mathrm{abs}}$, no matter the value of $\beta \in \R$. Thus,  while not technically part of the spectrum, intersections of $\Sigma_{*,\mathrm{abs}}$  with the right half-plane indicate absolute instabilities. We find
\begin{equation}\label{e:abs1}
 \Sigma_{*,\mathrm{abs}}:=\left\{1 - \frac{c^2}{4(1+\mathrm{i \alpha})} - (1+\mathrm{i}\alpha )\ell^2 : \ell\geq0  \right\}.
\end{equation}
In our case, and in many other prototypical equations, the right most part of the absolute spectrum consists of branch points which are ``double-roots" of the dispersion relation~\eqref{e:ldisp}.  That is, they are $(\lambda,\nu)$-pairs which satisfy
\begin{align*}
0 = d(\lambda, \nu;c),\qquad 0 = \partial_\nu d(\lambda,\nu;c).
\end{align*}
Such a root, readily calculated to be
$$
\lambda_{*,\mathrm{br}}(c)= 1-\frac{c^2}{4(1+\ri \alpha)},\qquad \nu_{*,\mathrm{br}}(c) =  \frac{c}{2(1+\ri \alpha)},
$$
thus dictates the absolute instability of the base state, with transitions occurring at speeds $c$ with $\mathrm{Re}\,\lambda_{*,\mathrm{br}}(c) = 0$.

\paragraph{Nonlinear invasion}
In our setting, with a supercritical cubic nonlinearity, the above linear information indeed characterizes the nonlinear invading front connecting a periodic pattern to the unstable base state. That is, the front invades with the \emph{linear spreading speed} $c_\mathrm{lin} = \sup\{ c\,:\, \mathrm{Re}\,\lambda_{*,\mathrm{br}}(c) >0\}$, has temporal oscillation frequency given by $\omega_\mathrm{lin} = \mathrm{Im} \, \lambda_{*,\mathrm{br}}(c_\mathrm{lin})$ and has leading-order spatial decay rate  $\nu_\mathrm{lin} := \nu_{*,\mathrm{br}}(c_\mathrm{lin})$. In our specific case, one calculates~\cite{GS14}
$$
c_\mathrm{lin} = 2\sqrt{1+\alpha^2},\qquad \omega_\mathrm{lin} =\alpha  ,\qquad \nu_\mathrm{lin} =  \frac{1-\ri\alpha}{\sqrt{1+\alpha^2}}.
$$

Periodic patterns in the homogeneous nonlinear equation~\eqref{e:cgl0} are relative equilibria with respect to the gauge action $A\mapsto \re^{s\ri}A$, $s\in \R$, and thus take the form $\sqrt{1-k^2} \re^{\ri (\omega t - kx)}$, with a \emph{nonlinear dispersion relation}, in the co-moving frame of speed $c$, of the form
\begin{equation}
\omega(k) = -\alpha k^2 - \gamma \left(1-k^2\right) - ck.\label{e:disp}
\end{equation}
Through this relation, the linear prediction for the temporal oscillation frequency $\omega_\mathrm{lin}$ then gives the selected spatial wavenumber of the pattern formed:
$$
k_\mathrm{lin} = -\frac{\gamma+\alpha}{\sqrt{1+\alpha^2}+\sqrt{1+\gamma^2}}.
$$
Existence of pattern-forming fronts in certain parameter regimes for the homogeneous cGL equation are considered, for example, in~\cite{van1992fronts} where a phase-amplitude decomposition was used to construct PDE fronts as traveling-waves in a real three-dimensional ODE.

\paragraph{Stability of freely invading front}
After this front in the homogeneous cGL equation~\eqref{e:cgl0} with $\chi \equiv 1$ has been established, one expects perturbations behind the front interface to decay diffusively as long as the periodic pattern at $x=-\infty$ is stable with spectrum lying in the closed left-half plane only touching the imaginary axis in a quadratic tangency at the origin, see~\S\ref{sec:ess} for the spectral calculations confirming this in our case. If the front were to propagate fast enough, that is with speed $c > c_{\text{lin}}$, exponentially localized perturbations are convected into the periodic bulk of the front after which, if remaining small, they decay diffusively. This idea, which was first proposed by Sattinger~\cite{SAT77}, can be used to prove stability of ``fast'' pattern-forming fronts in exponentially weighted spaces~\cite{COEC87,ECSCH00,ESCH02}. Thus, upon introducing an exponential weight, the unstable spectrum of the state ahead of the front can be stabilized.

If the front invades with the linear spreading speed $c_{\lin}$, established above, we are right on the boundary where spectrum could be stabilized with an exponential weight. In this case the stability argument is more subtle, since the linearization has, after introducing the exponential weight, spectrum up to the imaginary axis. Although stability analyses in this regime have been carried out for invading fronts connecting a stable state to an unstable homogeneous rest state~\cite{FH19a,GAL94,KIR92}, the authors are not aware of any results for \emph{pattern-forming} fronts propagating with the spreading speed $c_\lin$.

\subsection{Main result} \label{ss:main}

\paragraph{Previous existence result}
Our main result concerns the spectral stability of pattern-forming fronts in the cGL equation~\eqref{e:cgl0} with a step-function heterogeneity $\chi(\xi) = -\text{sgn}(\xi)$. The fronts take the form $A(x,t) = \re^{\ri\omega t} A_\mathrm{tf}(x-ct)$, where $\omega\in \R$ gives the temporal frequency, and $A_\mathrm{tf}(\xi)$ is a function of the co-moving frame variable $\xi = x - ct$, which satisfies the traveling-wave ODE
\begin{align}
0 = (1+\ri\alpha) A_{\xi\xi} + c A_\xi + (\chi(\xi) - \ri \omega) A - (1+\ri \gamma) A|A|^2, \label{e:TW}
\end{align}
and connects a periodic pattern to the trivial state. That is, it has the asymptotics
\begin{align}
\lim_{\xi\rightarrow-\infty} |A_\mathrm{tf}(\xi) - A_p(\xi)| = 0, \qquad  \lim_{\xi\rightarrow\infty} A_\tf(\xi)  = 0, \label{asymptrigger}
\end{align}
where $ A_p(\xi) := \sqrt{1-k^2}\re^{-\ri k \xi}$
is a periodic solution to~\eqref{e:TW} for $\chi \equiv 1$. The wavenumber $k\in (-1,1)$ relates to the frequency $\omega$ through the nonlinear dispersion relation~\eqref{e:disp}.

The gauge action in the traveling-wave ODE can be factored out by writing~\eqref{e:TW} as the first-order system in the variables $\rho = |A|$ and $z = A_\xi/A$,
\begin{align}
\begin{split}
z_\xi &= -\frac{1}{1+\ri\alpha}\left(c z + \chi(\xi) - \ri \omega - (1+\ri \gamma) \rho^2\right) - z^2,\\
\rho_\xi &= \frac{1}{2}\rho\left(z+\overline{z}\right).
\end{split} \label{e:TW3}
\end{align}
 The front solution then arises as a heteroclinic connection between the points $(-\ri k,1-k^2\big)$ and $(z_+,0)$ with
\begin{align*}
z_+ := -\frac{c + \sqrt{c^2 + 4(1+\ri \alpha)(1+\ri \omega)}}{2(1+\ri \alpha)},
\end{align*}
which are equilibria of~\eqref{e:TW3} for $\chi \equiv 1$ and $\chi \equiv -1$, respectively. Thus, in addition to~\eqref{asymptrigger}, the asymptotic behavior of $A_\tf$ is characterized by
\begin{align}
\lim_{\xi \to \infty} \frac{A_\tf'(\xi)}{A_\tf(\xi)} = z_+. \label{asymptrigger2}
\end{align}

The previous work~\cite{GS14} rigorously established existence of such fronts and determined expansions for frequency/wavenumber selection curves $(\omega_{\tf},k_{\tf})(c)$ for speeds $0\ll c < c_\mathrm{lin}$. In this regime, just below the linear invasion speed $c_\mathrm{lin}$, the pattern-forming instability wants to invade the domain faster than the speed of the inhomogeneity causing the front to ``lock'' to the quenching point at $\xi = 0$, see Figure~\ref{fig:prof}. In this situation, it was found that the leading-order dependence of the wavenumber on the quenching speed $c$ is determined by the intersection of the absolute spectrum with the imaginary axis which, using~\eqref{e:abs1}, is found to be
\begin{equation*}
\Sigma_{*,\mathrm{abs}}\cap \mathrm{i}\R = \{\mathrm{i} \omega_\mathrm{abs}\}, \qquad
\omega_\mathrm{abs} = -\alpha + \frac{\alpha c^2}{2(1+\alpha^2)}. 
\end{equation*}
Technically, the front is the outcome of a heteroclinic bifurcation analysis in~\eqref{e:TW3}, which employs geometric desingularization and invariant foliations to describe the unfolding in the parameters $(c,\omega)$ at $(c_\mathrm{lin},\omega_\mathrm{lin})$ of the equilibria $(z_{0,\pm},0)$ in~\eqref{e:TW3} for $\chi \equiv 1$ given by
\begin{align*}
z_{0,\pm} := -\frac{c \pm \sqrt{c^2 - 4(1+\ri \alpha)(1-\ri \omega)}}{2(1+\ri \alpha)}.
\end{align*}
We summarize this result in the following statement. The corresponding trajectory, wavenumber, and front interface location curves are depicted in Figure~\ref{fig:front_info}.

\begin{theorem}[\cite{GS14}] \label{t:ex0}
Let $\delta > 0$ and $\alpha, \gamma \in \R$ be fixed such that $|\gamma - \alpha|$ and $\delta$ are sufficiently small. Then, provided $0<\Delta c:=c_\mathrm{lin}-c\ll1$, there exists a pattern-forming front solution $\re^{\ri \omega_{\tf} t} A_\mathrm{tf}(x-ct)$ to~\eqref{e:cgl0} with frequency $\omega_\tf$, where $A_\tf(\xi)$ is a heteroclinic solution to~\eqref{e:TW} whose interface
\begin{align*}
\xi_\tf := \inf\big\{\xi : |A_\mathrm{tf}(y)| < \delta \text{ for } y > \xi \big\},
\end{align*}
is located to the left of the jump heterogeneity at $\xi = 0$ at which $A_\tf(\xi)$ is continuously differentiable. The asymptotic behavior of $A_\tf(\xi)$ is described by~\eqref{asymptrigger} and~\eqref{asymptrigger2} with wavenumber $k = k_\tf$ and frequency $\omega = \omega_\tf$, which are related through the nonlinear dispersion relation~\eqref{e:disp}. Finally, the frequency $\omega_\tf$, the wavenumber $k_\tf$ and the rescaled front interface $\smash{\sqrt{\Delta c} \, \xi_\tf}$ are smooth functions of $\smash{\sqrt{\Delta c}}$ and we have the expansions
\begin{align*}
\omega_\tf = \omega_\abs + \ord\left(\left(\Delta c\right)^{3/2}\right), \qquad k_\tf = k_{\lin} + \ord\left(\Delta c\right), \qquad \sqrt{\Delta c}\,\xi_\tf &= \pi\sqrt{1 + \alpha^2} + \ord\left(\sqrt{\Delta c}\right).
\end{align*} 
\end{theorem}
\begin{figure}[h!]
\centering
\includegraphics[trim=0.3in 4in 0 4in, clip,width = 0.525\textwidth]{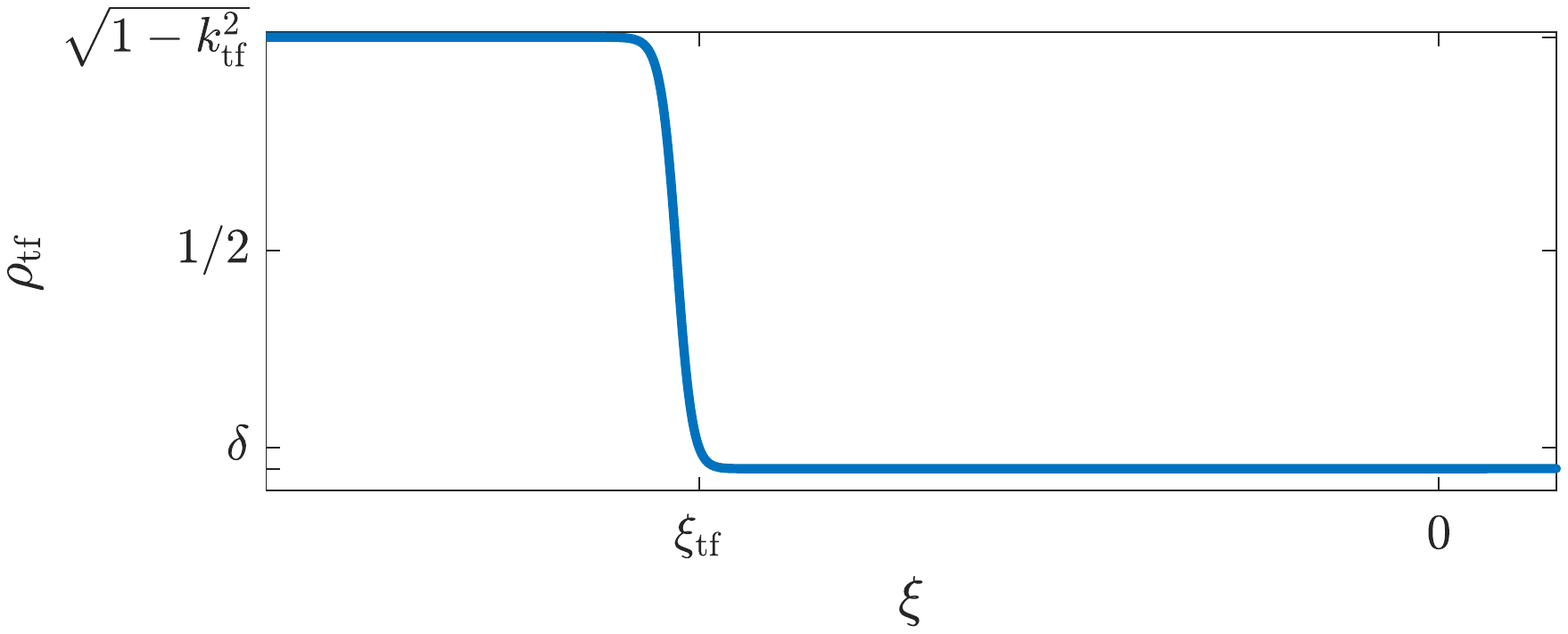}\hspace{-0.5in}
\includegraphics[trim=0 4in 1in 4in, clip,width = 0.5\textwidth]{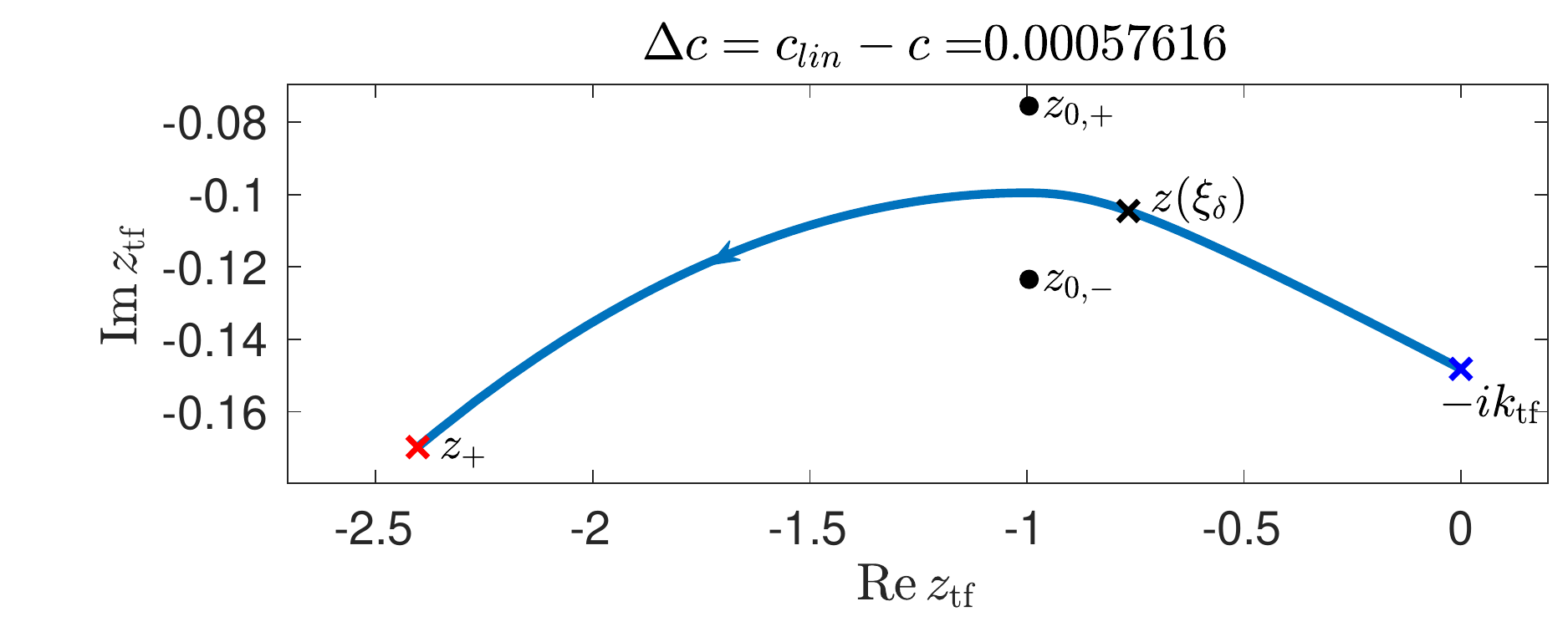}\\
  \includegraphics[trim=0in 2.5in 1in 2.75in, clip,width = 0.5\textwidth]{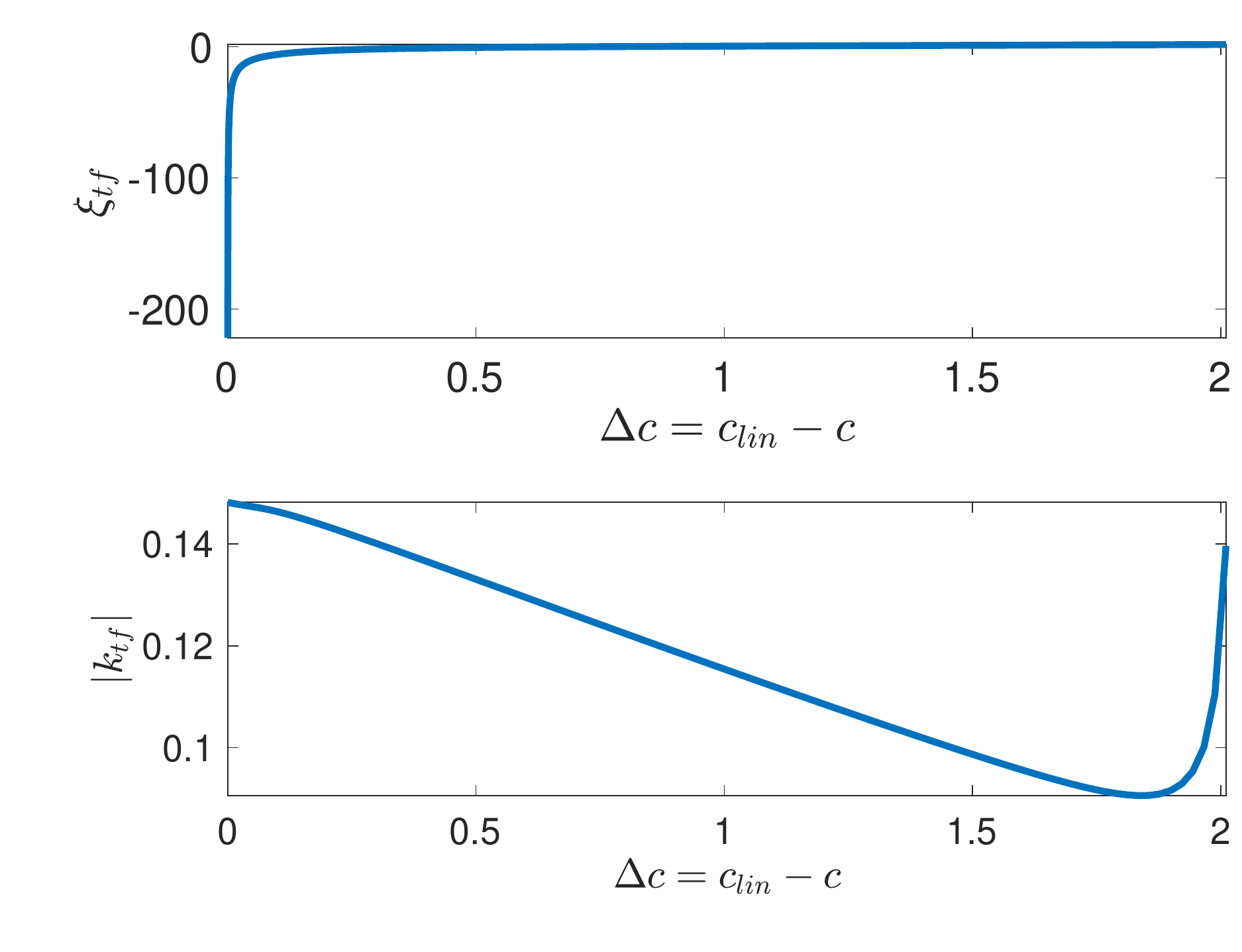}\hspace{-0.1in}
  \includegraphics[trim=0 2.5in 1in 2.75in, clip,width = 0.5\textwidth]{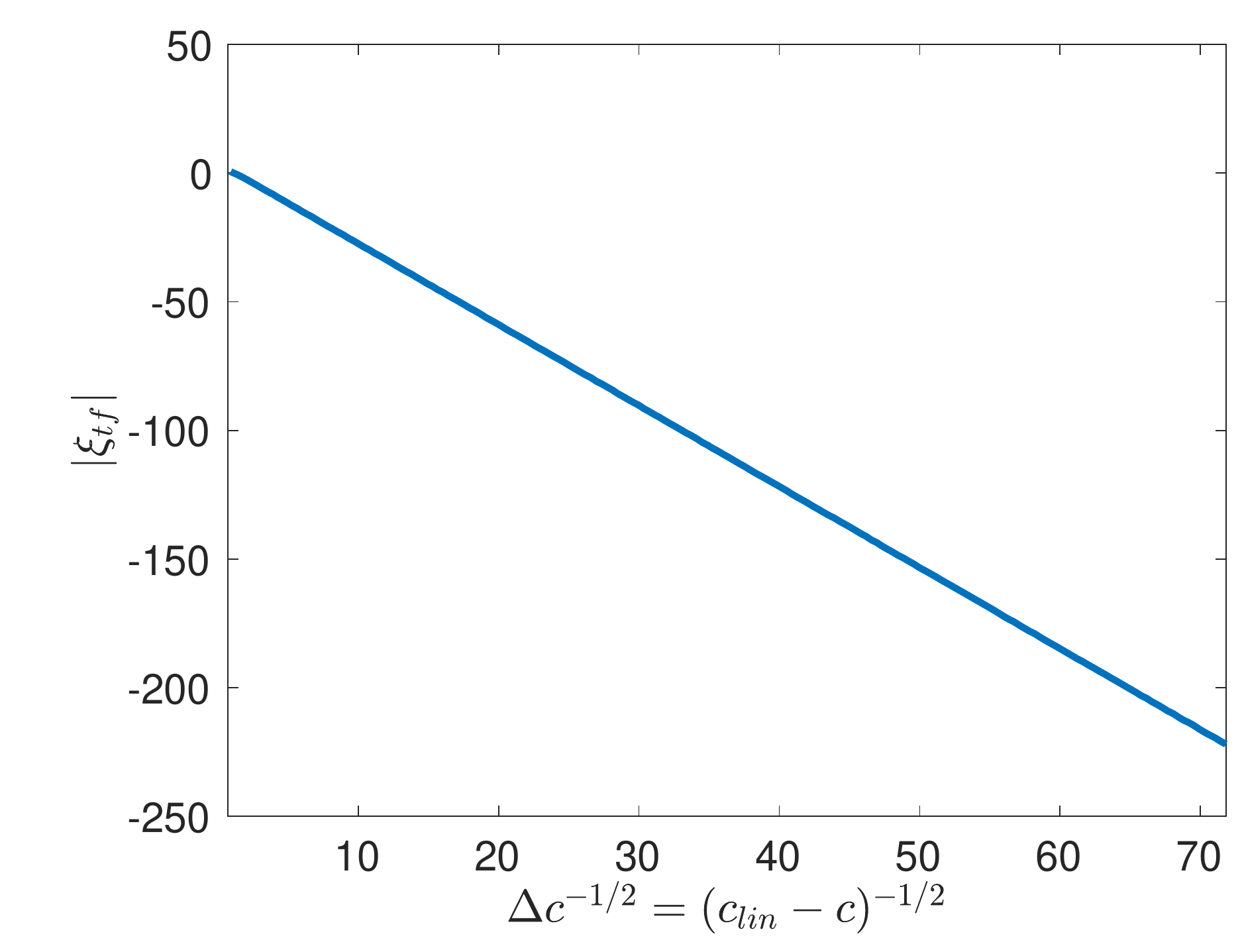}

   \caption{All for $\alpha = -0.1, \gamma = -0.2, c = 2.0094$ so $\Delta c = 0.00058$; (top left): Plot of the $\rho$-component of the solution trajectory; (top right): $z$-component of the solution trajectory, $z_{0,\pm}$ denote the equilibria of~\eqref{e:TW3} with $\chi \equiv 1$, $z_+$ denotes the $z$-unstable equilibrium of~\eqref{e:TW3} with $\chi \equiv -1$; (center left) plot of front interface location  dependence on quenching speed $c$; (bottom right): scaled plot depicting the leading-order $\mathcal{O}(|\Delta c|^{-1/2})$-behavior of $\xi_\mathrm{tf}$ for $\Delta c$ small; (bottom left) dependence of asymptotic wavenumber $k_\tf$ on $c$. Results obtained by continuing~\eqref{e:TW3} as a boundary value problem in parameters $(c,\omega)$ using AUTO07p (see \S~\ref{disc} and~\cite{GS14} for more detail.)} 
   \label{fig:front_info}
\end{figure}

This result shows that the quenching speed $c$ selects the pattern wavenumber $k_\tf$ and the location of the front interface $\xi_\mathrm{tf}$. In particular, the front interface moves away from $\xi = 0$ as $c \nearrow c_\mathrm{lin}$, leaving a long intermediate, or  ``plateau", region where its amplitude $|A|$ is small, see Figure~\ref{fig:front_info}. Thus, for $\xi\in[\xi_\mathrm{tf},0]$, the profile is close to the equilibrium state $A = 0$, which is unstable as a solution to~\eqref{e:cgl0} with $\chi \equiv 1$. Hence, one might intuitively expect quenched fronts to be unstable in this regime, with unstable modes arising in this plateau region. But, as $c\nearrow c_\mathrm{lin}$ and the plateau region expands, the corresponding base state $A=0$ becomes less unstable, in the sense that the absolute spectrum $\Sigma_{*,\mathrm{abs}}$, defined in~\eqref{e:abs1} above, moves to the left, intersecting less and less of the right-half plane.

\paragraph{Perturbation setup and stability result}
This begs the question of whether modes arising from the plateau state are actually unstable. This subtle mechanism, described heuristically in the next subsection, underpins our spectral stability analysis and makes the upcoming stability result somewhat unexpected, especially given that, in the spatially homogeneous setting as discussed in~\S\ref{ss:inv}, unstable absolute spectrum of the base state always yields unstable (essential) spectrum of the pattern-forming front (no matter the chosen exponential weight).

We introduce the necessary concepts to state our main result, which concerns the spectral stability of the pattern-forming front $\re^{\ri \omega_{\tf} t} A_\mathrm{tf}(x-ct)$ as a solution to~\eqref{e:cgl0}, or equivalently, of $A_\tf(\xi)$ as a stationary solution to
\begin{align}
A_t = (1+\ri\alpha) A_{\xi\xi} + c A_\xi + (\chi(\xi) - \ri \omega_\tf) A - (1+\ri \gamma) A|A|^2. \label{e:TW2}
\end{align}
Note that the temporal detuning by $\ri \omega_{\tf}$ moves the absolute spectrum~\eqref{e:abs1} of the base state via the shift $\lambda\mapsto\lambda - \ri\omega_{\tf}$. Thus, the absolute spectrum of the base state $A \equiv 0$ as a solution to~\eqref{e:TW2} with $\chi \equiv 1$ is now given by
\begin{align} \label{e:abs0}
\Sigma_{0,\mathrm{abs}}:=\left\{1 - \frac{c^2}{4(1+\mathrm{i \alpha})} - \ri \omega_\tf - (1+\mathrm{i}\alpha )\ell^2 : \ell\geq0  \right\},
\end{align}
with associated branch point $\lambda_{\mathrm{br}}(c) = 1 - \frac{c^2}{4(1+\ri \alpha)} - \ri \omega_\tf$.

We exploit the gauge invariance present in the cGL equation by decomposing in polar coordinates,
\begin{align} A_\tf(\xi) = r_\tf(\xi) \re^{\ri \phi_\tf(\xi)} \label{polar}.\end{align}
Substituting the perturbed solution $A(\xi,t) = A_\tf(\xi) + a(\xi,t)\re^{\ri\phi_\tf(\xi)}$ into~\eqref{e:TW2} yields a nonlinear evolution equation for the complex-valued perturbation $a(\xi,t)$. The stability of $A_\tf(\xi)$ as a solution to~\eqref{e:TW2} can be determined by studying the dynamics of small solutions to this perturbation equation. Therefore, we wish to split the perturbation equation in a linear and purely nonlinear part, which is at least quadratic in $a(\xi,t)$. However, since complex conjugation is not a linear operation, the obtained nonlinear part contains the term $-(1+\ri\gamma)r_\tf(\xi)^2 \overline{a}$, which is not quadratic in $a$. This problem can be resolved by introducing the variable $b = \overline{a}$. Thus, the resulting perturbation equation reads
\begin{align}
\begin{split}
a_t &= L_\tf(a,b) + \mathcal{N}_\tf(b,a),\\
b_t &= \overline{L_\tf\left(\overline{b},\overline{a}\right)} + \overline{\mathcal{N}_\tf\left(\overline{b},\overline{a}\right)},
\end{split} \label{pertubeq}
\end{align}
where $L_\tf$ denotes the asymptotically constant linear operator
\begin{align*}
\begin{split}
L_\tf(a,b) &= (1+\ri\alpha)\left(a_{\xi\xi} + 2\ri \phi_\tf'(\xi) a_\xi + \left(\ri \phi_\tf''(\xi) - \left( \phi_\tf'(\xi)\right)^2\right)a\right)\\
&\qquad + c\left(a_\xi+\ri\phi_\tf'(\xi) a \right) + \chi(\xi) a - \ri \omega_\tf a - (1+\ri \gamma)r_\tf(\xi)^2 \left(2a + b\right),
\end{split}
\end{align*}
and $\mathcal{N}_\tf$ is the nonlinearity
\begin{align*}
\mathcal{N}_\tf(a,b) = -(1+\ri\gamma) \left((r_\tf(\xi)+a)^2(r_\tf(\xi) + b) - r_\tf(\xi)^3 - r_\tf(\xi)^2(2a+b)\right).
\end{align*}
We observe that the nonlinearity in~\eqref{pertubeq} is indeed quadratic in $(a,b)$. 

The main result of this paper is a statement about the spectrum of the linearization $\El_\tf$ of~\eqref{e:TW2} about $A_\tf(\xi)$, which is given by the linear part of the perturbation equation~\eqref{pertubeq}, and reads
\begin{align*} \El_\tf\begin{pmatrix} a \\ b\end{pmatrix} = \begin{pmatrix} L_\tf(a,b) \\ \overline{L_\tf\left(\overline{b},\overline{a}\right)}\end{pmatrix}.\end{align*}
We note that $\El_\tf$ is a linear differential operator on $L^2\big(\R,\C^2\big)$ with domain $H^2\big(\R,\C^2\big)$, but can also be posed on the weighted space $L_\kappa^2\big(\R,\C^2\big)$ with domain $H^2_\kappa\big(\R,\C^2\big)$, where the Sobolev spaces $H^r_\kappa\big(\R,\C^2\big)$ are defined through their norms
$$\|u\|_{H^r_\kappa}^2 = \sum_{i = 0}^r \int_{\R} \re^{-2\kappa |\xi|} |\partial^i_\xi u(\xi)|^2 \mathrm d \xi,$$
for $r \in \N_{0}$ and $\kappa \geq 0$, and we denote $L^2_\kappa\big(\R,\C^2\big) := H^0_\kappa\big(\R,\C^2\big)$. We note that perturbations in $L^2\big(\R,\C^2\big)$ are localized, whereas perturbations in $L^2_\kappa\big(\R,\C^2\big)$ are allowed to grow as $\xi \to \pm \infty$ with exponential rate less than $ \kappa$. Hence, the spaces $L^2_\kappa\big(\R,\C^2\big)$ are different from the exponentially weighted spaces used in the spatially homogeneous setting in~\S\ref{ss:inv} to stabilize the spectrum of the unstable rest state at $+\infty$.

As mentioned before, we require that the periodic end state $A_p$ of the pattern-forming front at $-\infty$ is spectrally stable in $L^2(\R)$ as a solution to~\eqref{e:cgl0} with $\chi \equiv 1$. That is, the spectrum of its linearization is confined to the open left-half  plane except for a parabolic touching at the origin due to translational invariance. Necessary and sufficient conditions for spectral stability of periodic traveling waves (or wave trains) in the cGL equation have been obtained in~\cite{vH94}. In the relevant regime $|\alpha - \gamma| \ll 1$ and $0 < c_\lin - c \ll 1$ of Theorem~\ref{t:ex0}, a sufficient condition for spectral stability is $|\alpha| < \frac{1}{2}\smash{\sqrt{2}}$, whereas spectral instability holds for $|\alpha| > \frac{1}{2}\smash{\sqrt{2}}$.

We are now able to state our main result which is also depicted schematically in Figure~\ref{fig:basespec}.
\begin{theorem} \label{theospecstab} 
Let $\alpha\in\big(\!-\!\frac{1}{2}\sqrt{2}, \frac{1}{2}\sqrt{2}\big)$, fix $\kappa \in \smash{\big(0,\frac{1}{2c_{\lin}}\big)}$, and take the same assumptions as in Theorem~\ref{t:ex0}. Then, the pattern-forming front is spectrally stable as a solution to~\eqref{e:cgl0}, which entails:
\begin{itemize}
 \item[i)] The spectrum of $\El_\tf$, posed on $L^2\big(\R,\C^2\big)$, does not intersect the closed right-half plane, except at the origin as a parabolic curve.
 \item[ii)] Posed on the exponentially weighted space $L^2_\kappa\big(\R,\C^2\big)$, the operator $\El_\tf$ has no spectrum in the closed right-half plane, except for an algebraically simple eigenvalue, which resides at the origin. Furthermore, eigenvalues near the origin lie $\mathcal{O}(|c_\mathrm{lin} - c|)$-close to $\Sigma_{0,\abs} \cup \overline{\Sigma_{0,\abs}}$.
\end{itemize}
\end{theorem}
\begin{figure}[h!]
\centering
  \vspace{-0.2in}\includegraphics[trim=0 0.0in 0 0, clip, width = 0.35\textwidth]{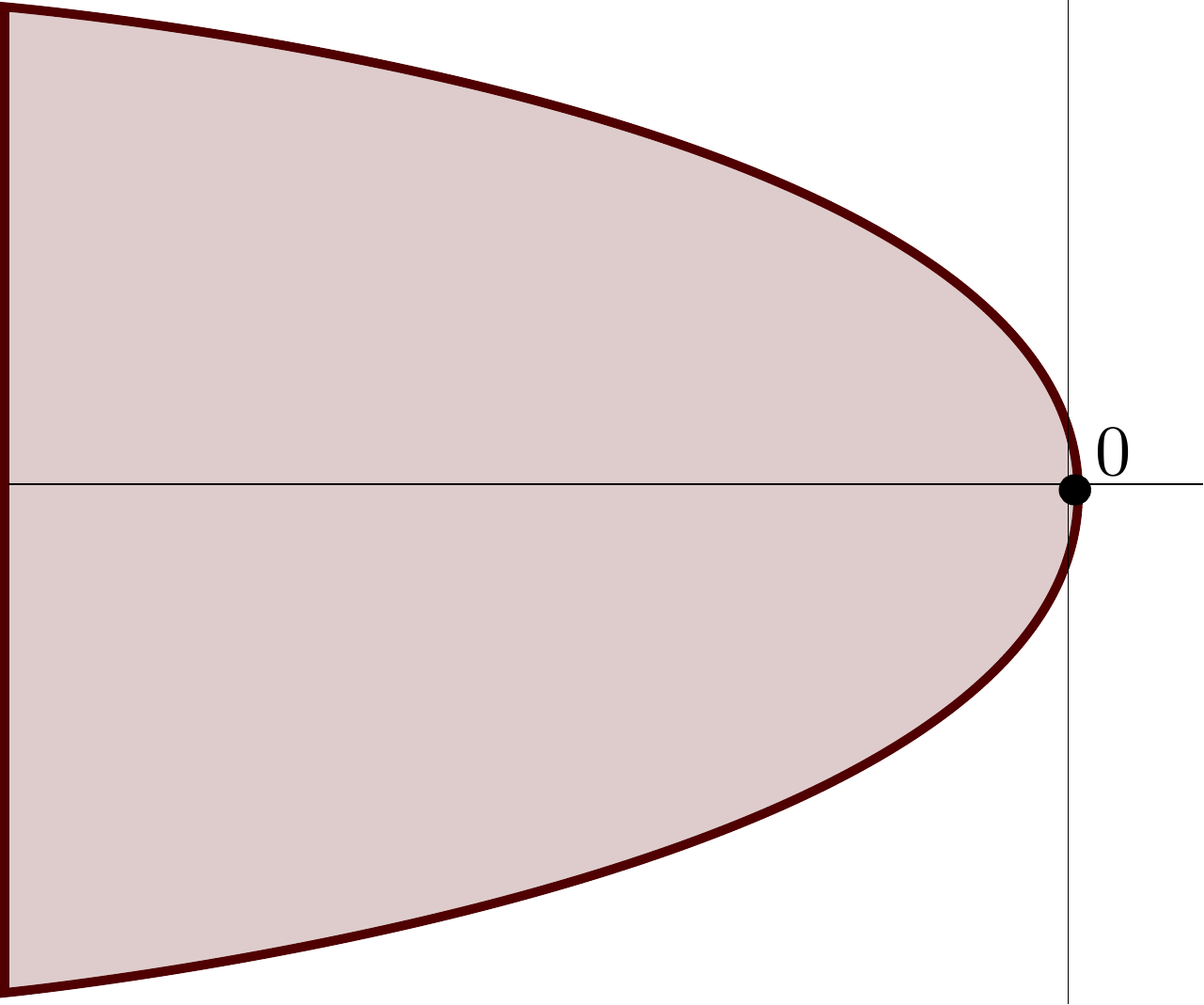}
   \includegraphics[trim=0 0 0 0.0in, clip , width = 0.35\textwidth]{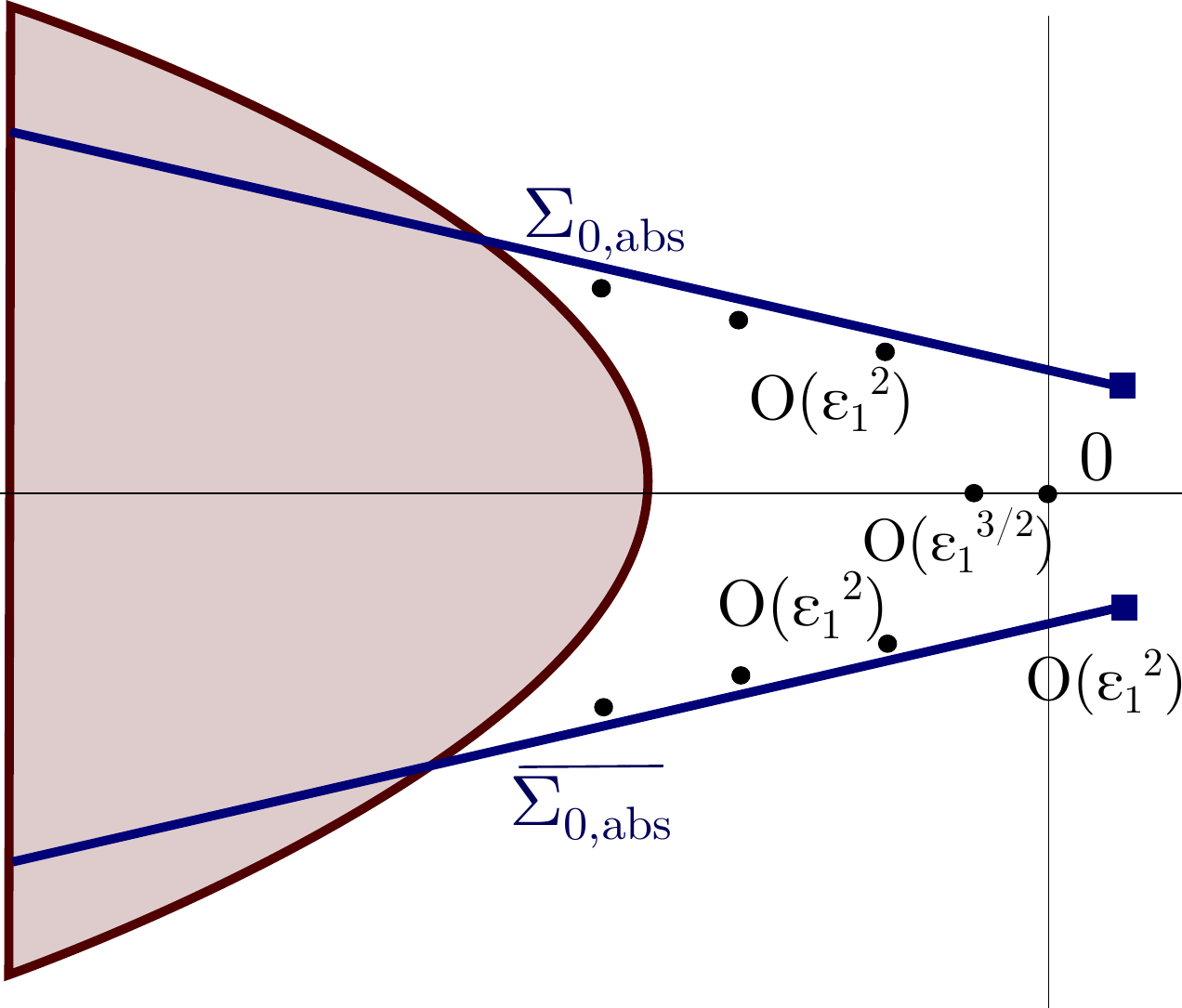}
   \caption{Schematic depiction of the spectrum of the linearization $\El_\tf$ about the pattern-forming front, in both the unweighted space $L^2(\R,\C^2)$ (left) and the weighted space $L^2_{\kappa}(\R,\C^2)$ (right), including essential spectrum (red) and point spectrum (black dots). The right plot also includes the absolute spectrum (blue curves), $\Sigma_{0,\mathrm{abs}}$, and its complex conjugate, of the trivial state $A\equiv0$ in~\eqref{e:TW2} for $\chi\equiv1$, with branch points $\lambda_\mathrm{br}(c),\overline{\lambda_\mathrm{br}}(c)$ denoted as blue squares. The various ``big-oh" notations give the distance of the eigenvalues and the branch points from the origin as $\epsilon_1 \sim \Delta c^{1/2} \searrow0.$} 
   \label{fig:basespec}
\end{figure}

\begin{remark} 
{\rm We emphasize for the parameter range in the above result that the unstable set $\Sigma_{0,\mathrm{abs}} \cup \overline{\Sigma_{0,\mathrm{abs}} }$ is not contained in the absolute spectrum of the linearization about the front, $\mathcal{L}_\tf$. This of course is because the absolute spectrum of the asymptotically constant operator $\mathcal{L}_\tf$ is determined by its end states
, see Appendix~\ref{ss:app_absp}. As the absolute spectrum of each of these states must lie to the left of the essential spectrum of the linear operator, our parameter assumptions imply that they must be contained in the open left-half plane, bounded away from the imaginary axis.}
\end{remark}

The first assertion in Theorem~\ref{theospecstab} yields spectral stability of the pattern-forming front in the translational invariant space $L^2(\R,\C^2)$, whereas the exponential weight in assertion ii) shifts the spectrum associated with the periodic end state at $-\infty$ to the left, and thus reveals the embedded eigenvalue at the origin, see Figure~\ref{fig:basespec}. This simple eigenvalue arises due to gauge invariance of the cGL equation. We expect that the spectral information in Theorem~\ref{theospecstab}, i.e.~assertions i) and ii) combined, is sufficient to prove nonlinear stability of the pattern-forming front as a solution to~\eqref{e:cgl0} via a similar approach as in~\cite{BEC} (cf.~Hypothesis~2.3 in~\cite{BEC}). We do note that, in contrast to the spatially homogeneous setting in~\cite{BEC}, the inhomogeneous cGL equation~\eqref{e:cgl0} is not translational invariant in $\xi$ and, thus, possesses no additional eigenvalue at $0$. We refer to~\S\ref{ss:nonlstab} for further discussion.

To gain intuitive understanding of Theorem~\ref{theospecstab}, and to stress the difference with the spatially homogeneous setting described in~\S\ref{ss:inv}, we next formally outline the expected mechanism for stability of pattern-forming fronts in~\eqref{e:cgl0} with $\chi(\xi) = -\text{sgn}(\xi)$.

\subsection{Heuristic mechanism: stability of a front with absolutely unstable plateau }

To develop understanding of how the spectrum of the linearization $\El_\tf$ about the pattern-forming front solution $A_\tf(\xi)$ to~\eqref{e:TW2} behaves, one can view this solution as a composite, matching, or ``gluing", of three states: the diffusively stable asymptotic periodic pattern $A_p(\xi)$ for $\xi \in( -\infty,\xi_{\tf})$, the stable rest state $A \equiv 0$ for $\xi\in(0,\infty)$, and the unstable plateau state $A \equiv 0$ for $\xi\in (\xi_{\mathrm{tf}},0)$ in between; see Figure~\ref{fig:prof}.
Here, we recall that $A\equiv 0$ is stable as a solution to~\eqref{e:TW2} for $\chi \equiv -1$ and (absolutely) unstable for $\chi \equiv 1$.

Since the two asymptotic states of the front $A_\tf(\xi)$ are stable (so that the essential spectrum of the front is also stable), one only needs to focus on point spectrum arising from the plateau region and from the interfaces between each state. Viewing the composite front as a gluing of two separate fronts, one between the stable rest state and the unstable rest state across the inhomogeneity at $\xi = 0$ and another between the unstable rest state and the stable periodic pattern, the main result of~\cite{SAN3} gives that all but finitely many eigenvalues accumulate onto the absolute spectrum of the plateau state as the length of the plateau region increases, or in other words, as $c\nearrow c_\mathrm{lin}$.  Due to the introduction of the new variable $b = \overline{a}$ in the perturbation equation~\eqref{pertubeq}, the absolute spectrum~\eqref{e:abs0} of the plateau state is now given by $\Sigma_{0,\mathrm{abs}} \cup \smash{\overline{\Sigma_{0,\abs}}}$.

Furthermore, the result in~\cite{SAN3} also implies that, as the plateau width $|\xi_{\tf}| \sim (\Delta c)^{-1/2}$ increases, the point spectrum accumulates with rate $\mathcal{O}(1/|\xi_\mathrm{tf}|^2) = \mathcal{O}(\Delta c)$ onto the branch points, $\lambda_{\mathrm{br}}(c),\overline{\lambda_{\mathrm{br}}(c)}$, which lie at the right most part of the absolute spectrum $\Sigma_{0,\mathrm{abs}} \cup \smash{\overline{\Sigma_{0,\abs}}}$.  At the same time, recall that as $\Delta c\searrow0$ (i.e.~$c\nearrow c_\mathrm{lin}$) these branch points become less unstable, satisfying
$$
\mathrm{Re}\, \lambda_{\mathrm{br}}(c_\mathrm{lin} - \Delta c) = \frac{\Delta c}{\sqrt{1+\alpha^2}} - \frac{\Delta c^2}{4(1+\alpha^2)}. 
$$
In sum, \emph{as $c\nearrow c_\mathrm{lin}$, point spectrum accumulates onto the absolute spectrum with the same rate as the absolute spectrum stabilizes}, indicating that it may be possible for point spectrum to in fact be \emph{stable}. See Figure~\ref{fig:abs_spec} for a schematic depiction of this phenomena, and Figure~\ref{fig:num_spec} below for numerical computations of the spectrum indicating that this is indeed the case. From a phenomenological point of view, one could interpret this potential stability as follows. If the pattern-forming front is locally perturbed in the plateau region,  the absolute instability of the nearby state $A = 0$ indicates that this perturbation should grow with rate $\Delta c$, but if the plateau domain is not long enough,  the perturbation might get convected into the bulk of the front and then diffusively decay away before it can grow and saturate the domain. 
\begin{figure}
\centering
     \hspace{-0.5in}\includegraphics[width = 0.6\textwidth]{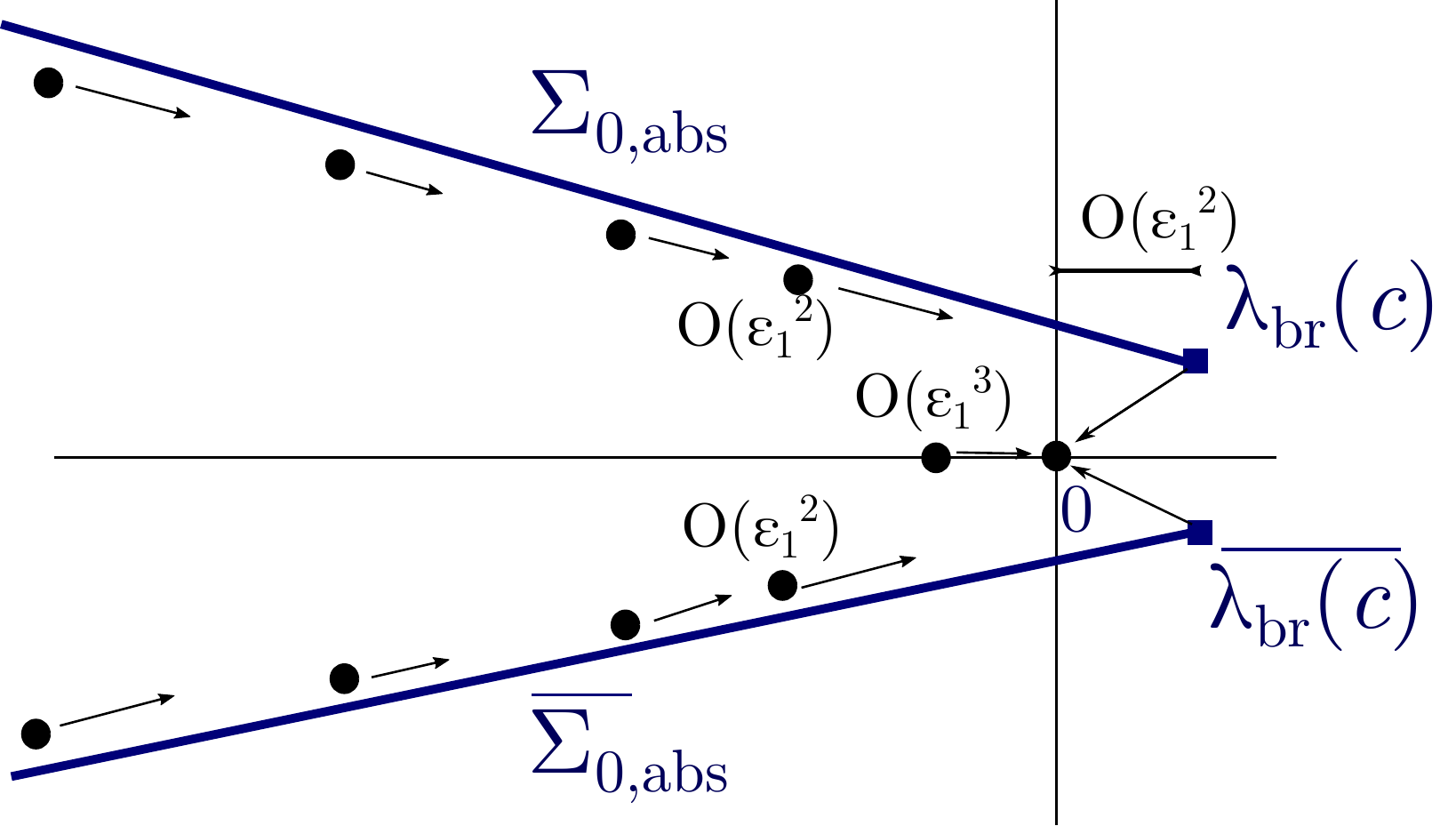}
   \hspace{-0.5in}
   \caption{ The absolute spectrum of the base state (blue lines) and eigenvalues (black dots) which limit onto $\Sigma_{0,\mathrm{abs}} \cup \smash{\overline{\Sigma_{0,\abs}}}$ and accumulate onto the branch-points $\lambda_\mathrm{br}, \overline{\lambda_{\mathrm{br}}}$, as $\epsilon_1\searrow0$. The various ``big-oh" notations give the rate at which the eigenvalues and the branch points of the absolute spectrum converge (as $\epsilon_1\searrow0$). } 
   \label{fig:abs_spec}
\end{figure}

\subsection{Overview of our approach} \label{sec:overview}

We start by reducing complexity in the existence and eigenvalue problems through rescaling and reparameterization, factoring out the gauge invariance in the existence problem and roughly eliminating the dispersion parameter $\alpha$. Since the (transformed) linearization has asymptotically constant coefficients, we can explicitly determine its essential spectrum. This leads us to focus on the point spectrum, posed on the exponentially weighted function space $L^2_\kappa\big(\R,\C^2\big)$ where eigenfunctions may grow exponentially as $\xi \to \pm \infty$ with rate smaller than $\kappa$. By formulating the eigenvalue problem as a first-order system with the unbounded spatial variable $\xi$ taking the role of the evolutionary variable, eigenfunctions can be constructed as intersections of invariant subspaces with certain asymptotic decay properties at $\xi = \pm \infty.$ In previous works studying stability of fronts, pulses, and coherent structures, such intersections were tracked using a complex analytic function of the spectral parameter $\lambda$, known as the \emph{Evans function}~\cite{KAP,SAN2}, whose zeros give eigenvalue locations, including (algebraic) multiplicity.

The main difficulty of this work arises for spectral parameter values $\lambda$ near $\Sigma_{0,\abs} \cup \smash{\overline{\Sigma_{0,\abs}}}$, the absolute spectrum of the plateau state.  In this region, the spatial eigenvalues of the associated linear system in the plateau region, $\xi\in (\xi_{\tf},0)$, lack a uniform spectral gap, precluding one from regaining hyperbolicity with an exponential weight as mentioned above. Thus, to evolve subspaces in this region, one would need a smooth, parameter-dependent change of coordinates (sometimes referred to as an Arnold normal form~\cite{arnold1971matrices,GS15}) to unfold the dynamics in $\lambda$ and $c$. Instead of taking this approach, we take inspiration from the existence problem~\cite{GS14} and projectivize the linear flow, studying the evolution of invariant subspaces as trajectories on the relevant Grassmannian manifold. By performing an analogous ``blow-up" of the linear system, and coordinatizing the Grassmannian with \emph{frame coordinates}, such trajectories are described as solutions to a matrix Riccati differential equation for each coordinate chart of the manifold. Hence one can construct eigenfunctions by finding intersections of corresponding trajectories in the matrix Riccati equation~\cite{beck2015computing,ledoux2010grassmannian,LED,SAN4}. Given a coordinate chart of the Grassmannian manifold, intersections can be located using a meromorphic function of the spectral parameter $\lambda$, known as the \emph{Riccati-Evans function}, whose zeros are in 1-to-1 correspondence to the eigenvalues, including (algebraic) multiplicity, and whose poles indicate $\lambda$-values at which trajectories have left the coordinate chart. Such poles often occur for $\lambda$-values close to the absolute spectrum where the eigenvalue problem exhibits highly oscillatory behavior. The Riccati-Evans function was introduced quite recently (mostly as a numerical tool) to study stability~\cite{HARLEY201536,harley2019instability}, but has not been used in the presence of absolute spectrum.

We split our spectral analysis into parts by dividing the complex plane into three regions, a neighborhood of the origin, where the branch points $\lambda_{\mathrm{br}}(c),\overline{\lambda_{\mathrm{br}}(c)}$ of the absolute spectrum are located and where we find most eigenvalues accumulate as $c \nearrow c_\lin$, an intermediate region bounded away from the absolute spectrum, and a neighborhood of infinity. In the last region, standard scaling arguments preclude the existence of eigenvalues. In the intermediate region, the eigenvalue problem reduces to a coupled Sturm-Liouville problem in the limit $c \nearrow c_\lin$, whose eigenvalues can be bounded away using an $L^2$-energy estimate, after which a perturbative approach allows us to approximate the Riccati-Evans function showing it does not vanish anywhere in the region. The first region, lying near the origin, is the most critical as it is where absolute spectrum lies and where eigenvalues accumulate as $c\nearrow c_\mathrm{lin}$. Here, we use a Riemann surface unfolding in combination with the superposition principle to track the relevant trajectories in the matrix Riccati equation. Using winding number arguments we can enclose eigenvalues in a discrete set of $\mathcal{O}((\Delta c)^{3/2})$-disks $D_j$, where the number $j \in \Z \setminus \{0\}$ can be interpreted as the number of times eigenfunctions wind around the fixed points of the associated Riccati equation, cf.~Remark~\ref{junifrem}. Only the most critical disk, which contains at most two eigenvalues and is centered at the origin, intersects the closed right-half plane. One of these eigenvalues must be 0 due to gauge invariance, whereas the other one must be real and negative by a subtle parity argument involving the derivative of the Riccati-Evans function at $0$.

\paragraph{Contributions}

To our knowledge, this work is the first rigorous result considering the spectral stability of a quenched pattern-forming front. Thereby, it is a first step towards addressing nonlinear stability of such pattern-forming fronts against perturbations which are small in translational invariant norms. Broadly speaking, our result indicates that, as long as pattern-formation is controlled by a spatial inhomogeneity progressively quenching the system in an unstable state, the invasion process is expected to be stable against ``natural'' classes of perturbations.

We expect our analysis to be prototypical in the sense that similar mechanisms (i.e.~the same subtle dance between accumulating point spectrum and stabilizing weakly unstable absolute spectrum) will govern the stability of pattern-forming fronts in other important spatially inhomogeneous models where the quenching is suitably ``fast", such as the Swift-Hohenberg equation, the Cahn-Hilliard equation, and certain reaction-diffusion systems. This is discussed more in~\S\ref{disc}. At the phenomenological level, our result shows the somewhat subtle and unexpected phenomena of a (spectrally) \emph{stable} pattern-forming front with a long plateau state lying near an \emph{absolutely unstable} base state. More generally, it contributes a novel and explanatory example to the recently growing set of works where absolute spectrum plays a role in governing the stability and bifurcation of coherent structures~\cite{CRS,DHM,SS20}. We note that our case is novel as we exhibit a situation where the absolute spectrum of the plateau state is unstable while the spectrum of the linearization about the front is stable. 

On the technical side, we give an example of how a matrix Riccati formulation can be used to rigorously explore subtle behaviors in the stability problem and unfold dynamics for spectral parameters where no hyperbolicity of spatial eigenvalues can be recovered. Our work should also provide useful information and insight into studying more complicated problems where the first-order system formulation of the eigenvalue problem is infinite-dimensional. 

Finally, to briefly comment on the other prototypical type of invasion front which can be controlled by a quenching mechanism, we also provide numerical results for spectral stability and instability of quenched fronts in the cGL equation with a \emph{subcritical} cubic-quintic nonlinearity. In this case the free invasion front in the homogeneous, non-quenched, system is not pulled, but \emph{pushed}, that is the front spreads faster than predicted by the linear information about the unstable base state. Our simulations indicate that stability is not governed by the absolute spectrum of the plateau state, but by a single fold eigenvalue, reminiscent of snaking phenomena; see~\S\ref{ss:push}.

\paragraph{Outline of paper}

First, we introduce the Riccati-Evans function in~\S\ref{sec:ric} as a tool to locate point spectrum of general second-order operators and outline its relation to the standard Evans function. Subsequently, in~\S\ref{sec:rep} we reduce complexity in the existence and eigenvalue problems through rescaling and reparameterization. In~\S\ref{sec:exis} we summarize and slightly extend the existence result of pattern-forming fronts in~\cite{GS14} to suit our needs. In~\S\ref{secsetup} we formulate our spectral stability result in the rescaled and reparameterized coordinates and sketch the set-up of our spectral analysis dividing the complex plane into three regions. The essential spectrum is then determined in~\S\ref{sec:ess}, whereas the analysis of the point spectrum, which is the core of this paper, is the content of~\S\ref{sec:point}-\S\ref{sec:pointR1}. We conclude the proof of our main result, Theorem~\ref{theospecstab}, in~\S\ref{sec:proof}. Finally, in~\S\ref{disc} we discuss potential applications and extensions of our work, as well as provide numerical results for spectral stability of fronts for parameters outside the regime rigorously considered here, and for quenched fronts in the cGL equation with subcritical  nonlinearity. The appendices~\ref{appexpdi} and~\ref{appess} provide some background information on exponential dichotomies, essential spectrum and absolute spectrum.

\section{The Riccati-Evans function} \label{sec:ric}

In this paper, we use the Riccati-Evans function as a tool to locate the critical point spectrum. Below, we construct the Riccati-Evans function for general second-order operators posed on exponentially weighted $L^2$-spaces.

\begin{definition} \label{defspaces} {\upshape
Let $\kappa_\pm \in \R$. For $r \in \N_0$ and $n \in \N$, we define the weighted Sobolev spaces $H^r_{\kappa_-,\kappa_+}(\R,\C^n)$ of $r$-times weakly differentiable functions $\phi \colon \R \to \C^n$ through
the associated norm
\begin{align*} \|\phi\|_{H^r_{\kappa_-,\kappa_+}}^2 = \sum_{j = 0}^r \left(\int_{-\infty}^0 \left\|\partial_x^j \phi(x)\right\|^2 \re^{-2\kappa_- x} \de x + \int_0^\infty \left\|\partial_x^j \phi(x)\right\|^2 \re^{-2\kappa_+ x} \de x\right).\end{align*}
We denote $L^2_{\kappa_-,\kappa_+}(\R,\C^n) = H^0_{\kappa_-,\kappa_+}(\R,\C^n)$ and abbreviate $H^r_{-\kappa,\kappa}(\R,\C^n) = H^r_\kappa(\R,\C^n)$ for $\kappa \in \R$.}
\end{definition}

Let $k_\pm \in \R$ and $n \in \N$. Let $\Omega \subset \C$ be an open and bounded region, where we wish to locate point spectrum of a second-order elliptic operator $\El \colon D(\El) \subset L^2_{\kappa_-,\kappa_+}(\R,\C^n) \to L^2_{\kappa_-,\kappa_+}(\R,\C^n)$, with domain $D(\El) = H^2_{\kappa_-,\kappa_+}(\R,\C^n)$, given by
\begin{align*} \El u = D\partial_{x}^2 u + A_1(x) \partial_x u + A_0(x)u, \end{align*}
with $D \in \C^{n \times n}$ a positive matrix and bounded and continuous coefficients functions $A_0, A_1 \colon \R \to \C^{n\times n}$. The eigenvalue problem $\El u = \lambda u$ can be written as a first-order system
\begin{align}
 \phi_x = A(x,\lambda) \phi, \qquad \phi \in \C^{2n}, \label{linsysan}
\end{align}
with $A \in C(\R \times \Omega, \C^{2n \times 2n})$ such that $A(x,\cdot)$ is analytic on $\Omega$ for each $x \in \R$, and $\partial_\lambda A(\cdot,\lambda)$ is bounded on $\R$ for each $\lambda \in \Omega$. We assume that the weighted eigenvalue problems
\begin{align}
 \phi_x = (A(x,\lambda) - \kappa_\pm) \phi, \qquad \phi \in \C^{2n}, \label{linsysan2}
\end{align}
admit for each $\lambda \in \Omega$ \emph{exponential dichotomies}, see Appendix~\ref{appexpdi}, on $[0,\infty)$ and $(-\infty,0]$ with rank $n$ projections $P_\pm(x,\lambda)$ on $\C^{2n}$, which depend analytically on $\lambda$. Typically, $P_+$ is called the stable projection for $x\rightarrow+\infty$, and $P_-$ the unstable projection for $x\rightarrow-\infty$. We emphasize that such exponential dichotomies exist as long as $\lambda$ lies to the right of the essential spectrum of the operator $\El$, cf.~\cite[Theorem 3.2]{SAN2}, because the Morse indices of the exponential dichotomies must be equal to $n$ for $\lambda$ to the right of the essential spectrum by the second-order structure of $\El$.

Thus, $\lambda \in \Omega$ lies in the point spectrum of $\El$ if and only if there exists a nontrivial solutions $\phi \in \smash{H^1_{\kappa_-,\kappa_+}(\R,\C^{2n})}$ to~\eqref{linsysan}, which is, by continuity and boundedness of the matrix $A$ and by the exponential dichotomies of~\eqref{linsysan2}, equivalent to finding nontrivial $C^1$-solutions $\phi$ to~\eqref{linsysan} with initial condition $\phi(0) \in \ker(P_-(0,\lambda)) \cap P_+(0,\lambda)[\C^{2n}]$. We choose bases
\begin{align*}\Phi_\pm(\pm x,\lambda) = \begin{pmatrix} X_\pm(\pm x,\lambda)\\ Y_\pm(\pm x,\lambda)\end{pmatrix} \in \C^{2n \times n}, \qquad x\geq 0,\lambda \in \Omega,\end{align*}
of the relevant subspaces $\ker(P_-(-x,\lambda))$ and $P_+(x,\lambda)[\C^{2n}]$, respectively, which depend analytically on $\lambda$. Consider the analytic maps $S_\pm \colon \R_\pm \times \Omega \to \C$ given by $S_\pm(\pm x,\lambda) = \det(X_\pm(\pm x,\lambda))$. For $x \geq 0$ and $\lambda \in \Omega$ with $S_\pm(\pm x,\lambda) \neq 0$ the relevant subspaces can be represented by
\begin{align*}T_\pm(\pm x,\lambda) := Y_\pm(\pm x,\lambda)X_\pm(\pm x,\lambda)^{-1} \in \C^{n \times n}.\end{align*}

In fact, for $x \geq 0$, the $n$-dimensional subspaces $\ker(P_-(-x,\lambda))$ and $P_+(x,\lambda)[\C^{2n}]$ in the \emph{Grassmannian manifold} $\mathrm{Gr}(n,\C^{2n})$ are mapped to $T_\pm(\pm x,\lambda)$ under the coordinate chart $\mathfrak{c}$, which maps any $n$-dimensional subspace $W$ of $\C^{2n}$ represented by a basis $\smash{\left(\begin{smallmatrix} X \\ Y \end{smallmatrix}\right) \in \C^{2n \times n}}$ with $\det(X) \neq 0$ to the matrix $YX^{-1} \in \C^{n \times n}$; see ~\cite{ledoux2010grassmannian} and references therein. We emphasize that this mapping is well-defined: if $W$ is also represented by the basis $\smash{\left(\begin{smallmatrix} \tilde{X} \\ \tilde{Y}\end{smallmatrix}\right)} \in \C^{2n \times n}$, there is an invertible matrix $\tau \in \C^{n \times n}$ such that $\smash{\left(\begin{smallmatrix} \tilde{X} \\ \tilde{Y} \end{smallmatrix}\right) = \left(\begin{smallmatrix} X \\ Y \end{smallmatrix}\right)\tau}$ yielding $\det(\tilde{X}) \neq 0$ and $\tilde{Y}\tilde{X}^{-1} = YX^{-1}$.

Writing the coefficient matrix $A(x,\lambda)$ of~\eqref{linsysan} as a block matrix
\begin{align*}
A(x,\lambda) := \begin{pmatrix} A_{11}(x,\lambda) & A_{12}(x,\lambda)\\ A_{21}(x,\lambda) & A_{22}(x,\lambda) \end{pmatrix},
\end{align*}
with $A_{ij}(x,\lambda) \in \C^{n \times n}$, we observe that $T_-(x,\lambda)$ and $T_+(x,\lambda)$ are solutions to the \emph{matrix Riccati equation} ~\cite{levin1959matrix,shayman1986phase} 
\begin{align}
T_x = A_{21}(x,\lambda) + A_{22}(x,\lambda) T - T A_{11}(x,\lambda) - T A_{12}(x,\lambda) T, \qquad T \in \C^{n \times n}, \label{ric}
\end{align}
for $x \in (-\infty,0]$ such that $S_-(x,\lambda) \neq 0$, and $x \in [0,\infty)$ such that $S_+(x,\lambda) \neq 0$, respectively. In other words, the evolution of the $n$-dimensional subspaces $\ker(P_-(-x,\lambda))$ and $P_+(x,\lambda)[\C^{2n}]$ is captured by the Riccati flow~\eqref{ric}, and nontrivial intersections can be located by the following determinantal function.

\begin{definition}
{\upshape
Let $\Es := \{\lambda \in \Omega : S_+(0,\lambda) = 0 \text{ or } S_-(0,\lambda) = 0\}$. We define the \emph{Riccati-Evans function} $\E \colon \Omega \setminus \Es \to \C$ by
\begin{align*} \E(\lambda) = \det(T_+(0,\lambda) - T_-(0,\lambda)).\end{align*}}
\end{definition}

The following result follows immediately by the above construction, see also~\cite{harley2019instability,HARLEY201536}.

\begin{proposition} \label{prop:ric}
The Riccati-Evans function $\E \colon \Omega \setminus \Es \to \C$ enjoys the following properties:
\begin{itemize}
\item[1.] The Riccati-Evans function can be related to the classical Evans function $E \colon \Omega \to \C$ given by $E(\lambda) = \det(\Phi_-(0,\lambda) \mid \Phi_+(0,\lambda))$ through the formula
\begin{align} \E(\lambda) = \frac{E(\lambda)}{S_-(0,\lambda)S_+(0,\lambda)}, \qquad \lambda \in \Omega \setminus \Es. \label{relEvans}\end{align}
\item[2.] $\E$ is meromorphic on $\Omega$.
\item[3.] $\E$ vanishes at some $\lambda_0 \in \Omega \setminus \Es$ if and only if $\lambda_0$ is an eigenvalue of $\El$. Moreover, the multiplicity of $\lambda_0$ as a root of $\E$ agrees with the algebraic multiplicity of $\lambda_0$ as an eigenvalue of $\El$.
\item[4.] The Riccati-Evans function is uniquely determined and does not depend on the choice of bases.
\end{itemize}
\end{proposition}

The advantage of the Riccati-Evans function over the classical Evans function, is that one tracks the flow of subspaces rather than individual solutions to~\eqref{linsysan}. Thus, the Riccati-Evans function does not depend on the choice of bases, which simplifies parity arguments in the spectral analysis significantly. However, one has to bare in mind that the dynamics in the matrix Riccati equation~\eqref{ric} can be highly complicated and solutions might exhibit singularities, which explains the meromorphic character of $\E$. Consequently, winding number arguments with the Riccati-Evans function $\E$ are more involved than with the classical Evans function, since, in addition to the winding number of $\E$, one needs to calculate the winding number of $S_\pm(0,\cdot)$ to compute the number of zeros of $\E$. In our upcoming spectral analysis, we can control the dynamics in a matrix Riccati equation in $\C^{2 \times 2}$ by using superposition principles to perturb from an invariant subset of diagonal solutions.

\subsection{Derivative of the Riccati-Evans function} \label{secderric}

In our spectral analysis we find that the two most critical eigenvalues correspond to two simple real roots of the Riccati-Evans function. One of these roots must reside at the origin due to gauge invariance. The position of the other critical eigenvalue can be determined via a parity argument, which involves the sign of the derivative $\E'(0)$. Below we establish, in the general setting of the previous subsection, an expression for the derivative of the Riccati-Evans function at a simple root.

Let $\lambda_0 \in \Omega \setminus \Es$ be a root of the Riccati-Evans function such that $\ker(P_-(0,\lambda_0)) \cap P_+(0,\lambda_0)[\C^{2n}]$ is one-dimensional, i.e.~such that the geometric multiplicity of $\lambda_0$ as an eigenvalue of the operator $\El$ equals one. Take a nonzero vector $\phi_0 \in \ker(P_-(0,\lambda_0)) \cap P_+(0,\lambda_0)[\C^{2n}]$. Choose bases $\tilde{\Phi}_-$ of $\ker(P_-(0,\lambda_0))$ and $\tilde{\Phi}_+$ of $P_+(0,\lambda_0)[\C^{2n}]$ such that the first column of $\tilde{\Phi}_\pm$ equals $\phi_0$. There exist invertible matrices $\tau_\pm \in \C^{n \times n}$ such that $\tilde{\Phi}_\pm = \Phi_\pm(0,\lambda_0)\tau_\pm$. Thus, $\tilde{\Phi}_\pm(\lambda) := \Phi_\pm(0,\lambda)\tau_\pm$ are also analytic bases of $\ker(P_-(0,\lambda))$ and $P_+(0,\lambda)[\C^{2n}]$, respectively, for $\lambda \in \Omega$.

The subspace $\ker(P_-(0,\lambda_0))^\perp \cap P_+(0,\lambda_0)[\C^{2n}]^\perp$ is one-dimensional and thus spanned by some vector $\psi_0 \in \C^{2n}$. Let $\phi_0(x)$ be the solution to~\eqref{linsysan} at $\lambda = \lambda_0$ with initial condition $\phi_0$ and let $\psi_0(x)$ be the solution to the adjoint problem
\begin{align}
\psi_x = -A(x,\lambda_0)^* \psi, \qquad \psi \in \C^{2n}, \label{adjointan}
\end{align}
with initial condition $\psi_0$. Since the systems~\eqref{linsysan2} at $\lambda = \lambda_0$ admit exponential dichotomies on $[0,\infty)$ and $(-\infty,0]$, respectively, we find that the same holds true for their adjoints
\begin{align}
\psi_x = \left(-A(x,\lambda_0)^* + k_\pm\right) \psi, \qquad \psi \in \C^{2n}, \label{adjointan2}
\end{align}
with associated rank $n$ projections $I_n - P_\pm(x,\lambda_0)^*$ on $\C^{2n}$. Since it holds
\begin{align*}\ker(P_-(0,\lambda_0))^\perp \cap P_+(0,\lambda_0)[\C^{2n}]^\perp = \ker(I_n - P_-(0,\lambda_0)^*) \cap (I_n - P_+(0,\lambda)^*)[\C^{2n}],\end{align*}
$\psi_0(x)$ is a nontrivial solution to~\eqref{adjointan} in $\smash{H^1_{-k_-,-k_+}(\R,\C^{2n})}$, which is unique up to scalar multiples.

The derivative of the classical Evans function $\tilde{E} \colon \Omega \to \C$ given by $\tilde{E}(\lambda) = \det(\tilde{\Phi}_-(\lambda) \mid \tilde{\Phi}_+(\lambda))$ can now be determined via a Lyapunov-Schmidt reduction procedure, which exploits the exponential dichotomies of~\eqref{linsysan2} and~\eqref{adjointan2}. As in~\cite[\S4.2.1]{SAN2}, we find
\begin{align}
 \tilde{E}'(\lambda_0) = \frac{1}{\|\psi_0\|^2} \det\left(\tilde{\Psi}_- \mid \tilde{\Phi}_+\right) \int_\R \psi_0(x)^* \partial_\lambda A(x,\lambda_0) \phi_0(x) \mathrm{d} x, \label{clasderiv1}
\end{align}
where $\tilde{\Psi}_-$ is obtained from $\tilde{\Phi}_-$ by replacing the first column by $\psi_0$. One readily observes via H\"older's inequality that the Melnikov-type integral in the latter is convergent, where we use that $\phi_0 \in L^2_{k_-,k_+}(\R,\C^{2n})$, $\psi_0 \in L^2_{-k_-,-k_+}(\R,\C^{2n})$ and $\partial_\lambda A(\cdot,\lambda_0)$ is bounded on $\R$. Taking derivatives in~\eqref{relEvans} we arrive at the derivative
\begin{align}
\E'(\lambda_0) = \frac{\tilde{E}'(\lambda_0)}{\det(\tilde{X}_-)\det(\tilde{X}_+)}, \label{ricderiv1}
\end{align}
of the Riccati-Evans function, where $\tilde{X}_\pm$ denote the upper $(n \times n)$-blocks of $\tilde{\Phi}_\pm$. Note that we used here that the Riccati-Evans function is independent of the choice of bases.

\section{Reducing complexity through rescaling and reparameterization} \label{sec:rep}

Theorem~\ref{t:ex0} is proved in~\cite{GS14} by reducing complexity in the traveling-wave equation~\eqref{e:TW}. We have already seen that the gauge action in~\eqref{e:TW} can be factored out by introducing the new variables $\rho = |A|$ and $z = A_\xi/A$, yielding the first-order problem~\eqref{e:TW3}. Subsequently, the dispersion parameter $\alpha$ is eliminated in~\cite{GS14} through appropriate rescaling and reparameterization in~\eqref{e:TW3}. In order to simplify the upcoming spectral analysis of the pattern-forming front, we follow the process in~\cite{GS14} and apply a similar rescaling and reparameterization to the linearization $\El_\tf$ of~\eqref{e:TW2} about the front solution $A_\tf(\xi)$.

\subsection{Rescaling and reparameterization in the existence problem} \label{sec:resc}

Completing the square in the $z$-equation in~\eqref{e:TW3} yields
\begin{align*}
z_\xi = -\left(z + \frac{\left(1-\ri \alpha\right)c}{2\left(1+\alpha^2\right)}\right)^2 + \frac{m^2}{1+\alpha^2}\left(\frac{\left(1-\chi(\xi)\right)(1-\ri\alpha)}{m^2} - \epsilon_1^2 + \frac{\epsilon_1^4}{4} + \ri\hat{\omega} + \left(1+\ri\epsilon_2\right)l^{2}\rho^2\right)
\end{align*}
where we denote
\begin{align}
 \begin{split}
m^2 = 1 + \frac{\left(c\alpha\right)^2}{4\left(1+\alpha^2\right)} - \alpha \omega, \qquad l^2 &= \frac{1+\alpha \gamma}{m^2}, \qquad
\epsilon_1 = \sqrt{2 - \frac{c}{m\sqrt{1+\alpha^2}}},\\
\hat{\omega} = \frac{\omega - \omega_{\mathrm{abs}}}{m^2}, &\qquad \epsilon_2 = \frac{\gamma - \alpha}{1+\alpha \gamma}.
\end{split}\label{reparameterization}
\end{align}
Note that it holds $m \geq 1$ in the relevant regime $0 < \Delta c := c_\lin - c \ll 1$ of Theorem~\ref{t:ex0}, where we use the expansion $\omega = \omega_\abs + \ord((\Delta c)^{3/2})$. Hence, the parameter $\epsilon_1 \sim \sqrt{\Delta c}$ is well-defined, and the parameter regime $0 < \Delta c \ll 1$ and $0 \leq |\gamma - \alpha| \ll 1$ in Theorem~\ref{t:ex0} is, after the reparameterization~\eqref{reparameterization}, captured by taking $0 < \epsilon_1 \ll 1$ and $0 \leq |\epsilon_2| \ll 1$. For notational convenience we abbreviate $\epsilon := (\epsilon_1,\epsilon_2)$ throughout the manuscript.

Now we observe that the parameter $\alpha$ can be eliminated in~\eqref{e:TW3} for $\xi < 0$ by rescaling $z$, $\rho$ and $\xi$. Indeed, setting
\begin{align}
\begin{split}
R &= l^{2} \rho^2, \qquad \zeta = \frac{m}{\sqrt{1+\alpha^2}} \xi, \qquad
\hat{z} = \left(z + \frac{\left(1-\ri \alpha\right)c}{2\left(1+\alpha^2\right)}\right)\frac{\sqrt{1+\alpha^2}}{m},
\end{split} \label{rescaling}
\end{align}
transforms system~\eqref{e:TW3} into
\begin{align}
\begin{split}
\hat{z}_\zeta &= -\hat{z}^2 + \mu + (1 + \ri \epsilon_2) R + \frac{1 - \chi}{m^2}(1 - \ri \alpha),\\
R_\zeta &= R\left(\hat{z} + \overline{\hat{z}} + \epsilon_1^2 - 2\right).
\end{split} \label{exprob}
\end{align}
with
$$\mu := -\epsilon_1^2 + \ri \hat{\omega} + \frac{\epsilon_1^4}{4}.$$
Since we can express $m$ as a smooth function of $\hat{\omega}$ and $\epsilon_1$ by
$$m^2 := \frac{1 + \alpha^2}{1 + \alpha^2 - \epsilon_1^2 \alpha^2 + \tfrac{1}{4}\epsilon_1^4 \alpha^2 + \alpha\hat{\omega}},$$
we find that system~\eqref{exprob} is smooth in $\alpha, \hat{\omega}$ and the small parameter $\epsilon = (\epsilon_1,\epsilon_2)$.

\subsection{Rescaling and reparameterization in the spectral problem}

Our idea is to reduce complexity in the spectral problem by applying a similar rescaling and reparameterization as in~\S\ref{sec:resc} to the linearization $\El_\tf$ of~\eqref{e:TW2} about the front solution $A_\tf(\xi)$. To do so, we first need to formulate the linear operator $\El_\tf$ in terms of the front solution $(z_\tf(\xi),\rho_\tf(\xi))$ to~\eqref{e:TW3}, where we of course denote $\rho_\tf(\xi) = |A_\tf(\xi)|$ and $z_\tf(\xi) = A_\tf'(\xi)/A_\tf(\xi)$. Thus, we set $a = w \rho_\tf$ and $b = y \rho_\tf$ and arrive, using~\eqref{asymptrigger} and~\eqref{asymptrigger2}, at the operator $\tilde{\El}_\tf$, posed on $\smash{L^2_{-\kappa, \kappa - 2\Re(z_+)}(\R,\C^2)}$, whose domain is $\smash{H^2_{-\kappa,\kappa-2\Re(z_+)}(\R,\C^2)}$ (see Definition~\ref{defspaces}), given by
\begin{align*}
\tilde{\El}_\tf\begin{pmatrix} w \\ y\end{pmatrix} = \begin{pmatrix} \tilde{L}_\tf(w,y) \\ \overline{\tilde{L}_\tf\left(\overline{w},\overline{y}\right)}\end{pmatrix},
\end{align*}
with
\begin{align*}
\tilde{L}_\tf(w,y) = \frac{1}{\rho_\tf} L_\tf\left(w \rho_\tf, y \rho_\tf\right) = (1+\ri \alpha) w_{\xi\xi} + \left(c + 2(1+\ri \alpha) z_\tf \right)w_\xi - (1+\ri \gamma) \rho_\tf^2 (w + y).
\end{align*}
We note that the rescaling of the perturbation by the front amplitude $\rho_\tf$, which is bounded away from the origin at $\xi = -\infty$, leaves the spectrum of the corresponding asymptotic state unchanged. Since $\rho_\tf\rightarrow0$ exponentially fast as $\xi\rightarrow\infty$, this shifts the spatial eigenvalues of the associated state to the right, and is the reason for the inclusion of an additional weight $-2\Re(z_+)$ for $\xi\geq0$, cf.~\eqref{asymptrigger2}.

Thus, by construction, we find that the spectra of the operators $\tilde{\El}_\tf$, posed on $\smash{L^2_{-\kappa,\kappa - 2\Re(z_+)}(\R,\C^2)}$, and of $\El_\tf$, posed on $L^2_\kappa(\R,\C^2)$, coincide, including multiplicities of eigenvalues.

Now, we can apply the reparameterization~\eqref{reparameterization} and rescaling~\eqref{rescaling} to the linear operator $\tilde{\El}_\tf$, posed on the space $\smash{L^2_{-\kappa,\kappa-2\Re(z_+)}(\R,\C^2)}$, and arrive at the operator $\hat{\El}_\tf$, posed on $\smash{L^2_{\hat{\kappa}_-,\hat{\kappa}_+}(\R,\C^2)}$ with domain $\smash{H^2_{\hat\kappa_-,\hat\kappa_+}(\R,\C^2)}$, which is given by
\begin{align*}
\hat{\El}_\tf\begin{pmatrix} w \\ y\end{pmatrix} = \begin{pmatrix} \hat{L}_\tf(w,y) \\ \overline{\hat{L}_\tf\left(\overline{w},\overline{y}\right)}\end{pmatrix},
\end{align*}
with
\begin{align*}
\hat{L}_\tf(w,y) = \tilde{L}_\tf\left(w\left(\cdot \ \frac{\sqrt{1+\alpha^2}}{m}\right), y\left(\cdot\ \frac{\sqrt{1+\alpha^2}}{m}\right)\right) = \left(w_{\zeta\zeta} + 2\hat{z}_\tf w_\zeta - (1 + \ri \epsilon_2) \rho_\tf^2 (w + y)\right)\frac{m^2}{(1- \ri \alpha)},
\end{align*}
and
\begin{align}
\begin{split}
\hat{\kappa}_- := -\frac{\kappa\sqrt{1+\alpha^2}}{m}, \qquad \hat{\kappa}_+ &:= \frac{\left(\kappa - 2\Re(z_+)\right)\sqrt{1+\alpha^2}}{m}\\ &= \frac{\kappa \sqrt{1+\alpha^2}}{m} + \Re \sqrt{\mu + \frac{2(1-\ri\alpha)}{m^2}} - \frac{\epsilon_1^2}{2} + 1,
\end{split}\label{defkappa}
\end{align}
where we note that rescaling the spatial coordinate by a factor $\smash{\sqrt{1+\alpha^2}}/m$ in~\eqref{rescaling} lead to a rescaling of the weights $-\kappa$ and $\kappa - 2\Re(z_+)$ for $\zeta \leq 0$ and $\zeta \geq 0$, respectively, by the same factor.

Thus, by construction, the spectra of the operators $\hat{\El}_\tf$, posed on $\smash{L^2_{\hat{\kappa}_-,\hat{\kappa}_+}(\R,\C^2)}$, and of $\tilde{\El}_\tf$, posed on $\smash{L^2_{-\kappa,\kappa - 2\Re(z_+)}(\R,\C^2)}$, coincide, including multiplicities of eigenvalues. In addition, the dispersion parameter $\alpha$ is, up to a scaling factor, eliminated from the operator $\hat{L}_\tf$. Therefore we have, as in the existence problem, reduced the complexity in the relevant spectral problem through the reparameterization~\eqref{reparameterization} and rescaling~\eqref{rescaling}.

\section{Overview of existence results} \label{sec:exis}

In this section, we collect the results from the existence analysis of the pattern-forming front in~\cite{GS14}, which are needed for our spectral analysis. The existence analysis in~\cite{GS14} is, in order to reduce complexity, performed in the rescaled and reparameterized traveling-wave equation~\eqref{exprob}. As explained in~\S\ref{sec:rep}, we adopt a similar reduction in our spectral analysis.

We start by reformulating the existence result, Theorem~\ref{t:ex0}, in terms of the rescaled and reparameterized system~\eqref{exprob}.

\begin{theorem} \label{t:ex02}
Let $\alpha \in \R$ and let $\delta > 0$ be fixed such that $\delta$ is sufficiently small. Then, provided $0 < \epsilon_1 \ll 1$ and $0 \leq |\epsilon_2| \ll 1$, there exists a front solution $\Psi_\mathrm{tf}(\zeta;\epsilon) = (\hat{z}_{\tf},R_{\tf})(\zeta;\epsilon)$ to~\eqref{exprob} for $\hat{\omega} = \hat{\omega}_{\tf}(\epsilon)$ between the hyperbolic fixed points
\begin{align*} \lim_{\zeta \to \infty} \Psi_\tf(\zeta;\epsilon) = \left(z_+(\epsilon),0\right), \qquad \lim_{\zeta \to -\infty} \Psi_\tf(\zeta;\epsilon) = \left(1 - \frac{\epsilon_1^2}{2} + \ri\hat{k}_\tf(\epsilon),{1-\hat{k}_\tf(\epsilon)^2}\right), \end{align*}
where
\begin{align}
\hat{z}_+(\epsilon) := -\sqrt{\frac{2(1-\ri\alpha)}{m(\epsilon)^2}-\epsilon_1^2 + \ri \hat{\omega} + \frac{\epsilon_1^4}{4}}, \label{defhatz}
\end{align}
and $\hat{k}_\tf(\epsilon)$ and $\hat{\omega}_\tf(\epsilon)$ are smooth at $\epsilon = (0,0)$ satisfying
\begin{align}
\hat{k}_\tf(0,\epsilon_2) = \frac{\epsilon_2}{1+\sqrt{1 + \epsilon_2^2}}, \qquad \partial_{\epsilon_1}^j \hat{\omega}_\tf(0,\epsilon_2) = 0, \quad j = 0,1,2, \label{preciseomegabound}\end{align}
The solution $\Psi_{\mathrm{tf}}(\zeta;\epsilon)$ is continuous at the jump heterogeneity at $\zeta = 0$. The position
\begin{align*}
\zeta_{\tf}(\epsilon) := \inf\left\{\zeta \,:\, R_\mathrm{tf}(\tilde\zeta;\epsilon) <\delta \text{ for all } \tilde\zeta > \zeta\right\},
\end{align*}
of the front interface satisfies
 \begin{align}
 \lim_{\epsilon \to (0,0)} \epsilon_1 \zeta_\tf(\epsilon) = \pi, \label{triggerbound}
\end{align}
and $\epsilon \mapsto \epsilon_1 \zeta_\tf(\epsilon)$ is smooth at $\epsilon = (0,0)$.
\end{theorem}

A consequence of Theorem~\ref{t:ex0} is that, by expressing $m$ and $\mu$ as
\begin{align}
m(\epsilon)^2 := \frac{1 + \alpha^2}{1 + \alpha^2 - \epsilon_1^2 \alpha^2 + \tfrac{1}{4}\epsilon_1^2\alpha^2 + \alpha\hat{\omega}_\tf(\epsilon)}, \qquad \mu(\epsilon) := -\epsilon_1^2 + \ri \hat{\omega}_\tf(\epsilon) + \frac{\epsilon_1^4}{4}, \label{defmm}
\end{align}
the number of parameters in~\eqref{exprob} has reduced to three: the fixed parameter $\alpha$ and the small parameter $\epsilon = (\epsilon_1,\epsilon_2)$. We find that $m(\epsilon)$ and $\mu(\epsilon)$ are smooth at $\epsilon = (0,0)$ and satisfy
\begin{align} m(0,0) = 1, \qquad \mu(0,0) = 0. \label{limmum}\end{align}

\subsection{Tracking the solution to the left of the front interface}

The front solution $\Psi_{\tf}(\zeta;\epsilon)$ in Theorem~\ref{t:ex02} is constructed in~\cite{GS14} by tracking the one-dimensional, unstable manifold $W^u_-$ of the fixed point
$$\left(1 - \tfrac{1}{2}\epsilon_1^2 + \ri\hat{k}_\tf(\epsilon),{1-\hat{k}_\tf(\epsilon)^2}\right),$$
in~\eqref{exprob}. Global control over this manifold can be obtained by perturbing from the `real limit' $\epsilon = (0,0)$. Indeed, setting $\epsilon_1, \epsilon_2 = 0$ and $\chi \equiv 1$ in~\eqref{exprob}, we obtain, using~\eqref{limmum}, the system
\begin{align}
\begin{split}
 \hat{z}_\zeta &= -\hat{z}^2 + R,\\
 R_\zeta &= 2R(\hat{z} - 1),
\end{split} \label{exproblim}
\end{align}
which is equivalent to the real Ginzburg-Landau equation
 \begin{align} r_{\zeta\zeta} = -2r_\zeta - r + r^3, \label{realGL}\end{align}
upon setting $R = r^2$ and $z = 1 + r_\zeta/r$. The unstable manifold $W^u_-$ in~\eqref{exproblim} is given by the solution
\begin{align*}
 \Psi_*(\zeta) = (\hat{z}_*(\zeta),R_*(\zeta)) = \left(1 + \frac{r_*'(\zeta)}{r_*(\zeta)}, r_*(\zeta)^2\right),
\end{align*}
where $r_*(\zeta)$ is a heteroclinic in~\eqref{realGL} connecting the hyperbolic saddle $(1,0)$ to the hyperbolic degenerate sink $(0,0)$. A simple phase plane analysis of~\eqref{realGL} shows that the solution $r_*(\zeta)$ is monotonically decreasing and does not lie in the strong stable manifold of $(0,0)$. Thus, $r_*(\zeta)$ decays exponentially to $0$, whereas $\frac{r_*'(\zeta)}{r_*(\zeta)}$ decays algebraically to $-1$  as $\zeta \to \infty$. All in all, one obtains the following result.

\begin{proposition}[\cite{GS14}] \label{prop:left}
Take $\delta > 0$ sufficiently small and let $\zeta_\delta \in \R$ be the value such that $R_*(\zeta_\delta) = \delta$. Then, there exists a $\delta$-independent constant $C>1$ and $\eta > 0$ such that
\begin{align}
0 < \hat{z}_*(\zeta_\delta) \leq C|\log(\delta)|^{-1}, \qquad & \frac{1}{\hat{z}_*(\zeta_\delta)} \leq C|\log(\delta)|, \label{deltaineq}
\end{align}
and
\begin{align}
\begin{split}
\left\|\Psi_*(\zeta) - (1,1)\right\| &\leq C\re^{\eta \zeta}, \qquad \zeta \leq 0,\\
\left\|R_*(\zeta)\right\| &\leq C\re^{-\eta \zeta}, \qquad \zeta \geq 0.
\end{split}  \label{Rsbound}
\end{align}
In addition, there exists a constant $C_\delta>1$, which only depends on $\delta$, such that, provided $0 < \epsilon_1 \ll 1$ and $0 \leq |\epsilon_2| \ll 1$, it holds
\begin{align}
\begin{split}
\left\|\Psi_\tf(\zeta_\tf(\epsilon) + \zeta;\epsilon) - \Psi_*(\zeta_\delta + \zeta)\right\| &\leq C_\delta \|\epsilon\|,\\
\left\|\Psi_\tf(\zeta_\tf(\epsilon) + \zeta;\epsilon) - \left(1 - \tfrac{1}{2}\epsilon_1^2 + \ri\hat{k}_\tf(\epsilon),1 -\hat{k}_\tf(\epsilon)^2\right)\right\| &\leq C_\delta \re^{\eta \zeta},
\end{split} \qquad \zeta \leq 0. \label{estex1}
\end{align}
\end{proposition}

\subsection{Tracking the solution to the right of the front interface}

To the right of the front interface, i.e.~for $\zeta \geq \zeta_\tf(\epsilon)$, the solution $\Psi_\tf(\zeta;\epsilon)$ lies in a neighborhood of the normally hyperbolic, attracting manifold $R = 0$ in~\eqref{exprob}. Thus, with the aid of geometric singular perturbation theory~\cite{FEN1,FEN2,FEN3}, it is established in~\cite{GS14} that there exists a solution on the manifold $R = 0$ that converges exponentially fast to $\Psi_\tf(\zeta;\epsilon)$ as $\zeta \to \infty$. Setting $R = 0$ in~\eqref{exprob} one obtains the scalar Riccati equation
\begin{align}
\hat{z}_0' &= -\hat{z}_0^2 + \mu(\epsilon) + \frac{1 - \chi}{m(\epsilon)^2}(1 - \ri \alpha). \label{dyninvman}
\end{align}
Parameters $\epsilon$ are chosen in such a way in~\cite{GS14} that the unstable manifold $W^u_-$ intersects with to the two-dimensional stable manifold $W^s_+$ of the sink $(\hat{z}_+(\epsilon),0)$ in~\eqref{exprob} (for $\chi \equiv -1$). Note that $\hat{z}_+(\epsilon)$, which is defined in~\eqref{defhatz}, is smooth at $\epsilon = (0,0)$ and satisfies
\begin{align} \hat{z}_+(0,0) = -\sqrt{2-2\ri \alpha}. \label{hatzapprox}\end{align}

Thus, the solution $\Psi_\tf(\zeta;\epsilon)$ in the unstable manifold $W^u_-$ converges exponentially fast to the solution $(\hat{z}_0(\zeta;\epsilon),0)$ on the manifold $R=0$ as $\zeta \to \infty$, where $\hat{z}_0(\zeta;\epsilon)$ is the solution to~\eqref{dyninvman} with initial condition $\hat{z}_0(0;\epsilon) = \hat{z}_+(\epsilon)$. Hence, as $\hat{z}_+(\epsilon)$ is a fixed point of~\eqref{dyninvman} for $\chi \equiv -1$, it holds $\hat{z}_0(\zeta;\epsilon) = \hat{z}_+(\epsilon)$ for \emph{all} $\zeta \geq 0$.

More precisely, the following estimates are obtained in~\cite{GS14}.

\begin{proposition} \label{prop:right}
There exist constants $C > 1$ and $\eta > 0$ such that, provided $0 < \epsilon_1 \ll 1$ and $0 \leq |\epsilon_2| \ll 1$, we have for $\zeta \geq \zeta_\tf(\epsilon)$,
 \begin{align}
 \left\|\Psi_\tf(\zeta;\epsilon) - \left(\hat{z}_0(\zeta;\epsilon),0\right)\right\| &\leq C\delta \re^{-\eta|\zeta - \zeta_\tf(\epsilon)|}, \label{estex2}\\
\left\|\hat{z}_0(\zeta;\epsilon)\right\| &\leq C. \label{estex3}
\end{align}
\end{proposition}

For the spectral analysis in this paper, we need to extend the estimates on the $R$-component in Proposition~\ref{prop:left} beyond the front interface, which follows by a simple application of Gr\"onwall's lemma in combination with the exponential decay obtained in Proposition~\ref{prop:right}.

\begin{proposition} \label{prop:rightleft}
Let $\delta > 0$ be sufficiently small and let $\zeta_\delta \in \R$ be the value such that $R_*(\zeta_\delta) = \delta$. Then, there exists a constant $C_\delta>1$, which only depends on $\delta$, such that, provided $0 < \epsilon_1 \ll 1$ and $0 \leq |\epsilon_2| \ll 1$, it holds
\begin{align*}
\left\|\Psi_\tf(\zeta_\tf(\epsilon) + \zeta;\epsilon) - \Psi_*(\zeta_\delta + \zeta)\right\| &\leq C_{\delta} \sqrt{\|\epsilon\|},
\end{align*}
for $\zeta \in [0,-\iota \log\|\epsilon\|]$, where $\iota > 0$ is a constant independent of $\epsilon$ and $\delta$.
\end{proposition}
\begin{proof}
By~\eqref{defmm} and Theorem~\ref{t:ex02}, system~\eqref{exprob} is of the abstract form
\begin{align*} \Psi_\zeta = F(\Psi;\epsilon),\end{align*}
where $F \colon \R \times \C \times U \to \R \times \C$ is smooth and $U \subset \R^2$ is a neighborhood of $(0,0)$. Clearly, by~\eqref{limmum}, system~\eqref{exproblim} equals
\begin{align*} \Psi_\zeta = F(\Psi;0).\end{align*}
We wish to bound the difference
$$d(\zeta;\epsilon) := \Psi_\tf(\zeta_\tf(\epsilon) + \zeta;\epsilon) - \Psi_*(\zeta_\delta + \zeta),$$
for $\zeta \in [0,-\zeta_\tf(\epsilon)]$.

The solution $\Psi_\tf(\zeta;\epsilon)$ to~\eqref{exprob} is $\delta$-uniformly bounded on $[\zeta_\tf(\epsilon),0] \times U$ by Proposition~\ref{prop:right}. Similarly, since $r_*(\zeta)$ is a heteroclinic solution to~\eqref{realGL}, the solution $\Psi_*(\zeta)$ to~\eqref{exproblim} is bounded on $\R$. Thus, bounding the right hand side of
\begin{align*} d(\zeta;\epsilon) = d(0;\epsilon) + \int_0^\zeta \left(F(\Psi_\tf(\zeta_\tf(\epsilon) + y;\epsilon);\epsilon) - F(\Psi_*(y);0)\right) \de y
\end{align*}
yields, by Proposition~\ref{prop:left}, the integral inequality
\begin{align*} \|d(\zeta;\epsilon)\| &\leq C_\delta\|\epsilon\| + \int_0^\zeta C \|d(y;\epsilon)\| \mathrm dy,\end{align*}
for $\zeta \in [0,-\zeta_\tf(\epsilon)]$, where $C > 1$ is an $\epsilon$- and $\delta$-independent constant and $C_\delta > 1$ is a constant depending on $\delta$ only. Taking $2\iota C < 1$ and applying Gr\"onwall's lemma, we arrive at
\begin{align*} \|d(\zeta;\epsilon)\| \leq C_\delta \sqrt{\|\epsilon\|},\end{align*}
for $\zeta \in [0,-\iota \log\|\epsilon\|]$, which proves the result.
\end{proof}

\section{Set-up of spectral analysis} \label{secsetup}

In an effort to reduce complexity, we have followed the rescaling and reparameterization performed in the existence analysis of the pattern-forming front in~\cite{GS14}, and applied a similar reduction to the linearized operator $\El_\tf$ in~\S\ref{sec:rep}. By construction, the spectrum of the obtained operator $\hat{\El}_\tf$, posed on $\smash{L^2_{\hat{\kappa}_-,\hat{\kappa}_+}(\R,\C^2)}$, coincides with the spectrum $\El_\tf$, posed on $L^2_\kappa(\R,\C^2)$, including multiplicities of eigenvalues. Of course, the same holds for the spectra of $\hat{\El}_\tf$, posed on $\smash{L^2_{0,\hat\kappa_0}(\R,\C^2)}$ with domain $\smash{H^2_{0,\hat\kappa_0}(\R,\C^2)}$, where we denote
\begin{align} \hat\kappa_0 = - \Re(\hat{z}_+(\epsilon)) - \tfrac{1}{2}\epsilon_1^2 + 1. \label{defkapp0}\end{align}
and of $\El_\tf$, posed on $L^2(\R,\C^2)$. Hence, our main result, Theorem~\ref{theospecstab}, follows by proving the following equivalent statement for $\hat{\El}_\tf$.

\begin{theorem} \label{theospecstab2}
Let $\alpha\in\big(-\frac{1}{2}\sqrt{2}, \frac{1}{2}\sqrt{2}\big)$. Fix $\kappa \in \big(0,\frac{1}{4\sqrt{1+\alpha^2}}\big)$. Let $\hat{\kappa}_\pm$ be as in~\eqref{defkappa} and let $\hat{\kappa}_0$ be as in~\eqref{defkapp0}. Then, provided $0 < \epsilon_1 \ll 1$ and $0 \leq |\epsilon_2| \ll 1$, the following assertions hold true:
\begin{itemize}
 \item[i)] The spectrum of $\hat{\El}_\tf$ posed on $\smash{L_{0,\hat\kappa_0}^2(\R,\C^2)}$ does not intersect the closed right-half plane, except at the origin as a parabolic curve.
 \item[ii)] When posed on $\smash{L^2_{\hat{\kappa}_-,\hat{\kappa}_+}(\R,\C^2)}$, the operator $\hat{\El}_\tf$ has no spectrum in the closed right-half plane, except for an algebraically simple eigenvalue, which resides at the origin. Furthermore, eigenvalues near the origin lie $\mathcal{O}(\epsilon_1^2)$-close to $\Sigma_{0,\abs} \cup \smash{\overline{\Sigma_{0,\abs}}}$. 
     \end{itemize}
\end{theorem}

In order to prove Theorem~\ref{theospecstab2}, we follow the approach as outlined in~\S\ref{sec:overview}. We cover, as in~\cite{CSR}, the critical spectrum of $\hat{\El}_\tf$ by the following three regions
\begin{align}
\begin{split}
R_1 &= R_1(\Theta_1) := \{\lambda \in \C : |\lambda| \leq \Theta_1\},\\
R_2 &= R_2(\theta_2,\Theta_1,\Theta_2) := \{\lambda \in \C : \Theta_1 \leq |\lambda| \leq \Theta_2 \wedge |\arg(\lambda)| \leq \tfrac{1}{2} \pi + \theta_2\},\\
R_3 &= R_3(\theta_3,\Theta_2) := \{\lambda \in \C : |\arg(\lambda)| \leq \tfrac{1}{2}\pi + \theta_3, |\lambda| \geq \Theta_2\},
\end{split}\label{regions}
\end{align} 
where $\theta_2,\theta_3 > 0$ and $\Theta_1,\Theta_2 > 1$ are constants which are independent of the small parameter $\epsilon$. We will study the spectrum of $\hat{\El}_\tf$ in the regions $R_1$, $R_2$ and $R_3$ separately. See Figure~\ref{fig:sp_ev} (right) for a schematic depiction of these regions.

We decompose the spectrum of $\hat{\El}_\tf$ into essential and point spectrum; we refer to~\cite{SAN2} for a general introduction. In~\S\ref{sec:ess} we show that the essential spectrum of $\hat{\El}_\tf$, posed on $L^2_{0,\hat\kappa_0}(\R,\C^2)$, is contained in the left-half plane and touches the imaginary axis only at the origin as a parabolic curve, whereas the essential spectrum of $\hat{\El}_\tf$, posed on $L^2_{\hat{\kappa}_-,\hat{\kappa}_+}(\R,\C^2)$, is confined to the open left-half plane and does not intersect the regions $R_1,R_2$ and $R_3$.

A point $\lambda \in \C$ lies in the point spectrum of $\hat{\El}_\tf$, posed on $L^2_{\hat{\kappa}_-,\hat{\kappa}_+}(\R,\C^2)$, if and only if it does not lie in its essential spectrum and the associated eigenvalue problem $\hat{\El}_\tf X = \lambda X$, which reads
\begin{align}
\begin{split}
w_{\zeta\zeta} + 2\hat{z}_\tf(\zeta;\epsilon) w_\zeta - (1 + \ri \epsilon_2) R_\tf(\zeta;\epsilon) (w + y) &= \frac{(1- \ri \alpha) \lambda}{ m(\epsilon)^2} w,\\
y_{\zeta\zeta} + 2\overline{\hat{z}_\tf(\zeta;\epsilon)} y_\zeta - (1 - \ri \epsilon_2) R_\tf(\zeta;\epsilon) (w + y) &= \frac{(1 + \ri \alpha) \lambda}{ m(\epsilon)^2} y,
\end{split} \label{evprob000}
\end{align}
admits a nontrivial solution in $\smash{H^2_{\hat{\kappa}_-,\hat{\kappa}_+}(\R,\C^2)}$.

Since $\smash{H^2_{0,\hat\kappa_0}(\R,\C^2)}$ is contained in $\smash{H^2_{\hat{\kappa}_-,\hat{\kappa}_+}(\R,\C^2)}$ as $\kappa > 0$, the point spectrum of $\hat{\El}_\tf$, posed on $\smash{L^2_{0,\hat\kappa_0}(\R,\C^2)}$, is contained in the spectrum of $\hat{\El}_\tf$, posed on $\smash{L^2_{\hat{\kappa}_-,\hat{\kappa}_+}(\R,\C^2)}$. We will show in~\S\ref{sec:point}-\S\ref{sec:pointR1} that the point spectrum of $\hat{\El}_\tf$, posed on $\smash{L^2_{\hat{\kappa}_-,\hat{\kappa}_+}(\R,\C^2)}$, is contained in the open left-half plane, except for an algebraically simple eigenvalue residing at the origin, which, in conjunction with the aforementioned results on the essential spectrum, proves Theorem~\ref{theospecstab2}, and thus also proves Theorem~\ref{theospecstab}.

\section{Analysis of the essential spectrum} \label{sec:ess}

In this section, we study the essential spectrum of the operator $\hat{\El}_\tf$, posed on the spaces $\smash{L^2_{0,\hat\kappa_0}(\R,\C^2)}$ and $\smash{L^2_{\hat{\kappa}_-,\hat{\kappa}_+}(\R,\C^2)}$. As outlined in Appendix~\ref{appess}, the essential spectrum of the asymptotically constant-coefficient operator $\hat{\El}_\tf$ is determined by the spatial eigenvalues of its limiting operators at $\pm \infty$.

In the limit $\zeta \to -\infty$, the eigenvalue problem~\eqref{evprob000} associated with $\hat{\El}_\tf$ is by Proposition~\ref{prop:left} of the form
\begin{align}
\begin{split}
\frac{\lambda}{m(\epsilon)^2} (1 - \ri \alpha) w &= w_{\zeta\zeta} + \left(2 - \epsilon_1^2 + 2\ri\hat{k}_\tf(\epsilon)\right) w_\zeta - (1 + \ri \dell) \left(1- \hat{k}_\tf(\epsilon)^2\right)(w + y),\\
\frac{\lambda}{m(\epsilon)^2} (1 + \ri \alpha) y &= y_{\zeta\zeta} + \left(2 - \epsilon_1^2 - 2\ri\hat{k}_\tf(\epsilon)\right) y_\zeta - (1- \ri \dell) \left(1- \hat{k}_\tf(\epsilon)^2\right) (w + y),
\end{split} \label{evprobmin}
\end{align}
whereas in the limit $\zeta \to \infty$ the eigenvalue problem reduces by Proposition~\ref{prop:right} to
\begin{align}
\begin{split}
\frac{\lambda}{m(\epsilon)^2} (1 - \ri \alpha) w &= w_{\zeta\zeta} + 2\hat{z}_+(\epsilon) w_\zeta,\\
\frac{\lambda}{m(\epsilon)^2} (1 + \ri \alpha) y &= y_{\zeta\zeta} + 2\overline{\hat{z}_+(\epsilon)} y_\zeta.
\end{split} \label{evprobplus}
\end{align}
The spatial eigenvalues can be identified as the values $\nu \in \C$ for which systems~\eqref{evprobmin} and~\eqref{evprobplus} admit a nontrivial solution of the form $\re^{\nu \zeta} X_0$ with $X_0 \in \R^2$. Taking determinants we find associated linear dispersions relations
\begin{align}
\det\begin{pmatrix} \displaystyle \vartheta(\nu,\epsilon) -\frac{\lambda (1 - \ri \alpha)}{m(\epsilon)^2}  & - (1 + \ri \dell) \left(1- \hat{k}_\tf(\epsilon)^2\right)\\
-(1 - \ri \dell) \left(1- \hat{k}_\tf(\epsilon)^2\right) & \displaystyle \overline{\vartheta(\overline{\nu},\epsilon)}-\frac{\lambda (1 + \ri \alpha)}{m(\epsilon)^2}\end{pmatrix} = 0, \label{spatev2}
\end{align}
and
\begin{align}
\det\begin{pmatrix}\nu^2 + 2\hat{z}_+(\epsilon) \nu - \frac{\lambda}{m(\epsilon)^2} (1 - \ri \alpha) & 0 \\ 0 & \nu^2 + 2\overline{\hat{z}_+(\epsilon)} \nu - \frac{\lambda}{m(\epsilon)^2} (1 + \ri \alpha)\end{pmatrix} = 0, \label{spatev1}
\end{align}
where we denote
\begin{align}\vartheta(\nu,\epsilon) := \nu^2 + \left(2 - \epsilon_1^2 + 2\ri\hat{k}_\tf(\epsilon)\right)\nu - (1 + \ri \dell) \left(1- \hat{k}_\tf(\epsilon)^2\right),\label{defvartheta}\end{align}
The spatial eigenvalues arise as the roots of~\eqref{spatev2} and~\eqref{spatev1}, and can be ordered by their real parts
\begin{align*} \Re\,\nu_{1,\pm}(\lambda,\epsilon) \leq \Re\,\nu_{2,\pm}(\lambda,\epsilon) \leq \Re\,\nu_{3,\pm}(\lambda,\epsilon) \leq \Re\,\nu_{4,\pm}(\lambda,\epsilon),\end{align*}
when counted with multiplicities.

In Appendix~\ref{appesscalc} we determine the spatial eigenvalues $\nu_{i,\pm}(\lambda,\epsilon), i = 1,\ldots,4$ up to leading-order. For each $\lambda$ in the regions $R_2$ and $R_3$, defined in~\eqref{regions}, we obtain the splittings
\begin{align} \Re\,\nu_{1,-}(\lambda) \leq \Re\,\nu_{2,-}(\lambda) < \hat{\kappa}_- < 0 < \Re\,\nu_{3,-}(\lambda) \leq \Re\,\nu_{4,-}(\lambda),\label{split1}\end{align}
and
\begin{align} \Re\,\nu_{1,+}(\lambda) \leq \Re\,\nu_{2,+}(\lambda) < \hat{\kappa}_0 < \hat{\kappa}_+ < \Re\,\nu_{3,+}(\lambda) \leq \Re\,\nu_{4,+}(\lambda), \label{split2}\end{align}
As outlined in Appendix~\ref{appess}, this implies that $R_2 \cup R_3$ contains no essential spectrum of the operator $\hat{\El}_\tf$, posed on $\smash{L^2_{\hat{\kappa}_-,\hat{\kappa}_+}(\R,\C^2)}$ or on $\smash{L^2_{0,\hat{\kappa}_0}(\R,\C^2)}$. Moreover, we show in Appendix~\ref{appesscalc} that the splitting~\eqref{split1} persists for $\lambda \in R_1$, whereas~\eqref{split2} no longer holds. However, for $\lambda \in R_1$ we still find
\begin{align} \Re\,\nu_{1,-}(\lambda) \leq \Re\,\nu_{2,-}(\lambda) < \hat{\kappa}_- < \Re\,\nu_{3,-}(\lambda) \leq \Re\,\nu_{4,-}(\lambda), \label{split3}\end{align}
and the curve $\{\lambda \in R_1 : \Re \, \nu_{3,-}(\lambda) = 0 \}$ is confined to the open left-half plane except for a parabolic touching with the imaginary axis at the origin. All in all, we have established the following result, which is schematically depicted in Figure~\ref{fig:sp_ev}. 

\begin{figure}[h!]
\centering
 \includegraphics[trim=0 0.0in 0 0, clip, width = 0.5\textwidth]{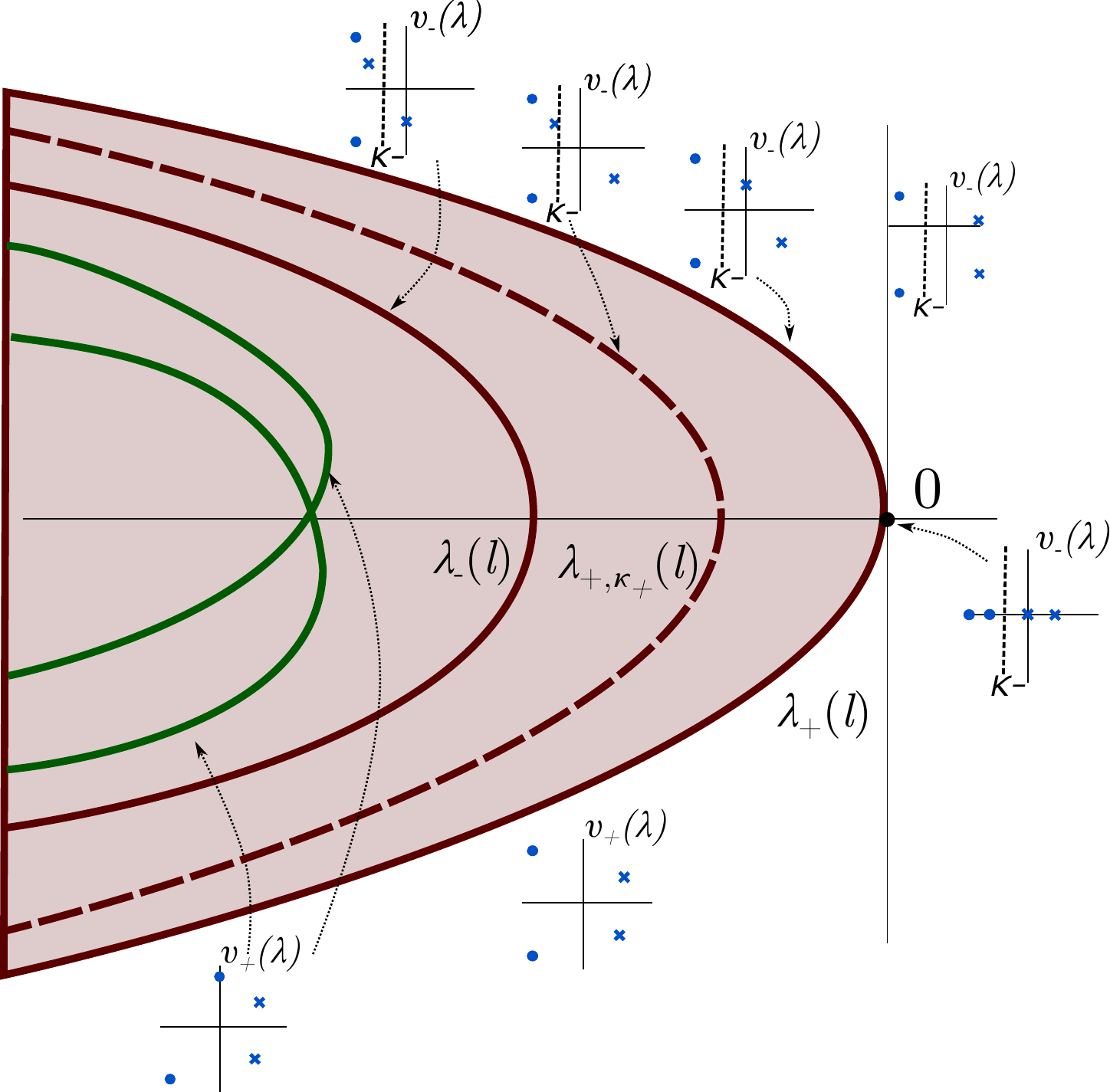}
 \hspace{-0.0in}
  \includegraphics[trim=1in 0.0in 0 0in, clip, width = 0.4\textwidth]{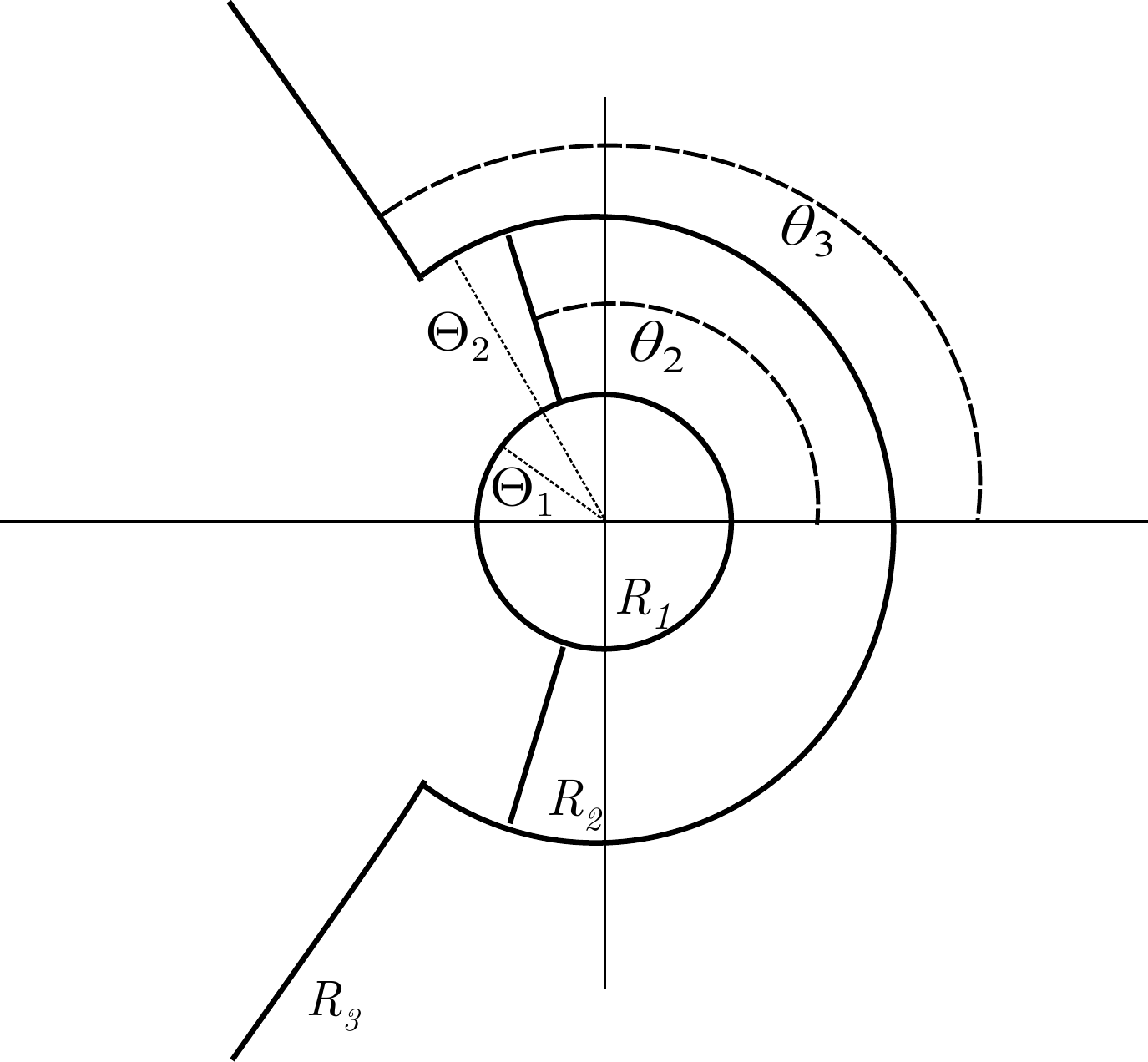}
\caption{(left): plot of the essential spectrum of the operator $\hat{\El}_\tf$, posed on $L^2_{0,\hat\kappa_0}(\R,\C^2)$ with insets of the most critical spatial eigenvalues $\nu_\pm(\lambda)$ associated with the asymptotic systems at $\xi = -\infty$ and $\xi = +\infty$, respectively. Solid red lines give Fredholm boundaries of spectrum coming from $\xi = -\infty$, dotted red line gives right-most Fredholm boundary when posed on $\smash{L^2_{\hat{\kappa}_-,\hat{\kappa}_+}(\R,\C^2)}$, green lines give Fredholm boundaries coming from $\xi = +\infty$; see Appendix~\ref{appess} for more detail. (right): schematic depiction of the regions $R_1,R_2$ and $R_3$ defined in~\eqref{regions}.}\label{fig:sp_ev}
\end{figure}

\begin{theorem} \label{concless}
Let $\alpha \in \big(-\frac{\sqrt{2}}{2}, \frac{\sqrt{2}}{2}\big)$ and fix $\kappa \in \smash{\big(0,\frac{1}{4\sqrt{1+\alpha^2}}\big)}$. For $\Theta_1,\theta_2,\theta_3 > 0$ sufficiently small and $\Theta_2 > 1$ sufficiently large, we have that, provided $0 < \epsilon_1 \ll 1$ and $0 \leq |\epsilon_2| \ll 1$, the following statements hold true:
\begin{itemize}
\item[i)] There is no essential spectrum of the operator $\hat{\El}_\tf$, posed on $\smash{L^2_{\hat{\kappa}_-,\hat{\kappa}_+}(\R,\C^2)}$, in the region $R_1(\Theta_1) \cup R_2(\theta_2,\Theta_1,\Theta_2) \cup R_3(\theta_3,\Theta_2)$.
\item[ii)] The essential spectrum of $\El_\tf$, posed on $\smash{L^2_{0,\kappa_0}(\R,\C^2)}$ does not touch the closed right-half plane, except at the origin as a parabolic curve.
\end{itemize}
\end{theorem}

\section{Preparations for the analysis of the point spectrum} \label{sec:point}

In the upcoming three sections we analyze the point spectrum of the operator $\hat{\El}_\tf$, posed on $\smash{L^2_{\hat{\kappa}_-,\hat{\kappa}_+}(\R,\C^2)}$, in the regions $R_1$, $R_2$ and $R_3$, defined in~\eqref{regions}, separately. In the region $R_3$, a standard scaling argument precludes the existence of point spectrum. In the regions $R_1$ and $R_2$, we employ the Riccati-Evans function, see~\S\ref{sec:ric}, as a tool to locate the point spectrum, which requires control over the evolution of the relevant subspaces as trajectories in the associated matrix Riccati equations. Such control is established in~\S\ref{sec:pointR2} and~\S\ref{sec:pointR1} for the regions $R_2$ and $R_1$, respectively.

In this section we make the necessary preparations for the analysis in the regions $R_1$ and $R_2$ in the upcoming two sections. After ruling out the presence of point spectrum in the region $R_3$, we apply a linear coordinate transform to the eigenvalue problem~\eqref{evprob000}, which leads to a simplification to the right of the front interface. Moreover, we formulate the Riccati-Evans function for our problem, and show that it remains invariant under the linear coordinate transform.

\subsection{Analysis in the region \texorpdfstring{$R_3$}{R3}}

We prove that the operator $\hat{\El}_\tf$ has no point spectrum in the region $R_3$. Our approach is to rewrite the eigenvalue problem~\eqref{evprob000} as an, appropriately scaled, first-order system and prove that this system admits an exponential dichotomy on $\R$. Therefore, it cannot have nontrivial solutions in $\smash{H^2_{\hat{\kappa}_-,\hat{\kappa}_+}(\R,\C^2)}$.

\begin{theorem} \label{concR3}
For $\theta_3 > 0$ sufficiently small and $\Theta_2 > 1$ sufficiently large, there is, provided $0 < \epsilon_1 \ll 1$ and $0 \leq |\epsilon_2| \ll 1$, no point spectrum of the operator $\hat{\El}_\tf$, posed on $\smash{L^2_{\hat{\kappa}_-,\hat{\kappa}_+}(\R,\C^2)}$, in the region $R_3(\theta_3,\Theta_2)$.
\end{theorem}
\begin{proof}
We set $y = |\lambda|^{1/2} \zeta$ and $W = (X,X_y)$ in~\eqref{evprob000}, to obtain the equivalent first-order problem
\begin{align}
\partial_y W = \left(\tilde{A}(\lambda) + \left(\frac{1}{\sqrt{|\lambda|}} + \epsilon_1\right)\tilde{B}(y;\lambda,\epsilon)\right) W, \qquad \tilde{A}(\lambda) := \begin{pmatrix}0 & 0 & 1 & 0 \\ 0 & 0 & 0 & 1 \\ \frac{\lambda}{|\lambda|} (1-\ri \alpha) & 0 & 0 & 0 \\ 0 & \frac{\lambda}{|\lambda|} (1+\ri \alpha) & 0 & 0\end{pmatrix}, \label{ressys}
\end{align}
where the coefficient function $\tilde B(\cdot;\cdot,\epsilon)$ is by Propositions~\ref{prop:left} and~\ref{prop:right} and by~\eqref{limmum} bounded on $\R \times R_3(\theta_3,\Theta_2)$ by an $\epsilon$-independent constant. One readily calculates the eigenvalues of $\tilde{A}(\lambda)$ to be
\begin{align*} \nu_{\pm,\pm}(\lambda) = \pm\sqrt{\frac{\lambda (1\pm\ri\alpha)}{|\lambda|}}. \end{align*}
Taking $\theta_3 > 0$ sufficiently small, we find that the eigenvalues of $\tilde{A}(\lambda)$ are $\lambda$-uniformly bounded away from the imaginary axis for $\lambda \in R_3(\theta_3,\Theta_2)$. Consequently, using~\cite[Proposition 4.2]{COP}, system~\eqref{ressys} has, provided $\Theta_2 > 1$ is sufficiently large and $\epsilon_1 > 0$ is sufficiently small, an exponential dichotomy on $\R$ for all $\lambda \in R_3(\theta_3,\Theta_2)$ with $\lambda$- and $\epsilon$-independent exponent $\eta > 0$. Hence, reverting back to the original spatial variable $\zeta$, the solution space of~\eqref{evprob000} is contained in the direct sum of the spaces $\smash{H^2_{-M^{1/2}\eta,M^{1/2}\eta}(\R,\C^2)}$ and $\smash{H^2_{M^{1/2}\eta,-M^{1/2}\eta}(\R,\C^2)}$. Consequently, taking $\Theta_2 > 1$ sufficiently large but independent of $\epsilon$, the eigenvalue problem~\eqref{evprob000} cannot possess a nontrivial solution in $\smash{H^2_{\hat{\kappa}_-,\hat{\kappa}_+}(\R,\C^2)}$.
\end{proof}

\subsection{A linear coordinate transform} 

In this subsection we apply a linear coordinate transform to the eigenvalue problem~\eqref{evprob000}, which simplifies the problem to the right of the front interface, i.e.~for $\zeta \in [\zeta_\tf(\epsilon),\infty)$.

Proposition~\ref{prop:right} implies that, to the right of the front interface, $R_\tf(\zeta;\epsilon)$ is small and $\hat{z}_\tf(\zeta;\epsilon)$ is approximated by the solution $\hat{z}_0(\zeta;\epsilon)$ to the scalar Riccati equation~\eqref{dyninvman}. Thus, setting $R_\tf$ to $0$ and $\hat{z}_\tf$ to $\hat{z}_0$ in~\eqref{evprob000}, we obtain the reduced eigenvalue problem
\begin{align}
\begin{split}
w_{\zeta \zeta} + 2 \hat{z}_0(\zeta;\epsilon) w_\zeta &= \frac{(1-\ri \alpha)\lambda}{ m(\epsilon)^2} w,\\
y_{\zeta \zeta} + 2 \overline{\hat{z}_0(\zeta;\epsilon)} y_\zeta &= \frac{(1+\ri \alpha)\lambda}{ m(\epsilon)^2} y.
\end{split} \label{redeiv}
\end{align}
We observe that, by setting $R_\tf$ to $0$, the eigenvalue problem~\eqref{evprob000} has decoupled into two Sturm-Liouville problems. In order to make the Sturm-Liouville problems self-adjoint, we apply the standard procedure of removing the first derivatives $w_\zeta$ and $y_\zeta$ from~\eqref{redeiv} through the linear coordinate transform $u = b w$ and $v = \overline{b} y$ with
\begin{align*}
b(\zeta;\epsilon) := \re^{\int_{\zeta_\tf(\epsilon)}^\zeta \hat{z}_0(y;\epsilon) \de y}.
\end{align*}
In the new coordinates, system~\eqref{redeiv} reads
\begin{align}
\begin{split}
u_{\zeta \zeta} - \left(\partial_\zeta \hat{z}_0(\zeta;\epsilon) + \hat{z}_0(\zeta;\epsilon)^2\right) u &= \frac{(1-\ri \alpha)\lambda}{ m(\epsilon)^2} u,\\
v_{\zeta \zeta} - \left(\partial_\zeta \overline{\hat{z}_0(\zeta;\epsilon)} + \overline{\hat{z}_0(\zeta;\epsilon)^2}\right) v &= \frac{(1+\ri \alpha)\lambda}{ m(\epsilon)^2} v.
\end{split} \label{redeiv2}
\end{align}
We can then exploit the specific structure of the underlying Ginzburg-Landau equation. Since $\hat{z}_0$ satisfies the scalar Riccati equation~\eqref{dyninvman}, we find that~\eqref{redeiv2} has in fact constant coefficients, and can be rewritten as
\begin{align}
\begin{split}
u_{\zeta \zeta} - \left(\frac{\lambda + 1 - \chi}{m(\epsilon)^2} (1 - \ri \alpha) + \mu(\epsilon)\right) u &= 0,\\
v_{\zeta \zeta} - \left(\frac{\lambda + 1 - \chi}{m(\epsilon)^2} (1 + \ri \alpha) + \overline{\mu(\epsilon)}\right) v &= 0.
\end{split} \label{redeiv3}
\end{align}
Hence, not only are the Sturm-Liouville problems in~\eqref{redeiv} self-adjoint after applying the coordinate transform $u = b w$ and $v = \overline{b} y$, they also have constant coefficients (except for the jump in $\chi$ at $\zeta = 0$), and are therefore explicitly solvable.

Inspired by the above simplification in the reduced eigenvalue problem~\eqref{redeiv}, we hope to simplify the full eigenvalue problem~\eqref{evprob000} by applying a similar linear coordinate transform. In order to apply the theory of exponential dichotomies later, it is convenient to first rewrite~\eqref{evprob000} as the first-order system
\begin{align}
\phi_\zeta = A(\zeta;\lambda,\epsilon)\phi, \qquad \phi \in \C^4, \label{evprob}
\end{align}
with
\begin{align*}
A(\zeta;\lambda,\epsilon):=\begin{pmatrix} 0 & 0 & 1 & 0 \\ 0 & 0 & 0 & 1 \\ \frac{\lambda}{m^2} (1 - \ri \alpha) + (1 + \ri \dell) R_\tf & (1 + \ri \dell) R_\tf & -2\hat{z}_\tf & 0 \\ (1 - \ri \dell) R_\tf & \frac{\lambda}{m^2} (1 + \ri \alpha) + (1 - \ri \dell) R_\tf & 0 & -2\overline{\hat{z}_\tf} \end{pmatrix},
\end{align*}
where we suppressed the arguments of $R_\tf$, $\hat{z}_\tf$ and $m$ in the coefficient matrix $A(\zeta;\lambda,\epsilon)$. Thus, inspired by the above, we apply the linear coordinate transformation
\begin{align} \hat{\phi} = B(\zeta;\epsilon)\phi,
\qquad B(\zeta;\epsilon) := \begin{pmatrix} \beta & 0 & 0 & 0 \\ 0 & \overline{\beta} & 0 & 0 \\   \beta\hat{z}_\tf & 0 & \beta & 0 \\ 0 & \overline{\beta\hat{z}_\tf} & 0 & \overline{\beta} \end{pmatrix},\label{lincoordtransform}\end{align}
with
\begin{align*} \beta(\zeta;\epsilon) := \exp\left(\int_{\zeta_\tf(\epsilon)}^\zeta \left(\hat{z}_\tf(y;\epsilon) - \ri \hat{k}_\tf(\epsilon)\right)\mathrm d y\right),
\end{align*}
to the eigenvalue problem~\eqref{evprob} yielding the system
\begin{align}
\hat{\phi}_\zeta = A_*(\zeta;\lambda,\epsilon)\hat{\phi}, \qquad \hat{\phi} \in \C^4, \label{evprob2}
\end{align}
with
\begin{align*}
A_*(\zeta;\lambda,\epsilon) := \begin{pmatrix} -\ri \hat{k}_\tf & 0 & 1 & 0 \\ 0 & \ri \hat{k}_\tf & 0 & 1 \\ \frac{\lambda + 1 - \chi}{m^2} (1 - \ri \alpha) + \mu + 2(1 + \ri \dell) R_\tf & (1 + \ri \dell)\frac{\beta}{\overline{\beta}} R_\tf & - \ri \hat{k}_\tf & 0 \\ (1 - \ri \dell)\frac{\overline{\beta}}{\beta} R_\tf & \frac{\lambda + 1 - \chi}{m^2} (1 + \ri \alpha) + \overline{\mu} + 2(1 - \ri \dell) R_\tf & 0 & \ri \hat{k}_\tf \end{pmatrix},
\end{align*}
where we suppress the arguments of $R_\tf,k_\tf,\mu,m$ and $\chi$ in the coefficient matrix $A_*(\zeta;\lambda,\epsilon)$. We find by~Proposition~\ref{prop:right} that, to the right of the front interface, i.e.~for $\zeta \in [\zeta_\tf(\epsilon),\infty)$, system~\eqref{evprob2} is close to a constant coefficient system. More specifically, along the plateau state, i.e.~for $\zeta \in [\zeta_\tf(\epsilon),0]$, system~\eqref{evprob2} is approximated by
\begin{align}
\hat{\phi}_\zeta = A_{*,-}(\lambda,\epsilon)\hat{\phi}, \qquad \hat{\phi} \in \C^4, \label{evprobp}
\end{align}
whereas to the right of the inhomogeneity, i.e.~for $\zeta \in [0,\infty)$, system~\eqref{evprob2} is approximated by
\begin{align}
\hat{\phi}_\zeta = A_{*,+}(\lambda,\epsilon)\hat{\phi}, \qquad \hat{\phi} \in \C^4, \label{evprob2red2}
\end{align}
with
\begin{align}
A_{*,\pm}(\lambda,\epsilon) := \begin{pmatrix} \D(\epsilon) & I_2 \\ \A_\pm(\lambda,\epsilon) & \D(\epsilon)\end{pmatrix}, \label{defApm}
\end{align}
and
\begin{align} \D(\epsilon) := \begin{pmatrix} -\ri \hat{k}_\tf & 0  \\ 0 & \ri \hat{k}_\tf\end{pmatrix}, \qquad \A_\pm(\lambda,\epsilon) := \begin{pmatrix} \frac{\lambda + 1 \pm 1}{m^2} (1 - \ri \alpha) + \mu & 0 \\ 0 & \frac{\lambda + 1 \pm 1}{m^2} (1 + \ri \alpha) + \overline{\mu} \end{pmatrix}, \label{defApm2}
\end{align}
where we suppress the arguments of $k_\tf,\mu$ and $m$ in the coefficient matrices $\D(\epsilon)$ and $\A_\pm(\lambda,\epsilon)$. We emphasize that, as in~\eqref{redeiv3}, systems~\eqref{evprobp} and~\eqref{evprob2red2} decouple into two systems in $\C^2$ with constant coefficients, which will simplify the upcoming analysis significantly. We remark that our choice for adding the factor $\re^{-\ri \hat{k}_\tf(\epsilon) (\zeta - \zeta_\tf(\epsilon))}$ in $\beta(\zeta;\epsilon)$ is motivated by the fact that system~\eqref{evprob2} has asymptotically constant coefficients, which is convenient for the definition of the Riccati-Evans function later in~\S\ref{sec:ric2}, see also~\S\ref{sec:ric}. Indeed, by Theorem~\ref{t:ex02} we have
\begin{align*} \lim_{\zeta \to -\infty} \frac{{\beta(\zeta;\epsilon)}}{\overline{\beta(\zeta;\epsilon)}}  R_\tf(\zeta;\epsilon) = 1, \qquad \lim_{\zeta \to \infty} \frac{{\beta(\zeta;\epsilon)}}{\overline{\beta(\zeta;\epsilon)}}  R_\tf(\zeta;\epsilon) = 0.\end{align*}

We conclude this subsection with the following technical result providing control over $\beta(\zeta;\epsilon)$.

\begin{lemma}\label{lem:beta}
Let $\iota$ be as in Proposition~\ref{prop:rightleft}. Provided $0 < \epsilon_1 \ll \delta \ll 1$ and $0 \leq |\epsilon_2| \ll \delta \ll 1$, we have the estimate
\begin{align}
\left|\frac{\beta(\zeta;\epsilon)}{\overline{\beta(\zeta;\epsilon)}} - 1\right| \leq C_\delta \sqrt{\|\epsilon\|}\left|\log\|\epsilon\|\right|, \qquad \zeta \leq \zeta_\tf(\epsilon) - \iota \log\|\epsilon\|, \label{beta:approx}\end{align}
where $C_\delta > 0$ depends on $\delta$ only. Moreover, it holds
\begin{align} \beta(\zeta;\epsilon)\re^{\hat\kappa_- \zeta} \to 0 \text{ as } \zeta \to -\infty, \label{betalim-}\end{align}
and
\begin{align} \beta(\zeta;\epsilon)^{-1} \re^{-\hat\kappa_+ \zeta} \to 0 \text{ as } \zeta \to \infty,\label{betalim+}\end{align}
where $\hat\kappa_\pm$ are defined in~\eqref{defkappa}.
\end{lemma}
\begin{proof}
By~\eqref{preciseomegabound} in Theorem~\ref{t:ex02}, estimate~\eqref{estex1} from Proposition~\ref{prop:left} and Proposition~\ref{prop:rightleft} we have
\begin{align*}
\begin{split}
&\left|\int_{\zeta_\tf(\epsilon)}^{\zeta} \left(\Im\left[\hat{z}_\tf(\zeta;\epsilon)\right] - \hat{k}_\tf(\epsilon)\right) \de \zeta\right|\\ &\qquad \leq \int_{\frac{\log\|\epsilon\|}{\eta}}^{-\iota \log\|\epsilon\|} \left|\Im\left[\hat{z}_\tf(\zeta_\tf(\epsilon) + \zeta,\epsilon)\right] - \hat{k}_\tf(\epsilon)\right| \de \zeta + \int_{-\infty}^{\frac{\log\|\epsilon\|}{\eta}} \left|\Im\left[\hat{z}_\tf(\zeta_\tf(\epsilon) + \zeta),\epsilon)\right] - \hat{k}_\tf(\epsilon)\right| \de \zeta
\\ &\qquad \leq C_\delta \sqrt{\|\epsilon\|}\left|\log\|\epsilon\|\right|,\end{split}
\end{align*}
for $\zeta \leq \zeta_\tf(\epsilon) - \iota \log\|\epsilon\|$, provided $0 < \epsilon_1 \ll \delta \ll 1$ and $0 \leq |\epsilon_2| \ll \delta \ll 1$. Hence,~\eqref{beta:approx} follows.

Moreover, Theorem~\ref{t:ex02} yields that $\Re(\hat{z}_\tf(\zeta;\epsilon))$ converges exponentially to $1 - \tfrac{1}{2}\epsilon_1^2$ as $\zeta \to -\infty$. Since it holds $-\frac{1}{4} < -\kappa\sqrt{1+\alpha^2} \leq 0$, it follows $-\frac{1}{2} < \hat\kappa_- < 0$ by~\eqref{limmum}, provided $0 < \epsilon_1 \ll 1$ and $0 \leq |\epsilon_2| \ll 1$. Hence,~\eqref{betalim-} follows.

Finally,~$-\Re(\hat{z}_\tf(\zeta;\epsilon))$ converges by Proposition~\ref{prop:right} exponentially to
$$-\Re(\hat{z}_+(\epsilon)) = \hat \kappa_+ - \frac{\kappa \sqrt{1+\alpha^2}}{m(\epsilon)} + \frac{\epsilon_1^2}{2} - 1 < \hat \kappa_+$$
as $\zeta \to \infty$, which yields~\eqref{betalim+}, provided $0 < \epsilon_1 \ll 1$ and $0 \leq |\epsilon_2| \ll 1$.
\end{proof}

\subsection{The Riccati-Evans function} \label{sec:ric2}

Since the essential spectrum of the operator $\hat{\El}_\tf$, posed on $\smash{L^2_{\hat{\kappa}_-,\hat{\kappa}_+}(\R,\C^2)}$, is, by Theorem~\ref{concless}, not intersecting the regions $R_1$ and $R_2$, a point $\lambda \in R_1 \cup R_2$ lies in its point spectrum if and only if the associated eigenvalue problem~\eqref{evprob000} admits a nontrivial solution in $\smash{H^2_{\hat{\kappa}_-,\hat{\kappa}_+}(\R,\C^2)}$ or, equivalently, if and only if the first-order reformulation~\eqref{evprob} admits a nontrivial solution in $\smash{H^1_{\hat{\kappa}_-,\hat{\kappa}_+}(\R,\C^4)}$. In this subsection we will define the Riccati-Evans function, which locates the point spectrum in $R_1 \cup R_2$. Thus,~\eqref{evprob} admits a nontrivial solution in $\smash{H^1_{\hat{\kappa}_-,\hat{\kappa}_+}(\R,\C^4)}$ for $\lambda \in R_1 \cup R_2$ if and only if $\lambda$ is a root of the Riccati-Evans function.

\subsubsection{Construction of the Riccati-Evans function} \label{sec:ricdef}

We construct the Riccati-Evans function by applying the general procedure in~\S\ref{sec:ric} to our operator $\hat{\El}_\tf$, posed on $\smash{L^2_{\hat{\kappa}_-,\hat{\kappa}_+}(\R,\C^2)}$. By Theorem~\ref{concless}, there exists an open and bounded neighborhood $\Omega \subset \C$ of $R_1 \cup R_2$ that lies to the right of its essential spectrum. Thus, as in~\S\ref{sec:ric}, we evoke~\cite[Theorem 3.2]{SAN2} to conclude that the systems
\begin{align*}
 \phi_\zeta = \left(A(\zeta;\lambda,\epsilon) - \hat{\kappa}_\pm\right)\phi, \qquad \phi \in \C^4,
\end{align*}
have for each $\lambda \in \Omega$, exponential dichotomies on $(-\infty,0]$ and $[0,\infty)$, respectively, with associated rank $2$ projections $P_\pm(\zeta;\lambda,\epsilon)$, which depend analytically on $\lambda$.

Consider the coordinate chart $\mathfrak{c}$, which maps any $2$-dimensional subspace $W \in \mathrm{Gr}(2,\C^4)$ represented by a basis $\smash{\left(\begin{smallmatrix} X \\ Y \end{smallmatrix}\right) } \in \C^{4 \times 2}$ with $\det(X) \neq 0$ to the matrix $YX^{-1} \in \C^{2 \times 2}$. Choose bases
\begin{align*} \begin{pmatrix} X_\pm(\pm \zeta;\lambda,\epsilon)\\ Y_\pm(\pm \zeta;\lambda,\epsilon)\end{pmatrix} \in \C^{4 \times 2}, \qquad \zeta \geq 0, \lambda \in \Omega,\end{align*}
of the relevant subspaces $\ker(P_-(-\zeta;\lambda,\epsilon))$ and $P_+(\zeta;\lambda,\epsilon)[\C^{4}]$, respectively. Define $S_{\pm,\epsilon} \colon \R_\pm \times \Omega \to \C$ by $S_{\pm,\epsilon}(\pm \zeta;\lambda) := \det(X_\pm(\pm \zeta;\lambda,\epsilon))$. For $\zeta \geq 0$ and for $\lambda \in \Omega$ with $S_{\pm,\epsilon}(\pm \zeta;\lambda) \neq 0$, the subspaces $\ker(P_-(-\zeta;\lambda,\epsilon))$ and $P_+(\zeta;\lambda,\epsilon)[\C^{4}]$ in the Grassmannian $\mathrm{Gr}(2,\C^4)$ are, under the coordinate chart $\mathfrak c$, represented by
\begin{align*}T_\pm(\pm \zeta;\lambda,\epsilon) := Y_\pm(\pm \zeta;\lambda,\epsilon)X_\pm(\pm \zeta;\lambda,\epsilon)^{-1} \in \C^{2 \times 2}.\end{align*}

Define $\Es_{\epsilon} := \{\lambda \in \Omega : S_{+,\epsilon}(0;\lambda) = 0 \text{ or } S_{-,\epsilon}(0;\lambda) = 0\}$. The associated {Riccati-Evans function} $\E_{\epsilon} \colon \Omega \setminus \Es_{\epsilon} \to \C$ is then given by
\begin{align*} \E_{\epsilon}(\lambda) = \det\left(T_+(0;\lambda,\epsilon) - T_-(0;\lambda,\epsilon)\right). \end{align*}
It follows from Proposition~\ref{prop:ric} that $\E_\epsilon$ is meromorphic on $\Omega$, and the eigenvalue problem~\eqref{evprob} has a nontrivial solution in $\smash{H^1_{\hat{\kappa}_-,\hat{\kappa}_+}(\R,\C^4)}$ for some $\lambda \in \Omega \setminus \Es_{\epsilon}$ if and only if $\E_{\epsilon}(\lambda) = 0$. In addition, the multiplicity of a root $\lambda_0 \in \Omega \setminus \Es_{\epsilon}$ of $\E_{\epsilon}$ corresponds to the algebraic multiplicity of $\lambda_0$ as an eigenvalue of $\hat{\El}_\tf$.

\subsubsection{Invariance under the linear coordinate transform} \label{sec:invric}

We study the behavior of the Riccati-Evans function $\E_{\epsilon}$ under the linear coordinate transform~\eqref{lincoordtransform}, which transforms the eigenvalue problem~\eqref{evprob} into~\eqref{evprob2}. For all $\lambda \in \Omega$ we set
\begin{align}
\begin{split}
W_+(\zeta;\lambda,\epsilon) := B(\zeta;\epsilon)P_+(\zeta;\lambda,\epsilon)[\C^4], \qquad & \zeta \geq 0, \\
W_-(\zeta;\lambda,\epsilon) := B(\zeta;\epsilon)\ker(P_-(\zeta;\lambda,\epsilon)), \qquad & \zeta \leq \zeta_\tf(\epsilon),
\end{split} \label{subspaces}
\end{align}
Flowing these subspaces forward and backward in the linear system~\eqref{evprob2} leads then to subspaces $W_\pm(\zeta;\lambda,\epsilon) \in \mathrm{Gr}(2,\C^4)$ for each $\zeta \in \R$ and $\lambda \in \Omega$, which can be represented by $\hat{T}_\pm(\zeta;\lambda,\epsilon) = \hat{Y}_\pm(\zeta;\lambda,\epsilon)\hat{X}_\pm(\zeta;\lambda,\epsilon)^{-1} \in \C^{2 \times 2}$, where
\begin{align}\begin{pmatrix} \hat{X}_\pm \\ \hat{Y}_\pm\end{pmatrix}(\zeta;\lambda,\epsilon) \in \C^{4 \times 2},\label{bases}\end{align}
are bases of $W_\pm(\zeta;\lambda,\epsilon)$, as long as $\det(\hat{X}_\pm(\zeta;\lambda,\epsilon)) \neq 0$. We emphasize that $\det(X_\pm(\zeta;\lambda,\epsilon)) = 0 \Leftrightarrow \det(\hat{X}_\pm(\zeta;\lambda,\epsilon)) = 0$, since the structure of the coordinate transform~\eqref{lincoordtransform} yield $\det(\hat{X}_\pm(\zeta;\lambda,\epsilon)) = |\beta(\zeta;\epsilon)|^2 \det(X_\pm(\zeta;\lambda,\epsilon))$.

Thus, the relevant subspaces $\ker(P_-(0;\lambda,\epsilon))$ and $P_+(0;\lambda,\epsilon)[\C^4]$ in $\mathrm{Gr}(2,\C^4)$, represented by $T_{\pm}(0;\lambda,\epsilon) \in \C^{2 \times 2}$ under the coordinate chart $\mathfrak{c}$, are mapped by the linear coordinate transform~\eqref{lincoordtransform} onto subspaces $W_\pm(0;\lambda,\epsilon) \in \mathrm{Gr}(2,\C^4)$ represented by $\hat{T}_\pm(0;\lambda,\epsilon) \in \C^{2 \times 2}$ given by
\begin{align*} \hat{T}_\pm(0) &:= \left(\begin{pmatrix} \beta(0) & 0 \\ 0 & \overline{\beta}(0) \end{pmatrix}Y_\pm(0) + \begin{pmatrix} \beta(0) \hat{z}_\tf(0) & 0 \\ 0 & \overline{\beta(0) \hat{z}_\tf(0)} \end{pmatrix} X_\pm(0)\right) X_\pm^{-1}(0) \begin{pmatrix} \beta(0)^{-1} & 0 \\ 0 & \overline{\beta(0)^{-1}} \end{pmatrix}\\ &= \begin{pmatrix} \beta(0) & 0 \\ 0 & \overline{\beta(0)} \end{pmatrix}T_\pm(0)\begin{pmatrix} \beta(0)^{-1} & 0 \\ 0 & \overline{\beta(0)^{-1}} \end{pmatrix} + \begin{pmatrix} \hat{z}_\tf(0) & 0 \\ 0 & \overline{\hat{z}_\tf(0)} \end{pmatrix},\end{align*}
where we suppress the dependency on $\lambda$ and $\epsilon$. Therefore, the coordinate transform leaves the Riccati-Evans function $\E_\epsilon \colon \Omega \setminus \Es_\epsilon \to \C$ invariant, i.e.~it holds
\begin{align}
 \E_{\epsilon}(\lambda) = \det\left(\hat{T}_+(0;\lambda,\epsilon) - \hat{T}_-(0;\lambda,\epsilon)\right), \qquad \lambda \in \Omega \setminus \Es_\epsilon. \label{Eveq}
\end{align}

Of course, one can also choose to evaluate at the front interface $\zeta_\tf(\epsilon)$ and define the \emph{alternative Riccati-Evans function} $\tilde{\E}_\epsilon \colon \Omega \setminus \tilde{\Es}_\epsilon \to \C$, with $\tilde{\Es}_\epsilon := \{\lambda \in \Omega : \det(X_+(\zeta_\tf(\epsilon);\lambda,\epsilon)) = 0 \text{ or } \det(X_-(\zeta_\tf(\epsilon);\lambda,\epsilon)) = 0\}$, by
\begin{align*}
 \tilde{\E}_{\epsilon}(\lambda) = \det\left(T_+(\zeta_\tf(\epsilon);\lambda,\epsilon) - T_-(\zeta_\tf(\epsilon);\lambda,\epsilon)\right) = \det\left(\hat{T}_+(\zeta_\tf(\epsilon);\lambda,\epsilon) - \hat{T}_-(\zeta_\tf(\epsilon);\lambda,\epsilon)\right),
\end{align*}
Analogous to the Riccati-Evans function $\E_\epsilon$, the alternative Riccati-Evans function $\tilde \E_\epsilon$ is meromorphic on $\Omega$ and its roots coincide (including multiplicity) with the point spectrum of $\hat{\El}_\tf$ in $\Omega \setminus \tilde{\Es}_\epsilon$ (including algebraic multiplicity of the eigenvalues).

\section{Analysis in the region \texorpdfstring{$R_2$}{R2}} \label{sec:pointR2}

In this section we study the point spectrum or, equivalently, the roots of the Riccati-Evans function $\E_\epsilon$ in the region $R_2$. Recall from~\S\ref{sec:invric} that the Riccati-Evans function is invariant under the linear coordinate transform~\eqref{lincoordtransform}, and thus can be defined in terms of the subspaces $W_\pm(\zeta;\lambda,\epsilon)$ given by~\eqref{subspaces}. Control over these subspaces can be obtained through exponential dichotomies, which arise by perturbing from the limit $\|\epsilon\| \to 0$ in which the transformed eigenvalue problem~\eqref{evprob2} along the front reduces to two coupled Sturm-Liouville problems. In order to extend the exponential dichotomies across the front we prove, with the aid of an $L^2$-energy estimate, that the coupled Sturm-Liouville problem admits no eigenvalues in the region $R_2$. Thus, using the control provided by the exponential dichotomies, we can approximate the Riccati-Evans function and show that it possesses neither zeros nor poles in the region $R_2$, which, by Proposition~\ref{prop:ric}, precludes the existence of point spectrum in $R_2$.

\subsection{Exponential dichotomy to the left of the inhomogeneity}

We establish an exponential dichotomy for system~\eqref{evprob2} on $(-\infty,0]$.

\begin{proposition} \label{prop:expdileftR2}
Let $\Theta_2 > 1$ be given. There exists $\tau > 0$ such that, provided $0 < \epsilon_1 \ll \Theta_1 \ll \delta, \theta_2 \ll 1$ and $0 \leq |\epsilon_2| \ll \Theta_1 \ll \delta,\theta_2 \ll 1$, system~\eqref{evprob2} admits for each $\lambda \in R_1(\Theta_2)$ an exponential dichotomy on $(-\infty,0]$ with $\lambda$- and $\epsilon$-independent constants and projections $P_{*,l}(\zeta;\lambda,\epsilon)$ on $\C^4$ satisfying
\begin{align}
 \left\|P_{*,l}(0;\lambda,\epsilon) - Q_{l}(\lambda)\right\| \leq C_\delta\|\epsilon\|^\tau, \label{leftbound6}
\end{align}
where $Q_l(\lambda)$ is the spectral projection onto the stable eigenspace of the matrix
\begin{align*}
A_l^{\infty}(\lambda) = \begin{pmatrix} 0 & 0 & 1 & 0 \\ 0 & 0 & 0 & 1\\ \lambda (1 - \ri \alpha) & 0 & 0 & 0\\ 0 & \lambda(1 + \ri\alpha) & 0 & 0 \end{pmatrix},
\end{align*}
and $C_\delta > 1$ is a constant depending only on $\delta$.
\end{proposition}
\begin{proof} In this proof $C_\delta > 1$ denotes any constant, which depends on $\delta$ only.

Let $\lambda \in R_2(\theta_2,\Theta_1,\Theta_2)$. We wish to approximate the coefficient matrix $A_*(\zeta;\lambda,\epsilon)$ of~\eqref{evprob2} for $\zeta$ to the left of the inhomogeneity at $\zeta = 0$. On the one hand, by~\eqref{preciseomegabound} in Theorem~\ref{t:ex02}, estimate~\eqref{estex1} in Proposition~\ref{prop:left}, identity~\eqref{limmum}, Proposition~\ref{prop:rightleft} and Lemma~\ref{lem:beta}, we establish the estimate
\begin{align}
 \left\|A_*(\zeta_\tf(\epsilon) + \zeta;\lambda,\epsilon) - A_l(\zeta_\delta + \zeta;\lambda)\right\| \leq C_\delta \sqrt{\|\epsilon\|} \left|\log\|\epsilon\|\right|, \qquad \zeta \in (-\infty,-\iota\log\|\epsilon\|], \label{leftbound7}
\end{align}
for $0 < \epsilon_1 \ll \delta \ll 1$ and $0 \leq |\epsilon_2| \ll \delta \ll 1$, where we denote
\begin{align*}
A_l(\zeta;\lambda) := \begin{pmatrix} 0 & 0 & 1 & 0 \\ 0 & 0 & 0 & 1\\ \lambda (1 - \ri \alpha) + 2R_*(\zeta) & R_*(\zeta) & 0 & 0 \\ R_*(\zeta) & \lambda(1 + \ri\alpha) + 2R_*(\zeta) & 0 & 0 \end{pmatrix}.
\end{align*}
On the other hand, by estimate~\eqref{Rsbound} in Proposition~\ref{prop:left}, and estimate~\eqref{estex2} in Proposition~\ref{prop:right}, both $R_*(\zeta)$ and $R_\tf(\zeta;\epsilon)$ converge to $0$ at an $\epsilon$- and $\delta$-independent exponential rate $\eta > 0$. Hence, combining this with~\eqref{preciseomegabound} in Theorem~\ref{t:ex02} and identity~\eqref{limmum}, we find
\begin{align}
 \left\|A_*(\zeta_\tf(\epsilon) + \zeta;\lambda,\epsilon) - A_l(\zeta_\delta + \zeta;\lambda)\right\| \leq C_\delta \|\epsilon\|^{\min\{1,\iota \eta\}}, \qquad \zeta \in [-\iota\log\|\epsilon\|,0]. \label{leftbound5}
\end{align}

By estimate~\eqref{Rsbound} in Proposition~\ref{prop:left}, the coefficient matrix $A_l(\zeta;\lambda)$ converges exponentially to
\begin{align*}
A_l^{-\infty}(\lambda) = \begin{pmatrix} 0 & 0 & 1 & 0 \\ 0 & 0 & 0 & 1\\ \lambda (1 - \ri \alpha) + 2 & 1 & 0 & 0 \\ 1 & \lambda(1 + \ri\alpha) + 2 & 0 & 0 \end{pmatrix},
\end{align*}
as $\zeta \to -\infty$, which has the four eigenvalues
\begin{align*}-\sqrt{2 + \lambda \pm \sqrt{1 - \alpha^2 \lambda^2}}, \qquad \sqrt{2 + \lambda \pm \sqrt{1 - \alpha^2 \lambda^2}}.\end{align*}
Moreover, as $\zeta \to \infty$ the coefficient matrix $A_l(\zeta;\lambda)$ converges exponentially to
\begin{align*}
A_l^{\infty}(\lambda) = \begin{pmatrix} 0 & 0 & 1 & 0 \\ 0 & 0 & 0 & 1\\ \lambda (1 - \ri \alpha) & 0 & 0 & 0\\ 0 & \lambda(1 + \ri\alpha) & 0 & 0 \end{pmatrix},
\end{align*}
which has the four eigenvalues
\begin{align*}-\sqrt{\lambda (1 \pm \ri \alpha)}, \qquad \sqrt{\lambda (1 \pm \ri \alpha)}.\end{align*}
In the regime $0 < \Theta_1 \ll \delta, \theta_2 \ll 1$, the matrices $A_l^{\pm \infty}(\lambda)$ are hyperbolic for each $\lambda \in R_2(\theta_2,\Theta_1,\Theta_2)$ with a $\lambda$-uniform spectral gap (which might depend on $\delta$). Hence, by~\cite[Theorem 1]{SAN}, system
\begin{align}
\hat{\phi}_\zeta = A_l(\zeta;\lambda)\hat{\phi}, \qquad \hat{\phi} \in \C^4, \label{evprob2red4}
\end{align}
has for every $\lambda$ in the compact set $R_2(\theta_2,\Theta_1,\Theta_2)$ exponential dichotomies on $(-\infty,0]$ and $[0,\infty)$ with $\lambda$-independent constants (which might depend on $\delta$). An Evans function $\E_l \colon \Omega_2 \to \C$ associated with~\eqref{evprob2red4} is therefore well-defined and analytic on a small enough open and bounded neighborhood $\Omega_2$ of $R_2(\theta_2,\Theta_1,\Theta_2)$, cf.~\cite[Theorem 1]{SAN}.

The roots of the Evans function $\E_l$ correspond to those $\lambda \in \Omega_2$ at which~\eqref{evprob2red4} admits a nontrivial exponentially localized solution. Hence, by~\cite[Proposition 2.1]{PAL}, system~\eqref{evprob2red4} has an exponential dichotomy on $\R$ if and only if $\lambda \in R_2(\theta_2,\Theta_1,\Theta_2)$ is not a root of $\E_l$. We prove, using an $L^2$-energy estimate,  that~\eqref{evprob2red4} admits no nontrivial $L^2$-localized solution and therefore has an exponential dichotomy on $\R$ for each $\lambda \in R_2$. Let $\lambda \in R_2(\theta_2,\Theta_1,\Theta_2)$ and $\hat{\phi}(\zeta) = (\hat{w},\hat{y},\hat{v},\hat{u})(\zeta)$ be a nontrivial solution to~\eqref{evprob2red4} in $L^2(\R,\C^4)$. Then, the first and second component satisfy the coupled Sturm-Liouville problem
\begin{align*}
\begin{split}
\partial_\zeta^2 \hat{w} &= \lambda (1 - \ri \alpha)\hat{w} + R_* \left(2\hat{w} + \hat{y}\right),\\
\partial_\zeta^2 \hat{y} &= \lambda (1 + \ri \alpha)\hat{y} + R_* \left(\hat{w} + 2\hat{y}\right).
\end{split}
\end{align*}
Taking the $L^2$-inner product of the above equations with $\hat{w}$ and $\hat{y}$, respectively, and integrating by parts yields
\begin{align*}
\begin{split}
\lambda (1 - \ri \alpha)\|\hat{w}\|_2^2 = -\|\hat{v}\|_2^2  - \int_{-\infty}^\infty R_*(\zeta) \left(2|\hat{w}(\zeta)|^2 + \hat{y}(\zeta)\overline{\hat{w}(\zeta)}\right)\de \zeta,\\
\lambda (1 + \ri \alpha)\|\hat{y}\|_2^2 = -\|\hat{u}\|_2^2  - \int_{-\infty}^\infty R_*(\zeta) \left(2|\hat{y}(\zeta)|^2 + \hat{w}(\zeta)\overline{\hat{y}(\zeta)}\right)\de \zeta.\\
\end{split}
\end{align*}
We add both equations and take imaginary parts to obtain
\begin{align}
\Im(\lambda)\left(\|\hat{w}\|_2^2 + \|\hat{y}\|_2^2\right) = \Re(\lambda)\alpha \left(\|\hat{w}\|_2^2 - \|\hat{y}\|_2^2\right). \label{idim}
\end{align}
Hence, as $\lambda \neq 0$ and $\hat{\phi} \in L^2(\R,\C^4)$ is a nontrivial solution to~\eqref{evprob2red4}, it must hold $\Re(\lambda) \neq 0$. So, adding both equations again, taking real parts now, we establish using Young's inequality and~\eqref{idim}
\begin{align*}
\frac{|\lambda|^2}{\Re(\lambda)}\left(\|\hat{w}\|_2^2 + \|\hat{y}\|_2^2\right) &= \Re(\lambda)\left(\|\hat{w}\|_2^2 + \|\hat{y}\|_2^2\right) + \Im(\lambda)\alpha \left(\|\hat{w}\|_2^2 - \|\hat{y}\|_2^2\right)\\ &= -\|\hat{v}\|_2^2  - \|\hat{u}\|_2^2 - \int_{-\infty}^\infty R_*(\zeta) \left(2|\hat{w}(\zeta)|^2 + 2\Re\left(\hat{y}(\zeta)\overline{\hat{w}(\zeta)}\right) + 2|\hat{y}(\zeta)|^2\right)\de \zeta \leq 0.
\end{align*}
Since $\hat{\phi} \in L^2(\R,\C^4)$ is a nontrivial solution to~\eqref{evprob2red4}, it follows $\Re(\lambda) < 0$. Hence, upon taking $\theta_2 > 0$ sufficiently small, we derive a contradiction, since the set of $\lambda \in R_2(\theta_2,\Theta_1,\Theta_2)$ at which~\eqref{evprob2red4} admits an $L^2$-localized solution corresponds to the isolated roots of the analytic Evans function $\E_l$.

Using~\cite[Theorem 1]{SAN}, we conclude that~\eqref{evprob2red4} has an exponential dichotomy on $\R$ for each $\lambda$ in the compact set $R_2(\theta_2,\Theta_1,\Theta_2)$ with $\lambda$-independent constants (which do depend on $\delta$) and associated projections $P_l(\zeta;\lambda)$. By~\eqref{Rsbound} in Proposition~\ref{prop:left}, there exist constants $C,\eta > 0$ such that
\begin{align*} \|A_l(\zeta;\lambda) - A_l^\infty(\lambda)\| \leq C\re^{-\eta \zeta}, \qquad \zeta \geq 0.\end{align*}
Hence, by~\cite[Lemma 3.4]{PAL}, the dichotomy projections satisfy
\begin{align} \|P_l(\zeta;\lambda) - Q_l(\lambda)\| \leq C_\delta \re^{-\eta \zeta}, \qquad \zeta \geq 0, \label{projboundleft}\end{align}
where $Q_l(\lambda)$ is the spectral projection onto the stable eigenspace of $A_l^\infty(\lambda)$.

Take $0 < \tau < \min\{\iota \eta, \frac{1}{2}\}$. By roughness of exponential dichotomies, cf.~\cite[Proposition 5.1]{COP} and estimates~\eqref{leftbound7} and~\eqref{leftbound5}, system~\eqref{evprob2} admits for each $\lambda \in R_2(\theta_2,\Theta_1,\Theta_2)$ an exponential dichotomy on $(-\infty,0]$ with $\lambda$- and $\epsilon$-independent constants and projections $P_{*,l}(\zeta;\lambda,\epsilon)$ satisfying
\begin{align*}
 \left\|P_{*,l}(\zeta_\tf(\epsilon) + \zeta;\lambda,\epsilon) - P_l(\zeta_\delta + \zeta;\lambda)\right\| \leq C_\delta \|\epsilon\|^\tau, \qquad \zeta \in (-\infty,-\zeta_\tf(\epsilon)],
\end{align*}
provided $0 < \epsilon_1 \ll \Theta_1 \ll \delta, \theta_2 \ll 1$ and $0 \leq |\epsilon_2| \ll \Theta_1 \ll \delta,\theta_2 \ll 1$. Finally, combining the latter with~\eqref{projboundleft} yields~\eqref{leftbound6}.
\end{proof}

\subsection{Exponential dichotomy to the right of the inhomogeneity}

We show that~\eqref{evprob2} admits an exponential dichotomy on $[0,\infty)$.

\begin{proposition} \label{prop:expdirightR2}
There exists a constant $C > 1$ such that, provided $0 < \epsilon_1 \ll \Theta_1 \ll \delta,\theta_2 \ll 1$ and $0 \leq |\epsilon_2| \ll \Theta_1 \ll \delta,\theta_2 \ll 1$, system~\eqref{evprob2} admits for each $\lambda \in R_2(\theta_2,\Theta_1,\Theta_2)$ an exponential dichotomy on $[0,\infty)$ with $\lambda$- and $\epsilon$-independent constants and projections $P_{*,r}(\zeta;\lambda,\epsilon)$ on $\C^4$ satisfying
\begin{align}
 \left\|P_{*,r}(\zeta;\lambda,\epsilon) - Q_r(\lambda)\right\| &\leq C\|\epsilon\|, \qquad \zeta \geq 0, \label{projboundright2}
\end{align}
where $Q_{r}(\lambda)$ is the spectral projection onto the stable eigenspace of the matrix
\begin{align*}A_r^\infty(\lambda) = \begin{pmatrix} 0 & 0 & 1 & 0 \\ 0 & 0 & 0 & 1\\ \left(\lambda + 2\right) (1 - \ri \alpha) & 0 & 0 & 0 \\ 0 &  \left(\lambda + 2\right)(1+\ri\alpha) & 0 & 0 \end{pmatrix}.\end{align*}
\end{proposition}
\begin{proof}
Using~\eqref{preciseomegabound} in Theorem~\ref{t:ex02},~\eqref{estex2} in Proposition~\ref{prop:right} and identity~\eqref{limmum} we establish the estimate
\begin{align}
 \left\|A_*(\zeta;\lambda,\epsilon) - A_r^\infty(\lambda)\right\| \leq C\|\epsilon\|, \qquad \zeta \geq \zeta_\tf(\epsilon), \label{rightbound2}
\end{align}
The eigenvalues of $A_r^\infty(\lambda)$ are given by
\begin{align*}-\sqrt{\left(\lambda + 2\right) (1 \pm \ri \alpha)}, \qquad \sqrt{\left(\lambda + 2\right) (1 \pm \ri \alpha)}.\end{align*}
Thus, in the regime $0 < \Theta_1 \ll \delta, \theta_2 \ll 1$, the matrix $A_r^\infty(\lambda)$ is hyperbolic for each $\lambda \in R_2(\theta_2,\Theta_1,\Theta_2)$ with a $\lambda$- and $\delta$-uniform spectral gap. Hence, system
\begin{align*}
\hat{\phi}_\zeta = A_r^\infty(\lambda)\hat{\phi}, \qquad \hat{\phi} \in \C^4,
\end{align*}
has an exponential dichotomy on $\R$ for each $\lambda$ in the compact set $R_2(\theta_2,\Theta_1,\Theta_2)$ with $\lambda$- and $\delta$-independent constants. The associated projection $Q_r(\lambda)$ coincides with the spectral projection onto the stable eigenspace of $A_r^\infty(\lambda)$. By roughness of exponential dichotomies~\cite[Proposition 5.1]{COP} and estimate~\eqref{rightbound2}, the eigenvalue problem~\eqref{evprob2} admits for each $\lambda \in R_2(\theta_2,\Theta_1,\Theta_2)$ exponential dichotomies on $[0,\infty)$ with $\lambda$-, $\delta$- and $\epsilon$-independent constants and projections $P_{*,r}(\zeta;\lambda,\epsilon)$ satisfying~\eqref{projboundright2}, provided $0 < \epsilon_1 \ll \Theta_1 \ll \delta,\theta_2 \ll 1$ and $0 \leq |\epsilon_2| \ll \Theta_1 \ll \delta,\theta_2 \ll 1$.
\end{proof}

\subsection{Conclusion} \label{sec:R2concl}

In this subsection we complete our spectral study of the operator $\hat{\El}_\tf$, posed on $\smash{L^2_{\hat{\kappa}_-,\hat{\kappa}_+}(\R,\C^2)}$, in the region $R_2$. We prove that the associated Riccati-Evans function $\E_\epsilon$ is analytic on $R_2$ and does not vanish. Recall from~\S\ref{sec:invric} that $\E_\epsilon$ can be defined in terms of the subspaces $W_\pm(\zeta;\lambda,\epsilon)$ given by~\eqref{subspaces}. The following lemma shows that these subspaces must coincide with the relevant subspaces of the exponential dichotomies for the transformed eigenvalue problem~\eqref{evprob2}, which were established in Propositions~\ref{prop:expdileftR2} and~\ref{prop:expdirightR2}.

\begin{lemma} \label{lem:subspace}
Let $\Theta_2 > 0$ be given. Provided $0 < \epsilon_1 \ll \Theta_1 \ll \delta, \theta_2 \ll 1$ and $0 \leq |\epsilon_2| \ll \Theta_1 \ll \delta,\theta_2 \ll 1$, it holds
\begin{align} W_-(\zeta;\lambda,\epsilon) &= \ker(P_{*,l}(\zeta;\lambda,\epsilon)), \qquad \zeta \in (-\infty,0], \, \lambda \in R_2(\theta_2,\Theta_1,\Theta_2), \label{expdieq1}\end{align}
where $P_{*,l}(\zeta;\lambda,\epsilon)$ is the projection associated with exponential dichotomy of~\eqref{evprob2} on $(-\infty,0]$ established in Proposition~\ref{prop:expdileftR2}. In addition, we have
\begin{align} W_+(\zeta;\lambda,\epsilon) &= P_{*,r}(\zeta;\lambda,\epsilon)[\C^4], \qquad \zeta \in [0,\infty), \, \lambda \in R_2(\theta_2,\Theta_1,\Theta_2), \label{expdieq2}\end{align}
where $P_{*,r}(\zeta;\lambda,\epsilon)$ is the projection associated with exponential dichotomy of~\eqref{evprob2} on $[0,\infty)$ established in Proposition~\ref{prop:expdirightR2}.
\end{lemma}
\begin{proof}
Let $\phi(\zeta)$ be a solution to~\eqref{evprob} with $\phi(\zeta) \in \ker(P_-(\zeta;\lambda,\epsilon))$ for $\zeta \leq 0$. Then, it holds $\phi(\zeta)\re^{-\hat \kappa_- \zeta} \to 0$ as $\zeta \to -\infty$. Hence, using~\eqref{betalim-} in Lemma~\ref{lem:beta} and using the fact that $\hat{z}_\tf(\zeta;\epsilon)$ is bounded on $(-\infty,\zeta_\tf(\epsilon)]$ by Proposition~\ref{prop:left}, the solution $B(\zeta;\epsilon){\phi}(\zeta)$ to~\eqref{evprob2} converges to $0$ as $\zeta \to -\infty$. Hence, we conclude that solutions to~\eqref{evprob2} in the subspace $W_-(\zeta;\lambda,\epsilon) = B(\zeta;\epsilon)\ker(P_-(\zeta;\lambda,\epsilon))$ converge to $0$ as $\zeta \to -\infty$, which yields~\eqref{expdieq1} by a simple dimension counting argument, using Proposition~\ref{prop:expdileftR2}.

Now let $\hat{\phi}(\zeta)$ be a solution to~\eqref{evprob2} that converges to $0$ as $\zeta \to \infty$. Then, using~\eqref{betalim+} in Lemma~\ref{lem:beta} and using the fact that $\hat{z}_\tf(\zeta;\epsilon)$ is bounded on $[0,\infty)$ by Proposition~\ref{prop:right}, it follows $B(\zeta,\epsilon)^{-1} \hat{\phi}(\zeta) \re^{-\hat \kappa_+ \zeta} \to 0$ as $\zeta \to \infty$. Hence, we conclude $B(\zeta,\epsilon)^{-1} \hat{\phi}(\zeta)$ is a solution to~\eqref{evprob}, which lies in $P_+(\zeta;\lambda,\epsilon)[\C^{4}]$ for $\zeta \geq 0$. So,~\eqref{expdieq2}
follows again by counting dimensions, using Proposition~\ref{prop:expdirightR2}.
\end{proof}

We are now in the position to establish that there is no point spectrum of the operator $\hat{\El}_\tf$ in $R_2$.

\begin{theorem} \label{conclR2}
Let $\Theta_2 > 0$ be given. Provided $0 < \epsilon_1 \ll \Theta_1 \ll \delta, \theta_2 \ll 1$ and $0 \leq |\epsilon_2| \ll \Theta_1 \ll \delta,\theta_2 \ll 1$, the operator $\hat{\El}_\tf$, posed on $\smash{L^2_{\hat{\kappa}_-,\hat{\kappa}_+}(\R,\C^2)}$, has no point spectrum in $R_2(\theta_2,\Theta_1,\Theta_2)$.
\end{theorem}
\begin{proof} In this proof we denote by $C_\delta > 1$ any constant which depends on $\delta$ only. We employ the notation and concepts introduced in~\S\ref{sec:ric2}

We start by proving that the Riccati-Evans function admits no poles in the region $R_2$. Consider the spectral projections $Q_l(\lambda)$ and $Q_r(\lambda)$ from Propositions~\ref{prop:expdileftR2} and~\ref{prop:expdirightR2}, respectively. First, observe that the kernel of $Q_l(\lambda)$ and the range of $Q_r(\lambda)$ have the bases
\begin{align*}
 \Phi_l(\lambda) := \begin{pmatrix}  1 & 0\\  0 & 1\\  \sqrt{\lambda(1-\ri \alpha)} & 0 \\ 0 & \sqrt{\lambda(1+\ri \alpha)} \end{pmatrix}, \qquad  \Phi_r(\lambda) := \begin{pmatrix}  1 & 0\\  0 & 1\\  -\sqrt{(\lambda+2)(1-\ri \alpha)} & 0 \\ 0 & -\sqrt{(\lambda+2)(1+\ri \alpha)} \end{pmatrix},
\end{align*}
respectively. Employing Propositions~\ref{prop:expdileftR2} and~\ref{prop:expdirightR2} and Lemma~\ref{lem:subspace}, we find that $\Psi_l(\lambda,\epsilon) := (I_4 - P_{*,l}(0;\lambda,\epsilon))\Phi_l(\lambda)$ and $\Psi_r(\lambda,\epsilon) := P_{*,r}(0;\lambda,\epsilon)\Phi_r(\lambda)$ are bases of $W_-(0;\lambda,\epsilon)$ and $W_+(0;\lambda,\epsilon)$, respectively, satisfying
\begin{align} \|\Psi_l(\lambda,\epsilon) - \Phi_l(\lambda)\|, \|\Psi_r(\lambda,\epsilon) - \Phi_r(\lambda)\| \leq C_\delta\|\epsilon\|^\tau, \qquad \lambda \in R_2(\theta_2,\Theta_1,\Theta_2).\label{basesest}\end{align}
We denote by $\hat{X}_\pm(\lambda,\epsilon)$ the upper $2\times2$-block of $\Psi_l(\lambda,\epsilon)$ and $\Psi_r(\lambda,\epsilon)$, respectively. By~\eqref{basesest} it follows
\begin{align*} \left|\det(\hat{X}_\pm(\lambda,\epsilon)) - 1\right| \leq C_\delta \|\epsilon\|^\tau, \qquad \lambda \in R_2(\theta_2,\Theta_1,\Theta_2),\end{align*}
yielding $\det(\hat{X}_\pm(\lambda,\epsilon)) \neq 0$. Hence, by~\eqref{Eveq}, $\E_\epsilon$ admits no poles in the region $R_2$.

Employing~\eqref{basesest}, we approximate
\begin{align*}
\begin{split}
\left\|\hat{T}_-(0;\lambda,\epsilon) - \begin{pmatrix} \sqrt{\lambda(1-\ri \alpha)} & 0 \\ 0 & \sqrt{\lambda(1+\ri \alpha)} \end{pmatrix}\right\| &\leq C_\delta \|\epsilon\|^\tau,\\ \left\|\hat{T}_+(0;\lambda,\epsilon) + \begin{pmatrix} \sqrt{(\lambda+2)(1-\ri \alpha)} & 0 \\ 0 & \sqrt{(\lambda+2)(1+\ri \alpha)} \end{pmatrix}\right\| &\leq C_\delta\|\epsilon\|^\tau,
\end{split}
\end{align*}
and, by~\eqref{Eveq}, we establish
\begin{align*}
\left\|\E_{\epsilon}(\lambda) - \left(\sqrt{\lambda(1-\ri \alpha)} + \sqrt{(\lambda+2)(1-\ri \alpha)}\right)\left(\sqrt{\lambda(1+\ri \alpha)} + \sqrt{(\lambda+2)(1+\ri \alpha)}\right)\right\| \leq C_\delta\|\epsilon\|^\tau,
\end{align*}
for $\lambda \in R_2(\theta_2,\Theta_1,\Theta_2)$. Therefore, the Riccati-Evans function $\E_{\epsilon}$, associated with the operator $\hat{\El}_\tf$, does not vanish on $R_2$. The result now follows directly from Proposition~\ref{prop:ric}.
\end{proof}

\section{Analysis in the region \texorpdfstring{$R_1$}{R1}} \label{sec:pointR1} 

\subsection{Approach} \label{sec:approachpoint}

In this section we locate the point spectrum in the region $R_1$ using the Riccati-Evans function $\E_\epsilon$. Recall from~\S\ref{sec:invric} that $\E_\epsilon$ is invariant under the linear coordinate transform~\eqref{lincoordtransform}, and thus can be defined in terms of the subspaces $W_\pm(\zeta;\lambda,\epsilon)$ of solutions to the transformed eigenvalue problem~\eqref{evprob2}, which were given by~\eqref{subspaces}.

In the region $R_1$, system~\eqref{evprob2} has exponential dichotomies to the left of the front interface at $\zeta = \zeta_\tf(\epsilon)$ and to the right of the inhomogeneity at $\zeta = 0$. However, along the intermediate plateau state of the front, i.e.~for $\zeta \in [\zeta_\tf(\epsilon),0]$, the eigenvalue problem loses hyperbolicity and the control over the relevant subspaces $W_\pm(\zeta;\lambda,\epsilon)$ through exponential dichotomies is lost. Indeed, all points $\lambda \in R_1$ lie close to the absolute spectrum of the plateau state and, thus, spatial eigenvalues cannot be separated uniformly. We regain control by observing that the eigenvalue problem is asymptotically close to a diagonal constant-coefficient system along the plateau state. Consequently, the leading-order dynamics in the matrix Riccati equation admits an invariant subset of diagonal solutions on which the flow is given by two scalar Riccati equations, which can be explicitly solved using Riemann-surface unfolding and a M\"obius transformation. We find that the relevant solutions to the scalar Riccati equations have, when evaluated at $\zeta = 0$, a discrete family of poles in $\lambda$ that accumulate on the absolute spectrum of the plateau state as $\epsilon_1 \searrow 0$. We prove that all these poles are confined to the open left-half plane except two of them that reside $\ord(\epsilon_1^3)$-close to the origin. With the aid of the superposition principle, we perturb from the invariant subset of diagonal solutions and, for $\lambda \in R_1$ lying $\ord(\epsilon_1^3)$-away from the poles, we establish sufficient control over the relevant trajectories in the matrix Riccati equation along the plateau state. As a result, for $\lambda \in R_1$ lying $\ord(\epsilon_1^3)$-away from the poles, we can approximate the Riccati-Evans function and prove that it admits neither zeros nor poles.

All in all, we will obtain that the closed right-half plane, except for a disk $D_1(\epsilon)$ centered at the origin with a radius of order $\ord(\epsilon_1^3)$, contains no point spectrum. To locate the point spectrum in $D_1(\epsilon)$ we proceed as follows. First, we establish that $\lambda = 0$ is a simple root of the Riccati-Evans function $\E_\epsilon$ with $\E_\epsilon'(0) > 0$, where we exploit that explicit solutions to the eigenvalue problem at $\lambda = 0$ arise through gauge and (almost) translational invariance. Second, we compute that the winding number of the meromorphic Riccati-Evans function on a contour enclosing the disk $D_1(\epsilon)$ equals $0$, which proves that its number of poles equals its number of zeros (including multiplicity). Then, we write the Riccati-Evans function, as in~\eqref{relEvans}, as a quotient of two analytic functions. We find, again using a winding number computation, that the number of zeros of its denominator equals $2$. Hence, because we cannot exclude zero-pole cancellation, we find that $\E_\epsilon$ has either one or two roots in $D_1(\epsilon)$. In case it has two roots, we use a parity argument, using $\E_\epsilon'(0) > 0$ and $\E_\epsilon(\lambda)$ is real for $\lambda \in \R$, to show that the other root must be real and negative. So, using Proposition~\ref{prop:ric}, we conclude that the point spectrum of $\hat{\El}_\tf$, posed on $\smash{L^2_{\hat{\kappa}_-,\hat{\kappa}_+}(\R,\C^2)}$, is contained in the open left-half plane, except for an algebraically simple eigenvalue residing at the origin.

The set-up of this section is as follows. In~\S\ref{sec:lambda0} we study the eigenvalue problem~\eqref{evprob} at $\lambda = 0$ and obtain the relevant solutions that arise due to gauge and (almost) translational invariance. In~\S\ref{sec:expdi1} and~\S\ref{sec:expdi2} we obtain exponential dichotomies for~\eqref{evprob2} to the left of the front interface and to the right of the inhomogeneity. In~\S\ref{sec:track} we then track the relevant subspaces along the plateau state, and vice versa, within the associated matrix Riccati equations. In~\S\ref{sec:disk} we deduce that the critical point spectrum in the region $R_1$ must be contained in the disk $D_1(\epsilon)$. Then, we make preparations for the final parity argument in~\S\ref{sec:parity}: we approximate the derivative $\E_\epsilon'(0)$ in~\S\ref{sec:ricderiv2}, we perform the necessary winding number computations in~\S\ref{sec:wind} and prove that the Riccati-Evans function $\E_\epsilon$ is real for real $\lambda$ in~\S\ref{sec:restr}.

\subsection{The eigenvalue problem at \texorpdfstring{$\lambda = 0$}{lambda = 0}} \label{sec:lambda0}

Due to gauge symmetry of the cGL equation~\eqref{e:cgl0}, $\lambda = 0$ is an eigenvalue of the operator $\hat{\El}_\tf$, posed on $\smash{L^2_{\hat{\kappa}_-,\hat{\kappa}_+}(\R,\C^2)}$. In this subsection we compute the associated eigenfunction and find two additional solutions to the eigenvalue problem at $\lambda = 0$, which arise through translational invariance of the homogeneous cGL equation~\eqref{e:cgl0} with $\chi$ constant.

It follows directly from gauge invariance of the cGL equation~\eqref{e:cgl0} that the time derivative of its pattern-forming front solution $\re^{\ri \omega_\tf t} A_\tf(x - ct)$ satisfies the associated variational equation. Switching to a co-moving frame and polar coordinates~\eqref{polar}, this yields the element $(r_\tf,-r_\tf)$ of the kernel of $\El_\tf$, when posed on the space $L^2_\kappa(\R,\C^2)$. We emphasize that $(r_\tf,-r_\tf)$ is not localized and, therefore, not an element of the kernel of $\El_\tf$, when posed on the space $L^2(\R,\C^2)$.

We note that, due to the spatially inhomogeneous term $\chi$, solutions to the cGL equation~\eqref{e:cgl0} are not translational invariant. However, upon switching to the co-moving frame, we find that $\chi(\xi)$ is constant except for a jump at $\xi = 0$. Hence, the spatial derivative $\partial_\xi A_\tf(\xi)$ of the front solution $A_\tf(\zeta)$ to~\eqref{e:TW2} is a solution to the associated variational equation, which is non-smooth at $\xi = 0$ only, where its derivative makes a jump. Switching to polar coordinates~\eqref{polar} again, yields the formal element $(r_\tf' + \ri \phi_\tf' r_\tf, r_\tf' - \ri \phi_\tf' r_\tf)$ of the kernel of $\El_\tf$.

Subsequently, we apply the rescaling and reparameterization from~\S\ref{sec:rep} to the obtained (formal) elements of the kernel of $\El_\tf$. We find the element $(1,-1)$ and the formal element $$\left(\hat{z}_\tf(\zeta;\epsilon) - 1 + \tfrac{1}{2}\epsilon_1^2, \overline{\hat{z}_\tf(\zeta;\epsilon)} - 1 + \tfrac{1}{2}\epsilon_1^2\right),$$
in the kernel of $\hat{\El}_\tf$, when posed on the space $\smash{L^2_{\hat{\kappa}_-,\hat{\kappa}_+}(\R,\C^2)}$. Thus, the eigenvalue problem~\eqref{evprob} at $\lambda = 0$, which reads
\begin{align}
\phi_\zeta = A(\zeta;0,\epsilon) \phi, \label{variationaleq}
\end{align}
admits the nontrivial solution $\phi_0 \in L^2_{\hat{\kappa}_-,\hat{\kappa}_+}(\R,\C^4)$ given by
\begin{align} \phi_0(\zeta) = \left(1, -1, 0, 0\right)^\top. \label{eigenfunction0} \end{align}
In addition, we find $C^1$-solutions $\phi_\pm \colon \R \to \C^4$ to~\eqref{variationaleq} satisfying
\begin{align} \phi_\pm(\pm \zeta;\epsilon) = \begin{pmatrix} \hat{z}_\tf(\pm \zeta;\epsilon) - 1 + \frac{\epsilon_1^2}{2} \\  \overline{\hat{z}_\tf(\pm \zeta;\epsilon)} - 1 + \frac{\epsilon_1^2}{2}\\ \partial_\zeta \hat{z}_\tf(\pm \zeta;\epsilon)\\ \overline{\partial_\zeta \hat{z}_\tf(\pm \zeta;\epsilon)}\end{pmatrix},\qquad \zeta > 0. \label{eigenfunction1}\end{align}

Finally, after applying the linear coordinate transform $B(\zeta;\epsilon)$, given by~\eqref{lincoordtransform}, the two solutions $\phi_0(\zeta)$ and $\phi_-(\zeta;\epsilon)$ to~\eqref{variationaleq} yield two linearly independent solutions $\hat{\phi}_{0}(\zeta;\epsilon) := B(\zeta;\epsilon)\phi_0(\zeta)$ and $\hat{\phi}_-(\zeta;\epsilon) := B(\zeta;\epsilon)\phi_{-}(\zeta;\epsilon)$ satisfying
\begin{align}
\hat{\phi}_{0}(\zeta;\epsilon) = \begin{pmatrix} \beta  \\ -\overline{\beta}  \\ \beta \hat{z}_\tf \\ -\overline{\beta\hat{z}_\tf}\end{pmatrix}, \qquad \hat{\phi}_-(\zeta;\epsilon) =
\begin{pmatrix} \beta\left(\hat{z}_\tf - 1 + \frac{\epsilon_1^2}{2}\right) \\  \overline{\beta}\left(\overline{\hat{z}_\tf} - 1 + \frac{\epsilon_1^2}{2}\right)\\ \beta\left(\left(\frac{\epsilon_1^2}{2} - 1\right)\hat{z}_\tf + \mu + \left(1+\ri \dell\right) R_\tf\right) \\ \overline{\beta}\left(\left(\frac{\epsilon_1^2}{2} - 1\right)\overline{\hat{z}_\tf} + \overline{\mu} + \left(1-\ri \dell\right) R_\tf\right)\end{pmatrix}, \qquad \zeta < 0, \label{twosolstrans}
\end{align}
to the transformed eigenvalue problem~\eqref{evprob2} at $\lambda = 0$, where we suppressed the arguments on the right hand sides and we used that $\Psi_\tf(\zeta;\epsilon) = (\hat{z}_\tf,R_\tf)(\zeta;\epsilon)$ satisfies equation~\eqref{exprob}.

The next result now follows immediately from Propositions~\ref{prop:left},~\ref{prop:right} and Lemma~\ref{lem:beta}.

\begin{lemma} \label{lem:approxphipm}
There exists a constant $C > 1$ such that, provided $0 < \epsilon_1 \ll 1$ and $0 \leq |\epsilon_2| \ll 1$, the solutions $\phi_\pm(\zeta;\epsilon)$ to~\eqref{variationaleq} are bounded as $\zeta \to \pm \infty$ and it holds
\begin{align*}
\left\|\phi_\pm(0;\epsilon) + \begin{pmatrix} \sqrt{2(1-\ri\alpha)} + 1 \\ \sqrt{2(1+\ri\alpha)} + 1 \\ (1 \mp 1)(1-\ri\alpha) \\ (1 \mp 1)(1+\ri\alpha)\end{pmatrix} \right\| \leq C\|\epsilon\|. \end{align*}
Moreover, the solutions $\hat{\phi}_0(\zeta;\epsilon)$ and $\hat{\phi}_-(\zeta;\epsilon)$ to~\eqref{evprob2} at $\lambda = 0$ converge to $0$ as $\zeta \to -\infty$.
\end{lemma}

\subsection{Exponential dichotomies to the left of the front interface} \label{sec:expdi1}

We show that~\eqref{evprob2} admits an exponential dichotomy on $(-\infty,\zeta_\tf(\epsilon)]$.

\begin{proposition} \label{prop:expdileftR1}
Provided $0 < \epsilon_1 \ll \Theta_1 \ll \delta \ll 1$ and $0 \leq |\epsilon_2| \ll \Theta_1 \ll \delta \ll 1$, system~\eqref{evprob2} admits for each $\lambda \in R_1(\Theta_1)$ an exponential dichotomy on $(-\infty,\zeta_\tf(\epsilon)]$ with $\lambda$- and $\epsilon$-independent constants and projections $P_{*,l}(\zeta;\lambda,\epsilon)$ on $\C^4$ satisfying
\begin{align} \ker(P_{*,l}(\zeta;0,\epsilon)) = \mathrm{Span}\left\{\hat{\phi}_0(\zeta;\epsilon), \hat{\phi}_-(\zeta;\epsilon)\right\}, \qquad \zeta \in (-\infty,\zeta_\tf(\epsilon)],\label{leftapprox2}\end{align}
where $\hat{\phi}_0(\zeta;\epsilon)$ and $\hat{\phi}_-(\zeta;\epsilon)$ are defined in~\eqref{twosolstrans}. In addition, we have the estimate
\begin{align}
 \left\|P_{*,l}(\zeta_\tf(\epsilon);\lambda,\epsilon) - Q_{*,l}(0)\right\| \leq C_\delta\left(\sqrt{\|\epsilon\|} \left|\log\|\epsilon\|\right| + |\lambda|\right), \label{leftbound4}
\end{align}
where $Q_{*,l}(0)$ is a projection on $\C^4$ with
\begin{align} \ker(Q_{*,l}(0)) = \mathrm{Span}\left\{\begin{pmatrix} 1 \\ -1 \\ \hat{z}_*(\zeta_\delta) \\ -\hat{z}_*(\zeta_\delta)\end{pmatrix},\begin{pmatrix} \hat{z}_*(\zeta_\delta) - 1 \\ \hat{z}_*(\zeta_\delta) - 1 \\ \delta - \hat{z}_*(\zeta_\delta) \\ \delta - \hat{z}_*(\zeta_\delta)\end{pmatrix}\right\}, \label{leftapprox}\end{align}
and $C_\delta > 1$ is a constant depending only on $\delta$.
\end{proposition}
\begin{proof}
Throughout this proof $C_\delta > 1$ is a constant depending on $\delta > 0$ only.

Let $\lambda \in R_1(\Theta_1)$. We approximate the coefficient matrix $A_*(\zeta;\lambda,\epsilon)$ of~\eqref{evprob2} for $\zeta$ to the left of the front interface $\zeta_\tf(\epsilon)$. By~\eqref{preciseomegabound} in Theorem~\ref{t:ex02}, estimate~\eqref{estex1} in Proposition~\ref{prop:left}, identity~\eqref{limmum} and Lemma~\ref{lem:beta}, we establish the estimate
\begin{align}
 \left\|A_*(\zeta_\tf(\epsilon) + \zeta;\lambda,\epsilon) - A_{*,l}(\zeta_\delta + \zeta)\right\| \leq C_\delta \sqrt{\|\epsilon\|} \left|\log\|\epsilon\|\right| + C|\lambda|, \qquad \zeta \leq 0,\, \lambda \in R_1(\Theta_1), \label{leftbound3}
\end{align}
for $0 < \epsilon_1 \ll \Theta_1 \ll \delta \ll 1$ and $0 \leq |\epsilon_2| \ll \Theta_1 \ll \delta \ll 1$, where we denote
\begin{align*}
A_{*,l}(\zeta) := \begin{pmatrix} 0 & 0 & 1 & 0 \\ 0 & 0 & 0 & 1\\ 2R_*(\zeta) & R_*(\zeta) & 0 & 0 \\ R_*(\zeta) & 2R_*(\zeta) & 0 & 0 \end{pmatrix},
\end{align*}
and $C > 1$ is a constant independent of $\delta$, $\epsilon$ and $\lambda$. By estimate~\eqref{Rsbound} in Proposition~\ref{prop:left}, the coefficient matrix $A_{*,l}(\zeta)$ converges exponentially to
\begin{align*}
A_{*,l}^{-\infty} = \begin{pmatrix} 0 & 0 & 1 & 0 \\ 0 & 0 & 0 & 1\\ 2 & 1 & 0 & 0 \\ 1 & 2 & 0 & 0 \end{pmatrix},
\end{align*}
as $\zeta \to -\infty$, which has the four eigenvalues $\pm \sqrt{3}, \pm 1$. So, the matrix $A_{*,l}^{-\infty}$ is hyperbolic and, by~\cite[Theorem 1]{SAN}, system
\begin{align}
\hat{\phi}_\zeta = A_{*,l}(\zeta)\hat{\phi}, \qquad \hat{\phi} \in \C^4, \label{evprob2red}
\end{align}
has an exponential dichotomy on $(-\infty,\zeta_\delta]$ with associated rank 2 projections $Q_{*,l}(\zeta)$. Note that the dichotomy constants might depend on $\delta$.

The two solutions $\hat{\phi}_0(\zeta;\epsilon)$ and $\hat{\phi}_-(\zeta;\epsilon)$ to~\eqref{evprob2} give rise, upon taking the limit $\|\epsilon\| \to 0$, to two linearly independent solutions
\begin{align*}
\hat{\phi}_{*,0}(\zeta) = \re^{\int_{\zeta_\delta}^\zeta \hat{z}_*(y) \mathrm dy} \begin{pmatrix} 1 \\ -1 \\ \hat{z}_*(\zeta) \\ -\hat{z}_*(\zeta\end{pmatrix}, \qquad \hat{\phi}_{*,-}(\zeta) = \re^{\int_{\zeta_\delta}^\zeta \hat{z}_*(y) \mathrm dy} \begin{pmatrix} \hat{z}_*(\zeta) - 1 \\ \hat{z}_*(\zeta) - 1 \\ R_*(\zeta) - \hat{z}_*(\zeta) \\ R_*(\zeta) - \hat{z}_*(\zeta)\end{pmatrix},
\end{align*}
to~\eqref{evprob2red}. Since $R_*(\zeta)$ and $\hat{z}_*(\zeta)$ converge exponentially to $1$ as $\zeta \to -\infty$ by~\eqref{Rsbound} in Proposition~\ref{prop:left}, it follows $\hat{\phi}_{*,0}(\zeta)$ and $\hat{\phi}_{*,-}(\zeta)$ decay exponentially to $0$ as $\zeta \to -\infty$. Thus, we obtain
\begin{align*} \ker(Q_{*,l}(\zeta)) = \mathrm{Span}\left\{\hat{\phi}_{*,0}(\zeta),\hat{\phi}_{*,-}(\zeta)\right\}, \end{align*}
for $\zeta \in (-\infty,\zeta_\delta]$, which yields~\eqref{leftapprox}.

By roughness of exponential dichotomies, cf.~\cite[Proposition 5.1]{COP}, and estimate~\eqref{leftbound3}, system~\eqref{evprob2} admits, provided $0 < \epsilon_1 \ll \Theta_1 \ll \delta \ll 1$ and $0 \leq |\epsilon_2| \ll \Theta_1 \ll \delta \ll 1$, for each $\lambda \in R_1(\Theta_1)$ an exponential dichotomy on $(-\infty,\zeta_\tf(\epsilon)]$ with $\lambda$- and $\epsilon$-independent constants and projections $P_{*,l}(\zeta;\lambda,\epsilon)$ satisfying~\eqref{leftbound4}.

Finally,~\eqref{leftapprox2} follows, because the linearly independent solutions $\hat{\phi}_0(\zeta;\epsilon)$ and $\hat{\phi}_-(\zeta;\epsilon)$ to~\eqref{evprob2} at $\lambda = 0$ converge to $0$ as $\zeta \to -\infty$ by Lemma~\ref{lem:approxphipm}.
\end{proof}

\subsection{Exponential dichotomies to the right of the inhomogeneity} \label{sec:expdi2}

We show that~\eqref{evprob2} admits an exponential dichotomy on $[0,\infty)$.

\begin{proposition} \label{prop:expdirightR1}
There exist constants $C > 1$ and $\eta > 0$ such that, provided $0 < \epsilon_1 \ll \Theta_1 \ll \delta \ll 1$ and $0 \leq |\epsilon_2| \ll \Theta_1 \ll \delta \ll 1$, system~\eqref{evprob2} admits for each $\lambda \in R_1(\Theta_1)$ an exponential dichotomy on $[0,\infty)$ with $\lambda$- and $\epsilon$-independent constants and projections $P_{*,r}(\zeta;\lambda,\epsilon)$ on $\C^4$ satisfying
\begin{align}
 \left\|P_{*,r}(\zeta;\lambda,\epsilon) - Q_{*,r}(\lambda,\epsilon)\right\| &\leq C\re^{-\eta \epsilon_1}, \qquad \zeta \geq 0, \label{projboundright}
\end{align}
where $Q_{*,r}(\lambda,\epsilon)$ is the spectral projection onto the stable eigenspace of the matrix $A_{*,+}(\lambda,\epsilon)$ defined in~\eqref{defApm}, which is smooth at $\lambda = 0$ and $\epsilon = (0,0)$ and satisfies
\begin{align} Q_{*,r}(0,0)\left[\C^4\right] = \mathrm{Span}\left\{\begin{pmatrix}  1 \\  0\\  -\sqrt{2(1-\ri \alpha)} \\ 0 \end{pmatrix}, \begin{pmatrix}  0\\  1\\  0 \\  -\sqrt{2(1+\ri \alpha)} \end{pmatrix}\right\}. \label{rightapprox}\end{align}
\end{proposition}
\begin{proof}
With the aid of identity~\eqref{triggerbound} in Theorem~\ref{t:ex02} and~\eqref{estex2} in Proposition~\ref{prop:right}, we establish the estimate
\begin{align}
\left\|A_*(\zeta;\lambda,\epsilon) - A_{*,+}(\lambda,\epsilon)\right\| \leq C\re^{-\eta \epsilon_1}, \qquad \zeta \geq 0. \label{rightbound}
\end{align}
for $\lambda$- and $\epsilon$-independent constants $C > 1$ and $\eta > 0$. The eigenvalues $\nu_{\pm,\pm}(\lambda,\epsilon)$ of $A_{*,+}(\lambda,\epsilon)$ are smooth at $\lambda = 0$ and $\epsilon = (0,0)$ and satisfy
\begin{align*} \nu_{-,\pm}(0,0) = -\sqrt{2 (1 \pm \ri \alpha)}, \qquad \nu_{+,\pm}(0,0) = \sqrt{2(1 \pm \ri \alpha)}. \end{align*}
Hence, provided $0 < \epsilon_1 \ll \Theta_1 \ll \delta \ll 1$ and $0 \leq |\epsilon_2| \ll \Theta_1 \ll \delta \ll 1$, the matrix $A_{*,+}(\lambda,\epsilon)$ is hyperbolic with $\lambda$- and $\epsilon$-uniform spectral gap, and the constant-coefficient system~\eqref{evprob2red2} has an exponential dichotomy on $\R$ for each $\lambda \in R_1(\Theta_1)$ with $\epsilon$- and $\lambda$-independent constants and rank 2 projection $Q_{*,r}(\lambda,\epsilon)$, which coincide with the spectral projection onto the stable eigenspace of $A_{*,+}(\lambda)$. The spectral projection $Q_{*,r}(\lambda,\epsilon)$ is smooth at $\lambda = 0$ and $\epsilon = (0,0)$ and satisfies~\eqref{rightapprox} by~\eqref{preciseomegabound} and~\eqref{limmum}. By roughness of exponential dichotomies, cf.~\cite[Proposition 5.1]{COP}, and estimate~\eqref{rightbound}, the eigenvalue problem~\eqref{evprob2} admits, provided $0 < \epsilon_1 \ll \Theta_1 \ll \delta \ll 1$ and $0 \leq |\epsilon_2| \ll \Theta_1 \ll \delta \ll 1$, for each $\lambda \in R_1(\Theta_1)$ an exponential dichotomy on  $[0,\infty)$ with $\lambda$- and $\epsilon$-independent constants and projections $P_{*,r}(\zeta;\lambda,\epsilon)$ satisfying~\eqref{projboundright}.
\end{proof}

\subsection{Tracking subspaces along the absolutely unstable plateau} \label{sec:track}

We approximate system~\eqref{evprob2} along the plateau between the front interface at $\zeta = \zeta_\tf(\epsilon)$ and the inhomogeneity at $\zeta = 0$ by system~\eqref{evprobp}. By observing that the coefficients of the matrix $A_{*,-}(\lambda,\epsilon)$ are smooth at $\lambda = 0$ and $\epsilon = (0,0)$ and it holds
$$A_{*,-}(0,0) = \begin{pmatrix} 0_2 & I_2 \\ 0_2 & 0_2\end{pmatrix},$$
by~\eqref{preciseomegabound} and~\eqref{limmum}, one directly obtains the a priori estimate.

\begin{lemma}\label{lem:evolbound}
There exists a constant $C>1$ such that, provided $0 < \epsilon_1 \ll \Theta_1 \ll 1$ and $0 \leq |\epsilon_2| \ll \Theta_1 \ll 1$, the evolution $\mathcal{T}_{*,p}(\zeta,y;\lambda,\epsilon)$  of~\eqref{evprobp} satisfies
\begin{align*} \left\|\mathcal{T}_{*,p}(\zeta,y;\lambda,\epsilon)\right\| \leq C\left(1 + |\zeta - y|\right) \re^{\sqrt{\|\epsilon\| + |\lambda|} \, |\zeta - y|}, \qquad \zeta,y \in \R, \lambda \in R_1(\Theta_1). \end{align*}
\end{lemma}

The flow induced by~\eqref{evprobp} in the coordinate chart $\mathfrak{c}$, which maps any subspace $W \in \mathrm{Gr}(2,\C^4)$ represented by a basis $\smash{\left(\begin{smallmatrix} X \\ Y \end{smallmatrix}\right)} \in \C^{4 \times 2}$ with $\det(X) \neq 0$ to the matrix $T = YX^{-1} \in \C^{2 \times 2}$, is given by the matrix Riccati equation
\begin{align}
T_\zeta = -T^2 + \D T - T\D + \A_-, \qquad T \in \C^{2 \times 2},\label{matrixric}
\end{align}
where we suppress dependency on $\lambda$ and $\epsilon$.

In this subsection, we track the relevant subspaces $W_\pm(\zeta;\lambda,\epsilon)$, defined in~\eqref{subspaces}, along the absolutely unstable plateau. As in~\S\ref{sec:R2concl}, we relate these subspaces of solutions of the transformed eigenvalue problem~\eqref{evprob2} to its exponential dichotomies established in Propositions~\ref{prop:expdileftR1} and~\ref{prop:expdirightR1}.

\begin{lemma} \label{lem:subspace2}
Provided $0 < \epsilon_1 \ll \Theta_1 \ll \delta \ll 1$ and $0 \leq |\epsilon_2| \ll \Theta_1 \ll \delta \ll 1$, it holds
\begin{align*} W_-(\zeta;\lambda,\epsilon) &= \ker(P_{*,l}(\zeta;\lambda,\epsilon)), \qquad \zeta \in (-\infty,\zeta_\tf(\epsilon)], \, \lambda \in R_1(\Theta_1), \end{align*}
where $P_{*,l}(\zeta;\lambda,\epsilon)$ is the projection associated with exponential dichotomy of~\eqref{evprob2} on $(-\infty,\zeta_\tf(\epsilon)]$ established in Proposition~\ref{prop:expdileftR1}. In addition, we have
\begin{align*} W_+(\zeta;\lambda,\epsilon) &= P_{*,r}(\zeta;\lambda,\epsilon)[\C^4], \qquad \zeta \in [0,\infty), \, \lambda \in R_1(\Theta_1),\end{align*}
where $P_{*,r}(\zeta;\lambda,\epsilon)$ is the projection associated with exponential dichotomy of~\eqref{evprob2} on $[0,\infty)$ established in Proposition~\ref{prop:expdirightR1}.
\end{lemma}
\begin{proof}
The proof is completely analogous to the proof of Lemma~\ref{lem:subspace}.
\end{proof}

The estimate~\eqref{deltaineq} in Proposition~\ref{prop:left} in combination with Lemma~\ref{lem:subspace2} and the bounds in Propositions~\ref{prop:expdileftR1} and~\ref{prop:expdirightR1} now readily lead to the following approximation result.

\begin{lemma} \label{lem:approx}
There exist constants $C > 1$ and $\eta > 0$ such that, provided $0 < \epsilon_1 \ll \Theta_1 \ll \delta \ll 1$ and $0 \leq |\epsilon_2| \ll \Theta_1 \ll \delta \ll 1$, it holds
\begin{align} \left\|\hat{T}_-(\zeta_\tf(\epsilon);\lambda,\epsilon) - \begin{pmatrix} \hat{z}_*(\zeta_\delta) & 0 \\ 0 & \hat{z}_*(\zeta_\delta)\end{pmatrix}\right\| \leq C|\log(\delta)|^{-2}, \qquad \lambda \in R_1(\Theta_1), \label{approxricl}\end{align}
and 
\begin{align} \left\|\hat{T}_+(0;\lambda,\epsilon) - \begin{pmatrix} \hat{t}_+(\lambda,\epsilon) & 0 \\ 0 & \hat{s}_+(\lambda,\epsilon)\end{pmatrix}\right\| \leq C\re^{-\eta \epsilon_1}, \qquad \lambda \in R_1(\Theta_1), \label{approxricr}
\end{align}
where $\hat{t}_+(\lambda,\epsilon)$ and $\hat{s}_+(\lambda,\epsilon)$ are the principal square roots of the diagonal entries of the matrix $\A_+(\lambda,\epsilon)$, defined in~\eqref{defApm2}, which are smooth near $\lambda = 0$ and $\epsilon = 0$ with
\begin{align} \hat{t}_+(0,0) = -\sqrt{2 (1 - \ri \alpha)}, \qquad \hat{s}_+(0,0) = -\sqrt{2(1 + \ri \alpha)}. \label{expex0}
\end{align}
\end{lemma}

For later convenience, we now fix the bases~\eqref{bases} of $W_\pm(\zeta;\lambda,\epsilon)$ by setting
\begin{align} \hat{X}_-(\zeta_\tf(\epsilon);\lambda,\epsilon) = I_4 = \hat{X}_+(0;\lambda,\epsilon). \label{gaugechoice}\end{align}
so that $\hat{T}(\zeta_\tf(\epsilon);\lambda,\epsilon) = \hat{Y}_-(\zeta_\tf(\epsilon);\lambda,\epsilon)$ and $\hat{T}_+(0;\lambda,\epsilon) = \hat{Y}_+(0;\lambda,\epsilon)$. Note that is possible by Lemma~\ref{lem:approx}.

We expect that the evolution of the subspaces $W_\pm(\zeta;\lambda,\epsilon)$ in~\eqref{evprob2} along the plateau is to leading order governed by the dynamics of~\eqref{evprobp}. Thus, we expect the evolution of the representations $\hat{T}_\pm(\zeta;\lambda,\epsilon) \in \C^{2 \times 2}$ of $W_\pm(\zeta;\lambda,\epsilon)$ under the coordinate chart $\mathfrak{c}$ to be to leading order governed by the matrix Riccati equation~\eqref{matrixric}. Since $\A_-(\lambda,\epsilon)$ and $\D(\epsilon)$ are diagonal matrices, one readily observes that the flow of~\eqref{matrixric} leaves the subspace of diagonal matrices in $\C^{2 \times 2}$ invariant. The matrices $\hat{T}_+(0;\lambda,\epsilon)$ and $\hat{T}_-(\zeta_\tf(\epsilon);\lambda,\epsilon)$ are to leading-order diagonal by the estimates in Lemma~\ref{lem:approx}. Thus, by tracking the leading-order diagonal approximation of $\hat{T}_-(\zeta_\tf(\epsilon);\lambda,\epsilon)$ in~\eqref{matrixric} forward from $\zeta_\tf(\epsilon)$ to $0$, we expect to estimate $\hat{T}_-(0;\lambda,\epsilon)$. Similarly, by tracking the leading-order diagonal approximation of $\hat{T}_+(0;\lambda,\epsilon)$ in~\eqref{matrixric} backward from $0$ to $\zeta_\tf(\epsilon)$, we expect to estimate $\hat{T}_+(\zeta_\tf(\epsilon);\lambda,\epsilon)$. We emphasize that, depending on the precise location of $\lambda$ in the region $R_1$, it is advantageous to either approximate $\hat{T}_+(\zeta_\tf(\epsilon);\lambda,\epsilon)$ or $\hat{T}_-(0;\lambda,\epsilon)$, cf.~Remark~\ref{rem:forwardbackward}.

The dynamics of~\eqref{matrixric} on this subspace of diagonal matrices is given by the two scalar Riccati equations
\begin{align}
t_\zeta = -t^2 + \frac{\lambda}{m(\epsilon)^2} (1 - \ri \alpha) + \mu(\epsilon), \label{Riceq1}\\
s_\zeta = -s^2 + \frac{\lambda}{m(\epsilon)^2} (1 + \ri \alpha) + \overline{\mu(\epsilon)}. \label{Riceq2}
\end{align}
We emphasize that although these scalar Riccati equations are explicitly solvable, the dependence of their solutions on the parameters $\lambda$ and $\epsilon$ is rather complicated. In fact, given any solution $\upsilon(\zeta;\varpi)$ to the scalar Riccati equation
$$\upsilon_\zeta = -\upsilon^2 + \varpi,$$
with parameter $\varpi \in \C$ and fixed initial condition $\upsilon(0;\varpi) = \upsilon_0 \in \C$, one finds that $\upsilon(\zeta;\cdot)$ is, for each fixed $\zeta > 0$, a meromorphic function whose poles and zeros accumulate on the negative real axis as $\zeta \to \infty$, see Figure~\ref{fig:poles}.

\begin{figure}[h!]
\centering
  \hspace{-0.5in}\includegraphics[trim=1.5in 0.0in 1in 0.0in, clip,width = 0.5\textwidth]{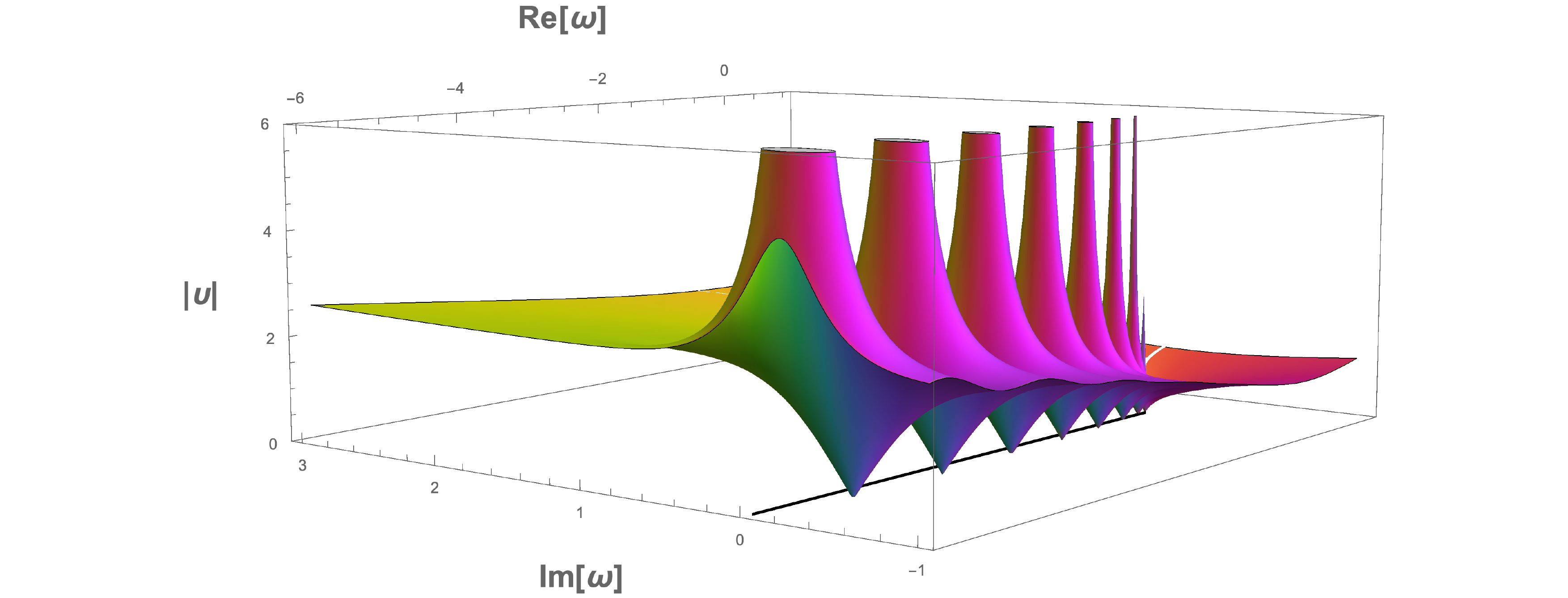}
  \hspace{-0.3in}\includegraphics[trim=1.5in 0.0in 2in 0.0in, clip,width = 0.5\textwidth]{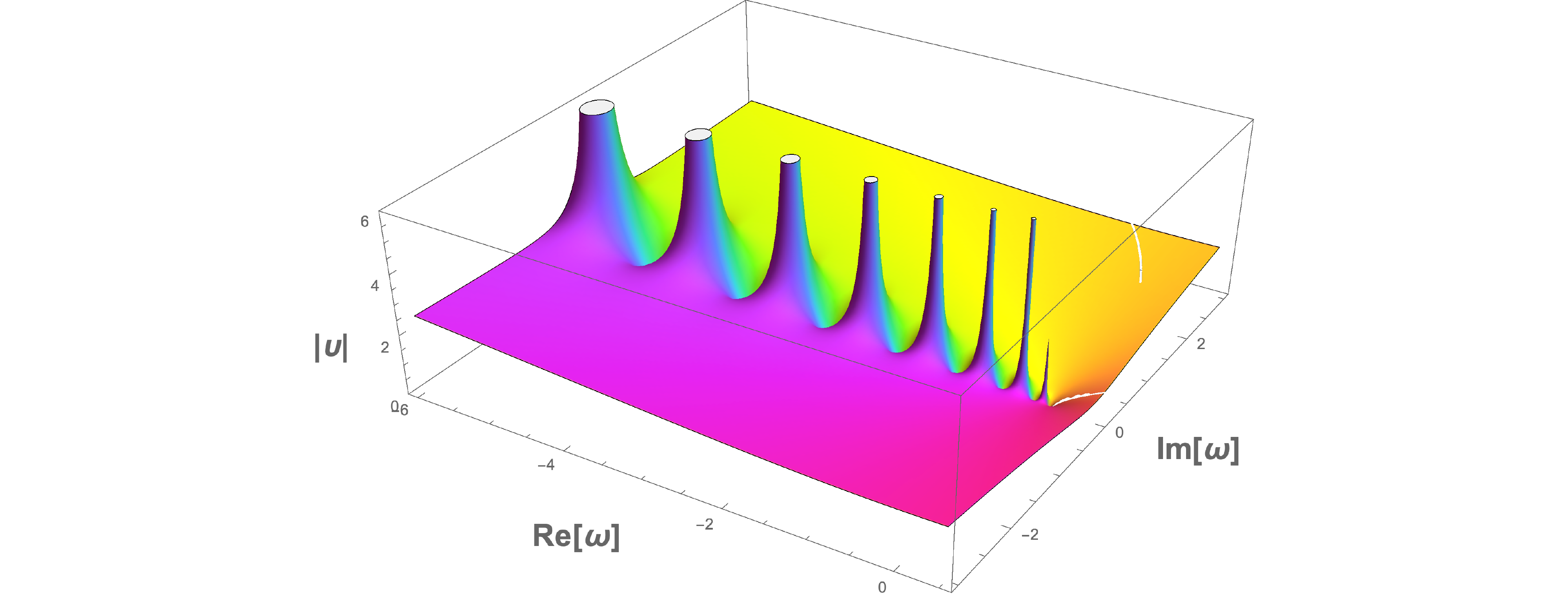}
   \caption{Two views of the plot of the modulus (vertical axis) and argument (color)  of the solution $\upsilon(\zeta;\varpi)$ to $\upsilon_\zeta = -\upsilon^2 + \varpi$ as a function of the parameter $\varpi$ with $\zeta = 10.0$ and $\Re(\upsilon(0;\varpi)) = 1.0 = \Im(\upsilon(0;\varpi))$. Left plot also has black line on the negative real axis with $|\upsilon| = 0$.}
   \label{fig:poles}
\end{figure}

To uncover the dependence of the dynamics in~\eqref{Riceq1} and~\eqref{Riceq2} on $\lambda \in R_1$, we proceed as in the existence analysis of the front in~\cite{GS14}. There, stable and unstable manifolds in~\eqref{exprob} are matched by projecting them onto the singular sphere $R = 0$, where the dynamics is governed by the scalar Riccati equation~\eqref{dyninvman}. Setting \begin{align}M(\epsilon) = \sqrt{\left|\mu(\epsilon)\right|},\label{defM}\end{align}
one observes, using~\eqref{preciseomegabound} and~\eqref{defmm}, that $M(\epsilon)$ is smooth at $\epsilon = (0,0)$ and satisfies
\begin{align} M(0,0) = 0, \qquad \partial_{\epsilon_1} M(0,0) = (1,0). \label{expex4}\end{align}
The matching procedure in~\cite[\S3.4]{GS14} yields, via the implicit function theorem, that
\begin{align}\mu(\epsilon) = M(\epsilon)^2 \re^{\ri \varsigma(M(\epsilon))}, \qquad m(\epsilon) = \tilde{m}(M(\epsilon)), \qquad M(\epsilon)\zeta_\tf(\epsilon) = \hat{\zeta}(M(\epsilon)),\label{expex5}\end{align}
where $\tilde{m}(M), \hat{\zeta}(M)$ and $\varsigma(M)$ are smooth at $M = 0$ satisfying
\begin{align}
\tilde{m}(0) = 1, \qquad \varsigma(0) = \pi, \qquad \hat{\zeta}(0) = -\pi, \qquad \hat{\zeta}'(0) = \Re\left(\frac{1}{\sqrt{2(1-\ri \alpha)}}\right) - \frac{1}{\hat{z}_*(\zeta_\delta)}. \label{expex1}
\end{align}

We will perform a similar matching procedure as in~\cite{GS14} to track the relevant diagonal solutions to~\eqref{matrixric}. We establish control for all $\lambda \in R_1(\Theta_1)$ outside a discrete collection of disks, whose centers lie on $\Sigma_{0,\mathrm{abs}}\cup \overline{\Sigma_{0,\mathrm{abs}}}$ and whose interior contains the poles of the diagonal solutions to~\eqref{matrixric} lie. See Figure~\ref{fig:disks} for a schematic depiction. 

\begin{figure}[h!]
\centering
  \hspace{-0.5in}\includegraphics[trim=0in 0in 0 0in, clip,width = 0.6\textwidth]{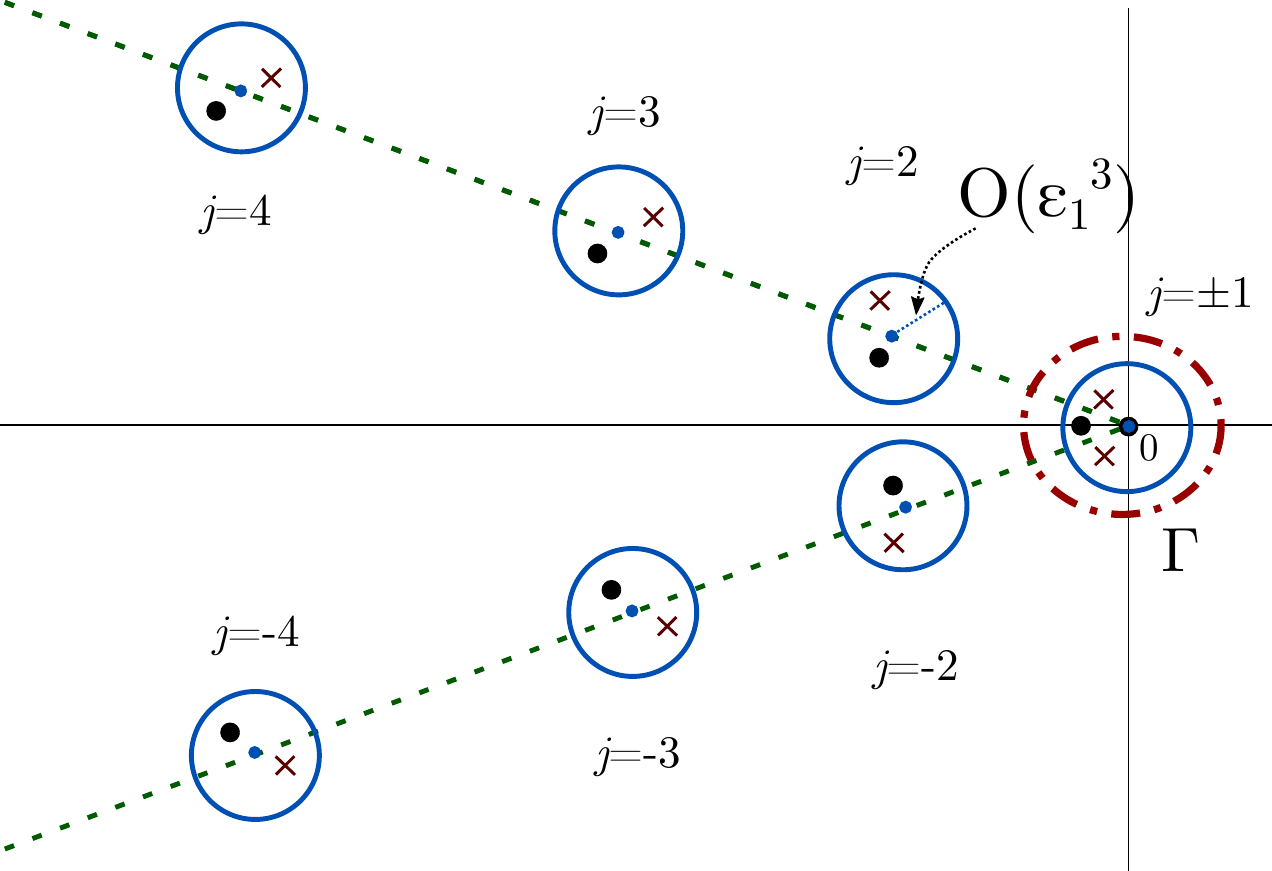}
   \caption{Schematic depiction of the disks $D_j(\epsilon), j \in \Z \setminus \{0\}$ (blue circles), with radii of order $\mathcal{O}(\epsilon_1^{3})$ and centers $\lambda_j$ (blue dot) as defined in~\eqref{e:dj}, lying on the lines $\{s(1\pm\mathrm{i}\alpha)\,:\, s\leq0\}$ (dotted green), which form the limit of the absolute spectrum~\eqref{e:abs0}, and its complex conjugate, as $\epsilon_1 \searrow 0$. The disks enclose the zeros (black dots) and poles (red x) of the Riccati-Evans function. The contour $\Gamma$ (red dotted-dashed circle) encloses the first disk $D_1(\epsilon)$.}
   \label{fig:disks}
\end{figure}

\begin{proposition} \label{prop:plateau1}
There exists a constant $C > 1$ such that, provided $0 < \epsilon_1 \ll \Theta_1 \ll \delta \ll 1$ and $0 \leq |\epsilon_2| \ll \Theta_1 \ll \delta \ll 1$, the following holds. There exists a discrete collection of disks $D_j(\epsilon) \subset \C, j \in \Z \setminus \{0\}$ with center
\begin{align} \label{e:dj}
\lambda_j := -\frac{\left(1-j^2\right) \epsilon_1^2 (1 + \mathrm{sign}(j)\, \ri \alpha)}{1+\alpha^2},
\end{align}
and radius $r_j$ satisfying
\begin{align} |r_j| \leq \min\left\{C_\delta j^2 M(\epsilon)^3, C j^2 \delta^{-1/4} M(\epsilon)^3 + C_{\delta,j} M(\epsilon)^4\right\} \leq C_{\delta} j^2 \|\epsilon\|^3, \label{radiusb}\end{align}
where $C_\delta > 1$ is a constant depending only on $\delta$, and $C_{\delta,j} > 1$ is a constant depending on $\delta$ and $j$ only, such that the diagonal solutions
\begin{align*}T_{d,\pm}(\zeta;\lambda,\epsilon) = \begin{pmatrix} t_\pm(\zeta;\lambda,\epsilon) & 0 \\ 0 & s_\pm(\zeta;\lambda,\epsilon))\end{pmatrix},\end{align*}
to~\eqref{matrixric} with initial data $t_+(0;\lambda,\epsilon) = \hat{t}_+(\lambda,\epsilon)$, $s_+(0;\lambda,\epsilon) = \hat{s}_+(\lambda,\epsilon)$ and $t_-(\zeta_\tf(\epsilon);\lambda,\epsilon) = \hat{z}_*(\zeta_\delta) = s_-(\zeta_\tf(\epsilon);\lambda,\epsilon)$ satisfy
\begin{align} \left\|T_{d,+}(\zeta_\tf(\lambda,\epsilon);\lambda,\epsilon)\right\| \leq \delta^{1/4}, \qquad \text{for } \lambda \in R_1(\Theta_1) \setminus \bigcup_{j \in \Z \setminus \{0\}} D_j(\epsilon), \label{dboundr}
\end{align}
and
\begin{align} \left\|T_{d,-}(0;\lambda,\epsilon)\right\| \leq \delta^{1/4}, \qquad \text{for } \lambda \in R_1(\Theta_1) \setminus \bigcup_{j \in \Z \setminus \{0\}} D_j(\epsilon).\label{dboundl}\end{align}
In addition, each of the meromorphic functions $t_-(0;\cdot,\epsilon)$ and $t_+(\zeta_\tf(\epsilon);\cdot,\epsilon)$ has precisely one pole in each disk $D_j(\epsilon), j \in \Z_{>0}$, which is simple. Similarly, each of the meromorphic functions $s_-(0;\cdot,\epsilon)$ and $s_+(\zeta_\tf(\epsilon);\cdot,\epsilon)$ possesses precisely one pole in each disk $D_j(\epsilon), j \in \Z_{<0}$, which is simple.
\end{proposition}
\begin{proof}
We start by tracking the solution $t_+(\zeta;\lambda,\epsilon)$ to~\eqref{Riceq1} backward from $0$ to $\zeta_\tf(\epsilon)$. As in~\cite{GS14}, the polar coordinate representation
\begin{align} \frac{\lambda}{m(\epsilon)^2} (1 - \ri \alpha) + \mu\left(\epsilon\right) =: N(\lambda,\epsilon)^2 \re^{\ri \varphi(\lambda,\epsilon)}, \label{unfolding}\end{align}
with $\varphi(\lambda,\epsilon) \in [0,2\pi)$ gives a Riemann-surface unfolding about the branch point in~\eqref{Riceq1} yielding
\begin{align} t_\zeta = -t^2 + N^2e^{\ri \varphi}\label{riceq11},\end{align}
where we suppress dependency on $\lambda$ and $\epsilon$ in~\eqref{riceq11}. Consequently, we have
\begin{align*} \hat{t}_+(\lambda,\epsilon) = \tilde{t}_+(M(\epsilon),N(\lambda,\epsilon),\varphi(\lambda,\epsilon)),\end{align*}
with
\begin{align*} \tilde{t}_+(M,N,\varphi) := \sqrt{\frac{2}{\tilde{m}(M)^2}(1-\ri\alpha) + N^2e^{\ri\varphi}}.\end{align*}
As in~\cite{GS14}, we scale $\xi = N\zeta$ in~\eqref{riceq11} and shift equilibria in~\eqref{riceq11} using the M\"obius transformation
\begin{align*}
\rho = \frac{t + N\eta}{t - N\eta}, \qquad \eta := \re^{\ri \frac{\varphi}{2}}.
\end{align*}
We set $r = \log(\rho)$ and, once again suppressing dependence on $\lambda$ and $\epsilon$,  find
\begin{align*}
r_{\xi} = 2\eta.
\end{align*}

Now we set $\xi_\tf(\lambda,\epsilon) := N(\lambda,\epsilon)\zeta_\tf(\epsilon)$, and define the closed disk
\begin{align*} \B_\delta := \left\{\lambda \in \C \colon |\lambda| \leq \delta^{-1/4}\right\} \end{align*}
at the origin in $\C$ (the choice of radius $\delta^{-1/4}$ becomes apparent later). Our goal is to solve the equation
\begin{align}
q_- = \frac{1}{t_+(\xi_\tf(\lambda,\epsilon);\lambda,\epsilon)},
\label{matching1}\end{align}
with respect to $\lambda \in R_1(\Theta_1)$ for $0 < \epsilon_1 \ll \Theta_1 \ll \delta \ll 1$, $0 \leq |\epsilon_2| \ll \Theta_1 \ll \delta \ll 1$ for $q_-$ in the compact set $\B_\delta$.

In the new coordinates, the endpoints of the resulting boundary value problem take the form
\begin{align*} r_+ = 2\pi \ri j_+ + \log\left(\frac{1 + \frac{N\eta}{\tilde{t}_+(M,N,\varphi)}}{1-\frac{N\eta}{\tilde{t}_+(M,N,\varphi)}}\right), \qquad r_- = 2\pi \ri j_- + \log\left(\frac{1 + N\eta q_-}{1 - N\eta q_-}\right),\end{align*}
for certain $j_\pm \in \Z$. Thus, after integration we obtain the equation
\begin{align}
-\xi_\tf = \frac{r_+ - r_-}{2\eta} = \re^{-\ri \frac{\varphi}{2}} \pi \ri j + \frac{e^{-\ri \frac{\varphi}{2}}}{2} \log\left(\frac{1 + \frac{Ne^{\ri \frac{\varphi}{2}}}{\tilde{t}_+(M,N,\varphi)}}{1-\frac{Ne^{\ri \frac{\varphi}{2}}}{\tilde{t}_+(M,N,\varphi)}}\right) - \frac{e^{-\ri \frac{\varphi}{2}}}{2} \log\left(\frac{1 + \re^{\ri \frac{\varphi}{2}} N q_-}{1-e^{\ri \frac{\varphi}{2}} N q_-}\right), \label{IFT1}
\end{align}
where we suppress dependence on $\lambda$ and $\epsilon$ and we denote $j = j_+ - j_- \in \Z$. We wish to solve~\eqref{IFT1} with respect to $N$ and $\varphi$ in terms of $M$ for $(N,M)$ in a neighborhood of $(0,0)$ in $\R^2$. Setting $N, M = 0$ in~\eqref{IFT1} and taking imaginary parts implies $\varphi = \pi$. Thus, by the implicit function theorem, we can only solve~\eqref{IFT1} for $(N,M)$ in a neighborhood of $(0,0)$ if $\varphi$ is close to $\pi$.

In case $j = 0$, we observe from~\eqref{IFT1} that $N(\lambda,\epsilon)^{-1} \xi_\tf(\lambda,\epsilon) = \zeta_\tf(\epsilon)$ is bounded as $\|\epsilon\| \to 0$, which contradicts~\eqref{triggerbound} in Theorem~\ref{t:ex02}. In addition, in case $j < 0$, we observe from~\eqref{IFT1} that $\xi_\tf > 0$ for $(N,M,\varphi)$ close to $(0,0,\pi)$, which contradicts~\eqref{triggerbound} again. Hence, it must hold $j \in \Z_{> 0}$.

Dividing~\eqref{IFT1} by $j \neq 0$, introducing $\jmath = 1/j$ and taking imaginary parts leads to the equation
\begin{align*}
\begin{split}
0 = \cos\left(\frac{\varphi}{2}\right) \pi + \frac{\jmath}{2} \Im\left[e^{-\ri \frac{\varphi}{2}} \log\left(\frac{1 + \frac{Ne^{\ri \frac{\varphi}{2}}}{\tilde{t}_+(M,N,\varphi)}}{1-\frac{Ne^{\ri \frac{\varphi}{2}}}{\tilde{t}_+(M,N,\varphi)}}\right) - \re^{-\ri \frac{\varphi}{2}}\log\left(\frac{1 + \re^{\ri \frac{\varphi}{2}}N q_-}{1-e^{\ri \frac{\varphi}{2}} N q_-}\right)\right],\end{split}
\end{align*}
which can, by the implicit function theorem and~\eqref{expex1}, in a neighborhood of $(N,M,\varphi) = (0,0,\pi)$ and for $(q_-,\jmath)$ in the compact set $\B_\delta \times B_1$, be solved for $\varphi$ in terms of $M, N, q_-$ and $\jmath$. This yields a unique smooth solution $\varphi(M,N;q_-,\jmath)$ defined in $U_2 \times \B_\delta \times B_1$, where $U_2$ is a neighborhood of $(0,0)$ in $\R^2$ and $B_1$ is the closed unit disk at the origin in $\C$. Subsequently expanding $\varphi$, we obtain
\begin{align} \left|\varphi(M,N;q_-,\jmath) - \pi - \frac{\jmath}{\pi} \Im\left(\frac{1}{\sqrt{2(1-\ri\alpha)}} - q_-\right) N\right| \leq C_{\delta}
|\jmath| N^2, \label{expex2}\end{align}
for $(M,N,q_-,\jmath) \in U_2 \times \B_\delta \times B_1$, where $C_\delta > 0$ is a $M$-, $N$-, $q_-$ and $\jmath$-independent constant that might be dependent on $\delta$.

Since the matching time $\hat\zeta$ must satisfy $M\xi_\tf - N \hat{\zeta}(M) = 0$, we subsequently fix $j \in \Z_{>0}$ and substitute the real part of equation~\eqref{IFT1} in for $\xi_\tf$ we find
\begin{align}
-N\hat{\zeta}(M) = M\sin\left(\frac{\varphi}{2}\right) \pi j + \frac{M}{2} \Re\left[e^{-\ri \frac{\varphi}{2}} \log\left(\frac{1 + \frac{Ne^{\ri \frac{\varphi}{2}}}{\tilde{t}_+(M,N,\varphi)}}{1-\frac{Ne^{\ri \frac{\varphi}{2}}}{\tilde{t}_+(M,N,\varphi)}}\right) - \re^{-\ri \frac{\varphi}{2}}\log\left(\frac{1 + \re^{\ri \frac{\varphi}{2}}N q_-}{1-e^{\ri \frac{\varphi}{2}} N q_-}\right)\right], \label{IFT2}
\end{align}
where we suppress the argument of $\varphi(M,N;q_-,\frac{1}{j})$. Equation~\eqref{IFT2} can, by the implicit function theorem,~\eqref{expex1} and~\eqref{expex2}, in a neighborhood of $(N,M) = (0,0)$ and for $q_-$ in the compact set $\B_\delta$, be solved for $N$ in terms of $M$ and $q_-$, yielding a unique smooth solution $N_j(M;q_-)$ defined in $U_1 \times \B_\delta$, where $U_1$ is a neighborhood of $0$ in $\R$. Subsequently expanding $N_j$ for fixed $j \in \Z_{>0}$, we obtain
\begin{align*} \left|N_j(M;q_-) - Mj - \frac{jM^2}{\pi} \left(\frac{1}{\hat{z}_*(\zeta_\delta)} - \Re(q_-)\right)\right| \leq C_{\delta,j} M^3,\end{align*}
yielding by~\eqref{deltaineq}:
\begin{align} \left|N_j(M;q_-) - Mj\right| \leq C \delta^{-1/4} |j| M^2 + C_{\delta,j} M^3, \label{expex6}\end{align}
for $(M,q_-) \in U_1 \times \B_\delta$, where $C_{\delta,j} > 0$ is a $M$- and $q_-$-independent constant that might be dependent on $\delta$ and $j$ and $C > 0$ is a $M$-, $j$-, $\delta$- and $q_-$-independent constant.

To establish a more $j$-uniform bound on $N_j$, see Remark~\ref{junifrem}, we use~\eqref{expex1},~\eqref{expex2},~\eqref{IFT2} and the fact that $N_j(M;q_-)$ is a priori small by~\eqref{unfolding} for $\lambda \in R_1(\Theta_1)$, $0 < \epsilon_1 \ll \Theta_1 \ll \delta \ll 1$ and $0 \leq |\epsilon_2| \ll \Theta_1 \ll \delta \ll 1$. We approximate $N_j(M;q_-)$ as follows
\begin{align*}
|N_j(M;q_-) - Mj| &\leq \frac{1}{\pi}\left(\left|\pi + \hat{\zeta}(M)\right|N_j(M;q_-) + \left|M\pi j + \hat{\zeta}(M)N_j(M;q_-)\right|\right)\\
&\leq C_\delta N_j(M;q_-) M \leq C_\delta\left(M |N_j(M;q_-) - Mj| + |j|M^2\right),
\end{align*}
yielding
\begin{align}
 |N_j(M;q_-) - Mj| \leq C_\delta |j| M^2, \label{expex3}
\end{align}
for $(M,q_-) \in U_1 \times \B_\delta$ and $j \in \Z_{>0}$, where $C_\delta > 0$ is a $M$-, $q_-$- and $j$-independent constant that might be dependent on $\delta$.

Substituting the obtained expressions for $N_j(M;q_-)$ and $\varphi(M,N_j(M;q_-);q_-,\frac{1}{j})$ into~\eqref{unfolding} yields by~\eqref{expex4},~\eqref{expex5},~\eqref{expex1},~\eqref{expex2},~\eqref{expex6} and~\eqref{expex3} that, given $(M,q_-) \in U_1 \times \B_\delta$, all solutions $\lambda \in R_1(\Theta_1)$ to the equation~\eqref{matching1} lie in one of the disks $D_j(\epsilon)$, $j \in \Z_{>0}$, with center
\begin{align*} \lambda_j := -\frac{\left(1-j^2\right) \epsilon_1^2 (1 + \ri \alpha)}{1+\alpha^2},\end{align*}
and radius $r_j$ satisfying~\eqref{radiusb}, provided $0 < \epsilon_1 \ll \Theta_1 \ll \delta \ll 1$ and $0 \leq |\epsilon_2| \ll \Theta_1 \ll \delta \ll 1$. In addition, given $(M,q_-) \in U_1 \times \B_\delta$, each disk $D_j(\epsilon) \cap R_1, j \in \Z_{>0}$ contains precisely one solution $\lambda \in \C$ to~\eqref{matching1}. Hence, the meromorphic function $\lambda \mapsto t_+(\zeta_\tf(\epsilon);\lambda,\epsilon)^{-1}$ admits, for each $j \in \Z_{>0}$, a unique zero $\lambda_{j,\epsilon}$ in $D_j(\epsilon)$ and is injective and analytic in a neighborhood of $\lambda_{j,\epsilon}$. Thus, since injective holomorphic functions have non-vanishing derivatives, $\lambda_{j,\epsilon}$ must be a simple zero of $t_+(\zeta_\tf(\epsilon);\,\cdot\,,\epsilon)^{-1}$ or, equivalently, a simple pole of $t_+(\zeta_\tf(\epsilon);\cdot,\epsilon)$. So, for each $j \in \Z_{>0}$, the number of poles of $t_+(\zeta_\tf(\epsilon);\,\cdot\,,\epsilon)$ in $D_j(\epsilon)$ is one (including multiplicity).

Similarly, by tracking the solution $s_+(\zeta;\lambda,\epsilon)$ with initial condition $\hat{s}_+(\lambda,\epsilon)$ backward in~\eqref{Riceq2} from $0$ to $\zeta_\tf(\epsilon)$, we obtain that all solutions $\lambda \in R_1(\Theta_1)$ to
\begin{align}q_- = \frac{1}{s_+(\xi_\tf(\lambda,\epsilon);\lambda,\epsilon)},\label{matching}\end{align}
for $q_- \in \B_\delta$, lie in one of the disks $D_j(\epsilon)$, $j \in \Z_{< 0}$ with center
\begin{align*} \lambda_j := -\frac{\left(1-j^2\right) \epsilon_1^2 (1 - \ri \alpha)}{1+\alpha^2},\end{align*}
and radius $r_j$ satisfying~\eqref{radiusb}. Given $(M,q_-) \in U_1 \times \B_\delta$, each disk $D_j(\epsilon) \cap R_1, j \in \Z_{<0}$ contains precisely one solution $\lambda \in \C$ to~\eqref{matching} and, thus, contains precisely one pole of $s_+(\zeta_\tf(\epsilon);\cdot,\epsilon)$, which is simple.

All in all, we obtain~\eqref{dboundr}, provided $0 < \epsilon_1 \ll \Theta_1 \ll \delta \ll 1$ and $0 \leq |\epsilon_2| \ll \Theta_1 \ll \delta \ll 1$.

Using an analogous procedure, one tracks the solution $T_{d,-}(\zeta;\lambda,\epsilon)$ forward in~\eqref{matrixric} from $\zeta_\tf(\epsilon)$ to $0$, where we use~\eqref{deltaineq} from Proposition~\ref{prop:left}. We obtain that all solutions $\lambda \in R_1(\Theta_1)$ to
\begin{align} q_+ = \frac{1}{t_-(0;\lambda,\epsilon)},\label{matching3}\end{align}
for some $q_+ \in \B_\delta$, lie in one of the disks $D_j(\epsilon)$, $j \in \Z_{>0}$. Given $(M,q_+) \in U_1 \times \B_\delta$, each disk $D_j(\epsilon) \cap R_1, j \in \Z_{>0}$ contains precisely one solution $\lambda \in \C$ to~\eqref{matching3} and, thus, contains precisely one pole of $t_-(0;\,\cdot\,,\epsilon)$, which is simple. In addition, all solutions $\lambda \in R_1(\Theta_1)$ to
\begin{align} q_+ = \frac{1}{s_-(0;\lambda,\epsilon)},\label{matching4}\end{align}
for some $q_+ \in \B_\delta$, lie in one of the disks $D_j(\epsilon)$, $j \in \Z_{<0}$. Given $(M,q_+) \in U_1 \times \B_\delta$, each disk $D_j(\epsilon) \cap R_1, j \in \Z_{<0}$ contains precisely one solution $\lambda \in \C$ to~\eqref{matching4} and, thus, contains precisely one pole of $s_-(0;\,\cdot\,,\epsilon)$, which is simple. Hence, we obtain~\eqref{dboundl}, provided $0 < \epsilon_1 \ll \Theta_1 \ll \delta \ll 1$ and $0 \leq |\epsilon_2| \ll \Theta_1 \ll \delta \ll 1$.
\end{proof}

\begin{remark} \label{junifrem}
{\rm 
It follows from the Riemann surface unfolding and M\"obius transformation in the proof of Proposition~\ref{prop:plateau1} that the parameter $j$ measures the winding number of solutions to the scalar Riccati equation~\eqref{Riceq1} around its fixed points. The estimate $|r_j| \leq C_\delta j^2 M(\epsilon)^3$ on the radius in Proposition~\ref{prop:plateau1} is necessary to bound the disks $D_j(\epsilon)$ uniformly in $j \in \Z \setminus \{0,\pm1\}$ away from the closed right-half plane. Indeed, provided $0 < \epsilon_1 \ll \delta \ll 1$ and $0 \leq |\epsilon_2| \ll \delta \ll 1$, it holds $|\Re \, \lambda_j| > C_\delta j^2 M(\epsilon)^3$, no matter the size of $j \in \Z \setminus \{0,\pm1\}$. We emphasize that such a bound was not necessary in the existence analysis in~\cite{GS14}.}\end{remark}

Using Proposition~\ref{prop:plateau1} we aim to approximate the solutions $\tilde{T}_\pm(\zeta;\lambda,\epsilon)$ to the matrix Riccati equation~\eqref{matrixric} with initial conditions $\tilde{T}_+(0;\lambda,\epsilon) = \hat{T}_+(0;\lambda,\epsilon)$ and $\tilde{T}_-(\zeta_\tf(\epsilon);\lambda,\epsilon) = \hat{T}_+(\zeta_\tf(\epsilon);\lambda,\epsilon)$ by the diagonal solutions $T_{d,\pm}(\zeta;\lambda,\epsilon)$. We do so by applying a superposition principle to~\eqref{matrixric}. Given two known solutions to the matrix Riccati equation~\eqref{matrixric}, the superposition principle reduces finding a third to solving a linear system in $\C^{2 \times 2}$. We refer to~\cite{HARN} for more theoretical background. It is readily seen that~\eqref{matrixric} has two fixed points, which correspond to square roots of the diagonal matrix $\A_-(\lambda,\epsilon)$. We use these two fixed point solutions as input for the superposition principle, which leads to the following result.

\begin{proposition} \label{prop:plateau2}
There exist constants $C > 1$ and $\eta > 0$ such that, provided $0 < \epsilon_1 \ll \Theta_1 \ll \delta \ll 1$ and $0 \leq |\epsilon_2| \ll \Theta_1 \ll \delta \ll 1$, the solutions $\tilde{T}_\pm(\zeta;\lambda,\epsilon)$ to the matrix Riccati equation~\eqref{matrixric} with initial conditions $\tilde{T}_+(0;\lambda,\epsilon) = \hat{T}_+(0;\lambda,\epsilon)$ and $\tilde{T}_-(\zeta_\tf(\epsilon);\lambda,\epsilon) = \hat{T}_+(\zeta_\tf(\epsilon);\lambda,\epsilon)$ enjoy the estimates
\begin{align}
\begin{split}
\left\|\tilde{T}_+(\zeta_\tf(\epsilon);\lambda,\epsilon) - T_{d,+}(\zeta_\tf(\epsilon);\lambda,\epsilon)\right\| &\leq C\re^{-\eta \epsilon_1},\\
\left\|\tilde{T}_-(0;\lambda,\epsilon) - T_{d,-}(0;\lambda,\epsilon)\right\| &\leq C\delta^{1/4},\\
\left\|\tilde{T}_+(\zeta_\tf(\epsilon);\lambda,\epsilon)\right\|,\left\|\tilde{T}_-(0;\lambda,\epsilon)\right\|  &\leq C\delta^{1/4},
\end{split}
\qquad \text{for } \lambda \in R_1(\Theta_1) \setminus \bigcup_{j \in \Z \setminus \{0\}} D_j(\epsilon), \label{tboundr}\end{align}
where $T_{d,\pm}(\zeta;\lambda,\epsilon)$ and $D_j(\epsilon)$ are as in Proposition~\ref{prop:plateau1}.
\end{proposition}
\begin{proof}
In this proof we denote by $C > 1$ any constant which is independent of $\epsilon$, $\lambda$, $\delta$ and $\Theta_1$.

Let $T_0(\lambda,\epsilon)$ be the matrix obtained by taking entry-wise principal square roots in the diagonal matrix $\A_-(\lambda,\epsilon)$. Then it holds $T_0(\lambda,\epsilon)^2 = \A_-(\lambda,\epsilon)$ and the entries of the diagonal matrix $T_0(\lambda,\epsilon)$ have nonnegative real parts. Consequently, $\pm T_0(\lambda,\epsilon)$ are fixed point solutions to~\eqref{matrixric}, which, by~\eqref{limmum}, satisfy
\begin{align} \|T_0(\lambda,\epsilon)\| \leq C\sqrt{\|\epsilon\| + |\lambda|}, \qquad \lambda \in R_1(\Theta_1). \label{T0bound}\end{align}

Following the superposition principle in~\cite[\S{III.B}]{HARN}, we consider the ratios
\begin{align*} \W_\pm(\zeta;\lambda,\epsilon) = \left(\tilde{T}_\pm(\zeta) \mp T_0\right)^{-1} \left(T_{d,\pm}(\zeta) - \tilde{T}_+(\zeta)\right) \left(T_{d,\pm}(\zeta) \mp T_0\right)^{-1} \in \C^{2 \times 2},\end{align*}
where we suppress $\lambda$- and $\epsilon$-dependency on the right hand side.

Using Lemma~\ref{lem:approx}, one finds $\det(\tilde{T}_+(0;\lambda,\epsilon) - T_0(\lambda,\epsilon)) \neq 0$, $\det(T_{d,+}(0;\lambda,\epsilon) - T_0(\lambda,\epsilon)) \neq 0$ and, more specifically,
\begin{align} \left\|\W_+(0;\lambda,\epsilon)\right\| \leq C\re^{-\eta \epsilon_1}, \qquad \lambda \in R_1(\Theta_1),\label{Wrbound}\end{align}
provided $0 < \epsilon_1 \ll \Theta_1 \ll \delta \ll 1$ and $0 \leq |\epsilon_2| \ll \Theta_1 \ll \delta \ll 1$. Similarly, with the aid of Lemma~\ref{lem:approx} and estimate~\eqref{deltaineq} in Proposition~\ref{prop:left}, one establishes $0 \neq \det(\tilde{T}_-(\zeta_\tf(\epsilon);\lambda,\epsilon) - T_0(\lambda,\epsilon))$, $\det(T_{d,-}(\zeta_\tf(\epsilon)\lambda,\epsilon) - T_0(\lambda,\epsilon)) \neq 0$ and, more specifically,
\begin{align} \left\|\W_-(\zeta_\tf(\epsilon);\lambda,\epsilon)\right\| \leq C, \qquad \lambda \in R_1(\Theta_1), \label{Wlbound}\end{align}
provided $0 < \epsilon_1 \ll \Theta_1 \ll \delta \ll 1$ and $0 \leq |\epsilon_2| \ll \Theta_1 \ll \delta \ll 1$. Thus, the subspaces corresponding to $\tilde{T}_\pm(\zeta;\lambda,\epsilon)$, $T_{d,\pm}(\zeta;\lambda,\epsilon)$ and $\pm T_0(\lambda,\epsilon)$ in the linear system~\eqref{evprobp} are complementary at $\zeta = 0$ and $\zeta = \zeta_\tf(\epsilon)$, respectively, and therefore at \emph{all} $\zeta \in \R$. It follows $\det(\tilde{T}_\pm(\zeta) - T_0), \det(T_{d,\pm}(\zeta) - T_0) \neq 0$ for all $\zeta \in \R$ and the ratio $\W_\pm(\zeta;\lambda,\epsilon)$ is well-defined for each $\zeta \in \R$ and $\lambda \in R_1(\Theta_1)$, where the $\lambda$-meromorphic functions $\tilde{T}_\pm(\zeta;\lambda,\epsilon)$ and $T_{d,\pm}(\zeta;\lambda,\epsilon)$ have no poles.

Through direct verification, see also~\cite{HARN}, one observes that $\W_\pm(\zeta;\lambda,\epsilon)$ satisfies the linear matrix equation
\begin{align*}
 \W_\zeta = (\D \pm T_0) \W - \W (\D \mp T_0), \qquad W \in \C^{2 \times 2},
\end{align*}
where we suppress $\lambda$- and $\epsilon$-dependency. Consequently, it holds
\begin{align*} \W_+(\zeta;\lambda,\epsilon) &= \re^{\left(T_0(\lambda,\epsilon) + \D(\epsilon)\right)\zeta} \W_+(0;\lambda,\epsilon) \re^{\left(T_0(\lambda,\epsilon) - \D(\epsilon)\right)\zeta},\\
 \W_-(\zeta;\lambda,\epsilon) &= \re^{-\left(T_0(\lambda,\epsilon) - \D(\epsilon)\right)\left(\zeta-\zeta_\tf(\epsilon)\right)} \W_-(\zeta_\tf(\epsilon);\lambda,\epsilon) \re^{-\left(T_0(\lambda,\epsilon) + \D(\epsilon)\right)\left(\zeta-\zeta_\tf(\epsilon)\right)},
\end{align*}
for $\zeta \in [\zeta_\tf(\epsilon),0]$, yielding
\begin{align}
\begin{split}
\|\W_+(\zeta_\tf(\epsilon);\lambda,\epsilon)\| &\leq C\re^{-\eta \epsilon_1},\\
\|\W_-(0;\lambda,\epsilon)\| &\leq C,
\end{split}
\label{Wrbound2}\end{align}
by~\eqref{Wrbound},~\eqref{Wlbound} and the fact that all entries of the diagonal matrices $T_0(\lambda,\epsilon) \pm \D(\epsilon)$ have nonnegative real part.

By combining
\begin{align*}
\left\|\tilde{T}_\pm(\zeta;\lambda,\epsilon) - T_{d,\pm}(\zeta;\lambda,\epsilon)\right\| &= \left\|\left(\tilde{T}_\pm \mp T_0\right) \W_\pm \left(T_{d,\pm} \mp T_0\right)\right\|\\
&\leq \left(\left\|\tilde{T}_\pm - T_{d,\pm}\right\| + \left\|T_{d,\pm}\right\| + \left\|T_0\right\|\right) \left\|\W_\pm\right\| \left(\left\|T_{d,\pm}\right\| + \left\|T_0\right\|\right),
\end{align*}
where we suppressed the arguments on the right hand side, with the estimates~\eqref{dboundr} and~\eqref{dboundl} established in Proposition~\ref{prop:plateau1} and the approximations~\eqref{T0bound} and~\eqref{Wrbound2}, we obtain~\eqref{tboundr}, provided $0 < \epsilon_1 \ll \Theta_1 \ll \delta \ll 1$ and $0 \leq |\epsilon_2| \ll \Theta_1 \ll \delta \ll 1$.
\end{proof}

Recall that, under the coordinate chart $\mathfrak{c}$, the matrix $\hat{T}_+(\zeta;\lambda,\epsilon) \in \C^{2 \times 2}$ represents the relevant subspace $W_+(\zeta;\lambda,\epsilon) \in \mathrm{Gr}(2,\C^4)$ of solutions to~\eqref{evprob2}, which was defined by~\eqref{subspaces}. Having established the bound~\eqref{tboundr} in Proposition~\ref{prop:plateau2}, we are now in the position to track $\hat{T}_+(\zeta;\lambda,\epsilon)$ along the plateau state back from $\zeta = 0$ to $\zeta = \zeta_\tf(\epsilon)$. We employ the variation of constant formula and exploit that system~\eqref{evprob2} is, by Proposition~\ref{prop:right}, converging at an $\epsilon$- and $\lambda$-independent exponential rate to~\eqref{evprobp} when propagating forward along the plateau, whereas the evolution of~\eqref{evprobp} can only grow with an exponential rate of order $\ord(\sqrt{\|\epsilon\| + |\lambda|})$ by Lemma~\ref{lem:evolbound}.

\begin{proposition} \label{prop:plateau3}
There exists a constant $C > 1$ such that, provided $0 < \epsilon_1 \ll \Theta_1 \ll \delta \ll 1$ and $0 \leq |\epsilon_2| \ll \Theta_1 \ll \delta \ll 1$, it holds
\begin{align}
\begin{split}
\left\|\hat{T}_+(\zeta_\tf(\epsilon);\lambda,\epsilon) - \tilde{T}_+(\zeta_\tf(\epsilon);\lambda,\epsilon)\right\| &\leq C\delta, \\
\left\|\hat{T}_+(\zeta_\tf(\epsilon);\lambda,\epsilon)\right\| &\leq C\delta^{1/4},
\end{split} \qquad \lambda \in R_1(\Theta_1) \setminus \bigcup_{j \in \Z \setminus \{0\}} D_j(\epsilon),
\label{tboundr2}\end{align}
where $\tilde{T}_+(\zeta;\lambda,\epsilon)$ is as in Proposition~\ref{prop:plateau2} and $D_j(\epsilon)$ is as in Proposition~\ref{prop:plateau1}.
\end{proposition}
\begin{proof} In this proof we denote by $C > 1$ any constant which is independent of $\epsilon$, $\lambda$, $\delta$ and $\Theta_1$. Fix $\lambda \in R_1(\Theta_1) \setminus \bigcup_{j \in \Z \setminus \{0\}} D_j(\epsilon)$.

Denote by
\begin{align*} B_{*,p}(\zeta;\epsilon) := A_*(\zeta;\lambda,\epsilon) - A_{*,-}(\lambda,\epsilon), \end{align*}
the difference of the coefficient matrices of~\eqref{evprob2} and~\eqref{evprobp}. By~\eqref{estex2} in Proposition~\ref{prop:right} it holds
\begin{align}
\left\|B_{*,p}(\zeta;\epsilon)\right\| \leq C\delta \re^{-\eta|\zeta - \zeta_\tf(\epsilon)|}, \qquad \zeta \in [\zeta_\tf(\epsilon),0]. \label{plateaubound}
\end{align}

Recall that $\mathcal{T}_{*,p}(\zeta,y;\lambda,\epsilon)$ denotes the evolution of the reduced system~\eqref{evprobp}. The subspace spanned by the matrix solution
\begin{align*} \mathcal{T}_{*,p}(\zeta,\zeta_\tf(\epsilon);\lambda,\epsilon)\begin{pmatrix} I_2 \\ \tilde{T}_+(\zeta_\tf(\epsilon);\lambda,\epsilon)\end{pmatrix} \in \C^{4 \times 2},\end{align*}
to~\eqref{evprobp} is, under the chart $\mathfrak{c}$, represented by the solution $\tilde{T}_+(\zeta;\lambda,\epsilon)$ to the corresponding matrix Riccati equation~\eqref{matrixric}. Thus, the variation of constants formula
\begin{align}
\begin{split}
\Y(\zeta;\lambda,\epsilon) &= \mathcal{T}_{*,p}(\zeta,\zeta_\tf(\epsilon);\lambda,\epsilon) \begin{pmatrix} I_2 \\ \tilde{T}_+(\zeta_\tf(\epsilon);\lambda,\epsilon)\end{pmatrix}\\
&\qquad + \int_0^\zeta
\mathcal{T}_{*,p}(\zeta,y;\lambda,\epsilon) B_{*,p}(y;\epsilon) \Y(y;\lambda,\epsilon) \mathrm dy,
\end{split} \label{varcon}
\end{align}
yields a matrix solution $\Y(\zeta;\lambda,\epsilon) \in \C^{4 \times 2}$ to~\eqref{evprob2}, which spans the subspace $W_+(\zeta;\lambda,\epsilon)$ for each $\zeta \in \R$. Indeed, $\Y(0;\lambda,\epsilon)$ spans the subspace which is represented by $\tilde{T}_+(0;\lambda,\epsilon) = \hat{T}_+(0;\lambda,\epsilon) \in \C^{2 \times 2}$ under the chart $\mathfrak{c}$.

By~\eqref{varcon}, we find that $\hat{\Y}(\zeta;\lambda,\epsilon) := \mathcal{T}_{*,p}(\zeta_\tf(\epsilon),\zeta;\lambda,\epsilon)\Y(\zeta;\lambda,\epsilon)$ satisfies
\begin{align}
\begin{split}
\hat{\Y}(\zeta;\lambda,\epsilon) &= \begin{pmatrix} I_2 \\ \tilde{T}_+(\zeta_\tf(\epsilon);\lambda,\epsilon)\end{pmatrix}\\
&\qquad + \int_0^\zeta
\mathcal{T}_{*,p}(\zeta_\tf(\epsilon),y;\lambda,\epsilon) B_{*,p}(y;\epsilon) \mathcal{T}_{*,p}(y,\zeta_\tf(\epsilon);\lambda,\epsilon) \hat{\Y}(y;\lambda,\epsilon) \mathrm dy.
\end{split}
\label{varcon2}
\end{align}
By~\eqref{plateaubound} and Lemma~\ref{lem:evolbound} the linear operator $\F_{\lambda,\epsilon}$ given by
\begin{align*} \left(\F_{\lambda,\epsilon} \hat{\Y}\right)(\zeta) = \int_0^\zeta \mathcal{T}_{*,p}(\zeta_\tf(\epsilon),y;\lambda,\epsilon) B_{*,p}(y;\epsilon) \mathcal{T}_{*,p}(y,\zeta_\tf(\epsilon);\lambda,\epsilon) \hat{\Y}(y) \de y,\end{align*}
on $C([\zeta_\tf(\epsilon),0],\C^{4 \times 2})$ has norm $\|\F_{\lambda,\epsilon}\| \leq C \delta$, provided $0 < \epsilon_1 \ll \Theta_1 \ll \delta \ll 1$ and $0 \leq |\epsilon_2| \ll \Theta_1 \ll \delta \ll 1$. Hence, $I - \F_{\lambda,\epsilon}$ is invertible on $C([\zeta_\tf(\epsilon),0],\C^{4 \times 2})$ and
\begin{align*}
 \hat{\Y}(\zeta;\lambda,\epsilon) = \left[(I - \F_{\lambda,\epsilon})^{-1} \begin{pmatrix} I_2 \\ \tilde{T}_+(\zeta_\tf(\epsilon);\lambda,\epsilon)\end{pmatrix} \right](\zeta),
\end{align*}
is the solution to~\eqref{varcon2}. Using Proposition~\ref{prop:plateau2}, we approximate
\begin{align*}
\left\|\hat{\Y}(\zeta;\lambda,\epsilon) - \begin{pmatrix} I_2 \\ \tilde{T}_+(\zeta_\tf(\epsilon);\lambda,\epsilon)\end{pmatrix}\right\| &\leq \|\F_{\lambda,\epsilon}\|\left\|(I - \F_{\lambda,\epsilon})^{-1}\right\| \left\|\begin{pmatrix} I_2 \\ \tilde{T}_+(\zeta_\tf(\epsilon);\lambda,\epsilon)\end{pmatrix}\right\| \leq C\delta.
\end{align*}
Hence, as $\hat{\Y}(\zeta_\tf(\epsilon);\lambda,\epsilon) = \Y(\zeta_\tf(\epsilon);\lambda,\epsilon)$, we establish
\begin{align} \left\|\Y(\zeta_\tf(\epsilon);\lambda,\epsilon) - \begin{pmatrix} I_2 \\ \tilde{T}_+(\zeta_\tf(\epsilon);\lambda,\epsilon)\end{pmatrix}\right\| \leq C\delta. \label{yboundr}\end{align}
Upon denoting $\Y(\zeta_\tf(\epsilon);\lambda,\epsilon) = (\Y_1(\lambda,\epsilon),\Y_2(\lambda,\epsilon))$, we approximate using Proposition~\ref{prop:plateau3} and estimate~\eqref{yboundr}:
\begin{align*} \left\|\hat{T}_+(\zeta_\tf(\epsilon);\lambda,\epsilon) - \tilde{T}_+(\zeta_\tf(\epsilon);\lambda,\epsilon)\right\|
 &= \left\|\Y_2 \Y_1^{-1} - \tilde{T}_+\right\| \leq \left\|\hat{T}_+\right\|\|I_2 - \Y_1\| + \|\Y_2-\tilde{T}_+\|\\
 &\leq \left(\left\|\hat{T}_+ - \tilde{T}_+\right\| + \|\tilde{T}_+\|\right) \|I_2 - \Y_1\| + \|\Y_2 - \tilde{T}_+\|\\
 &\leq C\left(\left\|\hat{T}_+ - \tilde{T}_+\right\| + \delta^{1/4}\right) \delta + C\delta,
\end{align*}
where we suppress the arguments on the right hand side. Hence, using Proposition~\ref{prop:plateau2}, we conclude~\eqref{tboundr2} holds, provided $0 < \epsilon_1 \ll \Theta_1 \ll \delta \ll 1$ and $0 \leq |\epsilon_2| \ll \Theta_1 \ll \delta \ll 1$.
\end{proof}

\begin{remark} \label{rem:forwardbackward}
{\rm 
We emphasize that a similar argument as in Proposition~\ref{prop:plateau3} to track the subspace $W_-(\zeta;\lambda,\epsilon)$ forward from $\zeta = \zeta_\tf(\epsilon)$ to $\zeta = 0$ fails, since the exponential growth of the evolution of~\eqref{evprobp} cannot be compensated as there is no exponential decay of system~\eqref{evprob2} towards system~\eqref{evprobp} when propagating backward along the plateau state. However, for the winding number argument in~\S\ref{sec:wind} on a contour $\Gamma$ enclosing the disk $D_1(\epsilon)$, we are able (and it is advantageous) to track $W_-(\zeta;\lambda,\epsilon)$ forward from $\zeta = \zeta_\tf(\epsilon)$ to $\zeta = 0$ by using that better bounds on the evolution of~\eqref{evprobp} hold for $\lambda \in \Gamma$.}
\end{remark}

\subsection{Reduction to disk \texorpdfstring{$D_1(\epsilon)$}{D1(epsilon)}} \label{sec:disk}

We prove that the alternative Riccati-Evans function $\tilde{\E}_\epsilon$, defined in~\S\ref{sec:invric}, has no roots in $R_1(\Theta_1) \setminus \bigcup_{j \in \Z \setminus \{0\}} D_j(\epsilon)$, where $D_j(\epsilon)$ are the disks defined in Proposition~\ref{prop:plateau1}.

\begin{theorem} \label{conclR1outD1}
Provided $0 < \epsilon_1 \ll \Theta_1 \ll \delta \ll 1$ and $0 \leq |\epsilon_2| \ll \Theta_1 \ll \delta \ll 1$, the operator $\hat{\El}_\tf$, posed on $\smash{L^2_{\hat{\kappa}_-,\hat{\kappa}_+}(\R,\C^2)}$, has no point spectrum in
$$R_1(\Theta_1) \setminus \bigcup_{j \in \Z \setminus \{0\}} D_j(\epsilon),$$
where the disks $D_j(\epsilon)$ are as in Proposition~\ref{prop:plateau1}.
\end{theorem}
\begin{proof}
Let $\tilde{\E}_\epsilon$ be as in~\S\ref{sec:invric}. Provided $0 < \epsilon_1 \ll \Theta_1 \ll \delta \ll 1$ and $0 \leq |\epsilon_2| \ll \Theta_1 \ll \delta \ll 1$, we employ~\eqref{deltaineq} in Proposition~\ref{prop:left},~\eqref{approxricl} in Lemma~\ref{lem:approx} and~\eqref{tboundr2} in Proposition~\ref{prop:plateau3} to yield
\begin{align}
\left|\tilde{\E}_{\epsilon}(\lambda) - \hat{z}_*(\zeta_\delta)^2\right| \leq C|\log(\delta)|^{-3}, \qquad \text{for } \lambda \in R_1(\Theta_1) \setminus \bigcup_{j \in \Z \setminus \{0\}} D_j(\epsilon), \label{rouche}\end{align}
where $C > 1$ is an $\epsilon$-, $\lambda$- and $\delta$-independent constant. Using~\eqref{deltaineq}, we conclude from~\eqref{rouche} that the meromorphic function $\tilde{\E}_\epsilon$ possesses neither poles nor roots on $R_1(\Theta_1) \setminus \bigcup_{j \in \Z \setminus \{0\}} D_j(\epsilon)$. Hence, by the considerations in~\S\ref{sec:invric}, the operator $\hat{\El}_\tf$, posed on $\smash{L^2_{\hat{\kappa}_-,\hat{\kappa}_+}(\R,\C^2)}$, has no point spectrum $R_1(\Theta_1) \setminus \bigcup_{j \in \Z \setminus \{0\}} D_j(\epsilon)$
\end{proof}

By estimate~\eqref{radiusb} in Proposition~\ref{prop:plateau1}, $D_1(\epsilon) = D_{-1}(\epsilon)$ is the only disk $D_j(\epsilon), j \in \Z \setminus \{0\}$ which intersects the closed right half plane. Thus, all that remains is to control the point spectrum in $D_1(\epsilon)$. In the following we proceed as outlined in~\S\ref{sec:approachpoint}.

\subsection{The derivative of the Riccati-Evans function at \texorpdfstring{$\lambda = 0$}{lambda = 0}} \label{sec:ricderiv2}

In this subsection, we show that $\lambda = 0$ is a simple root of the Riccati-Evans function $\E_\epsilon$ and we approximate the derivative $\E_\epsilon'(0)$.

We showed in~\S\ref{sec:lambda0} that gauge invariance yields the nontrivial solution $\phi_0 \in \smash{H^1_{\hat{\kappa}_-,\hat{\kappa}_+}(\R,\C^4)}$, given by~\eqref{eigenfunction0}, to the eigenvalue problem~\eqref{evprob} at $\lambda = 0$. Hence, $\lambda = 0$ is an eigenvalue of the operator $\hat{\El}_\tf$, posed on $\smash{L^2_{\hat{\kappa}_-,\hat{\kappa}_+}(\R,\C^4)}$. It readily follows by Proposition~\ref{prop:ric} that $\lambda = 0$ is a root of the associated Riccati-Evans function.

To approximate the derivative of the Riccati-Evans function we use the expressions~\eqref{clasderiv1} and~\eqref{ricderiv1} in~\S\ref{secderric}. Thus, on the one hand, we need bases of $\ker(P_-(0;0,\epsilon))$ and $P_+(0;0,\epsilon)[\C^4]$, where we recall that $P_-( \zeta;\lambda,\epsilon)$ and $P_+( \zeta;\lambda,\epsilon)$ are the projections associated with the exponential dichotomies of~\eqref{evprob} on $(-\infty,0]$ and $[0,\infty)$, respectively, which were established in~\S\ref{sec:ricdef}. On the other hand, we need to control a nontrivial solution $\psi_0(\zeta;\epsilon)$ to the adjoint system
\begin{align}
\psi_\zeta = -A(\zeta;0,\epsilon)^* \psi, \qquad \psi \in \C^4, \label{adjointpr}
\end{align}
with initial condition $\psi_0(0;\epsilon) \in \ker(P_-(0;0,\epsilon))^\perp \cap P_+(0;0,\epsilon)[\C^{4}]^\perp$.

To find the desired bases, we use the explicit solutions~\eqref{eigenfunction1} to~\eqref{variationaleq} obtained in~\S\ref{sec:lambda0}.

\begin{lemma} \label{lemma:bases}
Let $\phi_0(\zeta)$ and $\phi_\pm(\zeta;\epsilon)$ be as in~\eqref{eigenfunction0} and~\eqref{eigenfunction1}, respectively. Provided $0 < \epsilon_1 \ll 1$ and $0 \leq |\epsilon_2| \ll 1$, $\left(\phi_0(0) \mid \phi_\pm(0,\epsilon)\right)$ are bases of~$\ker(P_-(0;0,\epsilon))$ and $P_+(0;0,\epsilon)[\C^{4}]$, respectively. In addition, it holds $\ker(P_-(0;0,\epsilon)) \cap P_+(0;0,\epsilon)[\C^{4}] = \mathrm{Span}\{\phi_0(0)\}$.
\end{lemma}
\begin{proof}
By Lemma~\ref{lem:approxphipm} the solutions $\phi_\pm(\zeta;\epsilon)$ are bounded as $\zeta \to \pm \infty$, whereas the linearly independent solution $\phi_0(\zeta)$ is bounded on $\R$. Hence, $\left(\phi_0(0) \mid \phi_\pm(0,\epsilon)\right)$ are bases of~$\ker(P_-(0;0,\epsilon))$ and $P_+(0;0,\epsilon)[\C^{4}]$, respectively. In addition, the estimate Lemma~\ref{lem:approxphipm} shows that $\phi_+(0;\epsilon)$ and $\phi_-(0;\epsilon)$ are linearly independent, which proves $\ker(P_-(0;0,\epsilon)) \cap P_+(0;0,\epsilon)[\C^{4}] = \mathrm{Span}\{\phi_0(0)\}$.
\end{proof}

Lemma~\ref{lemma:bases} shows that the eigenvalue $\lambda = 0$ of the operator $\hat{\El}_\tf$, posed $\smash{L^2_{\hat{\kappa}_-,\hat{\kappa}_+}(\R,\C^4)}$, has geometric multiplicity $1$, and thus the formulas~\eqref{clasderiv1} and~\eqref{ricderiv1} in~\S\ref{secderric} can indeed be employed to compute the derivative of the Riccati-Evans function at $\lambda = 0$.

The next result provides control over the relevant nontrivial solution $\psi_0(\zeta;\epsilon)$ to the adjoint system~\eqref{adjointpr}.

\begin{lemma} \label{lemma:adjoint}
There exists a constant $C > 1$ such that, provided $0 < \epsilon_1 \ll \delta \ll 1$ and $0 \leq |\epsilon_2| \ll \delta \ll 1$, there exists a solution $\psi_0(\zeta;\epsilon)$ to the adjoint equation~\eqref{adjointpr} with initial condition $\psi_0(0;\epsilon) \in \ker(P_-(0;0,\epsilon))^\perp \cap P_+(0;0,\epsilon)[\C^{4}]^\perp$, which enjoys the estimates
\begin{align}
\begin{split}
\|\psi_0(\zeta;\epsilon)\| &\leq C|\log(\delta)|^2 \re^{\hat{\kappa}_-\left(\zeta_\tf(\epsilon) - \zeta\right)}, \qquad \zeta \leq \zeta_\tf(\epsilon), \\
\|\psi_0(\zeta;\epsilon)\| &\leq C \re^{-\hat{\kappa}_+\zeta}, \qquad \zeta \geq 0,
\end{split} \label{expdecayadjoint}
\end{align}
and
\begin{align}
\left\|\psi_0(\zeta;\epsilon) - \tilde{\psi}_0(\zeta;\epsilon)\right\| \leq C|\log(\delta)|^4\delta \re^{-\frac{\eta}{2}\left(\zeta - \zeta_\tf(\epsilon)\right)}, \qquad \zeta \in [\zeta_\tf(\epsilon),0], \label{adjointplateau}
\end{align}
where we denote
\begin{align*}
\tilde{\psi}_0(\zeta;\epsilon) := \left(0, 0,-\frac{1}{\overline{\partial_\zeta \hat{z}_0(\zeta;\epsilon)}}, \displaystyle\frac{1}{\partial_\zeta \hat{z}_0(\zeta;\epsilon)}\right)^\top.
\end{align*}
\end{lemma}
\begin{proof} In this proof $C > 1$ denotes any constant, which is independent of $\epsilon$ and $\delta$.

Using the bases of $\ker(P_-(0;0,\epsilon))$ and $P_+(0;0,\epsilon)[\C^{4}]$ derived in Lemma~\ref{lemma:bases}, we can, by Proposition~\ref{prop:right}, take a vector
\begin{align}
\begin{split}
\psi_0 &\in \ker(P_-(0;0,\epsilon))^\perp \cap P_+(0;0,\epsilon)[\C^{4}]^\perp\\ &= \ker(I_4 - P_-(0;0,\epsilon)^*) \cap (I_4 - P_+(0;0,\epsilon)^*)[\C^4].\end{split} \label{kernrang}
\end{align}
satisfying
\begin{align}
\left\|\psi_0 - \left(0, 0,-\frac{1}{\overline{\partial_\zeta \hat{z}_0(0;\epsilon)}}, \displaystyle\frac{1}{\partial_\zeta \hat{z}_0(0;\epsilon)}\right)^\top\right\| \leq Ce^{\eta \zeta_\tf(\epsilon)}, \label{adjointest0}
\end{align}
where we observe from~\eqref{dyninvman} and~\eqref{hatzapprox} that it holds
\begin{align}\left|\partial_\zeta \hat{z}_0(0;\epsilon) + 2-2\ri\alpha\right| \leq C\|\epsilon\|, \label{adjointest1}
\end{align}
noting $\chi(0) = -1$. Now, let $\psi_0(\zeta;\epsilon)$ be the solution to~\eqref{adjointpr} with initial condition $\psi_0(0;\epsilon) = \psi_0$. Since the inner product between solutions to~\eqref{variationaleq} and~\eqref{adjointpr} is preserved, we find that~\eqref{kernrang} implies that
\begin{align}
\psi_0(\zeta;\epsilon) &\in \ker(P_-(\zeta;0,\epsilon))^\perp = \ker(I_4 - P_-(\zeta;0,\epsilon)^*). \label{kernrang2}
\end{align}
holds for all $\zeta \leq 0$. Next, we prove the estimates~\eqref{expdecayadjoint} and~\eqref{adjointplateau} hold for $\psi_0(\zeta;\epsilon)$.

\paragraph{Exponential dichotomies.} To prove the estimates~\eqref{expdecayadjoint}, we establish exponential dichotomies for the weighted systems
\begin{align}
\psi_\zeta &= \left(-A(\zeta;0,\epsilon)^* + \hat{\kappa}_\pm\right) \psi, \qquad \psi \in \C^4, \label{shiftedad}
\end{align}
on $(-\infty,\zeta_\tf(\epsilon)]$ and $[0,\infty)$, respectively, with $\epsilon$-independent constants. These exponential dichotomies arise by transferring exponential dichotomies of
\begin{align}
 \phi_\zeta &= \left(A(\zeta;0,\epsilon) - \hat{\kappa}_-\right)\phi, \qquad \phi \in \C^4, \label{shiftedm}\\
 \phi_\zeta &= \left(A(\zeta;0,\epsilon) - \hat{\kappa}_+\right)\phi, \qquad \phi \in \C^4,\label{shifted+}
\end{align}
to their adjoint problems.

We start by establishing an exponential dichotomy for~\eqref{shiftedm} on $(-\infty,\zeta_\tf(\epsilon)]$. By Proposition~\ref{prop:left} and~\eqref{limmum} it holds
\begin{align}
 \left\|A(\zeta_\tf(\epsilon) + \zeta;0,\epsilon) - A_L(\zeta_\delta + \zeta)\right\| \leq C_\delta \|\epsilon\|, \qquad \zeta \leq 0,\label{leftbound8}
\end{align}
provided $0 < \epsilon_1 \ll \delta \ll 1$ and $0 \leq |\epsilon_2| \ll \delta \ll 1$, where $C_\delta > 0$ is an $\epsilon$-independent constant and we denote
\begin{align*}
A_L(\zeta) := \begin{pmatrix} 0 & 0 & 1 & 0 \\ 0 & 0 & 0 & 1\\ R_*(\zeta) & R_*(\zeta) & - 2\hat{z}_*(\zeta) & 0 \\ R_*(\zeta) & R_*(\zeta) & 0 & - 2\hat{z}_*(\zeta) \end{pmatrix}.
\end{align*}
By estimate~\eqref{Rsbound} in Proposition~\ref{prop:left}, the coefficient matrix $A_L(\zeta)$ converges exponentially to the asymptotic matrix
\begin{align*}
A_L^{-\infty} = \begin{pmatrix} 0 & 0 & 1 & 0 \\ 0 & 0 & 0 & 1\\ 1 & 1 & -2 & 0 \\ 1 & 1 & 0 & -2 \end{pmatrix},
\end{align*}
as $\zeta \to -\infty$, whose spectrum consists of the eigenvalues $0,-2,-1\pm\sqrt{3}$. The matrix $A_L^{- \infty}$ admits a spectral gap at $\Re(\nu) = \hat{\kappa}_- = -\kappa/\sqrt{1+\alpha^2} \in (-\tfrac{1}{4},0)$. Hence, by~\cite[Lemma 3.4]{PAL}, system
\begin{align*}
\phi_\zeta = \left(A_L(\zeta) - \hat{\kappa}_-\right)\phi, \qquad \phi \in \C^4,
\end{align*}
has an exponential dichotomy on $(-\infty,0]$ with associated rank 2 projections $P_L(\zeta)$. By roughness of exponential dichotomies, cf.~\cite[Proposition 5.1]{COP}, and estimate~\eqref{leftbound8}, system~\eqref{shiftedm} admits, provided $0 < \epsilon_1 \ll \delta \ll 1$ and $0 \leq |\epsilon_2| \ll \delta \ll 1$, an exponential dichotomy on $(-\infty,\zeta_\tf(\epsilon)]$ with $\epsilon$-independent constants and projections $P_L(\zeta;\epsilon)$. Thus, cf.~\cite[p. 17]{COP}, it must hold
\begin{align}
\ker(P_L(\zeta;\epsilon)) = \ker(P_-(\zeta;0,\epsilon)), \qquad \zeta \in (-\infty,\zeta_\tf(\epsilon)]. \label{kernels}
\end{align}

Next, we establish an exponential dichotomy for~\eqref{shifted+} on $(-\infty,\zeta_\tf(\epsilon)]$. By~\eqref{limmum},~\eqref{hatzapprox} and Proposition~\ref{prop:right} we have
\begin{align}
 \left\|A(\zeta;0,\epsilon) - A_R\right\| \leq C\|\epsilon\|, \qquad \zeta \geq 0, \label{rightbound8}
\end{align}
provided $0 < \epsilon_1 \ll \delta \ll 1$ and $0 \leq |\epsilon_2| \ll \delta \ll 1$, where we denote
\begin{align*}
A_R = \begin{pmatrix} 0 & 0 & 1 & 0 \\ 0 & 0 & 0 & 1\\ 0 & 0 & 2\sqrt{2-2\ri \alpha} & 0 \\ 0 & 0 & 0 & 2\sqrt{2+2\ri \alpha} \end{pmatrix}.
\end{align*}
Since $A_R$ admits the spatial eigenvalues $0,0$, $2\sqrt{2 \pm 2\ri \alpha}$ and it holds $$0 < \hat{\kappa}_+ = 1 + \Re\sqrt{2+2\ri \alpha} + \kappa\sqrt{1+\alpha^2} \leq \frac{5}{4} + \Re\sqrt{2+2\ri \alpha} < 2\Re\sqrt{2+2\ri \alpha},$$ the constant coefficient system
\begin{align*}
\phi_\zeta = \left(A_R - \hat{\kappa}_+\right)\phi, \qquad \phi \in \C^4,
\end{align*}
has an exponential dichotomy on $\R$. By roughness of exponential dichotomies, cf.~\cite[Proposition 5.1]{COP}, and estimate~\eqref{rightbound8}, system~\eqref{shifted+} admits, provided $0 < \epsilon_1 \ll \delta \ll 1$ and $0 \leq |\epsilon_2| \ll \delta \ll 1$, an exponential dichotomy on $[0,\infty)$ with $\epsilon$-independent constants and projections $P_R(\zeta;\epsilon)$. Thus, cf.~\cite[p. 17]{COP}, it must hold
\begin{align}
P_R(\zeta;\epsilon)[\C^4] = P_+(\zeta;0,\epsilon)[\C^4], \qquad \zeta \in [0,\infty). \label{ranges}
\end{align}

All in all, we conclude the adjoint systems~\eqref{shiftedad} have exponential dichotomies on $(-\infty,\zeta_\tf(\epsilon)]$ and $[0,\infty)$, respectively, with $\epsilon$-independent constants and corresponding projections $I_4-P_L(\zeta;\epsilon)^*$ and $I_4-P_R(\zeta;\epsilon)^*$. Thus, by identities~\eqref{kernrang},~\eqref{kernrang2},~\eqref{kernels} and~\eqref{ranges}, we obtain the exponential decay estimates
\begin{align}
\begin{split}
\|\psi_0(\zeta;\epsilon)\| &\leq C \re^{\hat{\kappa}_-\left(\zeta_\tf(\epsilon) - \zeta\right)}\|\psi_0(\zeta_\tf(\epsilon);\epsilon)\|, \qquad \zeta \leq \zeta_\tf(\epsilon),\\
\|\psi_0(\zeta;\epsilon)\| &\leq C \re^{-\hat{\kappa}_+\zeta}\|\psi_0(0;\epsilon)\| \leq C \re^{-\hat{\kappa}_+\zeta}, \qquad \zeta \geq 0,
\end{split} \label{expdecayadjoint2}
\end{align}
where the last inequality follows from~\eqref{adjointest0} and~\eqref{adjointest1}. Thus, proving~\eqref{expdecayadjoint} now reduces to finding an appropriate bound on $\|\psi_0(\zeta_\tf(\epsilon);\epsilon)\|$, which we will establish later.

\paragraph*{Tracking the adjoint solution along the absolutely unstable plateau.}
Along the absolutely unstable plateau, we approximate the coefficient matrix of the adjoint equation~\eqref{adjointpr} using~\eqref{estex2} in Proposition~\ref{prop:right} as
\begin{align}
\left\|A(\zeta;0,\epsilon)^* + A_{\ad}(\zeta;\epsilon)\right\| \leq C\delta \re^{-\eta(\zeta - \zeta_\tf(\epsilon))}, \qquad \qquad \zeta \in [\zeta_\tf(\epsilon),0], \label{estadjoint1}
\end{align}
where we denote
\begin{align*} A_{\ad}(\zeta;\epsilon):=\begin{pmatrix} 0 & 0 & 0 & 0 \\ 0 & 0 & 0 & 0 \\ -1 & 0 & 2\overline{\hat{z}_0(\zeta;\epsilon)} & 0 \\ 0 & -1 & 0 & 2\hat{z}_0(\zeta;\epsilon) \end{pmatrix}.\end{align*}
One observes that system
\begin{align*}
\psi_\zeta = A_\ad(\zeta;\epsilon)\psi,
\end{align*}
is explicitly solvable and its evolution matrix is given by
\begin{align*}
\mathcal{T}_{\ad}(\zeta,y;\epsilon) &:= B(\zeta;\epsilon)^{-1} \tilde{\mathcal{T}}_{\ad}(\zeta,y;\epsilon), \qquad B(\zeta;\epsilon) := \begin{pmatrix} \overline{\hat{z}_0'(\zeta)} & 0 & 0 & 0 \\ 0 & \hat{z}_0'(\zeta) & 0 & 0 \\ 0 & 0 & \overline{\hat{z}_0'(\zeta)} & 0 \\ 0 & 0 & 0 & \hat{z}_0'(\zeta)\end{pmatrix}, \\
 \tilde{\mathcal{T}}_{\ad}(\zeta,y;\epsilon) &:= \begin{pmatrix}  \overline{\hat{z}_0'(\zeta)} & 0 & 0 & 0 \\ 0 &  \hat{z}_0'(\zeta) & 0 & 0 \\ \overline{\hat{z}_0(y) - \hat{z}_0(\zeta)} & 0 & \overline{\hat{z}_0'(y)} & 0 \\ 0 & \hat{z}_0(y) - \hat{z}_0(\zeta) & 0 & \hat{z}_0'(y)\end{pmatrix},
\end{align*}
where we suppress the $\epsilon$-dependency in the coefficient matrices. We emphasize $\hat{z}_0'(\zeta) \neq 0$ for all $\zeta \in [\zeta_\tf(\epsilon),0]$, since the solution $\hat{z}_0(\zeta)$ does not lie on a stationary point of the planar system~\eqref{dyninvman}.

Using the variation of constants formula we find that $\psi_\ad(\zeta;\epsilon) := B(\zeta;\epsilon)\psi_0(\zeta;\epsilon)$ solves
\begin{align}
 \psi_\ad(\zeta;\epsilon) = \tilde{\mathcal{T}}_{\ad}(\zeta,0;\epsilon)\psi_0(0;\epsilon) + \int_0^\zeta \tilde{\mathcal{T}}_{\ad}(\zeta,y;\epsilon) F(y;\epsilon) \psi_\ad(y;\epsilon) \de y, \qquad \zeta \in [\zeta_\tf(\epsilon),0], \label{varconstadj}
\end{align}
with $F(y;\epsilon) := -\left(A(y;0,\epsilon)^* + A_{\ad}(y;\epsilon)\right)B(y;\epsilon)^{-1}$. Our plan is to approximate $\psi_\ad(\zeta;\epsilon)$ by estimating the terms in~\eqref{varconstadj}. First, using~\eqref{estex3} in Proposition~\ref{prop:right}, it holds
\begin{align}
\left\|\tilde{\mathcal{T}}_{\ad}(\zeta,y;\epsilon)\right\| \leq C, \qquad \zeta,y \in [\zeta_\tf(\epsilon),0]. \label{adjointest2}
\end{align}
Second, using that $\hat{z}_0(\zeta;\epsilon)$ solves~\eqref{dyninvman}, we obtain the identity
\begin{align}
 \frac{\hat{z}_0'(y)}{\hat{z}_0'(\zeta)} = \re^{2\int_y^\zeta \hat{z}_0(x) \de x}, \qquad \zeta,y \in [\zeta_\tf(\epsilon),0], \label{adjointid}
\end{align}
where we suppress $\epsilon$-dependency. On the other hand, by~\eqref{deltaineq} and~\eqref{estex1} in Proposition~\ref{prop:left} and by~\eqref{estex2} in Proposition~\ref{prop:right}, it holds
\begin{align}
C\,\Re\left(\hat{z}_0\left(\zeta_\tf(\epsilon);\epsilon\right)\right) \geq |\log(\delta)|^{-1}. \label{lowerboundz0}
\end{align}
Now let $\zeta \in [\zeta_\tf(\epsilon),0]$ and consider the $\epsilon$- and $\delta$-independent constant $\eta > 0$ in~\eqref{estadjoint1}. Take $\zeta_1$ to be the smallest number in $[\zeta_\tf(\epsilon),\zeta]$ such that $|\hat{z}_0(y;\epsilon)| \leq \frac{\eta}{8}$ for $y \in [\zeta_1,\zeta]$, if such a number exists, otherwise take $\zeta_1 = \zeta$. Thus, it either holds $|\hat{z}_0(\zeta_1,\epsilon)| \geq \frac{\eta}{8}$ or $\zeta_1 = \zeta_\tf(\epsilon)$. Using~\eqref{limmum},~\eqref{adjointid} and~\eqref{lowerboundz0}, we approximate
\begin{align}
\frac{1}{\left|\partial_\zeta \hat{z}_0(\zeta;\epsilon)\right|} \leq \frac{e^{2\Re\int_{\zeta_1}^\zeta \hat{z}_0(y;\epsilon) \de y}}{\left|\partial_\zeta \hat{z}_0(\zeta_1,\epsilon)\right|} \leq
 \frac{e^{2\int_{\zeta_1}^\zeta |\hat{z}_0(y;\epsilon)| \de y}}{\left|\hat{z}_0(\zeta_1,\epsilon)\right|^2 - |\mu(\epsilon)|} \leq C |\log(\delta)|^2 \re^{\frac{\eta}{4}\left(\zeta - \zeta_\tf(\epsilon)\right)}, \label{adjointest4}
\end{align}
for $\zeta \in [\zeta_\tf(\epsilon),0]$, provided $0 < \epsilon_1 \ll \delta \ll 1$ and $0 \leq |\epsilon_2| \ll \delta \ll 1$. Combining the latter with~\eqref{estadjoint1} yields
\begin{align}
\|F(\zeta;\epsilon)\| \leq C\delta|\log(\delta)|^2 \re^{-\frac{3\eta}{4}\left(\zeta - \zeta_\tf(\epsilon)\right)}, \qquad \zeta \in [\zeta_\tf(\epsilon),0]. \label{adjointest3}
\end{align}
Fix $\zeta \in [\zeta_\tf(\epsilon),0]$. By~\eqref{adjointest2} and~\eqref{adjointest3} the linear operator $\F_{\zeta,\epsilon}$ given by
\begin{align*} \left(\F_{\zeta,\epsilon}\psi\right)(x) = \int_0^x \tilde{\mathcal{T}}_{\ad}(x,y;\epsilon) F(y;\epsilon) \psi(y) \de y,\end{align*}
on $C([\zeta,0],\C^4)$ has norm $\|\F_{\zeta,\epsilon}\| \leq C|\log(\delta)|^2\delta \re^{-\frac{3\eta}{4}\left(\zeta - \zeta_\tf(\epsilon)\right)}$. Upon taking $\delta > 0$ sufficiently small (but independent of $\zeta$ and $\epsilon$), $I - \F_{\zeta,\epsilon}$ is invertible on $C([\zeta,0],\C^4)$ and
\begin{align*}
 \psi_{\ad}(\zeta;\epsilon) = \left[(I - \F_{\zeta,\epsilon})^{-1} \psi_{0,\ad}(\cdot,\epsilon)\right](\zeta),
\end{align*}
is the solution to~\eqref{varconstadj}, where we denote $\psi_{0,\ad}(y;\epsilon) := \tilde{\mathcal{T}}_{\ad}(y,0;\epsilon)\psi_0(0;\epsilon)$ for $y \in [\zeta_\tf(\epsilon),0]$. Using~\eqref{adjointest0} and~\eqref{adjointest2} we approximate
\begin{align*}
\left\|\psi_{0,\ad}(\zeta;\epsilon) - \left(0, 0, -1, 1\right)^\top\right\| &\leq Ce^{\eta \zeta_\tf(\epsilon)},
\end{align*}
and
\begin{align*}
\|\psi_{\ad}(\zeta;\epsilon) - \psi_{0,\ad}(\zeta;\epsilon)\| &\leq \|\F_{\zeta,\epsilon}\|\|(I - \F_{\zeta,\epsilon})^{-1}\| \sup_{y \in [\zeta,0]} \|\psi_{0,\ad}(y;\epsilon)\|
\leq C\frac{|\log(\delta)|^2\delta}{e^{\frac{3\eta}{4}\left(\zeta - \zeta_\tf(\epsilon)\right)}}.
\end{align*}

Combining the latter with~\eqref{adjointest4}, proves~\eqref{adjointplateau}. Finally, evaluating~\eqref{adjointplateau} at $\zeta = \zeta_\tf(\epsilon)$ yields~\eqref{expdecayadjoint} by combining~\eqref{expdecayadjoint2} with~\eqref{lowerboundz0}, where we use that $\hat{z}_0(\zeta;\epsilon)$ solves~\eqref{dyninvman}.
\end{proof}

We are now in a position to compute the derivative of the Riccati-Evans function at $\lambda = 0$ using the expressions~\eqref{clasderiv1} and~\eqref{ricderiv1} in~\S\ref{secderric}.

\begin{theorem} \label{theo:deriv}
There exists a constant $C > 1$ such that, provided $0 < \epsilon_1 \ll \delta \ll 1$ and $0 \leq |\epsilon_2| \ll \delta \ll 1$, it holds
\begin{align}
\left|\epsilon_1^3 \E_{\epsilon}'(0) - \frac{2\pi\left(1+\alpha^2\right)}{1 + \Re\sqrt{2+2\ri\alpha}}\right| \leq C\|\epsilon\|. \label{derivest}
\end{align}
\end{theorem}
\begin{proof}  In this proof $C > 1$ denotes any $\epsilon$- and $\delta$-independent constant.

By Lemma~\ref{lemma:bases} the subspace $\ker(P_-(0;0,\epsilon)) \cap P_+(0;0,\epsilon)[\C^{4}] = \mathrm{Span}\{\phi_0(0)\}$ is one-dimensional. Thus, we can substitute the relevant expressions in~\eqref{clasderiv1} and~\eqref{ricderiv1} and obtain
\begin{align}
 \E_{\epsilon}'(0) = \frac{\det\left(\psi_0(0;\epsilon) \mid \phi_-(0;\epsilon) \mid \phi_0(0) \mid \phi_+(0;\epsilon)\right) \displaystyle \int_\R \psi_0(\zeta;\epsilon)^* \partial_\lambda A(\zeta;0,\epsilon) \phi_0(\zeta) \mathrm{d} \zeta}{\|\psi_0(0;\epsilon)\|^2 \det\left(\tilde{\phi}_0(0) \mid \tilde{\phi}_-(0;\epsilon)\right) \det\left(\tilde{\phi}_0(0) \mid \tilde{\phi}_+(0;\epsilon)\right)}, \label{ricderiv}
\end{align}
where we recall $\left(\phi_0(0) \mid \phi_\pm(0;\epsilon)\right) \in \C^{4 \times 2}$ are the bases of $\ker(P_-(0;0,\epsilon))$ and $P_+(0;0,\epsilon)[\C^{4}]$, respectively, obtained in Lemma~\ref{lemma:bases}, $\psi_0(\zeta;\epsilon)$ is the nontrivial solution with $$\psi_0(0;\epsilon) \in \ker(P_-(0;0,\epsilon))^\perp \cap P_+(0;0,\epsilon)[\C^{4}]^\perp,$$ to the adjoint problem~\eqref{adjointpr}, obtained in Lemma~\ref{lemma:adjoint}, and $\tilde{\phi}_0(0), \tilde{\phi}_\pm(0;\epsilon) \in \C^2$ denote the upper two entries of the vectors $\phi_0(0)$ and $\phi_\pm(0;\epsilon)$, respectively.

We approximate the Melnikov integral in~\eqref{ricderiv}, with the aid of Lemma~\ref{lemma:bases}, as
\begin{align}
 \left|\int_\R \psi_0(\zeta;\epsilon)^* \partial_\lambda A(\zeta;0,\epsilon) \phi_0(\zeta) \mathrm{d} \zeta - \int_{\zeta_\tf(\epsilon)}^0 \tilde{\psi}_0(\zeta;\epsilon)^* \partial_\lambda A(\zeta;0,\epsilon) \phi_0(\zeta) \mathrm{d} \zeta\right| \leq C|\log(\delta)|^2. \label{melnikov2}
\end{align}
We will use Cauchy's integral theorem to approximate the remaining integral
\begin{align}
 \int_{\zeta_\tf(\epsilon)}^0 \tilde{\psi}_0(\zeta;\epsilon)^*\, \partial_\lambda A(\zeta;0,\epsilon)\, \phi_0(\zeta) \mathrm{d} \zeta =
 \frac{-2}{m(\epsilon)^2} \Re\left(\int_{\zeta_\tf(\epsilon)}^0 \frac{1 - \ri \alpha}{\partial_\zeta \hat{z}_0(\zeta;\epsilon)} \mathrm{d} \zeta\right). \label{melnikov}
\end{align}
Using~\eqref{preciseomegabound} in Theorem~\ref{t:ex02}, we deduce from~\eqref{defmm} that it holds
\begin{align} \left|\mu\left(\epsilon\right) + \epsilon_1^2\right| \leq C\epsilon_1^3, \qquad \left|m(\epsilon) - 1\right| \leq C\epsilon_1^2. \label{approxmu}\end{align}
The $C^1$-curve $\nu_{\epsilon} \colon [\zeta_\tf(\epsilon),0] \to \C$ given by $\nu_{\epsilon}(\zeta) = \frac{\hat{z}_0(\zeta;\epsilon)}{\sqrt{-\mu(\epsilon)}},$
satisfies the differential equation
\begin{align}
\nu_\zeta = -\sqrt{-\mu(\epsilon)} \left(\nu^2 + 1\right). \label{dyninv}
\end{align}
Thus, using estimates~\eqref{deltaineq} and~\eqref{estex1} in Proposition~\ref{prop:left}, estimate~\eqref{estex2} in Proposition~\ref{prop:right} and~\eqref{hatzapprox}, the begin and end point of the curve $\nu_\epsilon$ are approximated by
\begin{align}
 \left|\nu_{\epsilon}(0) + \frac{\sqrt{2-2\ri \alpha}}{\epsilon_1}\right| &\leq C, \qquad C\,\Re\left[\nu_{\epsilon}(\zeta_\tf(\epsilon))\right] \geq \frac{1}{|\log(\delta)|\epsilon_1}, \qquad \left|\nu_{\epsilon}(\zeta_\tf(\epsilon))\right| \leq \frac{C}{\epsilon_1}. \label{approxnu}
\end{align}
Since it holds $\Re \sqrt{-\mu(\epsilon)} > 0$ by~\eqref{approxmu}, the flow in~\eqref{dyninv} points into the left-half plane on the imaginary axis between the fixed points $\pm \ri$, whereas the flow points into the right-half plane on the other parts of the imaginary axis. Thus, since we have $\Re(\nu_{\epsilon}(0)) < 0 < \Re(\nu_{\epsilon}(\zeta_\tf(\epsilon)))$ by~\eqref{approxnu}, the curve $\nu_{\epsilon}$ must cross the imaginary at some point between the fixed points $-\ri$ and $\ri$ of~\eqref{dyninv}; see also Figure~\ref{fig:front_info}.

By~\eqref{approxnu} we can connect the point $\nu_{\epsilon}(0)$ to $\nu_{\epsilon}(\zeta_\tf(\epsilon))$ by a $C^1$-curve $\tilde{\nu}_{\epsilon} \colon [-1,1] \to \C$ of length $\leq C/\epsilon_1$ satisfying $C|\tilde{\nu}_{\epsilon}(\zeta)\epsilon_1\log(\delta)| \geq 1$ for all $\zeta \in [-1,1]$ and crossing the imaginary axis precisely once and doing so at a point above $\ri$. So, we can approximate
\begin{align} \left|\int_{\tilde{\nu}_{\epsilon}} \frac{1}{\left(z^2 + 1\right)^j}\de z\right| \leq C|\log(\delta)|^{2j} \epsilon_1^{j - 1/2}, \qquad j = 1,2. \label{intbound}\end{align}
By the previous considerations, the union of the curves $\nu_{\epsilon}$ and $\tilde{\nu}_{\epsilon}$ yields a closed $C^1$-curve $\check{\nu}_{\epsilon}$ winding $n_+ + 1$ times clockwise around $\ri$ and $n_-$ times counterclockwise around $-\ri$ for some $n_\pm \in \Z_{\geq 0}$. Applying separation of variables in~\eqref{dyninv} and using the residue theorem yields
\begin{align*}
\sqrt{-\mu(\epsilon)} \zeta_\tf(\epsilon) = \int_{\nu_{\epsilon}} \frac{1}{z^2 + 1} \de z &= \oint_{\check{\nu}_{\epsilon}} \frac{1}{z^2 + 1}\de z - \int_{\tilde{\nu}_{\epsilon}}\frac{1}{z^2 + 1} \de z = -\pi(n_+ + n_- + 1) - \int_{\tilde{\nu}_{\epsilon}} \frac{1}{z^2 + 1} \de z
\end{align*}
Combing the latter with~\eqref{triggerbound} in Theorem~\ref{t:ex02},~\eqref{approxmu} and~\eqref{intbound} yields $n_\pm = 0$. So, employing the residue theorem again, we compute
\begin{align*}
\int_{\nu_{\epsilon}} \frac{1}{\left(z^2 + 1\right)^2} \de z &= \oint_{\check{\nu}_{\epsilon}} \frac{1}{\left(z^2 + 1\right)^2}\de z - \int_{\tilde{\nu}_{\epsilon}}\frac{1}{\left(z^2 + 1\right)^2} \de z = -\frac{\pi}{2} - \int_{\tilde{\nu}_{\epsilon}}\frac{1}{\left(z^2 + 1\right)^2} \de z.
\end{align*}
Hence, by~\eqref{melnikov2},~\eqref{melnikov},~\eqref{approxmu},~\eqref{intbound} and
\begin{align*}
\int_{\zeta_\tf(\epsilon)}^0 \frac{1}{\partial_\zeta \hat{z}_0(\zeta;\epsilon)} \mathrm{d} \zeta = -\int_{\nu_{\epsilon}} \frac{1}{\mu(\epsilon_1)\sqrt{-\mu(\epsilon_1)} \left(z^2 + 1\right)^2}\de z,
\end{align*}
we thus obtain
\begin{align}
\left|\epsilon_1^3 \int_\R \tilde{\psi}_0(\zeta;\epsilon)^* \partial_\lambda A(\zeta;0,\epsilon) \phi_0(\zeta) \mathrm{d} \zeta - \pi\right|  \leq C\|\epsilon\|, \label{signmelnikov}
\end{align}
provided $0 < \epsilon_1 \ll \Theta_1 \ll 1$ and $0 \leq |\epsilon_2| \ll \delta \ll 1$. Finally, we observe from~\eqref{dyninvman} and~~\eqref{hatzapprox} that it holds
\begin{align*}\left|\partial_\zeta \hat{z}_0(0;\epsilon) + 2-2\ri\alpha\right| \leq C\|\epsilon\|,\end{align*}
noting $\chi(0) = -1$. So, using Lemma~\ref{lem:approxphipm}, estimate~\eqref{adjointplateau} in Lemma~\ref{lemma:adjoint} and~\eqref{signmelnikov} we approximate~\eqref{ricderiv} as~\eqref{derivest} and the result follows.
\end{proof}

We have now found a leading-order expression of the derivative $\E_\epsilon'(0)$, and conclude it must be of positive sign. In addition, it follows that $\lambda = 0$ is a simple zero of $\E_{\epsilon}$. Hence, by Proposition~\ref{prop:ric}, the eigenvalue $\lambda = 0$ of the operator $\hat{\El}_\tf$, posed $\smash{L^2_{\hat{\kappa}_-,\hat{\kappa}_+}(\R,\C^4)}$, has algebraic multiplicity $1$.

\subsection{Winding number computations} \label{sec:wind}

We start by computing the winding number of $\E_\epsilon$ on a simple contour $\Gamma$. As depicted in Figure~\ref{fig:disks}, we take $\Gamma$ to be a circle centered at the origin and tailor its radius such that, on the one hand, $\Gamma$ encloses the disk $D_1(\epsilon) = D_{-1}(\epsilon)$, but none of the other disks $D_j(\epsilon), j \in \Z \setminus\{0,\pm1\}$, and, on the other hand, the evolution of system~\eqref{evprobp} is well-approximated for $\lambda \in \Gamma$. Thus, we can, using variation of constants, approximate the Riccati-Evans function $\E_\epsilon(\lambda)$ for $\lambda \in \Gamma$ by a suitable analytic function and apply Rouch\'e's theorem to compute the winding number of $\E_\epsilon(\Gamma)$.

\begin{lemma} \label{lemma:windingnumber}
Provided $0 < \epsilon_1 \ll \Theta_1 \ll \delta \ll 1$ and $0 \leq |\epsilon_2| \ll \Theta_1 \ll \delta \ll 1$, the winding number of $\E_\epsilon(\Gamma)$ is $0$, where $\Gamma$ denotes the simple contour
\begin{align} \Gamma := \left\{\delta^{-1/3} M(\epsilon)^3 \re^{\ri \vartheta} : \vartheta \in [0,2\pi]\right\} \subset R_1(\Theta_1), \label{defGamma}\end{align}
and $M(\epsilon)$ is defined in~\eqref{defM}.
\end{lemma}
\begin{proof} In this proof $C > 1$ denotes any $\epsilon$-, $\lambda$- and $\delta$-independent constant.

We start by approximating the evolution $\mathcal{T}_{*,p}(\zeta,y;\lambda,\epsilon)$ of system~\eqref{evprobp} for $\lambda \in \Gamma$. Define
\begin{align*}
\varpi_1(\lambda,\epsilon) &= \sqrt{\frac{\lambda}{m(\epsilon)^2} (1 - \ri \alpha) + \mu\left(\epsilon\right)},\\
\varpi_2(\lambda,\epsilon) &= \sqrt{\frac{\lambda}{m(\epsilon)^2} (1 + \ri \alpha) + \overline{\mu\left(\epsilon\right)}}.
\end{align*}
Fix $\vartheta \in [0,2\pi]$ and let $\lambda = \delta^{-1/3} M(\epsilon)^3 \re^{\ri \vartheta} \in \Gamma$. Using~\eqref{expex4},~\eqref{expex5} and~\eqref{expex1} we obtain the expansions
\begin{align}
\begin{split}
\left|\varpi_1(\lambda,\epsilon) - M(\epsilon)\left(\ri + \frac{\delta^{-1/3} M(\epsilon)}{2} \re^{\ri\left(\vartheta - \frac{\pi}{2}\right)} (1-\ri \alpha)\right)\right| \leq CM(\epsilon)^2,\\
\left|\varpi_2(\lambda,\epsilon) - M(\epsilon)\left(-\ri + \frac{\delta^{-1/3} M(\epsilon)}{2} \re^{\ri\left(\vartheta + \frac{\pi}{2}\right)} (1+\ri \alpha)\right)\right| \leq CM(\epsilon)^2,
\end{split} \label{approxpi}
\end{align}
provided $0 < \epsilon_1 \ll \Theta_1 \ll \delta \ll 1$ and $0 \leq |\epsilon_2| \ll \Theta_1 \ll \delta \ll 1$. Directly solving the constant-coefficient system~\eqref{evprobp} yields
\begin{align} \begin{split}\mathcal{T}_{*,p}(\zeta,y;\lambda,\epsilon) &= \left(\begin{array}{cc} \cosh\left(\varpi_1(\lambda,\epsilon)(\zeta - y)\right) & 0 \\
0 & \cosh\left(\varpi_2(\lambda,\epsilon)(\zeta - y)\right)  \\
\sinh\left(\varpi_1(\lambda,\epsilon)(\zeta-y)\right) \varpi_1(\lambda,\epsilon) & 0 \\
0 & \sinh\left(\varpi_2(\lambda,\epsilon)(\zeta-y)\right) \varpi_2(\lambda,\epsilon) \end{array}\right. \ldots\\
& \qquad \qquad \qquad \qquad \ldots \left.\begin{array}{cc} \frac{\sinh\left(\varpi_1(\lambda,\epsilon)(\zeta-y)\right)}{\varpi_1(\lambda,\epsilon)} & 0\\
 0 & \frac{\sinh\left(\varpi_2(\lambda,\epsilon)(\zeta-y)\right)}{\varpi_2(\lambda,\epsilon)}\\
\cosh\left(\varpi_1(\lambda,\epsilon)(\zeta - y)\right) & 0 \\
0 & \cosh\left(\varpi_2(\lambda,\epsilon)(\zeta - y)\right)\end{array}\right).\end{split} \label{explevol}\end{align}
Thus, using Lemma~\ref{lem:evolbound},~\eqref{expex5},~\eqref{expex1} and~\eqref{approxpi}, for $\zeta \in [\zeta_\tf(\epsilon),0]$ we arrive at
\begin{align}
\begin{split}
\left\|\mathcal{T}_{*,p}(\zeta,0;\lambda,\epsilon)\right\| &\leq \left\|\mathcal{T}_{*,p}(\zeta,\zeta_\tf(\epsilon);\lambda,\epsilon)\right\| \left\|\mathcal{T}_{*,p}(\zeta_\tf(\epsilon),0;\lambda,\epsilon)\right\| \\
&\leq C\delta^{-1/3} \left(1 + |\zeta - \zeta_\tf(\epsilon)|\right) \re^{\sqrt{\|\epsilon\| + |\lambda|}|\zeta - \zeta_\tf(\epsilon)|}, \\
 \left\|\mathcal{T}_{*,p}(0,\zeta;\lambda,\epsilon)\right\| &\leq \left\|\mathcal{T}_{*,p}(0,\zeta_\tf(\epsilon);\lambda,\epsilon)\right\| \left\|\mathcal{T}_{*,p}(\zeta_\tf(\epsilon),\zeta;\lambda,\epsilon)\right\| \\
&\leq C\delta^{-1/3} \left(1 + |\zeta - \zeta_\tf(\epsilon)|\right) \re^{\sqrt{\|\epsilon\| + |\lambda|}|\zeta - \zeta_\tf(\epsilon)|}.
\end{split}
\label{evolbound2}\end{align}

Next, we fix $\lambda \in \Gamma$ and approximate $\hat{T}_-(0;\lambda,\epsilon)$. Recall from~\S\ref{sec:invric} that, under the coordinate chart $\mathfrak{c}$, the matrix $\hat{T}_-(\zeta;\lambda,\epsilon) \in \C^{2 \times 2}$ represents the subspace $W_-(\zeta;\lambda,\epsilon) \in \mathrm{Gr}(2,\C^4)$ of solutions to~\eqref{evprob2}. As in the proof of Proposition~\ref{prop:plateau3}, we denote $B_{*,p}(\zeta;\epsilon) = A_*(\zeta;\lambda,\epsilon) - A_{*,-}(\lambda,\epsilon)$. The variation of constants formula
\begin{align}
\Zet(\zeta;\lambda,\epsilon) = \mathcal{T}_{*,p}(\zeta,0;\lambda,\epsilon) \begin{pmatrix} I_2 \\ \tilde{T}_-(0;\lambda,\epsilon)\end{pmatrix} + \int_{\zeta_\tf(\epsilon)}^\zeta
\mathcal{T}_{*,p}(\zeta,y;\lambda,\epsilon) B_{*,p}(y;\epsilon) \Zet(y;\lambda,\epsilon) \mathrm dy, \label{varcon3}
\end{align}
yields a matrix solution $\Zet(\zeta;\lambda,\epsilon) \in \C^{4 \times 2}$ to~\eqref{evprob2}, which spans the subspace $W_-(\zeta;\lambda,\epsilon)$ for each $\zeta \in \R$. Indeed, $\Zet(\zeta_\tf(\epsilon);\lambda,\epsilon)$ spans the subspace which is represented by $\tilde{T}_-(\zeta_\tf(\epsilon);\lambda,\epsilon) = \hat{T}_-(\zeta_\tf(\epsilon);\lambda,\epsilon) \in \C^{2 \times 2}$ under the chart $\mathfrak{c}$

By~\eqref{varcon3}, we find that $\hat{\Zet}(\zeta;\lambda,\epsilon) := \mathcal{T}_{*,p}(0,\zeta;\lambda,\epsilon)\Zet(\zeta;\lambda,\epsilon)$ satisfies
\begin{align}
\hat{\Zet}(\zeta;\lambda,\epsilon) = \begin{pmatrix} I_2 \\ \tilde{T}_-(0;\lambda,\epsilon)\end{pmatrix} + \int_{\zeta_\tf(\epsilon)}^\zeta
\mathcal{T}_{*,p}(0,y;\lambda,\epsilon) B_{*,p}(y;\epsilon) \mathcal{T}_{*,p}(y,0;\lambda,\epsilon) \hat{\Zet}(y;\lambda,\epsilon) \mathrm dy. \label{varcon4}
\end{align}
By estimate~\eqref{plateaubound} from the proof of Proposition~\ref{prop:plateau3} and~\eqref{evolbound2} the linear operator $\G_{\lambda,\epsilon}$ given by
\begin{align*} \left(\G_{\lambda,\epsilon} \hat{\Zet}\right)(\zeta) = \int_{\zeta_\tf(\epsilon)}^\zeta \mathcal{T}_{*,p}(0,y;\lambda,\epsilon) B_{*,p}(y;\epsilon) \mathcal{T}_{*,p}(y,0;\lambda,\epsilon) \hat{\Zet}(y) \de y,\end{align*}
on $C([\zeta_\tf(\epsilon),0],\C^{4 \times 2})$ has norm $\|\G_{\lambda,\epsilon}\| \leq C \delta^{1/3}$. Upon taking $\delta > 0$ sufficiently small (but independent of $\lambda$ and $\epsilon$), $I - \G_{\lambda,\epsilon}$ is invertible on $C([\zeta_\tf(\epsilon),0],\C^{4 \times 2})$ and
\begin{align*}
 \hat{\Zet}(\zeta;\lambda,\epsilon) = \left[(I - \G_{\lambda,\epsilon})^{-1} \begin{pmatrix} I_2 \\ \tilde{T}_-(0;\lambda,\epsilon)\end{pmatrix} \right](\zeta),
\end{align*}
is the solution to~\eqref{varcon4}. With the aid of Proposition~\ref{prop:plateau2}, we approximate
\begin{align*}
\left\|\hat{\Zet}(\zeta;\lambda,\epsilon) - \begin{pmatrix} I_2 \\ \tilde{T}_-(0;\lambda,\epsilon)\end{pmatrix}\right\| &\leq \|\G_{\lambda,\epsilon}\|\left\|(I - \G_{\lambda,\epsilon})^{-1}\right\| \left\|\begin{pmatrix} I_2 \\ \tilde{T}_-(0;\lambda,\epsilon)\end{pmatrix}\right\| \leq C\delta^{1/3}.
\end{align*}
Hence, as $\hat{\Zet}(0;\lambda,\epsilon) = \Zet(0;\lambda,\epsilon)$, we establish
\begin{align} \left\|\Zet(0;\lambda,\epsilon) - \begin{pmatrix} I_2 \\ \tilde{T}_-(0;\lambda,\epsilon)\end{pmatrix}\right\| \leq C\delta^{1/3}. \label{yboundl}\end{align}
Upon denoting $\Zet(0;\lambda,\epsilon) = (\Zet_1(\lambda,\epsilon),\Zet_2(\lambda,\epsilon))$, we approximate using Proposition~\ref{prop:plateau2} and~\eqref{yboundl}:
\begin{align*} \left\|\hat{T}_-(0;\lambda,\epsilon) - \tilde{T}_-(0;\lambda,\epsilon)\right\|
 &= \left\|\Zet_2 \Zet_1^{-1} - \tilde{T}_-\right\| \leq \left\|\hat{T}_-\right\|\|I_2 - \Zet_1\| + \|\Zet_2-\tilde{T}_-\|\\
 &\leq \left(\left\|\hat{T}_- - \tilde{T}_-\right\| + \|\tilde{T}_-\|\right) \|I_2 - \Zet_1\| + \|\Zet_2 - \tilde{T}_-\|\\
 &\leq C\left(\left\|\hat{T}_- - \tilde{T}_-\right\| + \delta^{1/4}\right) \delta^{1/3} + C\delta^{1/3},
\end{align*}
where we suppress the arguments on the right hand side. Hence, we conclude
\begin{align} \left\|\hat{T}_-(0;\lambda,\epsilon) - \tilde{T}_-(0;\lambda,\epsilon)\right\| \leq C\delta^{1/3}, \qquad \left\|\hat{T}_-(0;\lambda,\epsilon)\right\| \leq C\delta^{1/4}, \label{tboundl2}\end{align}
using~Proposition~\ref{prop:plateau2}.

All in all,~\eqref{approxricr} and~\eqref{expex0} in Lemma~\ref{lem:approx} and~\eqref{tboundl2} yield
\begin{align*}
\left|\E_{\epsilon}(\lambda) - 2\sqrt{1+\alpha^2}\right| \leq C\delta^{1/4} < 2\sqrt{1+\alpha^2}, \qquad \text{for } \lambda \in \Gamma,
\end{align*}
provided $0 < \epsilon_1 \ll \Theta_1 \ll \delta \ll 1$ and $0 \leq |\epsilon_2| \ll \Theta_1 \ll \delta \ll 1$. Hence, applying Rouch\'e's Theorem on the contour $\Gamma$ yield that $\E_\epsilon(\Gamma)$ has winding number $0$ and, thus, the meromorphic function $\E_\epsilon$ possesses an equal number of zeros and poles in the interior of $\Gamma$.
\end{proof}

Hence, Lemma~\ref{lemma:windingnumber} shows that the Riccati-Evans function $\E_\epsilon$ has an equal number of zeros and poles in the interior of the contour $\Gamma$. Thus, to find its number of zeros, we can compute the number of poles of $\E_\epsilon$ in the interior of $\Gamma$. Our plan is to use the formula~\eqref{relEvans} in Proposition~\ref{prop:ric} to write the Riccati-Evans function $\E_\epsilon$ as a quotient of two analytic functions. Then, by applying Rouch\'e's theorem, we determine the number of zeros of the denominator in the interior of $\Gamma$, which then provides us an upper bound of the number of poles of $\E_\epsilon$. An additional outcome of the upcoming analysis is that we can approximate the Riccati-Evans function at $\Gamma \cap \R_{> 0}$, which is useful for the parity argument. All in all, we obtain the following result.

\begin{lemma} \label{lemma:windingnumber2}
Provided $0 < \epsilon_1 \ll \Theta_1 \ll \delta \ll 1$ and $0 \leq |\epsilon_2| \ll \Theta_1 \ll \delta \ll 1$, $\E_\epsilon$ has at most two poles (including multiplicity) in the interior of the simple contour $\Gamma$, given by~\eqref{defGamma}. If $\E_\epsilon$ has precisely two poles within $\Gamma$, then the zeros of $\det(\hat{X}_-(0;\cdot,\epsilon))$ and the poles of $\E_\epsilon$ in the interior of $\Gamma$ coincide (including multiplicity), where $\hat{X}_-(\zeta;\lambda,\epsilon) \in \C^{2 \times 2}$ is defined in~\S\ref{sec:track} by~\eqref{gaugechoice}. Finally, it holds
\begin{align} \left|\E_\epsilon\left(\delta^{-1/3} M(\epsilon)^3\right) - 2\sqrt{1+\alpha^2}\right| \leq C|\log(\delta)|^{-2}, \label{epsest}\end{align}
and
\begin{align} \left|\det\left(\hat{X}_-\left(0;\delta^{-1/3} M(\epsilon)^3,\epsilon\right)\right) - \frac{\pi^2}{4} \delta^{-2/3} \hat{z}_*(\zeta_\delta)^2 \left(1+\alpha^2\right)\right| \leq C\delta^{-2/3}|\log(\delta)|^{-3}. \label{epsest2}\end{align}
where $C > 1$ is an $\epsilon$- and $\delta$-independent constant and $M(\epsilon)$ is defined in~\eqref{defM}.
\end{lemma}
\begin{proof} In this proof $C > 1$ denotes any $\epsilon$-, $\lambda$- and $\delta$-independent constant.

Recall from~\S\ref{sec:invric} that the matrix solution
\begin{align*} \begin{pmatrix} \hat{X}_-(\zeta;\lambda,\epsilon) \\ \hat{Y}_-(\zeta;\lambda,\epsilon)\end{pmatrix} \in \C^{4 \times 2}, \qquad \text{with } \begin{pmatrix} \hat{X}_-(\zeta_\tf(\epsilon);\lambda,\epsilon) \\ \hat{Y}_-(\zeta_\tf(\epsilon);\lambda,\epsilon)\end{pmatrix} = \begin{pmatrix} I_2 \\ \hat{T}_-(\zeta_\tf(\epsilon);\lambda,\epsilon)\end{pmatrix},\end{align*}
to system~\eqref{evprob2} spans the relevant subspace $W_-(\zeta;\lambda,\epsilon)$, which is represented by $\hat{T}_-(\zeta;\lambda,\epsilon) = \hat{Y}_-(\zeta;\lambda,\epsilon)\hat{X}_-(\zeta;\lambda,\epsilon)^{-1}$, and is analytic in $\lambda \in R_1(\Theta_1)$, because the coefficient matrix of~\eqref{evprob2} depends analytically on $\lambda$ and the meromorphic function $\hat{T}_-(\zeta_\tf(\epsilon);\cdot,\epsilon)$ has no poles on $R_1(\Theta_1)$ by~\eqref{approxricl} in Lemma~\ref{lem:approx}. Thus, using~\eqref{approxricr} and~\eqref{expex0} in Lemma~\ref{lem:approx}, a pole of the Riccati-Evans function
\begin{align*}\E_\epsilon(\lambda) =\det\left(\hat{T}_+(0;\lambda,\epsilon) - \hat{T}_-(0;\lambda,\epsilon)\right) = \frac{\det\left(\hat{T}_+(0;\lambda,\epsilon)\hat{X}_-(0;\lambda,\epsilon) - \hat{Y}_-(0;\lambda,\epsilon)\right)}{\det(\hat{X}_-(0;\lambda,\epsilon))}, \end{align*}
of multiplicity $n$ at some $\lambda = \lambda_p$ yields a zero of $\det(\hat{X}_-(0;\lambda,\epsilon))$ at some point $\lambda = \lambda_p \in R_1(\Theta_1)$ of multiplicity $\geq n$. Note that, due to possible zero-pole cancellation, a zero of $\det(\hat{X}_-(0;\cdot,\epsilon))$ does not necessarily yield a pole of $\E_\epsilon$.

Our plan is to approximate $\det(\hat{X}_-(0;\lambda,\epsilon))$ for $\lambda \in \Gamma$ by an analytic function with a known number of zeros and find its number of zeros within the contour $\Gamma$ using Rouch\'e's theorem. Therefore, we consider the matrix solution
\begin{align*} \begin{pmatrix} \hat{X}_0(\zeta;\lambda,\epsilon) \\ \hat{Y}_0(\zeta;\lambda,\epsilon)\end{pmatrix} \in \C^{4 \times 2}, \qquad \text{with } \begin{pmatrix} \hat{X}_0(\zeta_\tf(\epsilon);\lambda,\epsilon) \\ \hat{Y}_0(\zeta_\tf(\epsilon);\lambda,\epsilon)\end{pmatrix} = \begin{pmatrix} I_2 \\ T_{d,-}(\zeta_\tf(\epsilon);\lambda,\epsilon)\end{pmatrix},\end{align*}
to system~\eqref{evprobp}, which spans the 2-dimensional subspace represented by the solution $T_{d,-}(\zeta;\lambda,\epsilon)$, defined in Proposition~\ref{prop:plateau1}, to~\eqref{matrixric} under the coordinate chart $\mathfrak{c}$. Because the coefficient matrix of~\eqref{evprobp} depends analytically on $\lambda$ and $T_{d,-}(\zeta_\tf(\epsilon);\lambda,\epsilon) = \mathrm{diag}\left(\hat{z}_*(\zeta_\delta),\hat{z}_*(\zeta_\delta)\right)$ is independent of $\lambda$, we find that $\hat{X}_0(0;\cdot,\epsilon)$ and $\hat{Y}_0(0;\cdot,\epsilon)$ are analytic on $R_1(\Theta_1)$. Consequently, the quotient
\begin{align}\hat{Y}_0(\zeta;\lambda,\epsilon)\hat{X}_0(\zeta;\lambda,\epsilon)^{-1} = T_{d,-}(\zeta;\lambda,\epsilon) = \mathrm{diag}\left(t_-(\zeta;\lambda,\epsilon),s_-(\zeta;\lambda,\epsilon)\right), \label{merom}\end{align}
is meromorphic.

It follows immediately by the diagonal structure of system~\eqref{evprobp} that $\hat{X}_0(\zeta;\lambda,\epsilon)$ and $\hat{Y}_0(\zeta;\lambda,\epsilon)$ are diagonal matrices for each $\zeta \in \R$. Consequently, the upper (or lower) diagonal elements of $\hat{X}_0(0;\lambda,\epsilon)$ and $\hat{Y}_0(0;\lambda,\epsilon)$ cannot vanish for the same $\lambda \in R_1(\Theta_1)$ as this would imply that the subspace spanned by $\left(\hat{X}_0, \hat{Y}_0\right)(\zeta;\lambda,\epsilon)$ is no longer 2-dimensional. So,~\eqref{merom} yields that the poles of $t_-(0;\cdot,\epsilon)s_-(0;\cdot,\epsilon)$ in $R_1(\Theta_1)$ coincide with the zeros of $\det(\hat{X}_0(0;\cdot,\epsilon))$, including multiplicity.

We obtained in Proposition~\ref{prop:plateau1} that the number of poles of $t_-(0;\cdot,\epsilon)$ and of $s_-(0;\cdot,\epsilon)$ in $D_1(\epsilon) = D_{-1}(\epsilon)$ is one (including multiplicity). We conclude that $\det(\hat{X}_0(0;\cdot,\epsilon))$ has precisely two zeros in $D_1(\epsilon)$, including multiplicity. In addition,~\eqref{dboundl} in Proposition~\ref{prop:plateau1} shows that $t_-(0;\cdot,\epsilon)^{-1}$ and $s_-(0,\cdot,\epsilon)^{-1}$ have no zeros in $R_1(\Theta_1)$ outside of the disks $D_j(\epsilon), j \in \Z \setminus \{0\}$. Thus, also $\det(\hat{X}_0(0;\cdot,\epsilon))$ admits no zeros in $R_1(\Theta_1)$ outside of the disks $D_j(\epsilon), j \in \Z \setminus \{0\}$. By~\eqref{radiusb} in Proposition~\ref{prop:plateau1}, the contour $\Gamma$ contains the disk $D_1(\epsilon) = D_{-1}(\epsilon)$, but none of the other disks $D_j(\epsilon), j \in \Z \setminus\{0,\pm1\}$, we establish that $\det(\hat{X}_0(0;\cdot,\epsilon))$ has precisely two zeros in the interior of $\Gamma$, including multiplicity.

All that remains is to approximate $\det(\hat{X}_-(0;\lambda,\epsilon))$ by $\det(\hat{X}_0(0;\lambda,\epsilon))$ for $\lambda \in \Gamma$ and apply Rouch\'e's theorem. The explicit expression~\eqref{explevol} of the evolution of~\eqref{evprobp}, obtained in the proof of Lemma~\ref{lemma:windingnumber}, yields that $\smash{\left(\begin{smallmatrix} \hat{X}_0\\ \hat{Y}_0\end{smallmatrix}\right)}(0;\lambda,\epsilon)$ is given by
{
\begin{align*} &\left(\begin{array}{c} \cosh\left(\varpi_1(\lambda,\epsilon)\zeta_\tf(\epsilon)\right) - \frac{\sinh(\varpi_1(\lambda,\epsilon)\zeta_\tf(\epsilon)) \hat{z}_*(\zeta_\delta)}{\varpi_1(\lambda,\epsilon)}  \\
0  \\
\cosh\left(\varpi_1(\lambda,\epsilon)\zeta_\tf(\epsilon)\right)\hat{z}_*(\zeta_\delta) - \sinh(\varpi_1(\lambda,\epsilon)\zeta_\tf(\epsilon))\varpi_1(\lambda,\epsilon) \\
0 \end{array}\right. \ldots \\
&\qquad \qquad \qquad \qquad \ldots
\left.\begin{array}{c}  0\\
\cosh\left(\varpi_2(\lambda,\epsilon)\zeta_\tf(\epsilon)\right) -  \frac{\sinh(\varpi_2(\lambda,\epsilon)\zeta_\tf(\epsilon)) \hat{z}_*(\zeta_\delta)}{\varpi_2(\lambda,\epsilon)}\\
0\\
\cosh\left(\varpi_2(\lambda,\epsilon)\zeta_\tf(\epsilon)\right)\hat{z}_*(\zeta_\delta) - \sinh(\varpi_2(\lambda,\epsilon)\zeta_\tf(\epsilon))\varpi_2(\lambda,\epsilon)
\end{array}\right).
\end{align*}}
So, using~\eqref{deltaineq} in Proposition~\ref{prop:left},~\eqref{expex5},~\eqref{expex1} and the estimate~\eqref{approxpi}, obtained in the proof of Lemma~\ref{lemma:windingnumber}, we establish
\begin{align}
\begin{split}
\left\|\hat{X}_0(0;\lambda,\epsilon) - \frac{\pi}{2} \delta^{-1/3}\re^{\ri\theta} \hat{z}_*(\zeta_\delta)\begin{pmatrix} (1-\ri \alpha) & 0 \\ 0 & (1+\ri\alpha)\end{pmatrix}\right\| &\leq C, \\
\left\|\hat{Y}_0(0;\lambda,\epsilon) + \hat{z}_*(\zeta_\delta)\begin{pmatrix} 1 & 0 \\ 0 & 1\end{pmatrix}\right\| &\leq CM(\epsilon),
\end{split}
\label{gammabound2}\end{align}
and therefore we arrive at
\begin{align}
\begin{split}
\left|\det(\hat{X}_0(0;\lambda,\epsilon)) - \frac{\pi^2}{4} \delta^{-2/3} \re^{2\ri\vartheta} \hat{z}_*(\zeta_\delta)^2 \left(1+\alpha^2\right)\right| &\leq C\delta^{-1/3},\\
\left|\det(\hat{Y}_0(0;\lambda,\epsilon)) - \hat{z}_*(\zeta_\delta)^2\right| &\leq CM(\epsilon),
\end{split}
\label{gammabound}\end{align}
for $\lambda = \delta^{-1/3} M(\epsilon)^3 \re^{\ri \vartheta} \in \Gamma$.

Next, we approximate $\det(\hat{X}_-(0;\lambda,\epsilon))$ by $\det(\hat{X}_0(0;\lambda,\epsilon))$. By~\eqref{gaugechoice} and the variation of constants formula we have
\begin{align*}
\begin{pmatrix} \hat{X}_-(\zeta;\lambda,\epsilon) \\ \hat{Y}_-(\zeta;\lambda,\epsilon)\end{pmatrix} = \mathcal{T}_{*,p}(\zeta,\zeta_\tf(\epsilon);\lambda,\epsilon) \begin{pmatrix} I_2 \\ \hat{T}_-(\zeta_\tf(\epsilon);\lambda,\epsilon)\end{pmatrix} + \int_{\zeta_\tf(\epsilon)}^\zeta \mathcal{T}_{*,p}(\zeta,y;\lambda,\epsilon) B_{*,p}(y;\epsilon) \begin{pmatrix} \hat{X}_-(y;\lambda,\epsilon) \\ \hat{Y}_-(y;\lambda,\epsilon)\end{pmatrix} \mathrm dy,
\end{align*}
where we denote $B_{*,p}(\zeta;\epsilon) = A_*(\zeta;\lambda,\epsilon) - A_{*,-}(\lambda,\epsilon)$ as in the proof of Proposition~\ref{prop:plateau3}. Therefore,
\begin{align*} \U(\zeta;\lambda,\epsilon) := \mathcal{T}_{*,p}(\zeta_\tf(\epsilon),\zeta;\lambda,\epsilon)\begin{pmatrix} \hat{X}_-(\zeta;\lambda,\epsilon) \\ \hat{Y}_-(\zeta;\lambda,\epsilon)\end{pmatrix},\end{align*}
solves
\begin{align}
\begin{split}
\U(\zeta;\lambda,\epsilon) &= \begin{pmatrix} I_2 \\ \hat{T}_-(\zeta_\tf(\epsilon);\lambda,\epsilon)\end{pmatrix}\\
&\qquad + \int_{\zeta_\tf(\epsilon)}^\zeta \mathcal{T}_{*,p}(\zeta_\tf(\epsilon),y;\lambda,\epsilon) B_{*,p}(y;\epsilon) \mathcal{T}_{*,p}(y,\zeta_\tf(\epsilon);\lambda,\epsilon) \U(y;\lambda,\epsilon) \mathrm dy.
\end{split}\label{varcon6}
\end{align}
By estimate~\eqref{plateaubound} from the proof of Proposition~\ref{prop:plateau3} and Lemma~\ref{lem:evolbound} the linear operator $F_{\lambda,\epsilon}$ given by
\begin{align*} \left(F_{\lambda,\epsilon} \U\right)(\zeta) = \int_{\zeta_\tf(\epsilon)}^\zeta \mathcal{T}_{*,p}(\zeta_\tf(\epsilon),y;\lambda,\epsilon) B_{*,p}(y;\epsilon) \mathcal{T}_{*,p}(y,\zeta_\tf(\epsilon);\lambda,\epsilon) \U(y) \mathrm dy,\end{align*}
on $C([\zeta_\tf(\epsilon),0],\C^{4 \times 2})$ has norm $\|F_{\lambda,\epsilon}\| \leq C \delta$. Upon taking $\delta > 0$ sufficiently small (but independent of $\epsilon$ and $\lambda$), $I - F_{\lambda,\epsilon}$ is invertible on $C([\zeta_\tf(\epsilon),0],\C^{4 \times 2})$ and
\begin{align*}
\U(\zeta;\lambda,\epsilon) = \left[(I - F_{\lambda,\epsilon})^{-1} \begin{pmatrix} I_2 \\ \hat{T}_-(\zeta_\tf(\epsilon);\lambda,\epsilon)\end{pmatrix} \right](\zeta),
\end{align*}
is the solution to~\eqref{varcon6} satisfying
\begin{align*}
\left\|\U(\zeta;\lambda,\epsilon) - \begin{pmatrix} I_2 \\ \hat{T}_-(\zeta_\tf(\epsilon);\lambda,\epsilon)\end{pmatrix}\right\|
\leq \|F_{\lambda,\epsilon}\|\|(I - F_{\lambda,\epsilon})^{-1}\| \left\|\begin{pmatrix} I_2 \\ \hat{T}_-(\zeta_\tf(\epsilon);\lambda,\epsilon)\end{pmatrix}\right\|
\leq C\delta,
\end{align*}
for $\zeta \in [\zeta_\tf,0]$, where we used~\eqref{deltaineq} from Proposition~\ref{prop:left} and~\eqref{approxricl} from Lemma~\ref{lem:approx}. So, combining the latter with~\eqref{approxricl} from Lemma~\ref{lem:approx} and identity~\eqref{evolbound2}, obtained in the proof of Lemma~\ref{lemma:windingnumber}, yields
\begin{align}
 \begin{split}
\left\|\begin{pmatrix} \hat{X}_-(0;\lambda,\epsilon) \\ \hat{Y}_-(0;\lambda,\epsilon)\end{pmatrix} - \begin{pmatrix} \hat{X}_0(0;\lambda,\epsilon) \\ \hat{Y}_0(0;\lambda,\epsilon)\end{pmatrix}\right\|
&= \left\|\mathcal{T}_{*,p}(0,\zeta_\tf(\epsilon);\lambda,\epsilon) \left(\U(0;\lambda,\epsilon) - \begin{pmatrix} I_2 \\ T_{d,-}(\zeta_\tf(\epsilon);\lambda,\epsilon)\end{pmatrix}\right)\right\|\\
&\leq \left\|\mathcal{T}_{*,p}(0,\zeta_\tf(\epsilon);\lambda,\epsilon)\right\|\left(\left\|\U(0;\lambda,\epsilon) - \begin{pmatrix} I_2 \\ \hat{T}_-(\zeta_\tf(\epsilon);\lambda,\epsilon)\end{pmatrix}\right\|\right. \\
&\left. \phantom{\left\|\begin{pmatrix} I_2 \\ T_{d,-}\end{pmatrix}\right\|} + \left\|T_{d,-}(\zeta_\tf(\epsilon);\lambda,\epsilon) - \hat{T}_-(\zeta_\tf(\epsilon);\lambda,\epsilon)\right\|\right)\\
&\leq C\delta^{-1/3} |\log(\delta)|^{-2},
\end{split}\label{Tboundr}\end{align}
for $\lambda \in \Gamma$. Thus, using~\eqref{deltaineq} from Proposition~\ref{prop:left},~\eqref{gammabound2},~\eqref{gammabound} and~\eqref{Tboundr}, we obtain
\begin{align}|\det(\hat{X}_-(0;\lambda,\epsilon)) - \det(\hat{X}_0(0;\lambda,\epsilon))| \leq C\delta^{-2/3}|\log(\delta)|^{-3} <
|\det(\hat{X}_0(\zeta_\tf(\epsilon);\lambda,\epsilon))|, \label{gammabound3}\end{align}
for $\lambda \in \Gamma$, provided $0 < \epsilon_1 \ll \Theta_1 \ll \delta \ll 1$ and $0 \leq |\epsilon_2| \ll \Theta_1 \ll \delta \ll 1$. Hence, by Rouch\'e's theorem the numbers of zeros within the contour $\Gamma$ of the analytic functions $\det(\hat{X}_-(0;\cdot,\epsilon))$ and $\det(\hat{X}_0(0;\cdot,\epsilon))$ coincide (including multiplicity).

We conclude that $\det(\hat{X}_-(0;\cdot,\epsilon))$ has precisely two zeros in the interior of $\Gamma$ and, therefore, the Riccati-Evans function $\E_\epsilon$ has at most two poles (including multiplicity) inside of $\Gamma$.

Finally, by combining~\eqref{deltaineq} in Proposition~\ref{prop:left},~\eqref{gammabound} and~\eqref{gammabound3} we establish~\eqref{epsest2}. On the other hand, with the aid of~\eqref{gammabound2},~\eqref{gammabound},~\eqref{Tboundr} and~\eqref{gammabound3} we obtain
\begin{align*}\left|\det\left(\hat{T}_-\left(0;\delta^{-1/3} M(\epsilon)^3,\epsilon\right)\right)\right| =
\left|\frac{\det\left(\hat{Y}_-\left(0;\delta^{-1/3} M(\epsilon)^3,\epsilon\right)\right)}{\det\left(\hat{X}_-\left(0;\delta^{-1/3} M(\epsilon)^3,\epsilon\right)\right)}\right| \leq C|\log(\delta)|^{-2},
\end{align*}
provided $0 < \epsilon_1 \ll \Theta_1 \ll \delta \ll 1$ and $0 \leq |\epsilon_2| \ll \Theta_1 \ll \delta \ll 1$. Thus,~\eqref{approxricr} and~\eqref{expex0} in Lemma~\ref{lem:approx} yield~\eqref{epsest}.
\end{proof}

\subsection{Restriction to the real line} \label{sec:restr}

We prove that the Riccati-Evans function $\E_\epsilon(\lambda)$ is real for real $\lambda \in \Omega \setminus \Es_\epsilon$ by exploiting that~\eqref{evprob2} obeys a conjugation symmetry for real $\lambda \in \R \cap \Omega \setminus \Es_\epsilon$.

\begin{proposition} \label{prop:Evans}
Provided $0 < \epsilon_1 \ll \Theta_1 \ll \delta \ll 1$ and $0 \leq |\epsilon_2| \ll \Theta_1 \ll \delta \ll 1$, it holds $\E_\epsilon(\lambda) \in \R$ for $\lambda \in \R \cap \Omega \setminus \Es_\epsilon$. In addition, $\det(\hat{X}_-(\zeta;\lambda,\epsilon))$ is real for $\zeta \in \R$ and $\lambda \in \R \cap R_1(\Theta_1)$, where $\hat{X}_-(\zeta;\lambda,\epsilon) \in \C^{2 \times 2}$ is defined in~\S\ref{sec:track} by~\eqref{gaugechoice}.
\end{proposition}
\begin{proof}  First, note that by conjugation symmetry, that if $\hat{\phi}(\zeta)$ is a solution to~\eqref{evprob2} for $\lambda \in \R$, then so is $\hat{\phi}_2(\zeta) = S_4\overline{\hat{\phi}(\zeta)}$ with
\begin{align*} S_4 := \begin{pmatrix} S_2 & 0 \\ 0 & S_2 \end{pmatrix}, \qquad S_2 := \begin{pmatrix} 0 & 1 \\ 1 & 0 \end{pmatrix}.\end{align*}

Let $\lambda \in \R \cap R_2(\theta_2,\Theta_1,\Theta_2)$. Consider the vector $\phi_r^-(\lambda) := (1, 0, \sqrt{\lambda(1-\ri \alpha)}, 0)$. System~\eqref{evprob2} admits, by Proposition~\ref{prop:expdirightR2}, an exponential dichotomy on $(-\infty,0]$ with projections $P_{*,l}(\zeta;\lambda,\epsilon)$. By~\eqref{leftbound6} and Lemma~\ref{lem:subspace}, $\psi_r^-(\lambda,\epsilon) := (I_4 - P_{*,l}(0;\lambda,\epsilon))\phi_r^-(\lambda)$ lies in $W_-(0;\lambda,\epsilon) = \ker(P_{*,l}(0;\lambda,\epsilon))$ and satisfies
\begin{align} \|\psi_r^-(\lambda,\epsilon) - \phi_r^-(\lambda)\| \leq C_\delta\|\epsilon\|^\tau, \qquad \lambda \in R_2, \label{vectorest}\end{align}
for some $\delta$-, $\epsilon$- and $\lambda$-independent constant $\tau > 0$ and some constant $C_\delta > 1$, which depends on $\delta$ only. Solutions $\hat{\phi}(\zeta)$ to~\eqref{evprob2} are bounded as $\zeta \to -\infty$ if and only if $\hat{\phi}(0) \in W_-(0;\lambda,\epsilon)$. Thus, the solution $\hat{\phi}(\zeta)$ to~\eqref{evprob2} with initial condition $\hat{\phi}(0) = \psi_r^-(\lambda,\epsilon)$ is bounded as $\zeta \to -\infty$. By the conjugation symmetry of~\eqref{evprob2}, $\hat{\phi}_2(\zeta) = S_4\overline{\hat{\phi}(\zeta)}$ is also a bounded solution to~\eqref{evprob2} as $\zeta \to -\infty$. Consequently, it holds $S_4\overline{\psi_r^-(\lambda,\epsilon)} \in W_-(0;\lambda,\epsilon)$ and, by~\eqref{vectorest}, $\tilde{\Phi}(\lambda,\epsilon) := \left(\psi_r^-(\lambda,\epsilon) \mid S_4\overline{\psi_r^-(\lambda,\epsilon)}\right)$ forms a basis of $W_-(0;\lambda,\epsilon)$ satisfying \begin{align*}\tilde{\Phi}(\lambda,\epsilon) = S_4\overline{\tilde{\Phi}(\lambda,\epsilon)}S_2,\end{align*}
Thus, since $\hat{T}_-(0;\lambda,\epsilon)$ is independent on the choice of basis of $W_-(0;\lambda,\epsilon)$, we establish
\begin{align*} \hat{T}_-(0;\lambda,\epsilon) = S_2 \overline{\hat{T}_-(0;\lambda,\epsilon)}S_2.\end{align*}
Analogously, one obtains
\begin{align*} \hat{T}_+(0;\lambda,\epsilon) = S_2 \overline{\hat{T}_+(0;\lambda,\epsilon)}S_2,\end{align*}
using Proposition~\ref{prop:expdirightR2} instead.

It follows the Riccati-Evans function $\E_{\epsilon}$ is real on $\R \cap R_2(\theta_2,\Theta_1,\Theta_2) \setminus \Es_\epsilon$. Since $\Omega \setminus \Es_\epsilon$ is a connected set containing $\R \cap R_2(\theta_2,\Theta_1,\Theta_2) \setminus \Es_\epsilon$, we conclude $\E_\epsilon$ is real on $\R \cap \Omega \setminus \Es_\epsilon$.

By~\eqref{gaugechoice}, it holds $\hat{Y}_-(\zeta_\tf(\epsilon);\lambda,\epsilon) = \hat{T}_-(\zeta_\tf(\epsilon);\lambda,\epsilon)$ for all $\lambda \in R_1(\Theta_1) \cap \R$. Thus, there exists a solution $\hat{\phi}(\zeta;\lambda,\epsilon)$ to~\eqref{evprob2} such that
\begin{align*}\begin{pmatrix} \hat{X}_-(\zeta;\lambda,\epsilon) \\ \hat{Y}_-(\zeta;\lambda,\epsilon)\end{pmatrix} = \left(\hat{\phi}(\zeta;\lambda,\epsilon) \mid S_4 \overline{\hat{\phi}(\zeta;\lambda,\epsilon)}\right),\end{align*}
for $\lambda \in \R \cap R_1$ and $\zeta \in \R$. We arrive at $\hat{X}_-(\zeta;\lambda,\epsilon) = S_2 \overline{\hat{X}_-(\zeta;\lambda,\epsilon)}S_2$, implying $\overline{\det(\hat{X}_-(\zeta;\lambda,\epsilon))} = \det(\hat{X}_-(\zeta;\lambda,\epsilon)$, and thus $\det(\hat{X}_-(\zeta;\lambda,\epsilon))$ is real for $\zeta \in \R$ and $\lambda \in \R \cap R_1(\Theta_1)$.
\end{proof}

\subsection{Parity argument} \label{sec:parity}

Consider the contour $\Gamma$ defined by~\eqref{defGamma}. By estimate~\eqref{radiusb} in Proposition~\ref{prop:plateau1}, the contour $\Gamma$ encloses the disk $D_1(\epsilon) = D_{-1}(\epsilon)$, but none of the other disks $D_j(\epsilon), j \in \Z \setminus\{0,\pm1\}$. We derived in Lemma~\ref{lemma:windingnumber} that the number of zeros of $\E_\epsilon$ in the interior of $\Gamma$ equals its number of poles within $\Gamma$ (including multiplicity). In addition, by Lemma~\ref{lemma:windingnumber2}, $\E_\epsilon$ has at most two poles, and thus at most two zeros, in the interior of $\Gamma$ (including multiplicity). We derived in~\S\ref{sec:ricderiv2} that $\lambda = 0$ is a simple zero of $\E_\epsilon$, which lies inside $\Gamma$. The following parity argument shows that, if a second zero of $\E_\epsilon$ exists in $\Gamma$, it must be real and negative.

\begin{theorem} \label{theo:parity}
Provided $0 < \epsilon_1 \ll \Theta_1 \ll \delta \ll 1$ and $0 \leq |\epsilon_2| \ll \Theta_1 \ll \delta \ll 1$, the Riccati-Evans function $\E_\epsilon$ has at most two zeros in the interior of the contour $\Gamma$, given by~\eqref{defGamma}. One of these zeros is the simple root $\lambda = 0$. If a second zero of $\E_\epsilon$ exists, it must be real and negative.
\end{theorem}
\begin{proof}
We derived in~\S\ref{sec:ricderiv2} that $\lambda = 0$ is a zero of $\E_\epsilon$. By Theorem~\ref{theo:deriv} it holds $\E_\epsilon'(0) > 0$. Hence, if $\E_\epsilon$ only contains one zero in the interior of $\Gamma$, we are done.

Now, assume this is not the case. Then, $\E_\epsilon$ possesses a second simple zero $\lambda_\epsilon \neq 0$ in the interior of $\Gamma$. By Lemmas~\ref{lemma:windingnumber} and~\ref{lemma:windingnumber2}, $\E_\epsilon$ possesses precisely two poles and two zeros (including multiplicity) inside $\Gamma$, and the zeros of $\det(\hat{X}_-(0;\cdot,\epsilon))$ and the poles of $\E_\epsilon$ inside $\Gamma$ coincide (including multiplicity), where we recall $\hat{X}_-(\zeta;\lambda,\epsilon) \in \C^{2 \times 2}$ is defined in~\S\ref{sec:track} by~\eqref{gaugechoice}.

We established in Proposition~\ref{prop:Evans} that both $\E_\epsilon(\lambda)$ and $\det(\hat{X}_-(0;\lambda,\epsilon))$ are real for real $\lambda \in \R \cap R_1(\Theta_1) \setminus \Es_\epsilon$. Therefore, $\lambda_\epsilon$ must be real. In addition, note that $\Gamma \cap \R_{>0} = \{\delta^{-1/3} M(\epsilon)^3\}$. So, by~\eqref{epsest} in Lemma~\ref{lemma:windingnumber2}, we must have $\E_\epsilon(\delta^{-1/3} M(\epsilon)^3) > 0$.

Assume now $\lambda_\epsilon > 0$. Then, since $\E_\epsilon$ has two poles within $\Gamma$ and it holds $\E_\epsilon'(0) > 0$ and $\E_\epsilon(\delta^{-1/3}M(\epsilon)^3) > 0$, it follows that $\E_\epsilon$ has precisely one positive real pole within $\Gamma$, which must be simple. Hence, $\det(\hat{X}_-(0;\cdot,\epsilon))$ has precisely one positive real zero inside the contour $\Gamma$, which must be simple.

Since $\det(\hat{X}_-(0;\cdot,\epsilon))$ has precisely one positive real zero inside the contour $\Gamma$, $\det(\hat{X}_-(0;0,\epsilon))$ and $\det(\hat{X}_-(0;\delta^{-1/3} M(\epsilon)^3,\epsilon))$ must have opposite signs. On the one hand, by~\eqref{epsest2} in Lemma~\ref{lemma:windingnumber2} we establish $\det(\hat{X}_-(0;\delta^{-1/3} M(\epsilon)^3,\epsilon)) > 0$. On the other hand, by Lemma~\ref{lem:subspace2}, the solution $\smash{\left(\begin{smallmatrix} \hat{X}_-\\\hat{Y}_-\end{smallmatrix}\right)}(\zeta;0,\epsilon)$ spans the subspace $W_-(\zeta;0,\epsilon) = \ker(P_{*,l}(\zeta;0,\epsilon))$. So, by Lemma~\ref{lem:subspace2}, estimate~\eqref{leftapprox2} in Proposition~\ref{prop:expdileftR1} and~\eqref{gaugechoice}, we have
\begin{align*} \begin{pmatrix} \hat{X}_-(\zeta;0,\epsilon) \\ \hat{Y}_-(\zeta;0,\epsilon)\end{pmatrix} = \left(\hat{\phi}_0(\zeta;\epsilon) \mid \hat{\phi}_-(\zeta;\epsilon)\right) \cdot \left[\left(I_2 \mid 0_2\right) \cdot \left(\hat{\phi}_0(\zeta_\tf(\epsilon);\epsilon) \mid \hat{\phi}_-(\zeta_\tf(\epsilon);\epsilon)\right)\right]^{-1}.\end{align*}
 Thus, we establish
\begin{align*} \det(\hat{X}_-(0;0,\epsilon)) &= \det\begin{pmatrix} \beta(0;\epsilon)  & \beta(0;\epsilon) \left(\hat{z}_\tf(0;\epsilon) - 1 + \frac{1}{2}\epsilon_1\right) \\ -\overline{\beta(0;\epsilon)}  & \overline{\beta(0;\epsilon)} \left(\overline{\hat{z}_\tf(0;\epsilon)} - 1 + \frac{1}{2}\epsilon_1\right)\end{pmatrix}\cdot\det\begin{pmatrix} 1 & \hat{z}_\tf(\zeta_\tf(\epsilon);\epsilon) - 1 + \frac{1}{2}\epsilon_1 \\
-1  & \overline{\hat{z}_\tf(\zeta_\tf(\epsilon);\epsilon)} - 1 + \frac{1}{2}\epsilon_1 \end{pmatrix}^{-1}\\
&= \frac{|\beta(0;\epsilon)|^2 \left(2\Re(\hat{z}_\tf(0;\epsilon)) - 2 + \epsilon_1\right)}{2\Re(\hat{z}_\tf(\zeta_\tf(\epsilon);\epsilon)) - 2 + \epsilon_1}.
\end{align*}
Thus, by~\eqref{hatzapprox},~\eqref{deltaineq} and~\eqref{estex1} in Proposition~\ref{prop:left} and estimate~\eqref{estex2} in Proposition~\ref{prop:right}, we conclude $\det(\hat{X}_-(0;0,\epsilon)) > 0$ for $0 < \epsilon_1 \ll \delta \ll 1$ and $0 \leq |\epsilon_2| \ll \delta \ll 1$. So, $\det(\hat{X}_-(0;0,\epsilon))$ and $\det(\hat{X}_-(0;\delta^{-1/3} M(\epsilon)^3,\epsilon))$ are both positive and do not have opposite signs. We have arrived at a contradiction. Therefore, our assumption that the real zero $\lambda_\epsilon$ of $\E_\epsilon$ is positive, was false, which concludes the proof.
\end{proof}

Hence, Theorems~\ref{conclR1outD1} and~\ref{theo:parity} yield that the region $R_1(\Theta_1)$ has no point spectrum of positive real part, except for the simple eigenvalue $\lambda = 0$ residing at the origin.

\begin{corollary} \label{conclR1}
Provided $0 < \epsilon_1 \ll \Theta_1 \ll \delta \ll 1$ and $0 \leq |\epsilon_2| \ll \Theta_1 \ll \delta \ll 1$, the operator $\hat{\El}_\tf$, posed on $\smash{L^2_{\hat{\kappa}_-,\hat{\kappa}_+}(\R,\C^2)}$, has no point spectrum in $R_1(\Theta_1) \cap \{\lambda \in \C : \Re(\lambda) \geq 0\}$, except for a simple eigenvalue at $\lambda = 0$.
\end{corollary}
\begin{proof}
By estimate~\eqref{radiusb} in Proposition~\ref{prop:plateau1}, the simple contour $\Gamma$, defined in~\eqref{defGamma}, encloses the disk $D_1(\epsilon) = D_{-1}(\epsilon)$ in its interior, but none of the other disks $D_j(\epsilon), j \in \Z \setminus\{0,\pm1\}$. The result now follows by combining Proposition~\ref{prop:ric} and Theorems~\ref{conclR1outD1} and~\ref{theo:parity}.
\end{proof}

\section{Proof of Theorem~\ref{theospecstab}} \label{sec:proof}

The proof of Theorem~\ref{theospecstab} follows, as outlined in~\S\ref{secsetup}, from the proof of Theorem~\ref{theospecstab2}, which is a direct result of Theorems~\ref{concless},~\ref{concR3} and~\ref{conclR2} and Corollary~\ref{conclR1}. $\hfill \Box$

\section{Numerical results and discussion}\label{disc}

\subsection{Numerical results and evidence} \label{ss:numpulled}

Figures~\ref{fig:num_spec} and~\ref{fig:spec_ep2} give the results of numerical computation of the spectrum of $\mathcal{L}_\tf$ posed on a finite, but large spatial domain, $\xi\in[-500,500]$.  We took the numerically continued front solution $\Psi_\tf$ depicted in Figure~\ref{fig:front_info} (obtained using AUTO07p~\cite{AUTO}, see~\cite{GS14} for more detail), interpolated onto a uniform grid, and discretized $\mathcal{L}_\tf$ with fourth-order finite differences, Neumann boundary conditions, and step-size $d\xi = 0.05$.  The linearization was also conjugated with exponential weights to precondition the resulting discretized operator, aiding in the convergence of the eigenvalue algorithms employed. Both the ``eigs" and ``eig" functions of MATLAB2019b were used to locate eigenvalues of the discretized operator\footnote{The repository \url{https://github.com/ryan-goh/cgl_stability} contains a set of MATLAB and AUTO07p codes used to create the numerical results and figures in ~\S\ref{disc}}. 

\begin{figure}[h!]
\centering
\vspace{-.3in}
  \hspace{-0.5in}
  \includegraphics[trim=0.4in 2.5in 0 2in, clip,width = 0.5\textwidth]{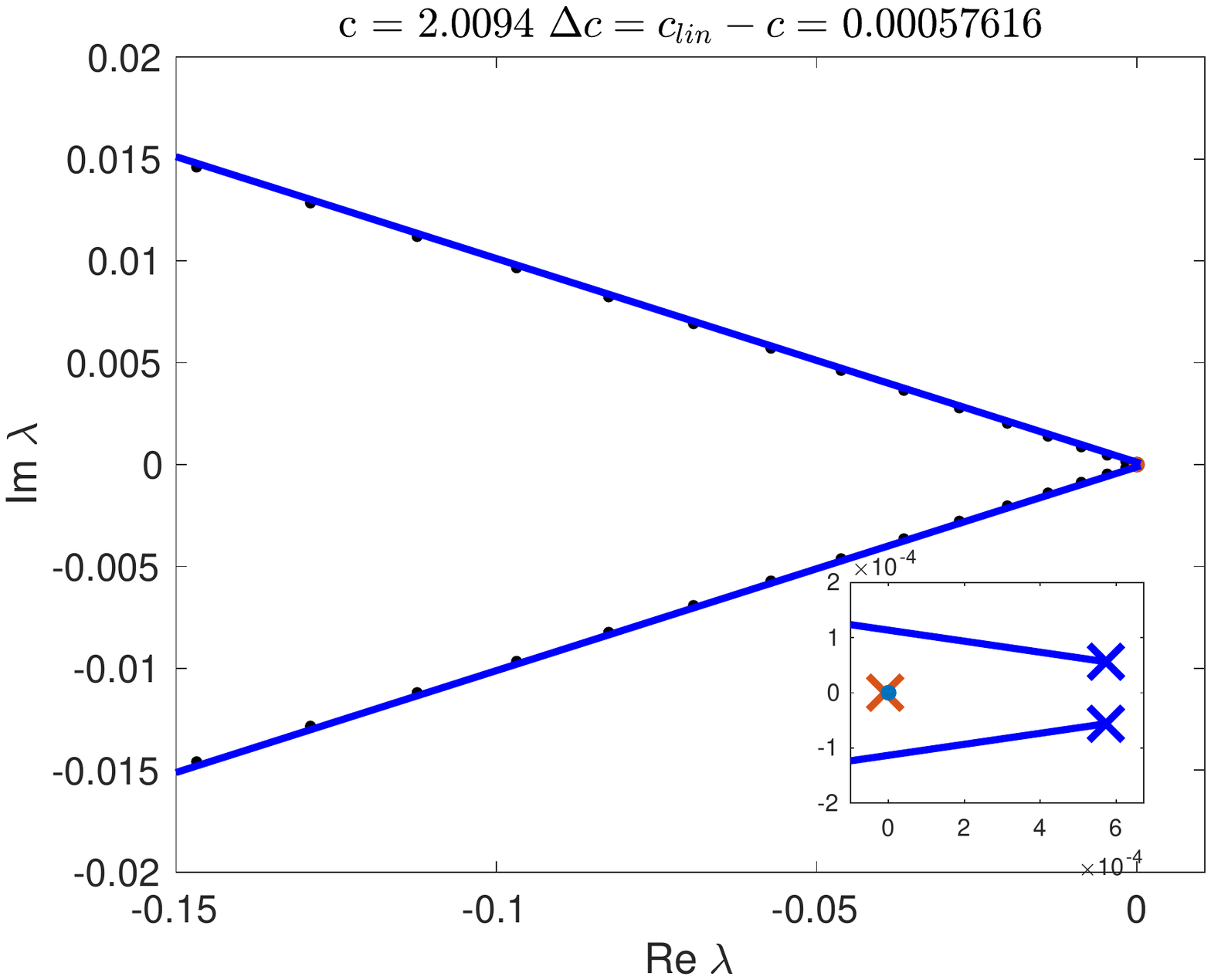}
  \hspace{-0.5in}
  \includegraphics[trim=0 2.5in 0.4in 2in, clip,width = 0.5\textwidth]{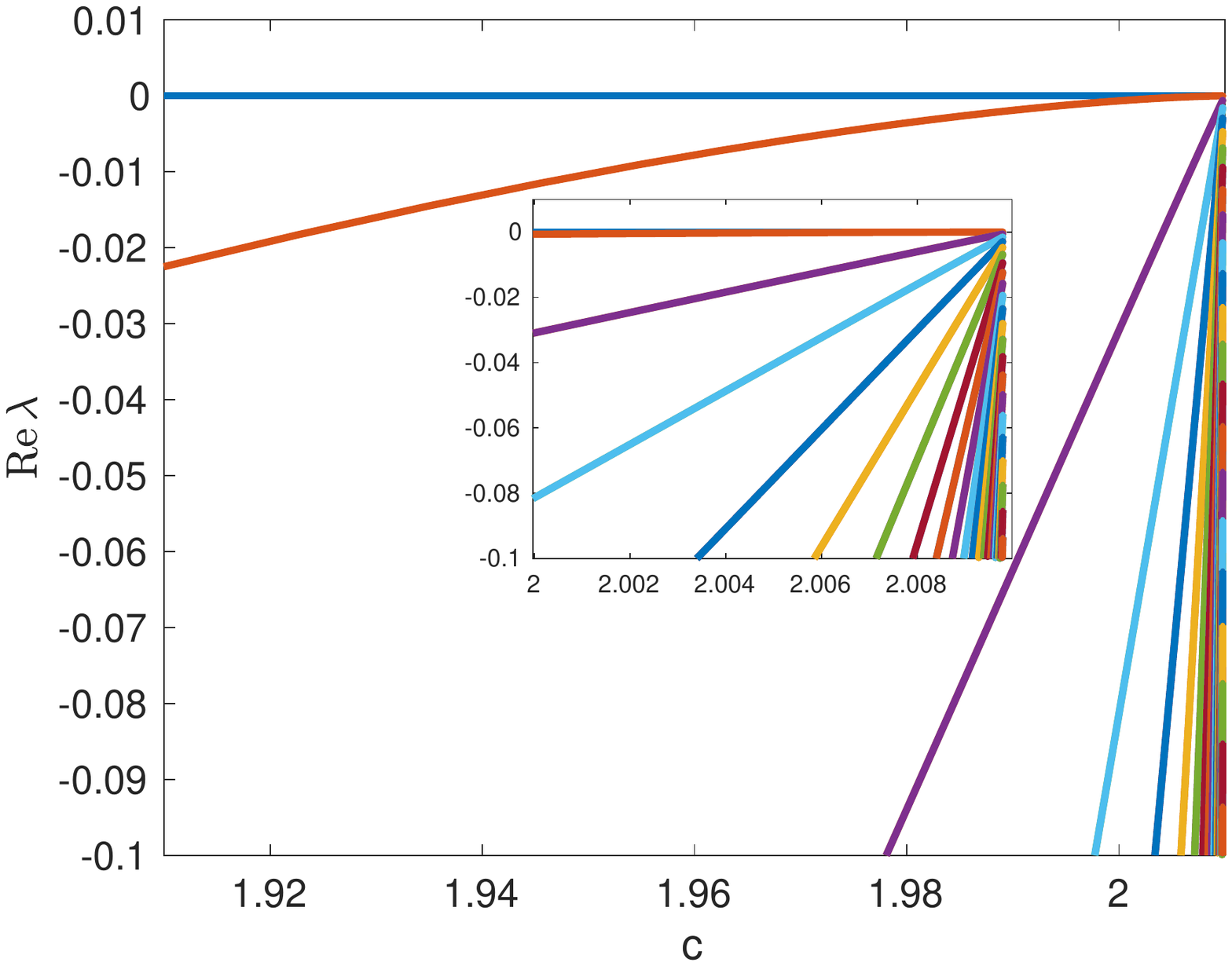}\vspace{-0.2in}\\
   \hspace{-0.5in}
     \hspace{-0.25in}
   \includegraphics[trim=0 2.5in 0 2.5in, clip,width = 0.5\textwidth]{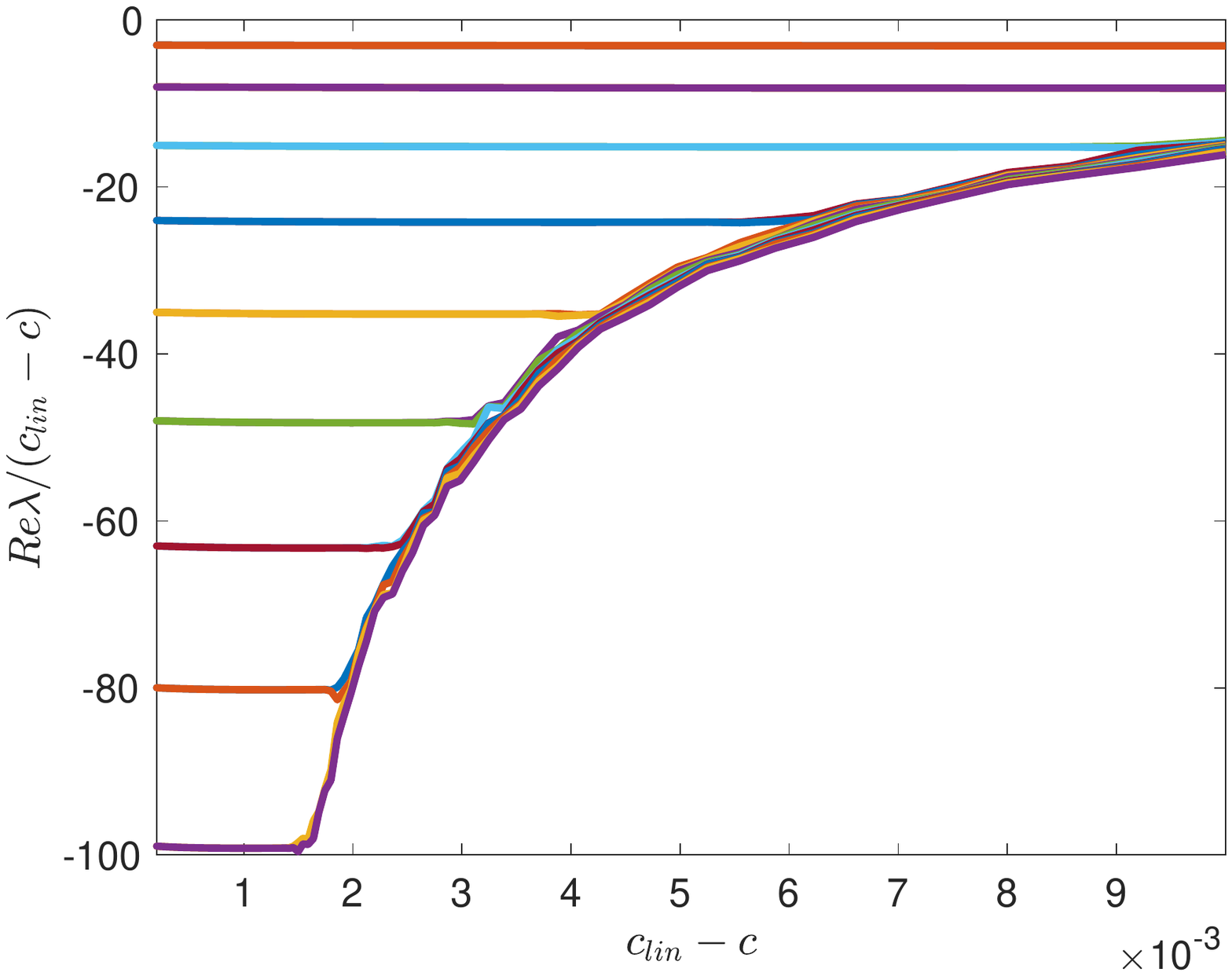}
  \hspace{-0.3in}
  \includegraphics[trim=0 2.55in 0 2.5in, clip,width = 0.5\textwidth]{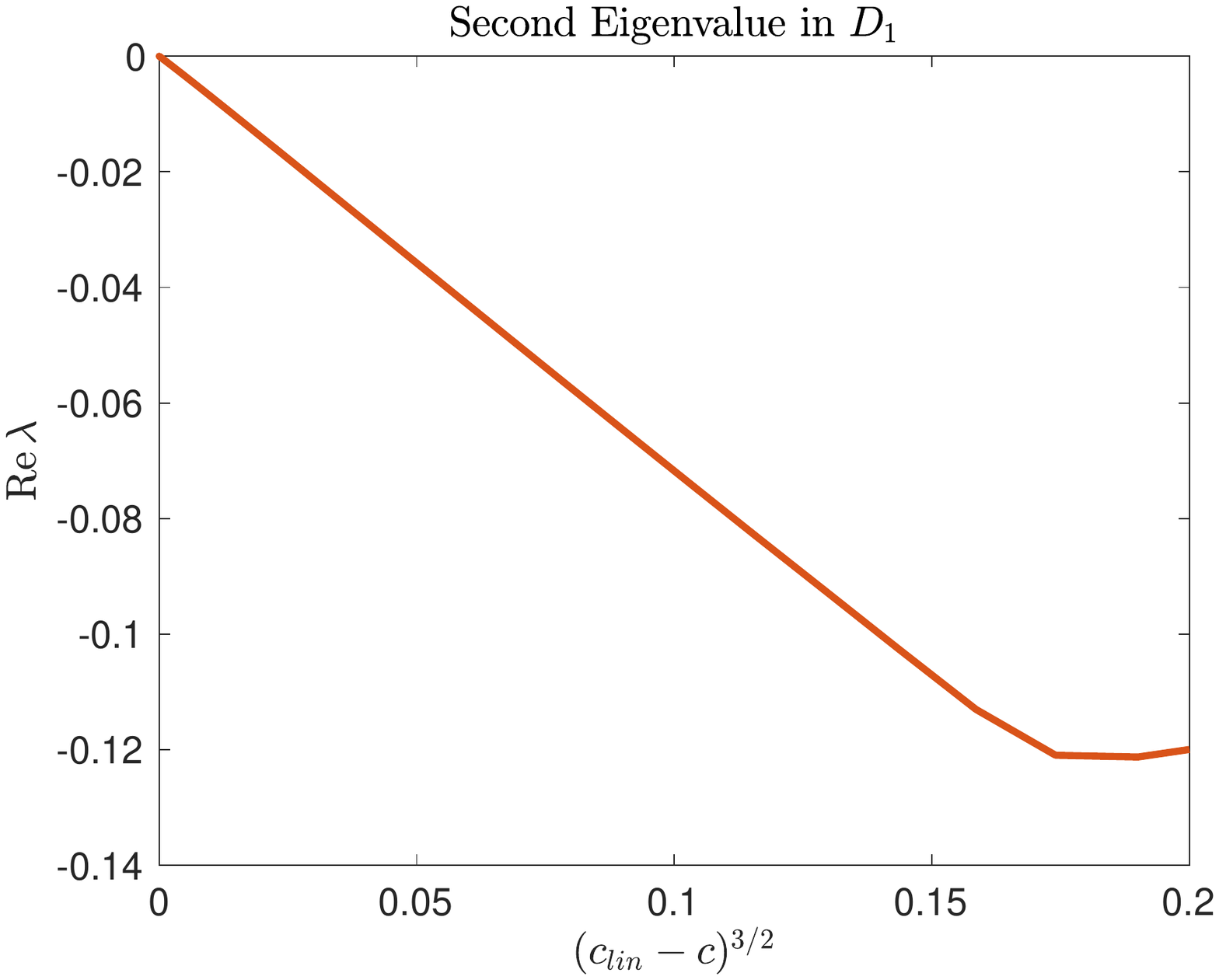}
   \hspace{-0.5in}
   \caption{Equation~\eqref{pertubeq} with parameters: $\alpha = -0.1,\gamma = -0.2;$ note $c_\mathrm{lin} \approx 2.01$; (top left): Numerical plot of the point spectrum of $\El_\tf$ (black dots) with absolute spectrum $\Sigma_{0,\mathrm{abs}}\cup\overline{\Sigma_{0,\mathrm{abs}}}$ (blue lines), light-blue dot gives the time-translation zero eigenvalue,  orange x denotes $\mathcal{O}(\Delta c^{3/2})$-eigenvalue, inset zooms into region near origin, with blue x's denoting branch points of $\Sigma_{0,\mathrm{abs}}\cup\overline{\Sigma_{0,\mathrm{abs}}}$; (top right): Real part of the 20 least negative eigenvalues for $0 \ll c < c_\mathrm{lin}$, all curves are asymptotically linear except for the $\mathcal{O}(\Delta c^{3/2})$-eigenvalue in orange, inset zooms in for speeds $c$ even closer to $c_\mathrm{lin}$; (bottom left): Plot confirming leading-order linear dependence of eigenvalues lying outside of $D_{1}(\epsilon)$ as $c_\mathrm{lin}-c\searrow0$,
(bottom right): Numerical evidence of $\mathcal{O}((\Delta c)^{3/2})$-dependence and stability of non-zero mode in $D_{\pm1}(\epsilon)$.} 
   \label{fig:num_spec}
\end{figure}

The lower left panel in Figure~\ref{fig:num_spec} confirms the $\mathcal{O}(\Delta c)$-dependence of the (real part of the) eigenvalues lying outside of the disk $D_1(\epsilon)$. Moreover, while our rigorous spectral stability result does not determine whether one or two eigenvalues lie inside of $D_1(\epsilon)$, the numerics strongly suggest that the latter case is indeed true. In both of the top right and bottom right plots of Figure~\ref{fig:num_spec} as well as the top two panels in Figure~\ref{fig:spec_ep2},  the second eigenvalue accumulates onto the origin with rate $\mathcal{O}(\epsilon_1^{3}) = \mathcal{O}(\Delta c^{3/2})$, consistent with our rigorous results (i.e. the radius of $D_1(\epsilon)$).  Heuristically, one can make sense of this interfacial eigenvalue as an approximate spatial translation, or Goldstone mode.  Here of course the heterogeneity at $\xi = 0$ precludes such a mode, but as $\epsilon_1\searrow0$, we found that the associated eigenfunction, being localized near the front interface at $\xi = \xi_\tf$, which moves farther and farther away $\xi = 0$ as $c\nearrow c_\lin$, approximately resembles the spatial derivative of the front, and hence an approximate translational mode. We note that similar behavior can be observed for the standard Fitzhugh-Nagume pulse with a small ``critical" eigenvalue resembling an approximate spatial derivative of the Nagumo pulse along the back, see~\cite{CRS} for numerical computations.

While we focused our rigorous efforts in the regime $0 \ll c < c_\mathrm{lin}, \alpha\sim \gamma$, where the existence results of~\cite{GS14} hold, our numerical continuation and spectral computations indicate that pattern-forming fronts continue to exist and are in fact spectrally stable for $0 < c_\mathrm{lin} - c = \mathcal{O}(1)$ and  $\gamma - \alpha = \mathcal{O}(1)$ as long as the asymptotic periodic pattern is still diffusively stable. Furthermore, eigenvalues still accumulate onto the branch points $(\lambda_\mathrm{br},\overline{\lambda}_\mathrm{br}$ as $\Delta c\searrow0$ with rate $\mathcal{O}(\Delta c)$, and one negative real eigenvalue approaches the origin with faster rate $\mathcal{O}((\Delta c)^{3/2})$;  see Figure~\ref{fig:spec_ep2}. Figure~\ref{fig:spec_ep2} also depicts the spectrum found for a range of speeds $c$. Here we note that, since our numerical approach uses separated Neumann boundary conditions, much of the numerical spectrum of the bounded domain approximation lies near  the absolute spectrum of the stable asymptotic states of the pattern-forming front, cf.~\cite{SAN}.

\begin{figure}[h!]
\centering
\includegraphics[trim=1in 0in 1in 0in, clip,width = 1\textwidth]{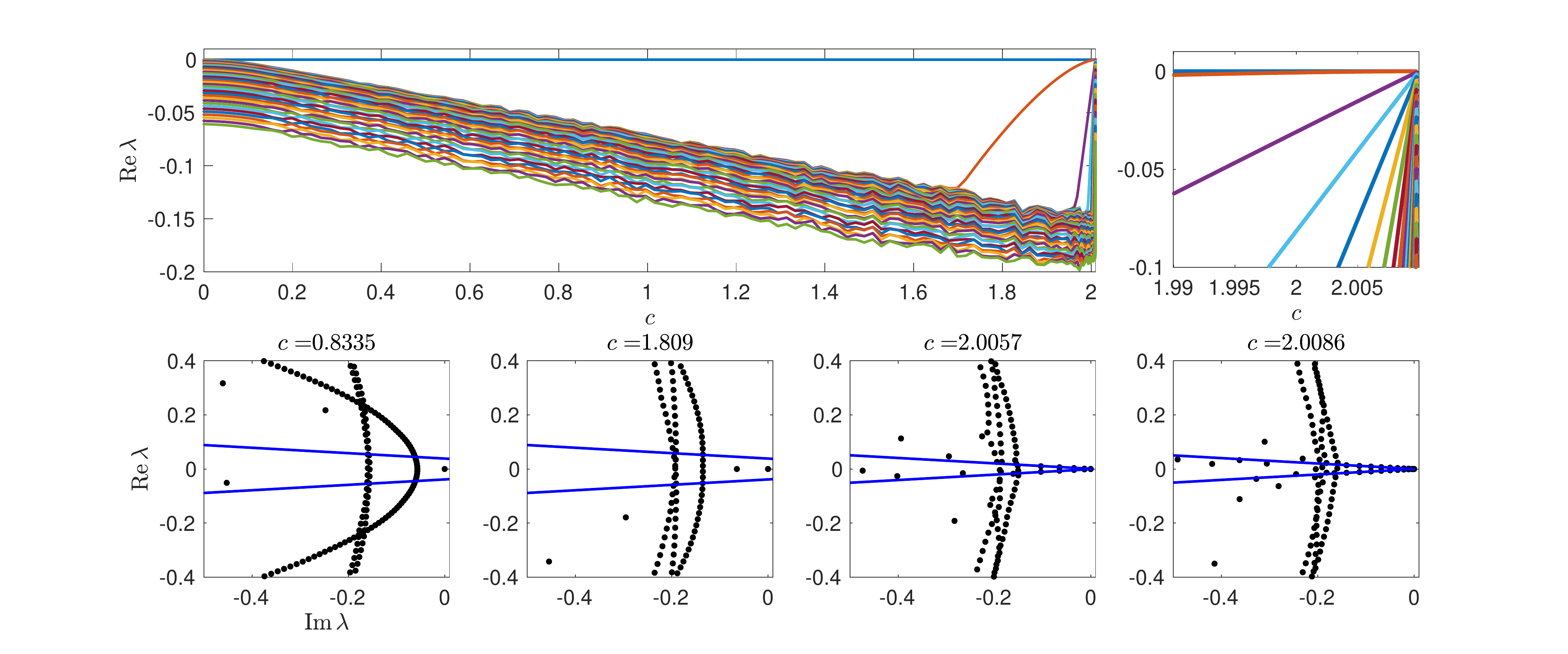}
\caption{Equation~\eqref{pertubeq} with parameters: $\alpha = -0.1,\gamma = -0.9$; (top): Plot of real parts of eigenvalues which are least negative, for speeds $c\in [0,c_\mathrm{lin}]$ (left) and $0 \ll c < c_\mathrm{lin}$ (right). The zero eigenvalue is denoted by the horizontal blue line on top, the $\mathcal{O}((\Delta c)^{3/2})$-eigenvalue is depicted in orange, and the other $\mathcal{O}(\Delta c)$-eigenvalues converge linearly to 0 near the right-hand side of the plot; (bottom row): Plots of eigenvalues for four quenching speeds $c$ with the absolute spectrum $\Sigma_{0,\mathrm{abs}}\cup\overline{\Sigma_{0,\mathrm{abs}}}$ of the plateau state (blue lines), all other spectra not depicted was found to be contained in the open left-half plane bounded away from the imaginary axis. }\label{fig:spec_ep2}
\end{figure}

\subsection{Stability of quenched pushed fronts}

We compare our findings in the supercritical cGL equation~\eqref{e:cgl0} with numerical results for the spectral stability and instability of quenched fronts in the cGL equation with a \emph{subcritical}, cubic-quintic nonlinearity,
\begin{align}
A_t = (1+\mathrm{i}\alpha) A_{xx} + \chi(x-ct)A +(\rho + \mathrm{i}\gamma)A|A|^2 - (1+\mathrm{i}\beta)A|A|^4,\qquad \rho>1,\label{e:cglp}
\end{align}
having dispersion parameters $\alpha,\gamma,\beta \in \R$. For $\chi\equiv1$ it is known that patterned fronts invade the unstable state $A\equiv0$ at a speed $c_\mathrm{p}$ greater than the linear spreading speed $c_\mathrm{lin} = 2\sqrt{1+\alpha^2}$. Since the strong nonlinear growth behind the front interface dictates the invasion process, such fronts are known as \emph{pushed fronts},~\cite{vS,van1992fronts}.   It was shown in~\cite{gs2} that for a quenching inhomogeneity, $\chi(\xi) = -\mathrm{sign}(\xi)$, traveling at speeds $c$ close to $c_\mathrm{p}$, the equation~\eqref{e:cglp} has pattern-forming front solutions. Here the interaction of the heterogeneity with the strongly decaying oscillatory tail of the front interface induces a multi-valued and non-monotonic wavenumber selection curve $(c,\omega) = (c(\xi_\tf),\omega_\tf(\xi_\tf))$, parameterized by the front interface distance, which takes the form of a deformed logarithmic spiral limiting on the free invasion parameters $(c_\mathrm{p},\omega_{p})$. The spirally nature of this curve indicates existence of multiple fronts for a single quenching speed $c$, and hence hysteretic front dynamics. Furthermore, this oscillatory tail interaction is reminiscent of ``collapsed snaking'' phenomenon in localized pattern-formation~\cite{TSEL}. The left plots of Figure~\ref{fig:pushed}, which are results of numerical continuation of the front solution, depict the wavenumber selection curve as well as the relationship of the front interface location and quenching speed $c$.

The right plot of Figure~\ref{fig:pushed} depicts the results of numerical eigenvalue computations of the linearization about front solutions to~\eqref{e:cglp}. Front solutions were obtained using numerical continuation as in the pulled case (see~\cite{gs2} for more detail on these computations) and the linearization was obtained using the same type of perturbation used to derive $\mathcal{L}_\tf$ in the pulled case above in~\S\ref{ss:numpulled}. Parameters $\alpha,\gamma,\rho,\beta$ were chosen such that the asymptotic periodic pattern is diffusively stable as a solution to~\eqref{e:cglp} with $\chi \equiv 1$.

\begin{figure}[h!]
\centering
\hspace{-0.5in}
\includegraphics[trim=0.3in 0.1in 0.3in 0.3in, clip,width = 0.5\textwidth]{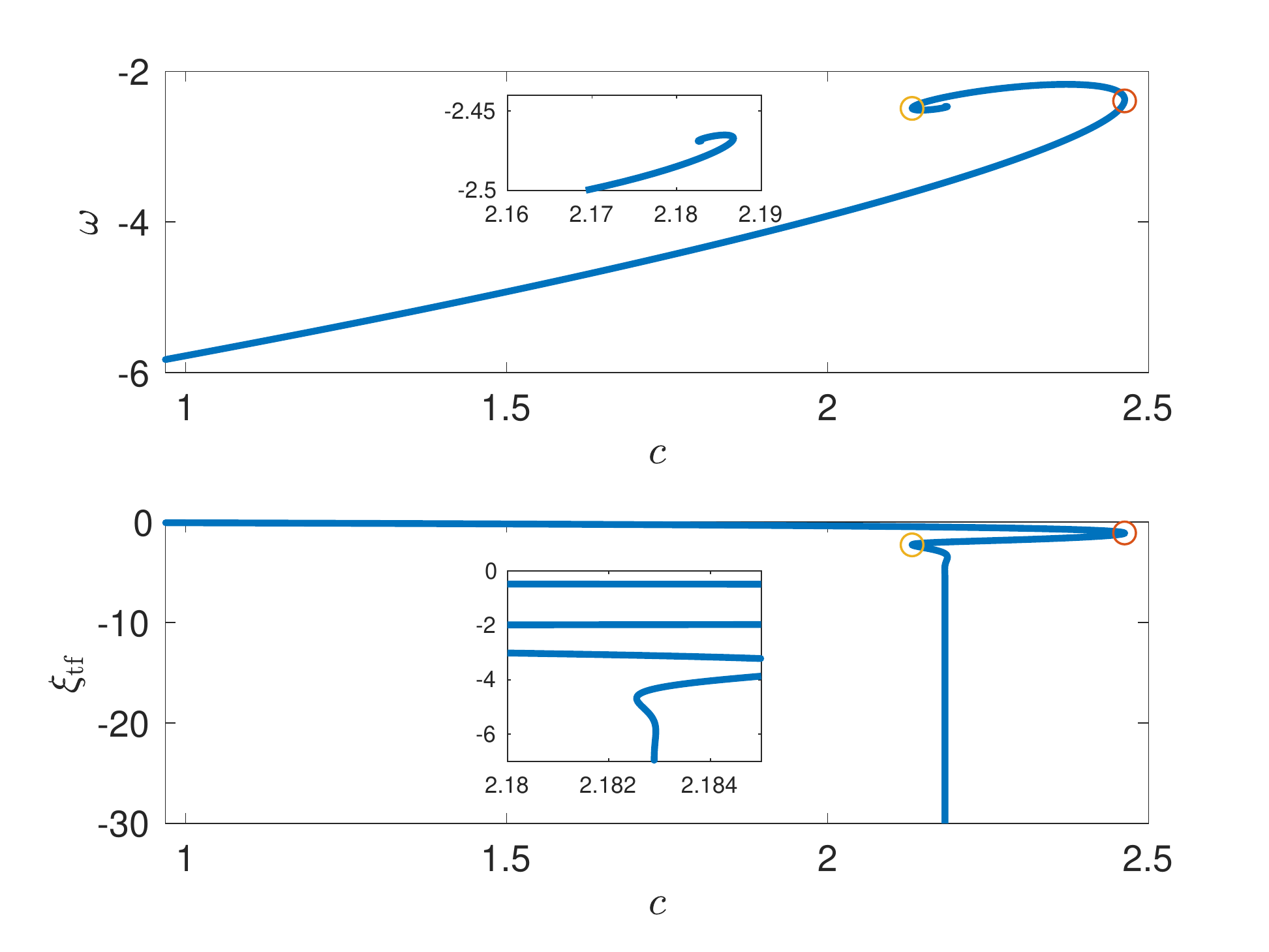}
\includegraphics[trim=0.2in 0in 0.25in 0in, clip,width = 0.5\textwidth]{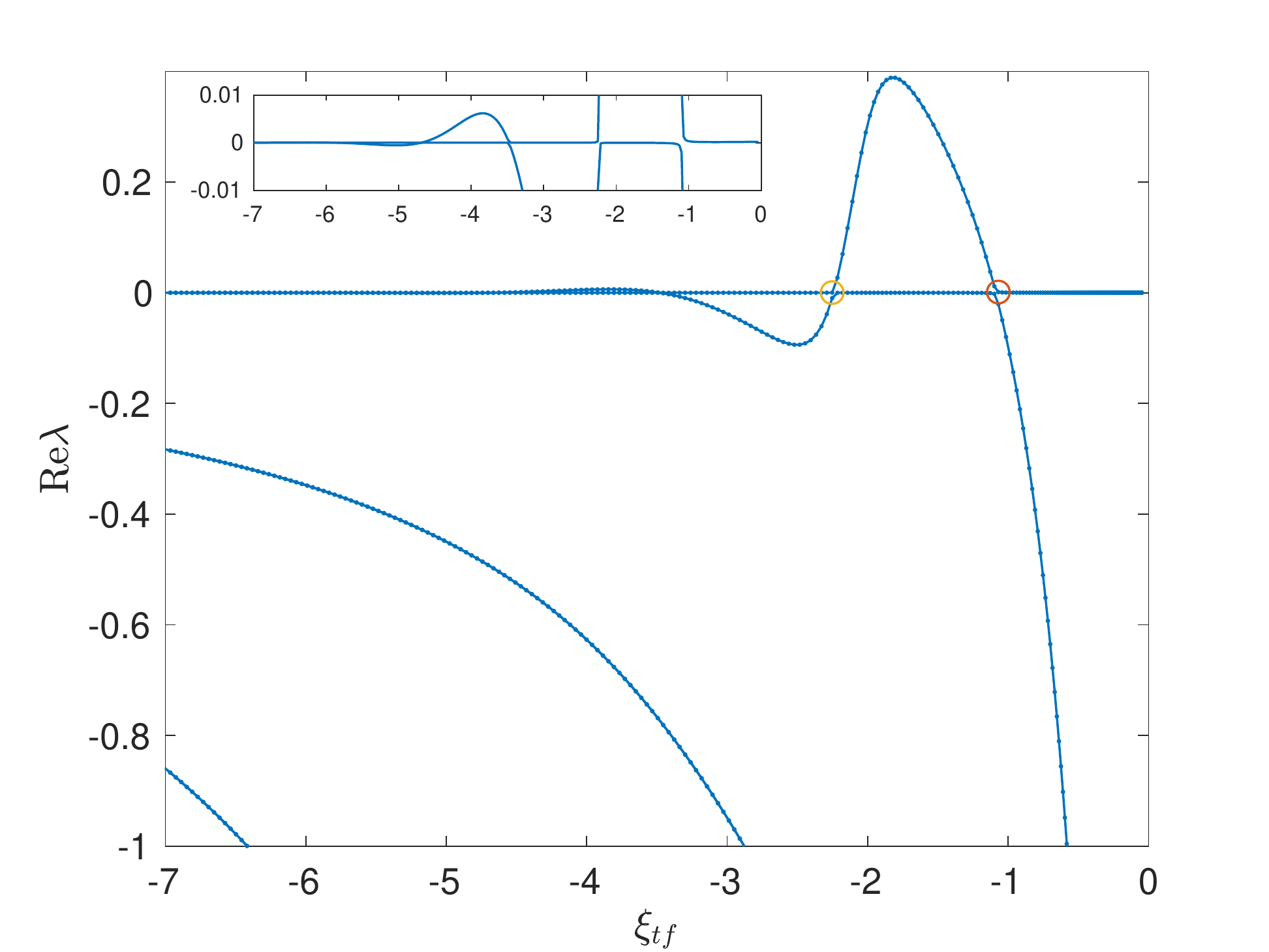}
\hspace{-0.5in}
\caption{Equation~\eqref{e:cglp} with parameters: $\alpha = 0.3,\gamma = -0.2;\rho = 4,\beta = 0.4;  c_\mathrm{p} = 2.183, \omega_\mathrm{p} = -2.468$, spatial domain $\xi\in[-100,100]$  (upper left): non-monotonic wavenumber selection curve spiralling into $(c_\mathrm{p},\omega_\mathrm{p})$; (lower left): plot of the front interface location $\xi_\tf$ against the quenching speed $c$; (right): results of eigenvalue computations of discretized linearization with $dx = 0.01$ and Neumann boundary conditions; Plot of real parts of right-most eigenvalues for a range of interface values $\xi_\tf\in(-7,0)$ which corresponds to $c$-values ranging roughly from 1 to 2.5. We note that the first two eigenvalues are purely real, with one fixed at 0 and the other oscillating back and forth. The other eigenvalues are bounded away from the imaginary axis in the left-half plane. In all plots, the orange and yellow circles denote the first two bifurcation points. Insets illustrate the oscillatory nature of the wavenumber selection curves and the corresponding fold-eigenvalue. }\label{fig:pushed}
\end{figure}

Further strengthening the connection to snaking phenomena, we find that as $\xi_\tf\rightarrow-\infty$, a single real eigenvalue controls the stability of the front, oscillating back and forth across the origin (which once again has an eigenvalue at zero due to the gauge symmetry present in~\eqref{e:cglp}), with bifurcations occurring at the fold points $c'(\xi_\tf) = 0$ of the wavenumber selection curve; see Figure~\ref{fig:pushed}. We anticipate such changes in stability could be rigorously tracked using an orientation index calculation similar to those in~\S\ref{sec:pointR1}. For such pushed fronts, we anticipate the analysis of the Riccati-Evans function to be less delicate due to the fact that we have $c>c_\mathrm{lin}$ which makes the plateau state $A\equiv0$ for $\xi\in(\xi_\tf,0)$ only convectively unstable, not absolutely unstable as is the case in this paper. This analysis will be the subject of future research.

\subsection{Nonlinear stability} \label{ss:nonlstab}

Looking outward from our results, a natural next step would be to consider nonlinear stability of the pattern-forming front solutions studied here.  Several approaches have been developed which could be implemented to obtain nonlinear stability in a suitable sense from the spectral results established in this paper. First, one could follow the approach of~\cite{B09,SAT77}. Posing the system in a co-moving frame with the quenching speed $c$, one could use an exponential weight to push continuous spectrum away from the imaginary axis, and then perform a center-manifold reduction onto the zero-mode coming from gauge symmetry. The small spectral gap caused by the presence of the $\mathcal{O}(\epsilon_1^2)$ and $\mathcal{O}(\epsilon_1^{3})$ stable eigenvalues would prohibitively restrict the size of perturbations considered.  To partially alleviate this, one could perform a center-stable manifold reduction onto the $0$-mode and the $\mathcal{O}(\|\epsilon\|^3)$-mode, with nearby dynamics foliated by the $\mathcal{O}(\|\epsilon\|^2)$ directions.

However, to obtain nonlinear stability against perturbations from spaces with translational invariant norms an alternative approach would have to be adopted. An option would be to use pointwise methods as developed in~\cite{BEC,Z98}. The fronts considered here behave somewhat like the source defects in~\cite{BEC}, where the heterogeneity at $\xi = 0$ behaves like the source center and the group velocity of waves points outwards towards $\xi = +\infty$, with $L^2$-spectrum consisting of a quadratic tangency at the origin and a single embedded eigenvalue at zero. One would hope to form the resolvent kernel using exponential dichotomies and formulate the temporal Green's function with the aid of a contour integral outside the essential spectrum. One would then obtain pointwise algebraic decay by separating the contribution of the 0-mode in the Green's function. Then, setting up a suitable perturbation ansatz, one would aim to characterize the precise pointwise decay of the perturbation in different regions of spacetime and close a nonlinear iteration scheme, showing that the perturbed front converges towards a time translate of the original front.

\subsection{Stability mechanism in other quenched pattern-forming models and higher spatial dimensions}\label{ss:push}
As mentioned in the introduction, we expect the phenomena and stability mechanisms observed here to be prototypical. In particular, we expect that, given any pattern-forming system where patterns are diffusively stable and invade an unstable homogenous state with a fixed speed predicted by the linear information ahead of the front (i.e.~the free invasion front is pulled), then, under directional quenching with speeds $0 \ll c < c_\mathrm{lin}$, the absolute spectrum of the base state governs both the leading-order behavior of the selected wavenumber of the pattern formed in the wake and the point spectrum of the linearization about the quenched front in an appropriate exponentially weighted space.  For example, such wavenumber selection dynamics have been observed in other important models of pattern formation, such as a directionally quenched Swift-Hohenberg model with supercritical cubic nonlinearity~\cite{avery2019growing,gs3}
$$
u_t = -(1+\partial_x^2)^2u + \chi(x-ct) u - u^3
$$ or the Cahn-Hilliard equation~\cite{GS15} $$u_t = -\partial_{x}^2\left[\partial_x^2u + \chi(x-ct) u - u^3\right],$$ and it would be interesting to characterize them for suitable reaction-diffusion models. In such models, since periodic patterns are generally not relative equilibria under the action of a gauge symmetry so temporal dynamics cannot be factored out, pattern-forming front solutions are \emph{modulated traveling waves}. This means the associated spatial dynamical system for the nonlinear existence problem as well as the linearized spectral problem is infinite-dimensional.  One would hope to perform a center-manifold analysis to reduce the eigenvalue problem down to a finite-dimensional ODE system whose dynamics resemble the system considered in this work. It would also be interesting to study the spectral stability problem for directionally quenched patterns in higher spatial dimensions where perturbations can act transversely to the direction of the front in addition to along it.

  \begin{Acknowledgment}
R.~Goh was partially supported by the National Science Foundation through grant NSF-DMS-1603416 and would like to acknowledge the kind hospitality of the IADM at the University of Stuttgart during several visits where some of this work was completed. B.~de Rijk would like to acknowledge the kind hospitality of Eugene Wayne during his visit at Boston University where part of this work was completed. The authors would like to acknowledge Margaret Beck for useful conversations which aided in this work.
  \end{Acknowledgment}

\appendix

\section{Exponential dichotomies} \label{appexpdi}

Exponential dichotomies enable us to track solutions in linear systems by separating the solution space in solutions that either decay exponentially in forward time or else in backward
time. For an extensive introduction on exponential dichotomies the reader is referred to~\cite{COP}. Throughout this paper, we employ the following standard definition.

\begin{definition} {\upshape
Let $n \in \N_{0}$, $J \subset \R$ an interval and $A \in C(J,\mathrm{Mat}_{n \times n}(\C))$. Denote by $\Phi(x,y)$ the evolution operator of
\begin{align} \phi_x = A(x)\phi, \quad \phi \in \C^n. \label{linsys}\end{align}
Equation~\eqref{linsys} has \emph{an exponential dichotomy on $J$ with constants $K,\mu > 0$ and projections $P(x) \colon \C^n \to \C^n$} if for all $x,y \in J$ it holds
\begin{itemize}
 \item $P(x)\Phi(x,y) = \Phi(x,y) P(y)$;
 \item $\|\Phi(x,y)P(y)\| \leq Ke^{-\mu(x-y)}$ for $x \geq y$;
 \item $\|\Phi(x,y)(I_n-P(y))\| \leq Ke^{-\mu(y-x)}$ for $y \geq x$.
\end{itemize}}
\end{definition}

\section{Properties of the essential spectrum and absolute spectrum} \label{appess}
\subsection{Essential spectrum}\label{appess1}
Let $k_\pm \in \R$ and $n \in \mathbb N$. The essential spectrum of a second-order operator $L$, posed on $L^2_{k_-,k_+}(\R,\C^n)$, given by
\begin{align*} LX = DX_{\xi\xi} + A_1(\xi) X_{\xi} + A_0(\xi) X, \end{align*}
where $D \in \C^{n \times n}$ is a positive matrix and $A_0,A_1 \colon \R \to \C^{n \times n}$ are piecewise continuous coefficient functions with well-defined limits $A_{i,\pm} = \lim_{\xi \to \pm \infty} A_i(\xi)$, is defined as the set of those $\lambda \in \C$ such that $L-\lambda$ is not Fredholm of index $0$. We shortly collect the necessary results to determine the essential spectrum. For more background information we refer to~\cite{KAP,SAN2}.

The essential spectrum of $L$ is determined by the \emph{spatial eigenvalues} of its limiting operators $L_\pm$ on $L^2_{k_-,k_+}(\R,\C^n)$ given by
\begin{align*} L_\pm X = D X_{\xi\xi} + A_{1,\pm} X_{\xi} + A_{0,\pm} X.\end{align*}
The spatial eigenvalues arise as the roots $\nu \in \C$ of the linear dispersion relations
\begin{align}
 \det\left(D \nu^2 + A_{1,\pm} \nu + A_{0,\pm} - \lambda\right) = 0, \label{spatialev}
\end{align}
which can be ordered by their real parts
\begin{align*} \Re\,\nu_{1,\pm}(\lambda) \leq \ldots \leq \Re\,\nu_{2n,\pm}(\lambda),\end{align*}
when counted with multiplicities. The \emph{Morse indices} $i_\pm(\lambda)$ equal the number of roots $\nu$ of~\eqref{spatialev} with $\Re(\nu) > k_\pm$. One then finds that $L-\lambda$ is Fredholm if and only if $\lambda$ lies in the intersection $\rho(L_+) \cap \rho(L_-)$ of the resolvent sets of $L_+$ and $L_-$. Using the Fourier transform one readily observes that this is precisely the case if there are no roots $\nu \in \C$ of~\eqref{spatialev} with $\Re(\nu) = k_\pm$. In addition, the Fredholm index of $L$ equals the difference $i_-(\lambda) - i_+(\lambda)$ of the Morse indices.

Thus, the essential spectrum
\begin{align*} \sigma_{\text{ess}}(L) = \C \setminus \{\lambda \in \rho(L_+) \cap \sigma(L_-) : i_-(\lambda) = i_+(\lambda)\},\end{align*}
is the union of $\sigma(L_+) \cup \sigma(L_-)$ together with some connected components of $\rho(L_+) \cap \rho(L_-)$. The second-order character of $L$ dictates that the spatial eigenvalues admit the splitting
\begin{align} \Re\,\nu_{n,\pm}(\lambda) < k_\pm < \Re\,\nu_{n+1,\pm}(\lambda). \label{splitev}\end{align}
for $\lambda > 0$ sufficiently large. Therefore, the right-most connected component of $\rho(L_+) \cap \rho(L_-)$ does not lie in the essential spectrum of $L$ and the right-most boundary of the essential spectrum of $L$ is contained in the set
\begin{align} \left\{\lambda \in \C : \Re\,\nu_{n,+(\lambda)} = k_+ \vee \Re\,\nu_{n+1,+}(\lambda) = k_+ \right\} \cup \left\{\lambda \in \C : \Re\,\nu_{n,-(\lambda)} = k_- \vee \Re\,\nu_{n+1,-}(\lambda) = k_- \right\}. \label{rightmost} \end{align}

\subsection{Absolute spectrum}\label{ss:app_absp}
The absolute spectrum of one of the asymptotic operators $L_\pm$ above is defined as the set
\begin{align*}
\Sigma_{\pm, \mathrm{abs}}:= \{ \lambda\in \C\,:\, \mathrm{Re}\,\nu_{n,\pm}(\lambda) = \mathrm{Re} \nu_{n+1,\pm}(\lambda)\}.
\end{align*}
The absolute spectrum of $L$ is then defined as the union $\Sigma_{\abs} = \Sigma_{+,abs} \cup \Sigma_{-,abs}$. Observe that for $\lambda$-values in $\Sigma_{\pm,\abs}$, one cannot choose a weight $k_\pm$ to recover the splitting~\eqref{splitev}. Thus, as the set~\eqref{rightmost} is contained in the essential spectrum, there must be essential spectrum of $L$ lying to the right of $\Sigma_{\abs}$, no matter the choice of $k_\pm \in \R$.

In addition to locating absolute instabilities we remark that the absolute spectrum is also important in the approximation of spectra of operators posed on large bounded domains. The work~\cite{SAN} showed that, when posed on a large bounded domain $[-h,h]\subset\R$ with separated boundary conditions, all but finitely many points of the spectrum of $L$ (which consists entirely of point spectrum) accumulate onto the absolute spectrum as $h\rightarrow+\infty.$

In the case of the constant-coefficient operator $\mathcal{L}_0$ defined in Section~\ref{ss:inv}, spatial eigenvalues are given as roots of the dispersion relation~\eqref{e:ldisp} and one can thus find curves of absolute spectrum by solving the following set of equations
$$
0 = d(\lambda,\nu;c),\qquad 0=d(\lambda,\nu + \mathrm{i}\ell;c),\qquad \ell\in \R,
$$
for $(\lambda,\nu)$ in terms of $\ell$. One readily finds
$$
\nu_{\mathrm{abs}}(\ell) = -\frac{c}{2(1+\mathrm{i}\alpha)} - \frac{\mathrm{i} \ell}{2},\qquad
\lambda_\mathrm{abs}(\ell) = 1 - \frac{c^2}{4(1+\mathrm{i}\alpha)} - \frac{(1+\mathrm{i}\alpha)\ell^2}{4}.
$$
Here the branch points discussed in Section~\ref{ss:inv} satisfy $$\lambda_{*,\mathrm{br}}(c) = \lambda_\mathrm{abs}(0) = 1 - \frac{c^2}{4(1+\mathrm{i}\alpha)}.$$
We remark that when we detune the linear operator with the change of variables $A\mapsto \re^{\ri \omega t}A$, this shifts the location of the absolute spectrum vertically in $\lambda$, so that $$\lambda_\mathrm{abs}(\ell) =  1 - \frac{c^2}{4(1+\mathrm{i}\alpha)} - \frac{(1+\mathrm{i}\alpha)\ell^2}{4} - \ri\omega.$$

\subsection{Analysis of spatial eigenvalues} \label{appesscalc}

In this appendix we analyze the spatial eigenvalues of the operator $\hat{\El}_\tf$, which confirms the claims made in~\S\ref{sec:ess}, explains the schematic depicted in Figure~\ref{fig:sp_ev} and thereby provides a proof of Theorem~\ref{concless}.

We will first show that the spatial eigenvalues $\nu_{i,+}(\lambda,\epsilon), i = 1,\ldots,4$ of $\hat{\El}_\tf$ at $+\infty$ have $\lambda$- and $\epsilon$-uniform spectral gaps at $\Re(\nu) = 1 + \Re\sqrt{2+2\ri \alpha}$ and $\Re(\nu) = 1 + \Re\sqrt{2+2\ri \alpha} + \kappa\sqrt{1+\alpha^2}$ for all $\lambda \in R_1 \cup R_2 \cup R_3$. Since it holds
\begin{align*}
\lim_{\epsilon \to (0,0)} \hat{\kappa}_+(\epsilon) := 1+\Re \sqrt{2+2\ri\alpha} +  \kappa \sqrt{1+\alpha^2}, \qquad \lim_{\epsilon \to (0,0)} \hat\kappa_0(\epsilon) = 1+\Re \sqrt{2+2\ri\alpha},
\end{align*}
by~\eqref{defkappa},~\eqref{limmum},~\eqref{hatzapprox} and~\eqref{defkapp0}, this proves the splitting~\eqref{split2} for $\lambda \in R_1 \cup R_2 \cup R_3$.

Next, we will show for $\lambda \in R_1$ that the spatial eigenvalues $\nu_{i,-}(\lambda,\epsilon), i = 1,\ldots,4$ of $\hat{\El}_\tf$ at $-\infty$ have a $\lambda$- and $\epsilon$-uniform spectral gap at $\Re(\nu) =  -\kappa\sqrt{1+\alpha^2}$. In addition, we will prove that the curve of $\lambda \in R_1$ such that there is a spatial eigenvalue $\nu$ of $\hat{\El}_\tf$ at $-\infty$ with $\Re(\nu) = 0$ lies in the open left-half plane, except for a parabolic touching with the imaginary axis at the origin. Finally, we will establish for $\lambda \in R_2 \cup R_3$ that the spatial eigenvalues $\nu_{i,-}(\lambda,\epsilon)$ of $\hat{\El}_\tf$ at $-\infty$ admit $\lambda$- and $\epsilon$-uniform spectral gaps at $\Re(\nu) = 0$ and $\Re(\nu) = -\kappa\sqrt{1+\alpha^2}$. Since it holds
\begin{align*}
\lim_{\epsilon \to (0,0)} \hat{\kappa}_-(\epsilon) = -\kappa\sqrt{1+\alpha^2},
\end{align*}
by~\eqref{defkappa}, this proves~\eqref{split1} for $\lambda \in R_2 \cup R_3$ and~\eqref{split3} for $\lambda \in R_1$. The proof of Theorem~\ref{concless} follows then as outlined in~\S\ref{sec:ess}.

\paragraph*{Leading-order expressions.} The four spatial eigenvalues $\nu_{i,-}(\lambda,\epsilon)$ associated with~\eqref{evprobmin} arise as the roots of the linear dispersion relation~\ref{spatev2}. By Theorem~\ref{t:ex02} and by identities~\eqref{defmm} and~\eqref{hatzapprox}, we find that $\nu_{i,-}(\lambda,\epsilon)$ depend continuously on $(\lambda,\epsilon)$ and satisfy
\begin{align*}
\nu_{1/2,-}(\lambda,0) = -1 - \sqrt{2 + \lambda \pm \sqrt{1 - \alpha^2 \lambda^2}}, \qquad
\nu_{3/4,-}(\lambda,0) = -1 + \sqrt{2 + \lambda \pm \sqrt{1 - \alpha^2 \lambda^2}}.
\end{align*}
Similarly, the four spatial eigenvalues $\nu_{i,+}(\lambda,\epsilon)$ associated with~\eqref{evprobplus}, which arise as the roots of the linear dispersion relation~\eqref{spatev1}, also depend continuously on $(\lambda,\epsilon)$ and satisfy
\begin{align*}
\nu_{1/2,+}(\lambda,0) = \sqrt{2 \pm 2 \ri \alpha} - \sqrt{(1 \pm \ri \alpha) (2 + \lambda)}, \qquad
\nu_{3/4,+}(\lambda,0) = \sqrt{2 \pm 2 \ri \alpha} + \sqrt{(1 \pm \ri \alpha) (2 + \lambda)}.
\end{align*}

\paragraph*{Analysis of spatial eigenvalues for $\lambda \in R_3$.} Provided $\Theta_2 > 1$ is sufficiently large and $\epsilon_1 > 0$ and $|\epsilon_2| \geq 0$ are sufficiently small, one readily observes that for $\lambda \in R_3(\theta_3,\Theta_2)$ it holds:
\begin{align*}
\left\|\frac{\nu_{1/2,\cdot}(\lambda,\epsilon)}{\sqrt{|\lambda|}} - \sqrt{\frac{\lambda}{|\lambda|} (1\pm \ri\alpha)}\right\|, \left\|\frac{\nu_{3/4,\cdot}(\lambda,\epsilon)}{\sqrt{|\lambda|}} + \sqrt{\frac{\lambda}{|\lambda|} (1\pm \ri\alpha)}\right\| \leq \frac{C}{\sqrt{|\lambda|}}.
\end{align*}
So, upon taking $\Theta > 1$ sufficiently large and $\theta_3 > 0$ sufficiently small, the spatial eigenvalues $\nu_{i,+}(\lambda,\epsilon)$ of~\eqref{evprobplus} admit a $\lambda$- and $\epsilon$-uniform spectral gap at $\Re(\nu) = 1 + \Re\sqrt{2+2\ri \alpha}$ and $\Re(\nu) = 1 + \Re\sqrt{2+2\ri \alpha} + \kappa\sqrt{1+\alpha^2}$, whereas the spatial eigenvalues $\nu_{i,-}(\lambda,\epsilon)$ of~\eqref{evprobmin} admit a $\lambda$- and $\epsilon$-uniform spectral gap at $\Re(\nu) = 0$ and $\Re(\nu) = -\kappa\sqrt{1+\alpha^2}$ for $\lambda \in R_3(\theta_3,\Theta_2)$.

\paragraph*{Analysis of spatial eigenvalues at $+\infty$.} Recall $\kappa > 0$ was fixed such that $\kappa\sqrt{1+\alpha^2} \in (0,\frac{1}{4})$. Therefore, $\Re \sqrt{(1 \pm \ri \alpha) (2 + \lambda)} \leq 1 + \kappa\sqrt{1+\alpha^2}$ implies $\Re(\lambda) \leq -\frac{7}{16}$. Hence, provided $\Theta_1,\theta_2 > 0$ sufficiently small, the spatial eigenvalues $\nu_{i,+}(\lambda,\epsilon)$ of~\eqref{evprobplus} have $\lambda$- and $\epsilon$-uniform spectral gaps at $\Re(\nu) = 1 + \Re \sqrt{2+2\ri \alpha}$ and at $\Re(\nu) =  1 + \Re \sqrt{2+2\ri \alpha} + \kappa\sqrt{1+\alpha^2}$ for $\lambda \in R_1(\Theta_1) \cup R_2(\theta_2,\Theta_1,\Theta_2)$.

\paragraph*{Analysis of spatial eigenvalues at $-\infty$.} At $\lambda = 0$ and $\epsilon = (0,0)$ the spatial eigenvalues of~\eqref{evprobmin} are $0,-2,-1\pm\sqrt{3}$. Hence, taking $\Theta_1 > 0$ sufficiently small, we may conclude that the spatial eigenvalues $\nu_{i,-}(\lambda,\epsilon)$ of~\eqref{evprobmin} admit a $\lambda$- and $\epsilon$-uniform spectral gap at $\Re(\nu) = -\kappa\sqrt{1+\alpha^2} \in (-\tfrac{1}{4},0)$ for $\lambda \in R_1(\Theta_1)$.

Next, we solve $\Re(\nu_{i,-}(\lambda,\epsilon)) = 0$ for $\lambda \in R_1(\Theta_1)$. Let $H \colon \C \times \R \times U$ be given by
\begin{align*}
 H(\lambda,l,\epsilon) = \det\begin{pmatrix} \displaystyle \vartheta(\ri l,\epsilon) -\frac{\lambda (1 - \ri \alpha)}{m(\epsilon)^2}  & - (1 + \ri \dell) \left(1- \hat{k}_\tf(\epsilon)^2\right)\\
-(1 - \ri \dell) \left(1- \hat{k}_\tf(\epsilon)^2\right) & \displaystyle \overline{\vartheta(-\ri l,\epsilon)}-\frac{\lambda (1 + \ri \alpha)}{m(\epsilon)^2}\end{pmatrix},
\end{align*}
where $\vartheta$ is defined as~\eqref{defvartheta}, and $U$ is a neighborhood of $(0,0)$ in $\R^2$. It holds $\partial_{\lambda} H(0,0,0) = -4\ri \neq 0$ and $H(0,0,\epsilon) = 0$ for all $\epsilon \in U$. Hence, by the implicit function theorem, we find a neighborhood $V \subset \R^3$ of $(0,0,0)$ and a locally unique curve $\lambda_* \colon V \to \C$ satisfying $H(\lambda_*(l,\epsilon),l,\epsilon) = 0$ and
\begin{align*}
\lambda_*(0,\epsilon) = 0, \qquad \partial_l \lambda_*(0,0,0) = 2\ri, \qquad \partial_{ll} \lambda_*(0,0,0) = -2 + 4 \alpha^2
\end{align*}
for all $(l,\epsilon) \in V$. Since by assumption it holds $\alpha^2 < 1/2$ the curve $\lambda_*(\cdot,\epsilon)$ intersects, upon making $V$ smaller if necessary, the closed right-half plane only in the origin as a parabolic curve. Hence, the set of $\lambda \in R_1(\Theta_1)$ such that there are spatial eigenvalues $\nu_{i,-}(\lambda,\epsilon)$ of~\eqref{evprobmin} having $\Re(\nu) = 0$ does not intersect the closed right-half plane, except at the origin as a parabolic curve.

Finally, we solve $\Re(\nu_{i,-}(\lambda,\epsilon)) = 0$ for $\lambda \in R_2(\theta_2,\Theta_1,\Theta_2)$. Thus, equating $F(\lambda,l,0) = 0$ yields two solution continuous curves $\lambda_\pm \colon \R \to \C$ given by
\begin{align*} \lambda_\pm(l) = \frac{-l^2+2 \ri l-1\pm\sqrt{-\alpha^2 l^4+4 \ri \alpha^2 l^3+2 \alpha^2 l^2+4 \ri \alpha^2 l+1}}{\alpha^2+1}.
\end{align*}
Clearly, for $|l|$ large, the curves lie in the left-half plane. Hence, if for some $\lambda \in R_2(\theta_2,\Theta_1,\Theta_2)$ there is a spatial eigenvalue $\nu_{i,-}(\lambda,0)$ of nonnegative real part, then one of the curves $\lambda_\pm(l)$ must intersect the imaginary axis, i.e.~there exists $l_*,\varrho \in \R$ such that $\lambda_{\pm}(l_*) = \ri \varrho$ . Consequently, it holds $H(\ri \varrho,l_*,0) = 0$. Equating real and imaginary parts of $H(\ri \varrho,l_*,0) = 0$ yields
\begin{align*}
(\varrho,l_*) = (0,0) \quad \text{ or } \quad (\varrho,l_*) = \left(\pm 2\sqrt{-2 + 4 \alpha^2}, \pm\sqrt{-2 + 4 \alpha^2}\right).
\end{align*}
Since $l_*$ must be real and $\alpha^2 < \frac{1}{2}$, it must hold $(\varrho,l_*) = (0,0)$. So, for $\theta_2 > 0$ sufficiently small, the curves $\lambda_\pm(l)$ do not intersect the region $R_2(\theta_2,\Theta_1,\Theta_2)$. Thus, it must hold $\Re(\nu_{1/2,-}(\lambda,0)) \leq -1 < -\kappa\sqrt{1+\alpha^2} < 0 < \Re(\nu_{3/4,-}(\lambda,0))$ for all $\lambda \in R_2(\theta_2,\Theta_1,\Theta_2)$. So, by compactness of the region $R_2$, we conclude that the spatial eigenvalues $\nu_{i,-}(\lambda,\epsilon)$ of~\eqref{evprobmin} admit a $\lambda$- and $\epsilon$-uniform spectral gap at $\Re(\nu) = -\kappa\sqrt{1+\alpha^2}$ and $\Re(\nu) = 0$ for $\lambda \in R_2(\theta_2,\Theta_1,\Theta_2)$.

We also note that since the eigenvalues $\nu_{1/2,j}(\lambda,\epsilon)$ and $\nu_{3/4,j}(\lambda,\epsilon)$ have uniformly separated real parts in $R_1\cup R_2\cup R_3$ for $j = \pm$, the absolute spectrum of the operator $\hat{\El}_\tf$ is contained in the open left-half plane bounded away from $\ri\R$.

\bibliographystyle{plain}
\bibliography{cgl_stab}

\end{document}